\documentclass[11pt,onecolumn]{article}%
\usepackage{amsmath}
\usepackage{amsfonts}
\usepackage{amssymb}
\usepackage{graphicx}
\usepackage{fullpage}%
\providecommand{\U}[1]{\protect\rule{.1in}{.1in}}
\newtheorem{theorem}{Theorem}

\newtheorem{conjecture}[theorem]{Conjecture}
\newtheorem{corollary}[theorem]{Corollary}

\newtheorem{definition}[theorem]{Definition}

\newtheorem{lemma}[theorem]{Lemma}

\newtheorem{problem}[theorem]{Problem}
\newtheorem{proposition}[theorem]{Proposition}

\newenvironment{proof}[1][Proof]{\noindent\textbf{#1.} }{\ \rule{0.5em}{0.5em}}
\begin{document}

\title{The Computational Complexity of Linear Optics}
\author{Scott Aaronson\thanks{MIT. \ Email: aaronson@csail.mit.edu. \ \ This material
is based upon work supported by the National Science Foundation under Grant
No. 0844626. \ Also supported by a DARPA YFA grant and a Sloan Fellowship.}
\and Alex Arkhipov\thanks{MIT. \ Email: arkhipov@mit.edu. \ Supported by an Akamai
Foundation Fellowship.}}
\date{}
\maketitle

\begin{abstract}
We give new evidence that quantum computers---moreover, rudimentary quantum
computers built entirely out of linear-optical elements---cannot be
efficiently simulated by classical computers. \ In particular, we define a
model of computation in which identical photons are generated, sent through a
linear-optical network, then nonadaptively\ measured\ to count the number of
photons in each mode. \ This model is not known or believed to be universal
for quantum computation, and indeed, we discuss the prospects for realizing
the model using current technology. \ On the other hand, we prove that the
model is able to solve sampling problems and search problems that are
classically intractable under plausible assumptions.

Our first result says that, if there exists a polynomial-time classical
algorithm that samples from the same probability distribution as a
linear-optical network, then $\mathsf{P}^{\mathsf{\#P}}=\mathsf{BPP}%
^{\mathsf{NP}}$, and hence the polynomial hierarchy collapses to the third
level. \ Unfortunately, this result assumes an extremely accurate simulation.

Our main result suggests that even an approximate or noisy classical
simulation would already imply a collapse of the polynomial hierarchy. \ For
this, we need two unproven conjectures: the \textit{Permanent-of-Gaussians
Conjecture}, which says that it is $\mathsf{\#P}$-hard to approximate the
permanent of a matrix $A$\ of independent $\mathcal{N}\left(  0,1\right)
$\ Gaussian entries, with high probability over $A$; and the \textit{Permanent
Anti-Concentration Conjecture}, which says that $\left\vert
\operatorname*{Per}\left(  A\right)  \right\vert \geq\sqrt{n!}%
/\operatorname*{poly}\left(  n\right)  $ with high probability over $A$. We
present evidence for these conjectures,\ both of which seem interesting even
apart from our application.

This paper does not assume knowledge of quantum optics.\ Indeed, part of its
goal is to develop the beautiful theory of noninteracting bosons underlying
our model, and its connection to the permanent function, in a self-contained
way accessible to theoretical computer scientists.

\end{abstract}
\tableofcontents

\section{Introduction\label{INTRO}}

The Extended Church-Turing Thesis says that all computational problems that
are efficiently solvable by realistic physical devices, are efficiently
solvable by a probabilistic Turing machine. \ Ever since Shor's algorithm
\cite{shor}, we have known that this thesis is in severe tension with the
currently-accepted laws of physics. \ One way to state Shor's discovery is this:

\begin{quotation}
\noindent\textit{Predicting the results of a given quantum-mechanical
experiment, to finite accuracy, cannot be done by a classical computer in
probabilistic polynomial time, unless factoring integers can as well.}
\end{quotation}

As the above formulation makes clear, Shor's result is not merely about some
hypothetical future in which large-scale quantum computers are built.\ \ It is
also a hardness result for a practical problem. \ For \textit{simulating
quantum systems} is one of the central computational problems of modern
science, with applications from drug design to nanofabrication to nuclear
physics. \ It has long been a major application of high-performance computing,
and Nobel Prizes have been awarded for methods (such as the Density Functional
Theory)\ to handle special cases. \ What Shor's result shows is that, if we
had an efficient, \textit{general-purpose} solution to the quantum simulation
problem, then we could also break widely-used cryptosystems such as RSA.

However, as evidence against the Extended Church-Turing Thesis, Shor's
algorithm has two significant drawbacks. \ The first is that, even by the
conjecture-happy standards of complexity theory, it is no means settled that
factoring is classically hard. \ Yes, we believe this enough to base modern
cryptography on it---but as far as anyone knows, factoring could be in
$\mathsf{BPP}$\ without causing any collapse of complexity classes or other
disastrous theoretical consequences. \ Also, of course, there \textit{are}
subexponential-time factoring algorithms (such as the number field sieve), and
few would express confidence that they cannot be further improved. \ And thus,
ever since Bernstein and Vazirani\ \cite{bv} defined the class $\mathsf{BQP}%
$\ of quantumly feasible problems, it has been a dream of quantum computing
theory to show (for example) that, if $\mathsf{BPP}=\mathsf{BQP}$, then the
polynomial hierarchy would collapse, or some other \textquotedblleft generic,
foundational\textquotedblright\ assumption of theoretical computer science
would fail. \ In this paper, we do not \textit{quite} achieve that dream, but
we come closer than one might have thought possible.

The second, even more obvious drawback of Shor's algorithm is that
implementing it scalably is well beyond current technology. \ To run Shor's
algorithm, one needs to be able to perform arithmetic (including modular
exponentiation) on a coherent superposition of integers encoded in binary.
\ This does not seem much easier than building a \textit{universal} quantum
computer.\footnote{One caveat is a result of Cleve and Watrous \cite{cw}, that
Shor's algorithm can be implemented using \textit{log-depth} quantum circuits
(that is, in $\mathsf{BPP}^{\mathsf{BQNC}}$). \ But even here, fault-tolerance
will presumably be needed, among other reasons because one still has
polynomial \textit{latency} (the log-depth circuit does not obey spatial
locality constraints).} \ In particular, it appears one first needs to solve
the problem of \textit{fault-tolerant quantum computation}, which is known to
be possible in principle if quantum mechanics is valid \cite{ab,klz}, but
might require decoherence rates that are several orders of magnitude below
what is achievable today.

Thus, one might suspect that proving a quantum system's computational power by
having it factor integers encoded in binary is a bit like proving a dolphin's
intelligence by teaching it to solve arithmetic problems. \ Yes, with heroic
effort, we can probably do this, and perhaps we have good reasons to.
\ However, if we just watched the dolphin in its natural habitat, then we
might see it display equal intelligence with no special training at all.

Following this analogy, we can ask: are there more \textquotedblleft
natural\textquotedblright\ quantum systems that \textit{already} provide
evidence against the Extended Church-Turing Thesis? \ Indeed, there are
countless quantum systems accessible to current experiments---including
high-temperature superconductors, Bose-Einstein condensates, and even just
large nuclei and molecules---that seem intractable to simulate on a classical
computer, and largely for the reason a theoretical computer scientist would
expect: namely, that the dimension of a quantum state increases exponentially
with the number of particles. \ The difficulty is that it is not clear how to
interpret these systems as \textit{solving computational problems}. \ For
example, what is the \textquotedblleft input\textquotedblright\ to a
Bose-Einstein condensate? \ In other words, while these systems might be hard
to simulate, we would not know how to justify that conclusion using the one
formal tool (reductions) that is currently available to us.

So perhaps the real question is this:\ do there exist quantum systems that are
\textquotedblleft intermediate\textquotedblright\ between Shor's algorithm and
a Bose-Einstein condensate---in the sense that

\begin{enumerate}
\item[(1)] they are significantly closer to experimental reality than
universal quantum computers, but

\item[(2)] they can be proved, under plausible complexity assumptions (the
more \textquotedblleft generic\textquotedblright\ the better), to be
intractable to simulate classically?
\end{enumerate}

In this paper, we will argue that the answer is yes.

\subsection{Our Model\label{INTMODEL}}

We define and study a formal model of \textit{quantum computation with
noninteracting bosons}. \ Physically, our model could be implemented using a
\textit{linear-optical network}, in which $n$ identical photons pass through a
collection of simple optical elements (beamsplitters and phaseshifters), and
are then measured to determine their locations. \ In Section \ref{MODEL}, we
give a detailed exposition of the model that does not presuppose any physics
knowledge. \ For now, though, it is helpful to imagine a rudimentary
\textquotedblleft computer\textquotedblright\ consisting of $n$ identical
balls, which are dropped one by one into a vertical lattice of pegs, each of
which randomly scatters each incoming ball onto one of two other pegs. \ Such
an arrangement---called \textit{Galton's board}---is sometimes used in science
museums to illustrate the binomial distribution (see Figure \ref{galtonfig}).
\ The \textquotedblleft input\textquotedblright\ to the computer is the exact
arrangement $A$\ of the pegs, while the \textquotedblleft
output\textquotedblright\ is the number of balls that have landed at each
location on the bottom (or rather, a sample from the joint distribution
$\mathcal{D}_{A}$ over these numbers). \ There is no interaction between pairs
of balls.%
%TCIMACRO{\FRAME{ftbpFU}{1.4641in}{1.6596in}{0pt}{\Qcb{Galton's board, a simple
%\textquotedblleft computer\textquotedblright\ to output samples from the
%binomial distribution. \ From MathWorld,
%http://mathworld.wolfram.com/GaltonBoard.html}}{\Qlb{galtonfig}}%
%{galton.eps}{\special{ language "Scientific Word";  type "GRAPHIC";
%width 1.4641in;  height 1.6596in;  depth 0pt;  original-width 10.6017in;
%original-height 8.0782in;  cropleft "0";  croptop "1";  cropright "1";
%cropbottom "0";  filename 'galton.eps';file-properties "XNPEU";}} }%
%BeginExpansion
\begin{figure}[ptb]%
\centering
\includegraphics[
height=1.6596in,
width=1.4641in
]%
{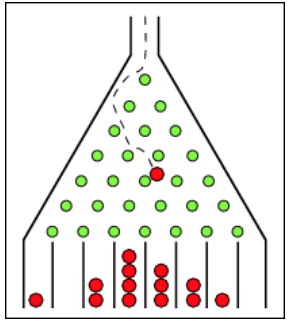}%
\caption{Galton's board, a simple \textquotedblleft computer\textquotedblright%
\ to output samples from the binomial distribution. \ From MathWorld,
http://mathworld.wolfram.com/GaltonBoard.html}%
\label{galtonfig}%
\end{figure}
%EndExpansion

Our model is essentially the same as that shown in Figure \ref{galtonfig},
except that instead of identical balls, we use \textit{identical bosons}
governed by quantum statistics. \ Other minor differences are that, in our
model, the \textquotedblleft balls\textquotedblright\ are each dropped from
different starting locations, rather than a single location; and the
\textquotedblleft pegs,\textquotedblright\ rather than being arranged in a
regular lattice, can be arranged arbitrarily to encode a problem of interest.

Mathematically, the key point about our model is that, to find the probability
of any particular output of the computer, one needs to calculate the
\textit{permanent} of an $n\times n$ matrix. \ This can be seen even in the
classical case: suppose there are $n$ balls and $n$ final locations, and ball
$i$ has probability $a_{ij}$\ of landing at location $j$. \ Then the
probability of one ball landing in each of the $n$ locations is%
\[
\operatorname*{Per}\left(  A\right)  =\sum_{\sigma\in S_{n}}\prod_{i=1}%
^{n}a_{i\sigma\left(  i\right)  },
\]
where $A=\left(  a_{ij}\right)  _{i,j\in\left[  n\right]  }$. \ Of course, in
the classical case, the $a_{ij}$'s are nonnegative real numbers---which means
that we can approximate $\operatorname*{Per}\left(  A\right)  $\ in
probabilistic polynomial time, by using the celebrated algorithm of Jerrum,
Sinclair, and Vigoda \cite{jsv}. \ In the quantum case, by contrast, the
$a_{ij}$'s are complex numbers. \ And it is not hard to show that, given a
general matrix $A\in\mathbb{C}^{n\times n}$, even \textit{approximating}
$\operatorname*{Per}\left(  A\right)  $\ to within a constant factor is
$\mathsf{\#P}$-complete. \ This fundamental difference between nonnegative and
complex matrices is the starting point for everything we do in this paper.

It is not hard to show that a boson computer can be simulated by a
\textquotedblleft standard\textquotedblright\ quantum computer (that is, in
$\mathsf{BQP}$). \ But the other direction seems extremely unlikely---indeed,
it even seems unlikely that a boson computer can do universal
\textit{classical} computation! \ Nor do we have any evidence that a boson
computer could factor integers, or solve any other decision or promise problem
not in $\mathsf{BPP}$. \ However, if we broaden the notion of a computational
problem to encompass \textit{sampling} and \textit{search} problems, then the
situation is quite different.

\subsection{Our Results\label{RESULTS}}

In this paper we study \textsc{BosonSampling}: the problem of sampling, either
exactly or approximately, from the output distribution of a boson computer.
\ Our goal is to give evidence that this problem is hard for a classical
computer. \ Our main results fall into three categories:

\begin{enumerate}
\item[(1)] Hardness results for exact \textsc{BosonSampling}, which give an
essentially complete picture of that case.

\item[(2)] Hardness results for \textit{approximate} \textsc{BosonSampling}%
,\ which depend on plausible conjectures about the permanents of
i.i.d.\ Gaussian matrices.

\item[(3)] A program aimed at understanding and proving the conjectures.
\end{enumerate}

We now discuss these in turn.

\subsubsection{The Exact Case\label{EXACTCASE}}

Our first (easy) result, proved in Section \ref{WARMUPSEC}, says the following.

\begin{theorem}
\label{warmup}\noindent The exact \textsc{BosonSampling}\ problem is not
efficiently solvable by a classical computer, unless\textit{ }$\mathsf{P}%
^{\mathsf{\#P}}=\mathsf{BPP}^{\mathsf{NP}}$\textit{\ }and the polynomial
hierarchy collapses to the third level.

More generally, let $\mathcal{O}$ be any oracle that \textquotedblleft
simulates boson computers,\textquotedblright\ in the sense that $\mathcal{O}%
$\ takes as input a random string $r$ (which $\mathcal{O}$\ uses as its only
source of randomness) and a description of a boson computer $A$, and returns a
sample $\mathcal{O}_{A}\left(  r\right)  $ from the probability distribution
$\mathcal{D}_{A}$\ over possible outputs of $A$. \ Then $\mathsf{P}%
^{\mathsf{\#P}}\subseteq\mathsf{BPP}^{\mathsf{NP}^{\mathcal{O}}}$.
\end{theorem}

In particular, even if the exact \textsc{BosonSampling} problem were solvable
by a classical computer \textit{with an oracle for a }$\mathsf{PH}%
$\textit{\ problem}, Theorem \ref{warmup}\ would still imply that
$\mathsf{P}^{\mathsf{\#P}}\subseteq\mathsf{BPP}^{\mathsf{PH}}$---and therefore
that the polynomial hierarchy would collapse, by Toda's Theorem \cite{toda}.
\ This provides evidence that quantum computers have capabilities outside the
entire polynomial hierarchy, complementing the recent evidence of Aaronson
\cite{aar:ph} and Fefferman and Umans \cite{feffumans}.

At least for a computer scientist, it is tempting to interpret Theorem
\ref{warmup} as saying that \textquotedblleft the exact \textsc{BosonSampling}
problem is $\mathsf{\#P}$-hard under $\mathsf{BPP}^{\mathsf{NP}}%
$-reductions.\textquotedblright\ \ Notice that this would have a shocking
implication: that quantum computers (indeed, quantum computers of a
particularly simple kind) could efficiently solve a $\mathsf{\#P}$-hard\ problem!

There is a catch, though, arising from the fact that \textsc{BosonSampling}%
\ is a sampling problem rather than a decision problem. \ Namely, if
$\mathcal{O}$\ is an oracle for sampling from the boson distribution
$\mathcal{D}_{A}$, then Theorem \ref{warmup}\ shows that $\mathsf{P}%
^{\mathsf{\#P}}\subseteq\mathsf{BPP}^{\mathsf{NP}^{\mathcal{O}}}%
$---\textit{but only if the} $\mathsf{BPP}^{\mathsf{NP}}$ \textit{machine gets
to fix the random bits\ used by} $\mathcal{O}$. \ This condition is clearly
met if $\mathcal{O}$ is a classical randomized algorithm, since we can always
interpret a randomized algorithm as just a deterministic algorithm that takes
a random string $r$\ as part of its input. \ On the other hand, the condition
would \textit{not} be met if we implemented $\mathcal{O}$\ (for example) using
the boson computer itself. \ In other words, our \textquotedblleft
reduction\textquotedblright\ from $\mathsf{\#P}$-complete problems to
\textsc{BosonSampling}\ makes essential use of the hypothesis that we have a
\textit{classical} \textsc{BosonSampling}\ algorithm.

We will give two proofs of Theorem \ref{warmup}. \ In the first proof, we
consider the probability $p$ of some particular basis state when a boson
computer is measured. \ We then prove two facts:

\begin{enumerate}
\item[(1)] Even \textit{approximating} $p$ to within a multiplicative constant
is a $\mathsf{\#P}$-hard problem.

\item[(2)] \textit{If} we had a polynomial-time classical algorithm for exact
\textsc{BosonSampling},\ \textit{then} we could approximate $p$ to within a
multiplicative constant in the class $\mathsf{BPP}^{\mathsf{NP}}$, by using a
standard technique called \textit{universal hashing}.
\end{enumerate}

Combining facts (1) and (2), we find that, if the classical
\textsc{BosonSampling} algorithm exists, then $\mathsf{P}^{\mathsf{\#P}%
}=\mathsf{BPP}^{\mathsf{NP}}$, and therefore the polynomial hierarchy collapses.

Our second proof was inspired by independent work of Bremner, Jozsa, and
Shepherd \cite{bjs}. \ In this proof, we start with a result of Knill,
Laflamme, and Milburn \cite{klm}, which says that linear optics with
\textit{adaptive measurements} is universal for $\mathsf{BQP}$. \ A
straightforward modification of their construction shows that linear optics
with \textit{postselected} measurements is universal for $\mathsf{PostBQP}$
(that is, quantum polynomial-time with postselection on possibly
exponentially-unlikely measurement outcomes). \ Furthermore, Aaronson
\cite{aar:pp} showed that $\mathsf{PostBQP}=\mathsf{PP}$. \ On the other hand,
if a classical \textsc{BosonSampling} algorithm existed, then we will show
that we could simulate postselected linear optics in $\mathsf{PostBPP}$\ (that
is, \textit{classical} polynomial-time with postselection, also called
$\mathsf{BPP}_{\mathsf{path}}$). \ We would therefore get%
\[
\mathsf{BPP}_{\mathsf{path}}=\mathsf{PostBPP}=\mathsf{PostBQP}=\mathsf{PP},
\]
which is known to imply a collapse of the polynomial hierarchy.

Despite the simplicity of the above arguments, there is something conceptually
striking about them. \ Namely, starting from an algorithm to \textit{simulate
quantum mechanics}, we get an algorithm\footnote{Admittedly, a $\mathsf{BPP}%
^{\mathsf{NP}}$\ algorithm.} to \textit{solve }$\mathsf{\#P}$\textit{-complete
problems}---even though solving $\mathsf{\#P}$-complete problems is believed
to be well beyond what a quantum computer itself can do! \ Of course, one
price we pay is that we need to talk about sampling problems rather than
decision problems. \ If we do so, though, then we get to base our belief in
the power of quantum computers on $\mathsf{P}^{\mathsf{\#P}}\neq
\mathsf{BPP}^{\mathsf{NP}}$, which is a much more \textquotedblleft
generic\textquotedblright\ (many would say safer) assumption than
\textsc{Factoring}$\notin\mathsf{BPP}$.

As we see it, the central drawback of Theorem \ref{warmup} is that it only
addresses the consequences of a fast classical algorithm that \textit{exactly}
samples the boson distribution $\mathcal{D}_{A}$. \ One can relax this
condition slightly: if the oracle $\mathcal{O}$ samples from some distribution
$\mathcal{D}_{A}^{\prime}$ whose probabilities are all
\textit{multiplicatively close} to those in $\mathcal{D}_{A}$, then we still
get the conclusion that $\mathsf{P}^{\mathsf{\#P}}\subseteq\mathsf{BPP}%
^{\mathsf{NP}^{\mathcal{O}}}$. \ In our view, though, multiplicative closeness
is already too strong an assumption. \ At a minimum, given as input an error
parameter $\varepsilon>0$, we ought to let our simulation algorithm sample
from some distribution $\mathcal{D}_{A}^{\prime}$\ such that $\left\Vert
\mathcal{D}_{A}^{\prime}-\mathcal{D}_{A}\right\Vert \leq\varepsilon$ (where
$\left\Vert \cdot\right\Vert $\ represents total variation distance), using
$\operatorname*{poly}\left(  n,1/\varepsilon\right)  $\ time.

Why are we so worried about this issue? \ One obvious reason is that noise,
decoherence, photon losses, etc.\ will be unavoidable features in any real
implementation of a boson computer. \ As a result, not even the boson computer
\textit{itself} can sample exactly from the distribution $\mathcal{D}_{A}$!
\ So it seems arbitrary and unfair to require this of a classical simulation algorithm.

A second, more technical reason to allow error is that later, we would like to
show that a boson computer can solve classically-intractable \textit{search}
problems, in addition to sampling problems. \ However, while Aaronson
\cite{aar:samp} proved an extremely general connection between search problems
and sampling problems, that connection only works for \textit{approximate}
sampling, not exact sampling.

The third, most fundamental reason to allow error is that the connection we
are claiming, between quantum computing and $\mathsf{\#P}$-complete problems,
is so counterintuitive. \ One's first urge is to dismiss this connection as an
artifact of poor modeling choices. \ So the burden is on us to demonstrate the
connection's robustness.

Unfortunately, the proof of Theorem \ref{warmup} fails completely when we
consider approximate sampling algorithms. \ The reason is that the proof
hinges on the $\mathsf{\#P}$-completeness of estimating a single,
exponentially-small probability $p$. \ Thus, if a sampler \textquotedblleft
knew\textquotedblright\ which $p$ we wanted to estimate, then it could
adversarially choose to corrupt that $p$. \ It would still be a perfectly good
approximate sampler, but would no longer reveal the solution to
the\ $\mathsf{\#P}$-complete\ instance that we were trying to solve.

\subsubsection{The Approximate Case\label{APPROXCASE}}

To get around the above problem, we need to argue that a boson computer can
sample from a distribution $\mathcal{D}$\ that \textquotedblleft
robustly\textquotedblright\ encodes the solution to a $\mathsf{\#P}$-complete
problem. \ This means intuitively that, even if a sampler was badly wrong
about any $\varepsilon$ fraction of the probabilities in $\mathcal{D}$, the
remaining $1-\varepsilon$\ fraction would still allow the $\mathsf{\#P}%
$-complete problem\ to be solved.

It is well-known that there exist $\mathsf{\#P}$-complete\ problems with
\textit{worst-case/average-case equivalence}, and that one example of such a
problem is the permanent, at least over finite fields. \ This is a reason for
optimism that the sort of robust encoding we need might be possible. \ Indeed,
it was precisely our desire to encode the \textquotedblleft robustly
$\mathsf{\#P}$-complete\textquotedblright\ permanent function into a quantum
computer's amplitudes\ that led us to study the noninteracting-boson model in
the first place. \ That this model also has great experimental interest simply
came as a bonus.

In this paper, \textit{our main technical contribution is to prove a
connection between the ability of classical computers to solve the approximate
\textsc{BosonSampling} problem and their ability to approximate the
permanent.} \ This connection \textquotedblleft almost\textquotedblright%
\ shows that even approximate classical simulation of boson computers would
imply a collapse of the polynomial hierarchy. \ There is still a gap in the
argument, but it has nothing to do with quantum computing. \ The gap is simply
that it is not known, at present, how to extend the worst-case/average-case
equivalence of the permanent from finite fields to suitably analogous
statements over the reals or complex numbers. \ We will show that, \textit{if}
this gap can be bridged, then there exist search problems and approximate
sampling problems that are solvable in polynomial time by a boson computer,
but not by a $\mathsf{BPP}$\ machine unless $\mathsf{P}^{\mathsf{\#P}%
}=\mathsf{BPP}^{\mathsf{NP}}$.

More concretely, consider the following problem, where the \textsc{GPE} stands
for \textsc{Gaussian Permanent Estimation}:

\begin{problem}
[$\left\vert \text{\textsc{GPE}}\right\vert _{\pm}^{2}$]\label{gpe2+}Given as
input a matrix $X\thicksim\mathcal{N}\left(  0,1\right)  _{\mathbb{C}%
}^{n\times n}$ of i.i.d.\ Gaussians, together with error bounds $\varepsilon
,\delta>0$, estimate $\left\vert \operatorname*{Per}\left(  X\right)
\right\vert ^{2}$ to within additive error $\pm\varepsilon\cdot n!$, with
probability at least $1-\delta$\ over $X$, in $\operatorname*{poly}\left(
n,1/\varepsilon,1/\delta\right)  $\ time.
\end{problem}

Then our main result is the following.

\begin{theorem}
[Main Result]\label{mainresult}Let $\mathcal{D}_{A}$\ be the probability
distribution sampled by a boson computer $A$. \ Suppose there exists a
classical algorithm $C$ that takes as input a description of $A$ as well as an
error bound $\varepsilon$, and that samples from a probability distribution
$\mathcal{D}_{A}^{\prime}$\ such that $\left\Vert \mathcal{D}_{A}^{\prime
}-\mathcal{D}_{A}\right\Vert \leq\varepsilon$ in $\operatorname*{poly}\left(
\left\vert A\right\vert ,1/\varepsilon\right)  $\ time. \ Then the $\left\vert
\text{\textsc{GPE}}\right\vert _{\pm}^{2}$\ problem is solvable in
$\mathsf{BPP}^{\mathsf{NP}}$. \ Indeed, if we treat $C$ as a black box, then
$\left\vert \text{\textsc{GPE}}\right\vert _{\pm}^{2}\in\mathsf{BPP}%
^{\mathsf{NP}^{C}}$.
\end{theorem}

Theorem \ref{mainresult}\ is proved in Section \ref{MAIN}. \ The key idea of
the proof is to \textquotedblleft smuggle\textquotedblright\ the $\left\vert
\text{\textsc{GPE}}\right\vert _{\pm}^{2}$\ instance $X$ that we want to solve
into the probability of a \textit{random} output of a boson computer $A$.
\ That way, even if the classical sampling algorithm $C$\ is adversarial, it
will not know which of the exponentially many probabilities in $\mathcal{D}%
_{A}$\ is the one we care about. \ And therefore, provided $C$ correctly
approximates \textit{most} probabilities in $\mathcal{D}_{A}$, with high
probability it will correctly approximate \textquotedblleft
our\textquotedblright\ probability, and will therefore allow $\left\vert
\operatorname*{Per}\left(  X\right)  \right\vert ^{2}$\ to be estimated in
$\mathsf{BPP}^{\mathsf{NP}}$.

Besides this conceptual step, the proof of Theorem \ref{mainresult}\ also
contains a technical component that might find other applications in quantum
information. \ This is that, if we choose an $m\times m$\ unitary matrix
$U$\ randomly according to the Haar measure, then \textit{any }$n\times
n$\textit{\ submatrix of }$U$\textit{ will be close in variation distance to a
matrix of i.i.d.\ Gaussians, provided that }$n\leq m^{1/6}$\textit{.}
\ Indeed, the fact that i.i.d.\ Gaussian matrices naturally arise as
submatrices of Haar unitaries is the reason why we will be so interested in
Gaussian matrices in this paper, rather than Bernoulli matrices or other
well-studied ensembles.

In our view, Theorem \ref{mainresult} already shows that fast, approximate
classical simulation of boson computers would have a surprising complexity
consequence. \ For notice that, if $X\thicksim\mathcal{N}\left(  0,1\right)
_{\mathbb{C}}^{n\times n}$\ is a complex Gaussian matrix, then
$\operatorname*{Per}\left(  X\right)  $\ is a sum of $n!$\ complex terms,
almost all of which usually cancel each other out, leaving only a tiny residue
exponentially smaller than $n!$. \ \textit{A priori}, there seems to be little
reason to expect that residue to be approximable in the polynomial hierarchy,
let alone in $\mathsf{BPP}^{\mathsf{NP}}$.

\subsubsection{The Permanents of Gaussian Matrices\label{PGRESULTS}}

One could go further, though, and speculate that estimating
$\operatorname*{Per}\left(  X\right)  $ for Gaussian $X$\ is actually
$\mathsf{\#P}$-hard. \ We call this the \textit{Permanent-of-Gaussians
Conjecture}, or PGC.\footnote{The name is a pun on the well-known Unique Games
Conjecture (UGC) \cite{khot}, which says that a certain approximation problem
that \textquotedblleft ought\textquotedblright\ to be $\mathsf{NP}$-hard
really \textit{is} $\mathsf{NP}$-hard.} \ We prefer to state the PGC using a
more \textquotedblleft natural\textquotedblright\ variant of the
\textsc{Gaussian Permanent Estimation}\ problem than $\left\vert
\text{\textsc{GPE}}\right\vert _{\pm}^{2}$. \ The more natural variant talks
about estimating $\operatorname*{Per}\left(  X\right)  $\ itself, rather than
$\left\vert \operatorname*{Per}\left(  X\right)  \right\vert ^{2}$, and also
asks for a \textit{multiplicative} rather than additive approximation.

\begin{problem}
[\textsc{GPE}$_{\times}$]\label{gpe*}Given as input a matrix $X\thicksim
\mathcal{N}\left(  0,1\right)  _{\mathbb{C}}^{n\times n}$ of
i.i.d.\ Gaussians, together with error bounds $\varepsilon,\delta>0$, estimate
$\operatorname*{Per}\left(  X\right)  $ to within error $\pm\varepsilon
\cdot\left\vert \operatorname*{Per}\left(  X\right)  \right\vert $, with
probability at least $1-\delta$\ over $X$, in $\operatorname*{poly}\left(
n,1/\varepsilon,1/\delta\right)  $\ time.
\end{problem}

Then the main complexity-theoretic challenge we offer is to prove or disprove
the following:

\begin{conjecture}
[Permanent-of-Gaussians Conjecture or PGC]\label{pgc}\textsc{GPE}$_{\times}$
is $\mathsf{\#P}$-hard. \ In other words, if $\mathcal{O}$\ is any oracle that
solves \textsc{GPE}$_{\times}$,\ then $\mathsf{P}^{\mathsf{\#P}}%
\subseteq\mathsf{BPP}^{\mathcal{O}}$.
\end{conjecture}

Of course, a question arises as to whether one can bridge the gap between
the\ $\left\vert \text{\textsc{GPE}}\right\vert _{\pm}^{2}$ problem that
appears in Theorem \ref{mainresult}, and\ the more \textquotedblleft
natural\textquotedblright\ \textsc{GPE}$_{\times}$\ problem used in Conjecture
\ref{pgc}. \ We are able to do so assuming \textit{another} conjecture, this
one an extremely plausible anti-concentration bound for the permanents of
Gaussian random matrices.

\begin{conjecture}
[Permanent Anti-Concentration Conjecture]\label{pacc}There exists a polynomial
$p$ such that for all $n$ and $\delta>0$,%
\[
\Pr_{X\sim\mathcal{N}\left(  0,1\right)  _{\mathbb{C}}^{n\times n}}\left[
\left\vert \operatorname*{Per}\left(  X\right)  \right\vert <\frac{\sqrt{n!}%
}{p\left(  n,1/\delta\right)  }\right]  <\delta.
\]

\end{conjecture}

In Section \ref{GPE2GPE}, we give a complicated reduction that proves the following:

\begin{theorem}
\label{decompthm}Suppose the Permanent Anti-Concentration Conjecture holds.
\ Then $\left\vert \text{\textsc{GPE}}\right\vert _{\pm}^{2}$\ and
\textsc{GPE}$_{\times}$\ are polynomial-time equivalent.
\end{theorem}

Figure \ref{orgfig} summarizes the overall structure of our hardness argument
for approximate \textsc{BosonSampling}.%
%TCIMACRO{\FRAME{ftbpFU}{1.5636in}{2.2788in}{0pt}{\Qcb{Summary of our hardness
%argument (modulo conjectures). \ If there exists a polynomial-time classical
%algorithm for approximate \QTR{sc}{BosonSampling}, then Theorem
%\ref{mainresult} says that $\left\vert \text{\QTR{sc}{GPE}}\right\vert _{\pm
%}^{2}\in\QTR{sf}{BPP}^{\QTR{sf}{NP}}$. \ Assuming Conjecture \ref{pacc} (the
%PACC), Theorem \ref{decompthm} says that this is equivalent to
%\QTR{sc}{GPE}$_{\times}\in\QTR{sf}{BPP}^{\QTR{sf}{NP}}$. \ Assuming Conjecture
%\ref{pgc}\ (the PGC), this is in turn equivalent to
%$\QTR{sf}{P}^{\QTR{sf}{\#P}}=\QTR{sf}{BPP}^{\QTR{sf}{NP}}$, which collapses
%the polynomial hierarchy by Toda's Theorem \cite{toda}.}}{\Qlb{orgfig}%
%}{proofstruc.eps}{\special{ language "Scientific Word";  type "GRAPHIC";
%width 1.5636in;  height 2.2788in;  depth 0pt;  original-width 3.8398in;
%original-height 5.6265in;  cropleft "0";  croptop "1";  cropright "1";
%cropbottom "0";  filename 'proofstruc.eps';file-properties "XNPEU";}} }%
%BeginExpansion
\begin{figure}[ptb]%
\centering
\includegraphics[
height=2.2788in,
width=1.5636in
]%
{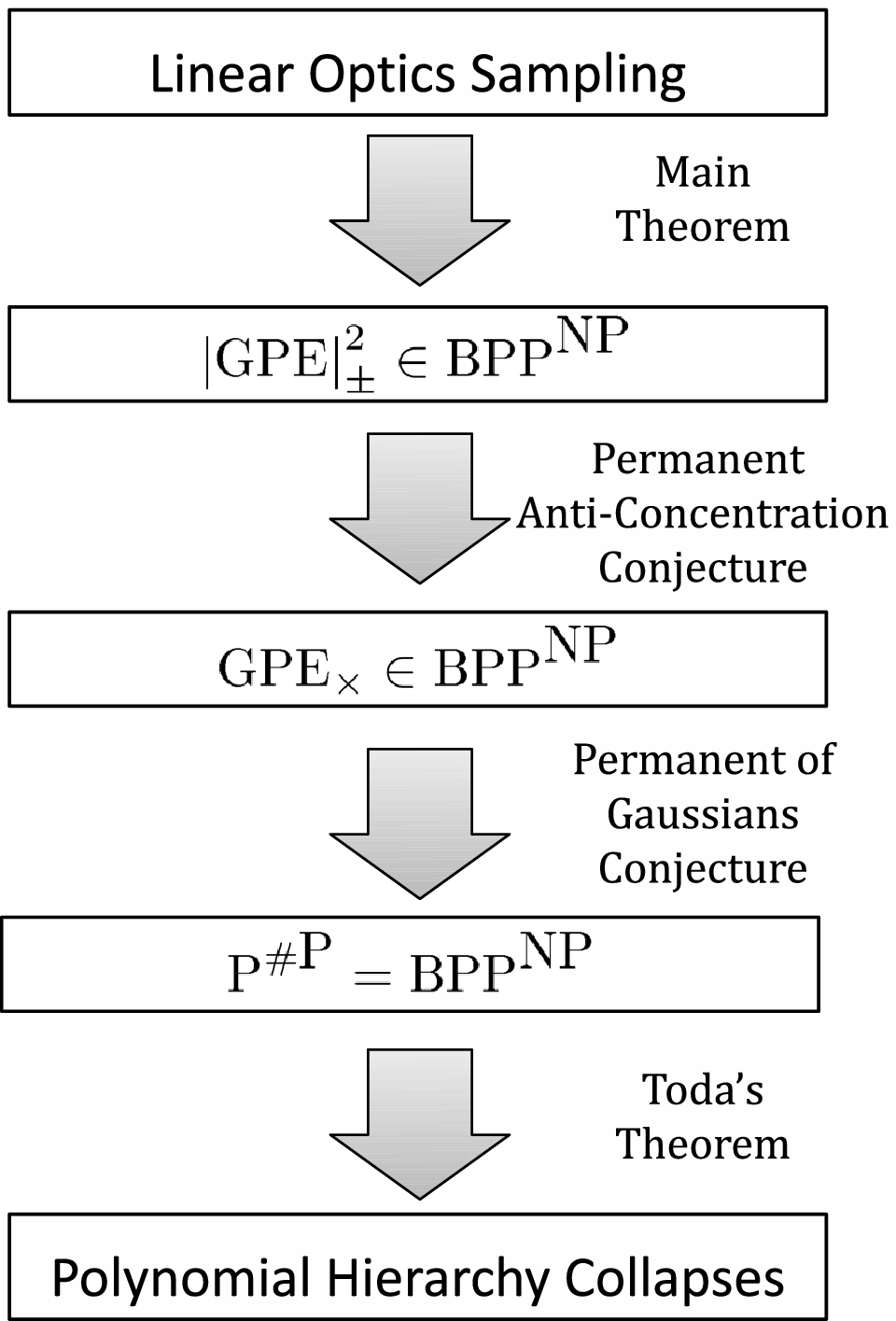}%
\caption{Summary of our hardness argument (modulo conjectures). \ If there
exists a polynomial-time classical algorithm for approximate
\textsc{BosonSampling}, then Theorem \ref{mainresult} says that $\left\vert
\text{\textsc{GPE}}\right\vert _{\pm}^{2}\in\mathsf{BPP}^{\mathsf{NP}}$.
\ Assuming Conjecture \ref{pacc} (the PACC), Theorem \ref{decompthm} says that
this is equivalent to \textsc{GPE}$_{\times}\in\mathsf{BPP}^{\mathsf{NP}}$.
\ Assuming Conjecture \ref{pgc}\ (the PGC), this is in turn equivalent to
$\mathsf{P}^{\mathsf{\#P}}=\mathsf{BPP}^{\mathsf{NP}}$, which collapses the
polynomial hierarchy by Toda's Theorem \cite{toda}.}%
\label{orgfig}%
\end{figure}
%EndExpansion

The rest of the body of the paper aims at a better understanding of
Conjectures \ref{pgc}\ and \ref{pacc}.

First, in Section \ref{DISTPER}, we summarize the considerable evidence for
the Permanent Anti-Concentration Conjecture. \ This includes numerical
results; a weaker anti-concentration bound for the permanent recently proved
by Tao and Vu \cite{taovu:perm}; another weaker bound that we prove; and the
analogue of Conjecture \ref{pacc}\ for the determinant.

Next, in Section \ref{HARDPER}, we discuss the less certain state of affairs
regarding the Permanent-of-Gaussians Conjecture. \ On the one hand, we extend
the random self-reducibility of permanents over finite fields proved by Lipton
\cite{lipton}, to show that \textit{exactly} computing the permanent of
\textit{most} Gaussian matrices $X\sim\mathcal{N}\left(  0,1\right)
_{\mathbb{C}}^{n\times n}$\ is $\mathsf{\#P}$-hard. \ On the other hand, we
also show that extending this result further, to show that
\textit{approximating} $\operatorname*{Per}\left(  X\right)  $ for Gaussian
$X$ is $\mathsf{\#P}$-hard, will require going beyond Lipton's polynomial
interpolation technique in a fundamental way.

Two appendices give some additional results. \ First, in Appendix \ref{ALGS},
we present two remarkable algorithms due to Gurvits \cite{gurvits:alg} (with
Gurvits's kind permission) for solving certain problems related to
linear-optical networks in classical polynomial time. \ We also explain why
these algorithms do not conflict with our hardness conjecture. \ Second, in
Appendix \ref{BIRTHDAY}, we bring out a useful fact that was implicit in our
proof of Theorem \ref{mainresult}, but seems to deserve its own treatment.
\ This is that, if we have $n$ identical bosons scattered among $m\gg n^{2}$
locations, with no two bosons in the same location, and if we apply a
Haar-random $m\times m$\ unitary transformation $U$ and then measure the
number of bosons in each location, with high probability we will
\textit{still} not find two bosons in the same location. \ In other words, at
least asymptotically, the birthday paradox works the same way for identical
bosons as for classical particles, in spite of bosons' well-known tendency to
cluster in the same state.

\subsection{Experimental Implications\label{EXPERIMP}}

An important motivation for our results is that they immediately suggest a
linear-optics experiment, which would use simple optical elements
(beamsplitters and phaseshifters) to induce a Haar-random $m\times m$\ unitary
transformation $U$ on an input state of $n$ photons, and would then check that
the probabilities of various final states of the photons correspond to the
permanents of $n\times n$\ submatrices of $U$, as predicted by quantum
mechanics. \ Were such an experiment successfully scaled to large values of
$n$, Theorem \ref{mainresult}\ asserts that no polynomial-time classical
algorithm could simulate the experiment even approximately, unless $\left\vert
\text{\textsc{GPE}}\right\vert _{\pm}^{2}\in\mathsf{BPP}^{\mathsf{NP}}$.

Of course, the question arises of how large $n$ has to be before one can draw
interesting conclusions. \ An obvious difficulty is that \textit{no} finite
experiment can hope to render a decisive verdict on the Extended Church-Turing
Thesis, since the ECT is a statement about the asymptotic limit as
$n\rightarrow\infty$. \ Indeed, this problem is actually \textit{worse} for us
than for (say) Shor's algorithm, since unlike with \textsc{Factoring}, we do
not believe there is any $\mathsf{NP}$\ witness for \textsc{BosonSampling}.
\ In other words, if $n$ is large enough that a classical computer cannot
solve \textsc{BosonSampling}, then $n$ is probably \textit{also} large enough
that a classical computer cannot even verify that a quantum computer is
solving \textsc{BosonSampling} correctly.

Yet while this sounds discouraging, it is not really an issue from the
perspective of near-term experiments. \ For the foreseeable future, $n$ being
\textit{too large} is likely to be the least of one's problems! \ If one could
implement our experiment with (say)\ $20\leq n\leq30$, then certainly a
classical computer could verify the answers---but at the same time, one would
be getting direct evidence that a quantum computer could efficiently solve an
\textquotedblleft interestingly difficult\textquotedblright\ problem, one for
which the best-known classical algorithms require many millions of operations.
\ While \textit{disproving} the Extended Church-Turing Thesis is formally
impossible, such an experiment would arguably constitute the strongest
evidence against the ECT to date.

Section \ref{EXPER} goes into more detail about the physical resource
requirements for our proposed experiment, as well as how one would interpret
the results. \ In Section \ref{EXPER}, we also show that the size and depth of
the linear-optical network needed for our experiment can both be improved by
polynomial factors over the na\"{\i}ve bounds. \ Complexity theorists who are
not interested in the \textquotedblleft practical side\textquotedblright\ of
boson computation can safely skip Section \ref{EXPER}, while experimentalists
who are \textit{only} interested the practical side can skip everything else.

While most further discussion of experimental issues is deferred to Section
\ref{EXPER}, there is one question we need to address now. \ Namely:
\textit{what, if any, are the advantages of doing our experiment, as opposed
simply to building a somewhat larger \textquotedblleft
conventional\textquotedblright\ quantum computer, able (for example) to factor
}$10$\textit{-digit numbers using Shor's algorithm?} \ While a full answer to
this question will need to await detailed analysis by experimentalists,\ let
us mention four aspects of \textsc{BosonSampling} that might make it
attractive for quantum computing experiments.

\begin{enumerate}
\item[(1)] Our proposal does not require any explicit \textit{coupling}
between pairs of photons. \ It therefore bypasses what has long been seen as
one of the central technological obstacles to building a scalable quantum
computer: namely, how to make arbitrary pairs of particles \textquotedblleft
talk to each other\textquotedblright\ (e.g., via two-qubit gates), in a manner
that still preserves the particles' coherence. \ One might ask how there is
any possibility of a quantum speedup, if the particles are never entangled.
\ The answer is that, because of the way boson statistics work, every two
identical photons are \textit{somewhat} entangled \textquotedblleft for
free,\textquotedblright\ in the sense that the amplitude for any process
involving both photons includes contributions in which the photons swap their
states. \ This \textquotedblleft free\textquotedblright\ entanglement is the
only kind that our model ever uses.

\item[(2)] Photons traveling through linear-optical networks are known to have
some of the best coherence properties of any quantum system accessible to
current experiments. \ From a \textquotedblleft traditional\textquotedblright%
\ quantum computing standpoint, the disadvantages of photons are that they
have no direct coupling to one another, and also that they are extremely
difficult to store (they are, after all, traveling at the speed of light).
\ There have been ingenious proposals for working around these problems, most
famously the adaptive scheme of Knill, Laflamme, and Milburn \cite{klm}. \ By
contrast, rather than trying to remedy photons' disadvantages as qubits, our
proposal simply never uses photons as qubits at all, and thereby gets the
coherence advantages of linear optics without having to address the disadvantages.

\item[(3)] To implement Shor's algorithm, one needs to perform modular
arithmetic on a coherent superposition of integers encoded in binary.
\ Unfortunately, this requirement causes significant constant blowups, and
helps to explain why the \textquotedblleft world record\textquotedblright\ for
implementations of Shor's algorithm is still the factoring of $15$ into
$3\times5$, first demonstrated in 2001 \cite{vsbysc}. \ By contrast, because
the \textsc{BosonSampling}\ problem is so close to the \textquotedblleft
native physics\textquotedblright\ of linear-optical networks, an $n$-photon
experiment corresponds directly to a problem instance of size $n$, which
involves the permanents of $n\times n$ matrices. \ This raises the hope that,
using current technology, one could sample quantum-mechanically from a
distribution in which the probabilities depended (for example) on the
permanents of $10\times10$ matrices of complex numbers.

\item[(4)] The resources that our experiment \textit{does} demand---including
reliable single-photon sources and photodetector arrays---are ones that
experimentalists, for their own reasons, have devoted large and successful
efforts to improving within the past decade. \ We see every reason to expect
further improvements.
\end{enumerate}

In implementing our experiment, the central difficulty is likely to be getting
a reasonably-large probability of an $n$\textit{-photon coincidence}: that is,
of all $n$ photons arriving at the photodetectors at the same time (or rather,
within a short enough time interval that interference is seen). \ If the
photons arrive at different times, then they effectively become
\textit{distinguishable} particles, and the experiment no longer solves the
\textsc{BosonSampling}\ problem. \ Of course, one solution is simply to repeat
the experiment many times, then \textit{postselect} on the $n$-photon
coincidences. \ However, if the probability of an $n$-photon coincidence
decreases exponentially with $n$, then this \textquotedblleft
solution\textquotedblright\ has obvious scalability problems.

\textit{If} one could scale our experiment to moderately large values of $n$
(say, $10$ or $20$),\ without the probability of an $n$-photon
coincidence\ falling off dramatically, then our experiment would raise the
exciting possibility of doing an interestingly-large quantum computation
without any need for explicit quantum error-correction. \ Whether or not this
is feasible is the main open problem we leave for experimentalists.

\subsection{Related Work\label{RELATED}}

By necessity, this paper brings together many ideas from quantum computing,
optical physics, and computational complexity. \ In this section, we try to
survey the large relevant literature, organizing it into eight
categories.\bigskip

\textbf{Quantum computing with linear optics.} \ There is a huge body of work,
both experimental and theoretical, on quantum computing with linear optics.
\ Much of that work builds on a seminal 2001 result of Knill, Laflamme, and
Milburn \cite{klm}, showing that linear optics combined with\textit{ adaptive
measurements} is universal for quantum computation. \ It is largely because of
this result that linear optics is considered a viable proposal for building a
universal quantum computer.\footnote{An earlier proposal for building a
universal optical quantum computer was to use \textit{nonlinear optics}: in
other words, explicit entangling interactions between pairs of photons. \ (See
Nielsen and Chuang \cite{nc}\ for discussion.) \ The problem is that, at least
at low energies, photons \textit{have} no direct coupling to one another. \ It
is therefore necessary to use other particles as intermediaries, which greatly
increases decoherence, and negates many of the advantages of using photons in
the first place.}

In the opposite direction, several interesting classes of linear-optics
experiments have been proved to be efficiently simulable on a classical
computer. \ For example, Bartlett and Sanders \cite{bartlettsanders} showed
that a linear-optics network with \textit{coherent-state inputs} and
possibly-adaptive \textit{Gaussian measurements} can be simulated in classical
polynomial time. \ (Intuitively, a coherent state---the output of a standard
laser---is a superposition over different numbers of photons that behaves
essentially like a classical wave, while a Gaussian measurement is a
measurement that preserves this classical wave behavior.) \ Also, Gurvits
\cite{gurvits:alg} showed that, in any $n$-photon linear-optics experiment,
the probability of measuring a particular basis state can be estimated to
within $\pm\varepsilon$\ additive error in $\operatorname*{poly}\left(
n,1/\varepsilon\right)  $\ time.\footnote{While beautiful, this result is of
limited use in practice---since in a typical linear-optics experiment, the
probability $p$\ of measuring any \textit{specific} basis state is so small
that $0$\ is a good additive estimate to $p$.} \ He also showed that the
marginal distribution over any $k$ photon modes can be computed
deterministically in $n^{O\left(  k\right)  }$\ time. \ We discuss Gurvits's
results in detail in Appendix \ref{ALGS}.

Our model can be seen as intermediate between the above two extremes: unlike
Knill et al.\ \cite{klm}, we do not allow adaptive measurements, and as a
result, our model is probably not universal for $\mathsf{BQP}$. \ On the other
hand, unlike Bartlett and Sanders, we \textit{do} allow single-photon inputs
and photon-number measurements; and unlike Gurvits \cite{gurvits:alg}, we
consider the complexity of sampling from the joint distribution over all
$\operatorname*{poly}\left(  n\right)  $\ photon modes. \ Our main result
gives strong evidence that the resulting model cannot be simulated in
classical polynomial time. \ On the other hand, it might be significantly
easier to implement than a universal quantum computer.\bigskip

\textbf{Intermediate models of quantum computation.} \ By now, several
interesting models of quantum computation have been proposed that are neither
known to be universal for $\mathsf{BQP}$, nor simulable in classical
polynomial time. \ A few examples, besides the ones mentioned elsewhere in the
paper, are the \textquotedblleft one-clean-qubit\textquotedblright\ model of
Knill and Laflamme \cite{knilaflamme}; the permutational quantum computing
model of Jordan \cite{jordan}; and stabilizer circuits with non-stabilizer
initial states (such as $\cos\frac{\pi}{8}\left\vert 0\right\rangle +\sin
\frac{\pi}{8}\left\vert 0\right\rangle $) and nonadaptive measurements
\cite{ag}. \ The noninteracting-boson model is another addition to this
list.\bigskip

\textbf{The Hong-Ou-Mandel dip.} \ In 1987, Hong, Ou, and Mandel
\cite{hom}\ performed a now-standard experiment that, in essence, directly
confirms that \textit{two}-photon amplitudes correspond to $2\times
2$\ permanents in the way predicted by quantum mechanics. \ From an
experimental perspective, what we are asking for could be seen as a
generalization of the so-called \textquotedblleft Hong-Ou-Mandel
dip\textquotedblright\ to the $n$-photon case, where $n$ is as large as
possible. \ Lim and Beige \cite{limbeige} previously proposed an $n$-photon
generalization of the Hong-Ou-Mandel dip, but without the computational
complexity motivation.\bigskip

\textbf{Bosons and the permanent.} \ \textit{Bosons} are one of the two basic
types of particle in the universe; they include photons and the carriers of
nuclear forces. \ It has been known since work by Caianiello \cite{caianiello}%
\ in 1953\ (if not earlier) that the amplitudes for $n$-boson processes can be
written as the permanents of $n\times n$\ matrices. \ Meanwhile, Valiant
\cite{valiant}\ proved in 1979 that the permanent is $\mathsf{\#P}$-complete.
\ Interestingly, according to Valiant (personal communication), he and others
put these two facts together immediately, and wondered what they might mean
for the computational complexity of simulating bosonic systems. \ To our
knowledge, however, the first authors to discuss this question in print were
Troyansky and Tishby \cite{troyanskytishby} in 1996. \ Given an arbitrary
matrix $A\in\mathbb{C}^{n\times n}$, these authors showed how to construct a
quantum observable with expectation value equal to $\operatorname*{Per}\left(
A\right)  $. \ However, they correctly pointed out that this did not imply a
polynomial-time quantum algorithm to \textit{calculate} $\operatorname*{Per}%
\left(  A\right)  $, since the variance of their observable was large enough
that exponentially many samples would be needed.

Later, Scheel \cite{scheel} explained how permanents arise as amplitudes in
linear-optical networks, and noted that calculations involving linear-optical
networks might be intractable because the permanent is $\mathsf{\#P}%
$-complete.\bigskip

\textbf{Fermions and the determinant.} \ Besides bosons, the other basic
particles in the universe are \textit{fermions}; these include matter
particles such as quarks and electrons. \ Remarkably, the amplitudes for
$n$-fermion processes are given not by permanents but by \textit{determinants}
of $n\times n$\ matrices. \ Despite the similarity of their definitions, it is
well-known that the permanent and determinant differ dramatically in their
computational properties; the former is $\mathsf{\#P}$-complete\ while the
latter is in $\mathsf{P}$. \ In a lecture in 2000, Wigderson called attention
to this striking connection between the boson-fermion dichotomy of physics and
the permanent-determinant dichotomy of computer science. \ He joked that,
between bosons and fermions, \textquotedblleft the bosons got the harder
job.\textquotedblright\ \ One could view this paper as a formalization of
Wigderson's joke.

To be fair, \textit{half} the work of formalizing Wigderson's joke has already
been carried out. \ In 2002, Valiant \cite{valiant:qc} defined a beautiful
subclass of quantum circuits called \textit{matchgate circuits},\ and showed
that these circuits could be efficiently simulated classically, via a
nontrivial algorithm that ultimately relied on computing
determinants.\footnote{Or rather, a closely-related matrix function called the
Pfaffian.} \ Shortly afterward, Terhal and DiVincenzo \cite{td:fermion} (see
also Knill \cite{knill:matchgate}) pointed out that matchgate circuits were
equivalent to systems of noninteracting fermions\footnote{Strictly speaking,
\textit{unitary} matchgate circuits are equivalent to noninteracting fermions
(Valiant also studied matchgates that violated unitarity).}: in that sense,
one could say Valiant had \textquotedblleft rediscovered
fermions\textquotedblright! \ Indeed, Valiant's matchgate model can be seen as
the direct counterpart of the model studied in this paper, but with
noninteracting fermions in place of noninteracting bosons.\footnote{However,
the noninteracting-boson model is somewhat more complicated to define, since
one can have multiple bosons occupying the same state, whereas fermions are
prohibited from this by the Pauli exclusion principle. \ This is why the basis
states in our model are lists of nonnegative integers, whereas the basis
states in Valiant's model are binary strings.} \ At a very high level,
Valiant's model is easy to simulate classically because the determinant is in
$\mathsf{P}$, whereas our model is hard to simulate because the permanent is
$\mathsf{\#P}$-complete.

Ironically, when the \textit{quantum Monte Carlo method} \cite{ceperley}\ is
used to approximate the ground states of many-body systems, the computational
situation regarding bosons and fermions is reversed. \ Bosonic ground states
tend to be \textit{easy} to approximate because one can exploit
non-negativity, while fermionic ground states tend to be \textit{hard} to
approximate because of cancellations between positive and negative terms, what
physicists call \textquotedblleft the sign problem.\textquotedblright\bigskip

\textbf{Quantum computing and }$\mathsf{\#P}$\textbf{-complete problems.
\ }Since amplitudes in quantum mechanics are the sums of exponentially many
complex numbers, it is natural to look for some formal connection between
quantum computing and the class $\mathsf{\#P}$ of counting problems. \ In
1993, Bernstein and Vazirani \cite{bv}\ proved that $\mathsf{BQP}%
\subseteq\mathsf{P}^{\mathsf{\#P}}$.\footnote{See also Rudolph \cite{rudolph}
for a direct encoding of quantum computations by matrix permanents.}
\ However, this result says only that $\mathsf{\#P}$\ is an \textit{upper}
bound on the power of quantum computation,\ so the question arises of whether
solving $\mathsf{\#P}$-complete problems is in any sense\ \textit{necessary}
for simulating quantum mechanics.

To be clear, we do not expect that $\mathsf{BQP}=\mathsf{P}^{\mathsf{\#P}}$;
indeed, it would be a scientific revolution even if $\mathsf{BQP}$\ were found
to contain $\mathsf{NP}$. \ However, already in 1999, Fenner, Green, Homer,
and\ Pruim \cite{fghp}\ noticed that, if we ask more refined questions about a
quantum circuit than

\begin{quotation}
\noindent\textit{\textquotedblleft does this circuit accept with probability
greater than }$1-\varepsilon$\textit{\ or less than }$\varepsilon$\textit{,
promised that one of those is true?,\textquotedblright}
\end{quotation}

\noindent then we can quickly encounter $\mathsf{\#P}$-completeness. \ In
particular, Fenner et al.\ showed that deciding whether a quantum circuit
accepts with \textit{nonzero or zero} probability is complete for the
complexity class $\mathsf{coC}_{\mathsf{=}}\mathsf{P}$. \ Since $\mathsf{P}%
^{\mathsf{\#P}}\subseteq\mathsf{NP}^{\mathsf{coC}_{\mathsf{=}}\mathsf{P}}$,
this means that the problem is $\mathsf{\#P}$-hard under nondeterministic reductions.

Later, Aaronson \cite{aar:pp} defined the class $\mathsf{PostBQP}$, or quantum
polynomial-time\ with \textit{postselection} on possibly
exponentially-unlikely measurement outcomes. \ He showed that
$\mathsf{PostBQP}$ is equal to the classical class $\mathsf{PP}$. \ Since
$\mathsf{P}^{\mathsf{PP}}=\mathsf{P}^{\mathsf{\#P}}$, this says that quantum
computers with postselection can already solve $\mathsf{\#P}$-complete
problems. \ Following \cite{bjs}, in Section \ref{ALTPROOF} we will use the
$\mathsf{PostBQP}=\mathsf{PP}$\ theorem to give an alternative proof of
Theorem \ref{warmup}, which does not require using the $\mathsf{\#P}%
$-completeness of the permanent.\bigskip

\textbf{Quantum speedups for sampling and search problems.} \ Ultimately, we
want a hardness result for simulating \textit{real} quantum experiments,
rather than postselected ones. \ To achieve that, a crucial step in this paper
will be to switch attention from \textit{decision} problems to
\textit{sampling} and \textit{search} problems. \ The value of that step in a
quantum computing context was recognized in several previous works.

In 2008, Shepherd and Bremner \cite{shepherd}\ defined and studied a
fascinating subclass of quantum computations, which they called
\textquotedblleft commuting\textquotedblright\ or \textquotedblleft
temporally-unstructured.\textquotedblright\ \ Their model is probably not
universal for $\mathsf{BQP}$, and there is no known example of a decision
problem solvable by their model that is not also in $\mathsf{BPP}$. \ However,
if we consider \textit{sampling} problems or interactive protocols, then
Shepherd and Bremner plausibly argued (without formal evidence) that their
model might be hard to simulate classically.

Recently, and independently of us, Bremner, Jozsa, and Shepherd \cite{bjs}
showed that commuting quantum computers can sample from probability
distributions that cannot be efficiently sampled classically, unless
$\mathsf{PP}=\mathsf{BPP}_{\mathsf{path}}$ and hence the polynomial hierarchy
collapses to the third level. \ This is analogous to our Theorem \ref{warmup},
except with commuting quantum computations instead of noninteracting-boson ones.

Previously, in 2002, Terhal and DiVincenzo \cite{td} showed that
constant-depth quantum circuits can sample from probability distributions that
cannot be efficiently sampled by a classical computer, unless $\mathsf{BQP}%
\subseteq\mathsf{AM}$. \ By using our arguments and Bremner et
al.'s\ \cite{bjs}, it is not hard to strengthen Terhal and DiVincenzo's
conclusion, to show that exact classical simulation of their model would also
imply $\mathsf{PP}=\mathsf{PostBQP}=\mathsf{BPP}_{\mathsf{path}}$, and hence
that the polynomial hierarchy collapses.

However, all of these results (including our Theorem \ref{warmup}) have the
drawback that they only address sampling from \textit{exactly} the same
distribution $\mathcal{D}$\ as the quantum algorithm---or at least, from some
distribution in which all the probabilities are multiplicatively close to the
ideal ones. \ Indeed, in these results, everything hinges on the
$\mathsf{\#P}$-completeness of estimating a single, exponentially-small
probability $p$. \ For this reason, such results might be considered
\textquotedblleft cheats\textquotedblright: presumably not even the quantum
device \textit{itself} can sample perfectly from the ideal distribution
$\mathcal{D}$! \ What if we allow \textquotedblleft realistic
noise,\textquotedblright\ so that one only needs to sample from some
probability distribution $\mathcal{D}^{\prime}$\ that is
$1/\operatorname*{poly}\left(  n\right)  $\textit{-close} to $\mathcal{D}$\ in
total variation distance? \ Is that\ \textit{still} a classically-intractable
problem? \ This is the question we took as our starting point.\bigskip

\textbf{Oracle results.} \ We know of one previous work that addressed the
hardness of sampling \textit{approximately} from a quantum computer's output
distribution. \ In 2010, Aaronson \cite{aar:ph}\ showed that, relative to a
random oracle $A$, quantum computers can sample from probability distributions
$\mathcal{D}$\ that are not even \textit{approximately} samplable in
$\mathsf{BPP}^{\mathsf{PH}^{A}}$\ (that is, by classical computers with
oracles for the polynomial hierarchy). \ Relative to a random oracle $A$,
quantum computers can also solve \textit{search} problems not in
$\mathsf{BPP}^{\mathsf{PH}^{A}}$. \ The point of these results was to give the
first formal evidence that quantum computers have \textquotedblleft
capabilities outside $\mathsf{PH}$.\textquotedblright

For us, though, what is more relevant is a striking feature of the
\textit{proofs} of these results. \ Namely, they showed that, if the sampling
and search problems in question were in $\mathsf{BPP}^{\mathsf{PH}^{A}}$, then
(via a nonuniform, nondeterministic reduction) one could extract small
constant-depth circuits for the $2^{n}$-bit \textsc{Majority}\ function,
thereby violating the celebrated circuit lower bounds of H\aa stad
\cite{hastad:book} and others. \ What made this surprising was that the
$2^{n}$-bit \textsc{Majority}\ function is $\mathsf{\#P}$%
-complete.\footnote{Here we are abusing terminology (but only slightly) by
speaking about the\ $\mathsf{\#P}$-completeness of an oracle problem. \ Also,
strictly speaking we mean $\mathsf{PP}$-complete---but since $\mathsf{P}%
^{\mathsf{PP}}=\mathsf{P}^{\mathsf{\#P}}$, the distinction is unimportant
here.} \ In other words, even though there is no evidence that quantum
computers can solve $\mathsf{\#P}$-complete problems, somehow we managed
to\textit{ prove the hardness of simulating a }$\mathsf{BQP}$\textit{\ machine
by using the hardness of }$\mathsf{\#P}$\textit{.}

Of course, a drawback of Aaronson's results \cite{aar:ph} is that they were
relative to an oracle. \ However, just like Simon's oracle algorithm
\cite{simon} led shortly afterward to Shor's algorithm \cite{shor}, so too in
this case one could hope to \textquotedblleft reify the
oracle\textquotedblright: that is, find a real, unrelativized problem with the
same behavior that the oracle problem illustrated more abstractly. \ That is
what we do here.

\section{Preliminaries\label{PRELIM}}

Throughout this paper, we use $\mathcal{G}$\ to denote $\mathcal{N}\left(
0,1\right)  _{\mathbb{C}}$, the complex Gaussian distribution with mean $0$
and variance $\operatorname*{E}_{z\sim\mathcal{G}}\left[  \left\vert
z\right\vert ^{2}\right]  =1$. \ (We often use the word \textquotedblleft
distribution\textquotedblright\ for continuous probability measures, as well
as for discrete distributions.) \ We will be especially interested in
$\mathcal{G}^{n\times n}$, the distribution over $n\times n$\ matrices with
i.i.d.\ Gaussian entries.

For $m\geq n$, we use $\mathcal{U}_{m,n}$\ to denote the set of matrices
$A\in\mathbb{C}^{m\times n}$ whose columns are orthonormal vectors,\ and
$\mathcal{H}_{m,n}$\ to denote the Haar measure over $\mathcal{U}_{m,n}$. \ So
in particular, $\mathcal{H}_{m,m}$\ is the Haar measure over the set
$\mathcal{U}_{m,m}$\ of $m\times m$\ unitary matrices.

We use $\overline{\alpha}$\ to denote the complex conjugate of $\alpha$. \ We
denote the set $\left\{  1,\ldots,n\right\}  $\ by\ $\left[  n\right]  $.
\ Let $v\in\mathbb{C}^{n}$ and $A\in\mathbb{C}^{n\times n}$. \ Then
$\left\Vert v\right\Vert :=\sqrt{\left\vert v_{1}\right\vert ^{2}%
+\cdots+\left\vert v_{n}\right\vert ^{2}}$, and $\left\Vert A\right\Vert
:=\max_{\left\Vert v\right\Vert =1}\left\Vert Av\right\Vert $. \ Equivalently,
$\left\Vert A\right\Vert =\sigma_{\max}\left(  A\right)  $\ is the largest
singular value of $A$.

We generally omit floor and ceiling signs, when it is clear that the relevant
quantities can be rounded to integers without changing the asymptotic
complexity. \ Likewise, we will talk about a polynomial-time algorithm
receiving as input a matrix $A\in\mathbb{C}^{n\times n}$, often drawn from the
Gaussian distribution $\mathcal{G}^{n\times n}$. \ Here it is understood that
the entries of $A$\ are rounded to $p\left(  n\right)  $\ bits of precision,
for some polynomial $p$. \ In all such cases, it will be straightforward to
verify that there exists a fixed polynomial $p$, such that none of the
relevant calculations are affected by precision issues.

We assume familiarity with standard computational complexity classes such as
$\mathsf{BQP}$ (Bounded-Error Quantum Polynomial-Time) and $\mathsf{PH}$ (the
Polynomial Hierarchy).\footnote{See the Complexity Zoo, www.complexityzoo.com,
for definitions of these and other classes.} \ We now define some other
complexity classes that will be important in this work. \ 

\begin{definition}
[$\mathsf{PostBPP}$ and $\mathsf{PostBQP}$]Say the algorithm $\mathcal{A}$
\textquotedblleft succeeds\textquotedblright\ if its first output bit is
measured to be $1$ and \textquotedblleft fails\textquotedblright\ otherwise;
conditioned on succeeding, say $\mathcal{A}$ \textquotedblleft
accepts\textquotedblright\ if its second output bit is measured to be $1$\ and
\textquotedblleft rejects\textquotedblright\ otherwise. \ Then
$\mathsf{PostBPP}$ is the class of languages $L\subseteq\left\{  0,1\right\}
^{\ast}$\ for which there exists a probabilistic polynomial-time algorithm
$\mathcal{A}$ such that, for all inputs $x$:

\begin{enumerate}
\item[(i)] $\Pr\left[  \mathcal{A}\left(  x\right)  \text{ succeeds}\right]
>0$.

\item[(ii)] If $x\in L$ then $\Pr\left[  \mathcal{A}\left(  x\right)  \text{
accepts
%TCIMACRO{\TEXTsymbol{\vert} }%
%BeginExpansion
$\vert$
%EndExpansion
}\mathcal{A}\left(  x\right)  \text{ succeeds}\right]  \geq\frac{2}{3}$.

\item[(iii)] If $x\notin L$ then $\Pr\left[  \mathcal{A}\left(  x\right)
\text{ accepts
%TCIMACRO{\TEXTsymbol{\vert} }%
%BeginExpansion
$\vert$
%EndExpansion
}\mathcal{A}\left(  x\right)  \text{ succeeds}\right]  \leq\frac{1}{3}$.
\end{enumerate}

$\mathsf{PostBQP}$\ is defined the same way, except that $\mathcal{A}$ is a
quantum algorithm rather than a classical one.
\end{definition}

$\mathsf{PostBPP}$ is easily seen to equal the complexity class $\mathsf{BPP}%
_{\mathsf{path}}$, which was defined by Han, Hemaspaandra, and Thierauf
\cite{hht}. \ In particular, it follows from Han et al.'s results that
$\mathsf{MA}\subseteq\mathsf{PostBPP}$\ and that $\mathsf{P}_{||}%
^{\mathsf{NP}}\subseteq\mathsf{PostBPP}\subseteq\mathsf{BPP}_{||}%
^{\mathsf{NP}}$, where $\mathsf{P}_{||}^{\mathsf{NP}}$\ and $\mathsf{BPP}%
_{||}^{\mathsf{NP}}$\ denote $\mathsf{P}$\ and $\mathsf{BPP}$\ respectively
with \textit{nonadaptive} queries to an $\mathsf{NP}$\ oracle. \ As for
$\mathsf{PostBQP}$, we have the following result of Aaronson \cite{aar:pp},
which characterizes $\mathsf{PostBQP}$\ in terms of the classical complexity
class $\mathsf{PP}$\ (Probabilistic Polynomial-Time).

\begin{theorem}
[Aaronson \cite{aar:pp}]\label{postbqpthm}$\mathsf{PostBQP}=\mathsf{PP}$.
\end{theorem}

It is well-known that $\mathsf{P}^{\mathsf{PP}}=\mathsf{P}^{\mathsf{\#P}}%
$---and thus, Theorem \ref{postbqpthm}\ has the surprising implication that
$\mathsf{BQP}$\ with postselection is as powerful as an oracle for
\textit{counting} problems.

Aaronson \cite{aar:pp}\ also observed that, just as intermediate measurements
do not affect the power of $\mathsf{BQP}$, so intermediate postselected
measurements do not affect the power of $\mathsf{PostBQP}$.

\subsection{Sampling and Search Problems\label{SAMPNREL}}

In this work, a central role is played not only by decision problems, but also
by \textit{sampling} and \textit{search} problems. \ By a \textit{sampling
problem} $S$, we mean a collection of probability distributions $\left(
\mathcal{D}_{x}\right)  _{x\in\left\{  0,1\right\}  ^{\ast}}$, one for each
input string $x\in\left\{  0,1\right\}  ^{n}$. \ Here $\mathcal{D}_{x}$\ is a
distribution over $\left\{  0,1\right\}  ^{p\left(  n\right)  }$, for some
fixed polynomial $p$. \ To \textquotedblleft solve\textquotedblright\ $S$
means to sample from $\mathcal{D}_{x}$, given $x$ as input, while to
solve\textit{ }$S$\textit{ }approximately means (informally) to sample from
some distribution that is $1/\operatorname*{poly}\left(  n\right)  $-close to
$\mathcal{D}_{x}$\ in variation distance. \ In this paper, we will be
interested in both notions, but especially approximate sampling.

We now define the classes $\mathsf{SampP}$\ and $\mathsf{SampBQP}$, consisting
of those sampling problems that are approximately solvable by polynomial-time
classical and quantum algorithms respectively.

\begin{definition}
[$\mathsf{SampP}$ and $\mathsf{SampBQP}$]$\mathsf{SampP}$ is the class of
sampling problems $S=\left(  \mathcal{D}_{x}\right)  _{x\in\left\{
0,1\right\}  ^{\ast}}$ for which there exists a probabilistic polynomial-time
algorithm $A$\ that, given $\left\langle x,0^{1/\varepsilon}\right\rangle $ as
input,\footnote{Giving $\left\langle x,0^{1/\varepsilon}\right\rangle $\ as
input (where $0^{1/\varepsilon}$\ represents $1/\varepsilon$\ encoded in
unary) is a standard trick for forcing an algorithm's running time to be
polynomial in $n$\ as well as $1/\varepsilon$.} samples from a probability
distribution $\mathcal{D}_{x}^{\prime}$\ such that $\left\Vert \mathcal{D}%
_{x}^{\prime}-\mathcal{D}_{x}\right\Vert \leq\varepsilon$. \ $\mathsf{SampBQP}%
$\ is defined the same way, except that $A$\ is a quantum algorithm rather
than a classical one.
\end{definition}

Another class of problems that will interest us are \textit{search problems}
(also confusingly called \textquotedblleft relation problems\textquotedblright%
\ or \textquotedblleft function problems\textquotedblright). \ In a search
problem, there is always at least one valid solution, and the problem is to
\textit{find} a solution: a famous example is finding a Nash equilibrium of a
game, the problem shown to be $\mathsf{PPAD}$-complete by Daskalakis et
al.\ \cite{dgp}. \ More formally, a search problem $R$ is a collection of
nonempty sets $\left(  B_{x}\right)  _{x\in\left\{  0,1\right\}  ^{\ast}}$,
one for each input $x\in\left\{  0,1\right\}  ^{n}$. \ Here $B_{x}%
\subseteq\left\{  0,1\right\}  ^{p\left(  n\right)  }$\ for some fixed
polynomial $p$. \ To solve $R$ means to output an element of $B_{x}$, given
$x$\ as input.

We now define the complexity classes $\mathsf{FBPP}$\ and $\mathsf{FBQP}$,
consisting of those search problems that are solvable by $\mathsf{BPP}$\ and
$\mathsf{BQP}$\ machines respectively.

\begin{definition}
[$\mathsf{FBPP}$ and $\mathsf{FBQP}$]$\mathsf{FBPP}$ is the class of search
problems $R=\left(  B_{x}\right)  _{x\in\left\{  0,1\right\}  ^{\ast}}$\ for
which there exists a probabilistic polynomial-time algorithm $A$ that, given
$\left\langle x,0^{1/\varepsilon}\right\rangle $ as input, produces an output
$y$ such that $\Pr\left[  y\in B_{x}\right]  \geq1-\varepsilon,$\ where the
probability is over $A$'s internal randomness. \ $\mathsf{FBQP}$\ is defined
the same way, except that $A$ is a quantum algorithm rather than a classical one.
\end{definition}

Recently, and directly motivated by the present work, Aaronson \cite{aar:samp}%
\ proved a general connection between sampling problems and search problems.

\begin{theorem}
[Sampling/Searching Equivalence Theorem \cite{aar:samp}]\label{samprel}Let
$S=\left(  \mathcal{D}_{x}\right)  _{x\in\left\{  0,1\right\}  ^{\ast}}$\ be
any approximate sampling problem. \ Then there exists a search problem
$R_{S}=\left(  B_{x}\right)  _{x\in\left\{  0,1\right\}  ^{\ast}}$\ that is
\textquotedblleft equivalent\textquotedblright\ to $S$ in the following two senses.

\begin{enumerate}
\item[(i)] Let $\mathcal{O}$\ be any oracle that, given $\left\langle
x,0^{1/\varepsilon},r\right\rangle $\ as input, outputs a sample from a
distribution $\mathcal{C}_{x}$\ such that $\left\Vert \mathcal{C}%
_{x}-\mathcal{D}_{x}\right\Vert \leq\varepsilon$, as we vary the random string
$r$. \ Then$\ R_{S}\in\mathsf{FBPP}^{\mathcal{O}}$.

\item[(ii)] Let $M$\ be any probabilistic Turing machine that, given
$\left\langle x,0^{1/\delta}\right\rangle $\ as input, outputs an element
$Y\in B_{x}$\ with probability at least $1-\delta$. \ Then$\ S\in
\mathsf{SampP}^{M}$.
\end{enumerate}
\end{theorem}

Briefly, Theorem \ref{samprel}\ is proved by using the notion of a
\textquotedblleft universal randomness test\textquotedblright\ from
algorithmic information theory. \ Intuitively, given a sampling problem $S$,
we define an \textquotedblleft equivalent\textquotedblright\ search problem
$R_{S}$\ as follows: \textquotedblleft output a collection of strings
$Y=\left(  y_{1},\ldots,y_{T}\right)  $\ in the support of $\mathcal{D}_{x}$,
most of which have large probability in $\mathcal{D}_{x}$\ and which
\textit{also}, conditioned on that, have close-to-maximal Kolmogorov
complexity.\textquotedblright\ \ Certainly, if we can sample from
$\mathcal{D}_{x}$, then we can solve this search problem as well. \ But the
converse also holds: if a probabilistic Turing machine is solving the search
problem $R_{S}$, it can \textit{only} be doing so by sampling approximately
from $\mathcal{D}_{x}$. \ For otherwise, the strings $y_{1},\ldots,y_{T}%
$\ would have short Turing machine descriptions, contrary to assumption.

In particular, Theorem \ref{samprel}\ implies that $S\in\mathsf{SampP}$\ if
and only if $R_{S}\in\mathsf{FBPP}$, $S\in\mathsf{SampBQP}$\ if and only if
$R_{S}\in\mathsf{FBQP}$, and so on. \ We therefore obtain the following consequence:

\begin{theorem}
[\cite{aar:samp}]\label{equivthm}$\mathsf{SampP}=\mathsf{SampBQP}$ if and only
if $\mathsf{FBPP}=\mathsf{FBQP}$.
\end{theorem}

\section{The Noninteracting-Boson Model of Computation\label{MODEL}}

In this section, we develop a formal model of computation based on
\textit{identical, noninteracting bosons}: as a concrete example, a
linear-optical network with single-photon inputs and nonadaptive photon-number
measurements. \ This model will yield a complexity class that, as far as we
know, is intermediate between $\mathsf{BPP}$\ and $\mathsf{BQP}$. \ The ideas
behind the model have been the basis for optical physics for almost a century.
\ To our knowledge, however, this is the first time the model has been
presented from a theoretical computer science perspective.

Like quantum mechanics itself, the noninteracting-boson model possesses a
mathematical beauty that can be appreciated even independently of its physical
origins. \ In an attempt to convey that beauty, we will define the model in
\textit{three} ways, and also prove those ways to be equivalent.\ \ The first
definition, in Section \ref{PHYSDEF}, is directly in terms of physical devices
(beamsplitters and phaseshifters) and the unitary transformations that they
induce. \ This definition should be easy to understand for those already
comfortable with quantum computing, and makes it apparent why our model can be
simulated on a standard quantum computer. \ The second definition, in Section
\ref{POLYDEF}, is in terms of multivariate polynomials with an unusual inner
product. \ This definition, which we learned from Gurvits \cite{gurvits:alg},
is the nicest one mathematically, and makes it easy to prove many statements
(for example, that the probabilities sum to $1$) that would otherwise require
tedious calculation. \ The third definition is in terms of permanents of
$n\times n$\ matrices, and is what lets us connect our model to the hardness
of the permanent. \ The second and third definitions do not use any quantum formalism.

Finally, Section \ref{BCT} defines \textsc{BosonSampling}, the basic
computational problem considered in this paper, as well as the complexity
class $\mathsf{BosonFP}$ of search problems solvable using a
\textsc{BosonSampling}\ oracle. \ It also proves the simple but important fact
that $\mathsf{BosonFP}\subseteq\mathsf{FBQP}$: in other words, boson computers
can be simulated efficiently\ by standard quantum computers.

\subsection{Physical Definition\label{PHYSDEF}}

The model that we are going to define involves a quantum system of $n$
identical photons\footnote{For concreteness, we will often talk about photons
in a linear-optical network, but the mathematics would be the same with any
other system of identical, noninteracting bosons (for example, bosonic
excitations in solid-state).} and $m$\ \textit{modes}\ (intuitively, places
that a photon can be in). \ We will usually be interested in the case where
$n\leq m\leq\operatorname*{poly}\left(  n\right)  $, though the model makes
sense for arbitrary $n$\ and $m$.\footnote{The one caveat is that our
\textquotedblleft standard initial state,\textquotedblright\ which consists of
one photon in each of the first $n$ modes, is only defined if $n\leq m$.}
\ Each computational basis state of this system has the form $\left\vert
S\right\rangle =\left\vert s_{1},\ldots,s_{m}\right\rangle $, where $s_{i}%
$\ represents the number of photons in the $i^{th}$\ mode ($s_{i}$ is also
called the $i^{th}$\ \textit{occupation number}). \ Here the $s_{i}$'s can be
any nonnegative integers summing to $n$; in particular, the $s_{i}$'s can be
greater than $1$. \ This corresponds to the fact that photons are bosons, and
(unlike with fermions) an unlimited number of bosons can be in the same place
at the same time.

During a computation, photons are never created or destroyed, but are only
moved from one mode to another. \ Mathematically, this means that the basis
states $\left\vert S\right\rangle $\ of our computer will always satisfy
$S\in\Phi_{m,n}$, where $\Phi_{m,n}$\ is the set of tuples $S=\left(
s_{1},\ldots,s_{m}\right)  $ satisfying $s_{1},\ldots,s_{m}\geq0$\ and
$s_{1}+\cdots+s_{m}=n$. \ Let $M=\left\vert \Phi_{m,n}\right\vert $\ be the
total number of basis states; then one can easily check that $M=\binom
{m+n-1}{n}$.

Since this is quantum mechanics, a general state of the computer has the form%
\[
\left\vert \psi\right\rangle =\sum_{S\in\Phi_{m,n}}\alpha_{S}\left\vert
S\right\rangle ,
\]
where the $\alpha_{S}$'s are complex numbers satisfying $\sum_{S\in\Phi_{m,n}%
}\left\vert \alpha_{S}\right\vert ^{2}=1$. \ In other words, $\left\vert
\psi\right\rangle $\ is a unit vector in the $M$-dimensional complex Hilbert
space spanned by elements of $\Phi_{m,n}$. \ Call this Hilbert space $H_{m,n}$.

Just like in standard quantum computing, the Hilbert space $H_{m,n}$\ is
exponentially large (as a function of $m+n$), which means that we can only
hope to explore a tiny fraction of it using polynomial-size circuits. \ On the
other hand, one difference from standard quantum computing is that $H_{m,n}%
$\ is \textit{not} built up as the tensor product of smaller Hilbert spaces.

Throughout this paper, we will assume that our computer starts in the state%
\[
\left\vert 1_{n}\right\rangle :=\left\vert 1,\ldots,1,0,\ldots,0\right\rangle
,
\]
where the first $n$ modes contain one photon each, and the remaining
$m-n$\ modes are unoccupied. \ We call $\left\vert 1_{n}\right\rangle $\ the
\textit{standard initial state}.

We will also assume that measurement only occurs at the end of the
computation, and that what is measured is the number of photons in each mode.
\ In other words, a measurement of the state $\left\vert \psi\right\rangle
=\sum_{S\in\Phi_{m,n}}\alpha_{S}\left\vert S\right\rangle $ returns an element
$S$ of $\Phi_{m,n}$,\ with probability equal to%
\[
\Pr\left[  S\right]  =\left\vert \alpha_{S}\right\vert ^{2}=\left\vert
\left\langle \psi|S\right\rangle \right\vert ^{2}.
\]

But which unitary transformations can we perform on the state $\left\vert
\psi\right\rangle $, after the initialization and before the final
measurement? \ For simplicity, let us consider the special case where there is
only one photon; later we will generalize to $n$ photons. \ In the one-photon
case, the Hilbert space $H_{m,1}$\ has dimension $M=m$, and the computational
basis states ($\left\vert 1,0,\ldots,0\right\rangle $, $\left\vert
0,1,0,\ldots,0\right\rangle $, etc.) simply record which mode the photon is
in. \ Thus, a general state is just a unit vector in $\mathbb{C}^{m}$: that
is, a superposition over modes.

In standard quantum computing, we know that any unitary transformation on $n$
qubits can be decomposed as a product of \textit{gates}, each of which acts
nontrivially on at most two qubits, and is the identity on the other qubits.
\ Likewise, in the linear optics model, any unitary transformation on $m$
modes can be decomposed into a product of \textit{optical elements}, each of
which acts nontrivially on at most two modes, and is the identity on the other
$m-2$\ modes. \ The two best-known optical elements are called
\textit{phaseshifters} and \textit{beamsplitters}. \ A phaseshifter multiplies
a single amplitude $\alpha_{S}$ by $e^{i\theta}$, for some specified angle
$\theta$, and acts as the identity on the other $m-1$\ amplitudes. \ A
beamsplitter modifies two amplitudes $\alpha_{S}$\ and $\alpha_{T}$\ as
follows, for some specified angle $\theta$:%
\[
\left(
\begin{array}
[c]{c}%
\alpha_{S}^{\prime}\\
\alpha_{T}^{\prime}%
\end{array}
\right)  :=\left(
\begin{array}
[c]{cc}%
\cos\theta & -\sin\theta\\
\sin\theta & \cos\theta
\end{array}
\right)  \left(
\begin{array}
[c]{c}%
\alpha_{S}\\
\alpha_{T}%
\end{array}
\right)  .
\]
It acts as the identity on the other $m-2$\ amplitudes. \ It is easy to see
that beamsplitters and phaseshifters generate all optical elements (that is,
all $2\times2$\ unitaries). \ Moreover, the optical elements generate all
$m\times m$\ unitaries, as shown by the following lemma of Reck et
al.\ \cite{rzbb}:

\begin{lemma}
[Reck et al.\ \cite{rzbb}]\label{decompose}Let $U$ be any $m\times m$\ unitary
matrix. \ Then one can decompose $U$ as a product $U=U_{T}\cdots U_{1}$, where
each $U_{t}$\ is an optical element\ (that is, a unitary matrix that acts
nontrivially on at most $2$ modes and as the identity on the remaining $m-2$
modes). \ \ Furthermore, this decomposition has size $T=O\left(  m^{2}\right)
$, and can be found in time polynomial in $m$.
\end{lemma}

\begin{proof}
[Proof Sketch]The task is to produce $U$ starting from the identity
matrix---or equivalently, to produce $I$\ starting from $U$---by successively
multiplying by block-diagonal unitary matrices, each of which contains a
single $2\times2$\ block\ and $m-2$\ blocks consisting of $1$.\footnote{Such
matrices are the generalizations of the so-called \textit{Givens rotations} to
the complex numbers.} \ To do so, we use a procedure similar to Gaussian
elimination, which zeroes out the $m^{2}-m$\ off-diagonal entries of $U$ one
by one. \ Then, once $U$ has been reduced to a diagonal matrix, we use $m$
phaseshifters to produce the identity matrix.
\end{proof}

We now come to the more interesting part: how do we describe the action of an
optical element on \textit{multiple} photons? \ In the case of a phaseshifter,
it is relatively obvious what should happen. \ Namely, phaseshifting the
$i^{th}$ mode by angle $\theta$\ should multiply the amplitude by $e^{i\theta
}$ once for each of the $s_{i}$\ photons in mode $i$. \ In other words, it
should effect the diagonal unitary transformation%
\[
\left\vert s_{1},\ldots,s_{m}\right\rangle \rightarrow e^{i\theta s_{i}%
}\left\vert s_{1},\ldots,s_{m}\right\rangle .
\]
However, it is much less obvious how to describe the action of a beamsplitter
on multiple photons.

As it turns out, there is a natural homomorphism $\varphi$, which maps an
$m\times m$\ unitary transformation\ $U$\ acting on a single photon to the
corresponding $M\times M$\ unitary transformation\ $\varphi\left(  U\right)
$\ acting on $n$ photons. \ Since $\varphi$ is a homomorphism, Lemma
\ref{decompose} implies that we can specify $\varphi$\ merely by describing
its behavior on $2\times2$\ unitaries. \ For given an arbitrary $m\times
m$\ unitary matrix $U$,\ we can write $\varphi\left(  U\right)  $\ as%
\[
\varphi\left(  U_{T}\cdots U_{1}\right)  =\varphi\left(  U_{T}\right)
\cdots\varphi\left(  U_{1}\right)  ,
\]
where each $U_{t}$ is an optical element (that is, a block-diagonal unitary
that acts nontrivially on at most $2$ modes). \ So let%
\[
U=\left(
\begin{array}
[c]{cc}%
a & b\\
c & d
\end{array}
\right)
\]
be any $2\times2$\ unitary matrix, which acts on the Hilbert space $H_{2,1}%
$\ spanned by $\left\vert 1,0\right\rangle $\ and $\left\vert 0,1\right\rangle
$. \ Then since $\varphi\left(  U\right)  $\ preserves photon number, we know
it must be a block-diagonal matrix that satisfies%
\[
\left\langle s,t\right\vert \varphi\left(  U\right)  \left\vert
u,v\right\rangle =0
\]
whenever $s+t\neq u+v$. \ But what about when $s+t=u+v$? \ Here the formula
for the appropriate entry of $\varphi\left(  U\right)  $\ is%
\begin{equation}
\left\langle s,t\right\vert \varphi\left(  U\right)  \left\vert
u,v\right\rangle =\sqrt{\frac{u!v!}{s!t!}}\sum_{k+\ell=u,~k\leq s,~\ell\leq
t}\binom{s}{k}\binom{t}{\ell}a^{k}b^{s-k}c^{\ell}d^{t-\ell}. \label{gate}%
\end{equation}
One can verify by calculation that $\varphi\left(  U\right)  $ is
unitary;\ however, a much more elegant proof of unitarity will follow from the
results in\ Section \ref{POLYDEF}.

One more piece of notation: let $\mathcal{D}_{U}$\ be the probability
distribution over $S\in\Phi_{m,n}$\ obtained by measuring the state
$\varphi\left(  U\right)  \left\vert 1_{n}\right\rangle $\ in the
computational basis. \ That is,%
\[
\Pr_{\mathcal{D}_{U}}\left[  S\right]  =\left\vert \left\langle 1_{n}%
|\varphi\left(  U\right)  |S\right\rangle \right\vert ^{2}.
\]
Notice that $\mathcal{D}_{U}$\ depends only on the first $n$ columns of $U$.
\ Therefore, instead of writing $\mathcal{D}_{U}$\ it will be better to write
$\mathcal{D}_{A}$, where $A\in\mathcal{U}_{m,n}$\ is the $m\times n$\ matrix
corresponding to the first $n$ columns of $U$.

\subsection{Polynomial Definition\label{POLYDEF}}

In this section, we present a beautiful alternative interpretation of the
noninteracting-boson model, in which the \textquotedblleft
states\textquotedblright\ are multivariate polynomials, the \textquotedblleft
operations\textquotedblright\ are unitary changes of variable, and a
\textquotedblleft measurement\textquotedblright\ samples from a probability
distribution over monomials weighted by their coefficients. \ We also prove
that this model is well-defined (i.e.\ that in any measurement, the
probabilities of the various outcomes sum to $1$), and that it is indeed
equivalent to the model from Section \ref{PHYSDEF}. \ Combining these facts
yields the simplest proof we know that the model from Section \ref{PHYSDEF}%
\ is well-defined.

Let $m\geq n$. \ Then the \textquotedblleft state\textquotedblright\ of our
computer, at any time, will be represented by a multivariate complex-valued
polynomial $p\left(  x_{1},\ldots.x_{m}\right)  $ of degree $n$. \ Here
the\ $x_{i}$'s can be thought of as just formal variables.\footnote{For
physicists, they are \textquotedblleft creation operators.\textquotedblright}
\ The standard initial state $\left\vert 1_{n}\right\rangle $\ corresponds to
the degree-$n$ polynomial $J_{m,n}\left(  x_{1},\ldots,x_{m}\right)
:=x_{1}\cdots x_{n}$, where $x_{1},\ldots,x_{n}$\ are the first $n$ variables.
\ To transform the state, we can apply any $m\times m$\ unitary transformation
$U$ we like to the vector of $x_{i}$'s:%
\[
\left(
\begin{array}
[c]{c}%
x_{1}^{\prime}\\
\vdots\\
x_{m}^{\prime}%
\end{array}
\right)  =\left(
\begin{array}
[c]{ccc}%
u_{11} & \cdots & u_{1m}\\
\vdots & \ddots & \vdots\\
u_{m1} & \cdots & u_{mm}%
\end{array}
\right)  \left(
\begin{array}
[c]{c}%
x_{1}\\
\vdots\\
x_{m}%
\end{array}
\right)  .
\]
The new state of our computer is then equal to%
\[
U\left[  J_{m,n}\right]  \left(  x_{1},\ldots.x_{m}\right)  =J_{m,n}\left(
x_{1}^{\prime},\ldots.x_{m}^{\prime}\right)  =\prod_{i=1}^{n}\left(
u_{i1}x_{1}+\cdots+u_{im}x_{m}\right)  .
\]
Here and throughout, we let $L\left[  p\right]  $\ be the polynomial obtained
by starting with $p$ and then applying the $m\times m$\ linear transformation
$L$ to the variables.

After applying one or more unitary transformations to the $x_{i}$'s, we then
get a single opportunity to measure the computer's state. \ Let the polynomial
$p$ at the time of measurement be%
\[
p\left(  x_{1},\ldots.x_{m}\right)  =\sum_{S=\left(  s_{1},\ldots
,s_{m}\right)  }a_{S}x_{1}^{s_{1}}\cdots x_{m}^{s_{m}},
\]
where $S$\ ranges over $\Phi_{m,n}$\ (i.e., lists of nonnegative integers such
that $s_{1}+\cdots+s_{m}=n$). \ Then the measurement returns the monomial
$x_{1}^{s_{1}}\cdots x_{m}^{s_{m}}$ (or equivalently, the list of integers
$S=\left(  s_{1},\ldots,s_{m}\right)  $) with probability equal to%
\[
\Pr\left[  S\right]  :=\left\vert a_{S}\right\vert ^{2}s_{1}!\cdots s_{m}!.
\]

From now on, we will use $x$\ as shorthand for $x_{1},\ldots.x_{m}$, and
$x^{S}$\ as shorthand for the monomial $x_{1}^{s_{1}}\cdots x_{m}^{s_{m}}$.
\ Given two polynomials%
\begin{align*}
p\left(  x\right)   &  =\sum_{S\in\Phi_{m,n}}a_{S}x^{S},\\
q\left(  x\right)   &  =\sum_{S\in\Phi_{m,n}}b_{S}x^{S},
\end{align*}
we can define an inner product between them---the so-called \textit{Fock-space
inner product}---as follows:%
\[
\left\langle p,q\right\rangle :=\sum_{S=\left(  s_{1},\ldots,s_{m}\right)
\in\Phi_{m,n}}\overline{a}_{S}b_{S}s_{1}!\cdots s_{m}!.
\]
The following key result gives a more intuitive interpretation of the
Fock-space inner product.

\begin{lemma}
[Interpretation of Fock Inner Product]\label{innerprod}$\left\langle
p,q\right\rangle =\operatorname*{E}_{x\sim\mathcal{G}^{m}}\left[  \overline
{p}\left(  x\right)  q\left(  x\right)  \right]  $, where $\mathcal{G}$\ is
the Gaussian distribution $\mathcal{N}\left(  0,1\right)  _{\mathbb{C}}$.
\end{lemma}

\begin{proof}
Since inner product and expectation are linear, it suffices to consider the
case where $p$\ and $q$ are monomials. \ Suppose $p\left(  x\right)  =x^{R}%
$\ and $q\left(  x\right)  =x^{S}$, for some $R=\left(  r_{1},\ldots
,r_{m}\right)  $ and $S=\left(  s_{1},\ldots,s_{m}\right)  $ in $\Phi_{m,n}$.
\ Then%
\[
\operatorname*{E}_{x\sim\mathcal{G}^{m}}\left[  \overline{p}\left(  x\right)
q\left(  x\right)  \right]  =\operatorname*{E}_{x\sim\mathcal{G}^{m}}\left[
\overline{x}^{R}x^{S}\right]  .
\]
If $p\neq q$---that is, if there exists an $i$ such that $r_{i}\neq s_{i}%
$---then the above expectation is clearly $0$, since the Gaussian distribution
is uniform over phases. \ If $p=q$, on the other hand, then the expectation
equals%
\begin{align*}
\operatorname*{E}_{x\sim\mathcal{G}^{m}}\left[  \left\vert x_{1}\right\vert
^{2s_{1}}\cdots\left\vert x_{m}\right\vert ^{2s_{m}}\right]   &
=\operatorname*{E}_{x_{1}\sim\mathcal{G}}\left[  \left\vert x_{1}\right\vert
^{2s_{1}}\right]  \cdots\operatorname*{E}_{x_{m}\sim\mathcal{G}}\left[
\left\vert x_{m}\right\vert ^{2s_{m}}\right] \\
&  =s_{1}!\cdots s_{m}!
\end{align*}
We conclude that%
\[
\operatorname*{E}_{x\sim\mathcal{G}^{m}}\left[  \overline{p}\left(  x\right)
q\left(  x\right)  \right]  =\sum_{S=\left(  s_{1},\ldots,s_{m}\right)
\in\Phi_{m,n}}\overline{a}_{S}b_{S}s_{1}!\cdots s_{m}!
\]
as desired.
\end{proof}

Recall that $U\left[  p\right]  $\ denotes the polynomial $p\left(  Ux\right)
$, obtained by applying the $m\times m$\ linear transformation $U$\ to the
variables $x=\left(  x_{1},\ldots,x_{m}\right)  $\ of $p$. \ Then Lemma
\ref{innerprod}\ has the following important consequence.

\begin{theorem}
[Unitary Invariance of Fock Inner Product]\label{invarthm}$\left\langle
p,q\right\rangle =\left\langle U\left[  p\right]  ,U\left[  q\right]
\right\rangle $ for all polynomials $p,q$\ and all unitary transformations $U$.
\end{theorem}

\begin{proof}
We have%
\begin{align*}
\left\langle U\left[  p\right]  ,U\left[  q\right]  \right\rangle  &
=\operatorname*{E}_{x\sim\mathcal{G}^{m}}\left[  \overline{U\left[  p\right]
}\left(  x\right)  U\left[  q\right]  \left(  x\right)  \right] \\
&  =\operatorname*{E}_{x\sim\mathcal{G}^{m}}\left[  \overline{p}\left(
Ux\right)  q\left(  Ux\right)  \right] \\
&  =\operatorname*{E}_{x\sim\mathcal{G}^{m}}\left[  \overline{p}\left(
x\right)  q\left(  x\right)  \right] \\
&  =\left\langle p,q\right\rangle ,
\end{align*}
where the third line follows from the rotational invariance of the Gaussian distribution.
\end{proof}

Indeed, we have a more general result:

\begin{theorem}
\label{invarthm2}$\left\langle p,L\left[  q\right]  \right\rangle
=\left\langle L^{\dagger}\left[  p\right]  ,q\right\rangle $ for all
polynomials $p,q$\ and all linear transformations $L$. \ (So in particular, if
$L$ is invertible, then $\left\langle p,q\right\rangle =\left\langle
L^{-\dagger}\left[  p\right]  ,L\left[  q\right]  \right\rangle $.)
\end{theorem}

\begin{proof}
Let $p\left(  x\right)  =\sum_{S\in\Phi_{m,n}}a_{S}x^{S}$\ and $q\left(
x\right)  =\sum_{S\in\Phi_{m,n}}b_{S}x^{S}$. \ First suppose $L$ is a diagonal
matrix, i.e.\ $L=\operatorname*{diag}\left(  \lambda\right)  $\ for some
$\lambda=\left(  \lambda_{1},\ldots,\lambda_{m}\right)  $. \ Then%
\begin{align*}
\left\langle p,L\left[  q\right]  \right\rangle  &  =\sum_{S=\left(
s_{1},\ldots,s_{m}\right)  \in\Phi_{m,n}}\overline{a}_{S}\left(  b_{S}%
\lambda^{S}\right)  s_{1}!\cdots s_{m}!\\
&  =\sum_{S=\left(  s_{1},\ldots,s_{m}\right)  \in\Phi_{m,n}}\left(
\overline{a_{S}\overline{\lambda}^{S}}\right)  b_{S}s_{1}!\cdots s_{m}!\\
&  =\left\langle L^{\dagger}\left[  p\right]  ,q\right\rangle .
\end{align*}
Now note that we can decompose an arbitrary $L$\ as $U\Lambda V$, where
$\Lambda$\ is diagonal and $U,V$\ are unitary. \ So%
\begin{align*}
\left\langle p,L\left[  q\right]  \right\rangle  &  =\left\langle p,U\Lambda
V\left[  q\right]  \right\rangle \\
&  =\left\langle U^{\dagger}\left[  p\right]  ,\Lambda V\left[  q\right]
\right\rangle \\
&  =\left\langle \Lambda^{\dagger}U^{\dagger}\left[  p\right]  ,V\left[
q\right]  \right\rangle \\
&  =\left\langle V^{\dagger}\Lambda^{\dagger}U^{\dagger}\left[  p\right]
,q\right\rangle \\
&  =\left\langle L^{\dagger}\left[  p\right]  ,q\right\rangle
\end{align*}
where the second and fourth lines follow from Theorem \ref{invarthm}.
\end{proof}

We can also define a \textit{Fock-space norm} as follows:%
\[
\left\Vert p\right\Vert _{\operatorname*{Fock}}^{2}=\left\langle
p,p\right\rangle =\sum_{S=\left(  s_{1},\ldots,s_{m}\right)  }\left\vert
a_{S}\right\vert ^{2}s_{1}!\cdots s_{m}!.
\]
Clearly $\left\Vert p\right\Vert _{\operatorname*{Fock}}^{2}\geq0$\ for all
$p$. \ We also have the following:

\begin{corollary}
\label{normcor}$\left\Vert U\left[  J_{m,n}\right]  \right\Vert
_{\operatorname*{Fock}}^{2}=1$ for all unitary matrices $U$.
\end{corollary}

\begin{proof}
By Theorem \ref{invarthm},%
\[
\left\Vert U\left[  J_{m,n}\right]  \right\Vert _{\operatorname*{Fock}}%
^{2}=\left\langle U\left[  J_{m,n}\right]  ,U\left[  J_{m,n}\right]
\right\rangle =\left\langle UU^{\dagger}\left[  J_{m,n}\right]  ,J_{m,n}%
\right\rangle =\left\langle J_{m,n},J_{m,n}\right\rangle =1.
\]

\end{proof}

Corollary \ref{normcor}\ implies, in particular, that our model of computation
based on multivariate polynomials is well-defined: that is, the probabilities
of the various measurement outcomes always sum to $\left\Vert U\left[
J_{m,n}\right]  \right\Vert _{\operatorname*{Fock}}^{2}=1$. \ We now show that
the polynomial-based model of this section is \textit{equivalent} to the
linear-optics model of Section \ref{PHYSDEF}. \ As an immediate consequence,
this implies that probabilities sum to $1$ in the linear-optics model as well.

Given any pure state%
\[
\left\vert \psi\right\rangle =\sum_{S\in\Phi_{m,n}}\alpha_{S}\left\vert
S\right\rangle
\]
in $H_{m,n}$, let $P_{\left\vert \psi\right\rangle }$\ be the multivariate
polynomial defined by%
\[
P_{\left\vert \psi\right\rangle }\left(  x\right)  :=\sum_{S=\left(
s_{1},\ldots,s_{m}\right)  \in\Phi_{m,n}}\frac{\alpha_{S}x^{S}}{\sqrt
{s_{1}!\cdots s_{m}!}}.
\]
In particular, for any computational basis state $\left\vert S\right\rangle $,
we have%
\[
P_{\left\vert S\right\rangle }\left(  x\right)  =\frac{x^{S}}{\sqrt
{s_{1}!\cdots s_{m}!}}.
\]

\begin{theorem}
[Equivalence of Physical and Polynomial Definitions]\label{isomthm}$\left\vert
\psi\right\rangle \longleftrightarrow P_{\left\vert \psi\right\rangle }%
$\ defines an isomorphism between quantum states and polynomials, which
commutes with inner products and unitary transformations in the following
senses:%
\begin{align*}
\left\langle \psi|\phi\right\rangle  &  =\left\langle P_{\left\vert
\psi\right\rangle },P_{\left\vert \phi\right\rangle }\right\rangle ,\\
P_{\varphi\left(  U\right)  \left\vert \psi\right\rangle }  &  =U\left[
P_{\left\vert \psi\right\rangle }\right]  .
\end{align*}

\end{theorem}

\begin{proof}
That $\left\langle \psi|\phi\right\rangle =\left\langle P_{\left\vert
\psi\right\rangle },P_{\left\vert \phi\right\rangle }\right\rangle $\ follows
immediately from the definitions of $P_{\left\vert \psi\right\rangle }$\ and
the Fock-space inner product. \ For $P_{\varphi\left(  U\right)  \left\vert
\psi\right\rangle }=U\left[  P_{\psi}\right]  $,\ notice that%
\begin{align*}
U\left[  P_{\left\vert \psi\right\rangle }\right]   &  =U\left[
\sum_{S=\left(  s_{1},\ldots,s_{m}\right)  \in\Phi_{m,n}}\frac{\alpha_{S}%
x^{S}}{\sqrt{s_{1}!\cdots s_{m}!}}\right] \\
&  =\sum_{S=\left(  s_{1},\ldots,s_{m}\right)  \in\Phi_{m,n}}\frac{\alpha_{S}%
}{\sqrt{s_{1}!\cdots s_{m}!}}\prod_{i=1}^{m}\left(  u_{i1}x_{1}+\cdots
+u_{im}x_{m}\right)  ^{s_{i}}.
\end{align*}
So in particular, transforming $P_{\left\vert \psi\right\rangle }$\ to
$U\left[  P_{\left\vert \psi\right\rangle }\right]  $\ simply effects a linear
transformation on the coefficients on $P_{\left\vert \psi\right\rangle }$.
\ This means that there must be \textit{some} $M\times M$\ linear
transformation $\varphi\left(  U\right)  $, depending on $U$, such that
$U\left[  P_{\left\vert \psi\right\rangle }\right]  =P_{\varphi\left(
U\right)  \left\vert \psi\right\rangle }$. \ Thus, in defining the
homomorphism $U\rightarrow\varphi\left(  U\right)  $\ in equation
(\ref{gate}), we simply chose it to yield that linear transformation. \ This
can be checked by explicit computation.\ \ By Lemma \ref{decompose}, we can
restrict attention to a $2\times2$\ unitary matrix%
\[
U=\left(
\begin{array}
[c]{cc}%
a & b\\
c & d
\end{array}
\right)  .
\]
By linearity, we can also restrict attention to the action of $\varphi\left(
U\right)  $\ on a computational basis state $\left\vert s,t\right\rangle $ (or
in the polynomial formalism, the action of $U$ on a monomial $x^{s}y^{t}$).
\ Then%
\begin{align*}
U\left[  x^{s}y^{t}\right]   &  =\left(  ax+by\right)  ^{s}\left(
cx+dy\right)  ^{t}\\
&  =\sum_{k=0}^{s}\sum_{\ell=0}^{t}\binom{s}{k}\binom{t}{\ell}a^{k}%
b^{s-k}c^{\ell}d^{t-\ell}x^{k+\ell}y^{s+t-k-\ell}\\
&  =\sum_{u+v=s+t}\sum_{k+\ell=u,~k\leq s,~\ell\leq t}\binom{s}{k}\binom
{t}{\ell}a^{k}b^{s-k}c^{\ell}d^{t-\ell}x^{u}y^{v}.
\end{align*}
Thus, inserting normalization,%
\[
U\left[  \frac{x^{s}y^{t}}{\sqrt{s!t!}}\right]  =\sum_{u+v=s+t}\left(
\sqrt{\frac{u!v!}{s!t!}}\sum_{k+\ell=u,~k\leq s,~\ell\leq t}\binom{s}{k}%
\binom{t}{\ell}a^{k}b^{s-k}c^{\ell}d^{t-\ell}\right)  \frac{x^{u}y^{v}}%
{\sqrt{u!v!}},
\]
which yields precisely the definition of $\varphi\left(  U\right)  $\ from
equation (\ref{gate}).
\end{proof}

As promised in Section \ref{PHYSDEF}, we can also show that $\varphi\left(
U\right)  $\ is unitary.

\begin{corollary}
\label{unitarycor}$\varphi\left(  U\right)  $\ is unitary.
\end{corollary}

\begin{proof}
One definition of a unitary matrix is that it preserves inner products. \ Let
us check that this is the case for $\varphi\left(  U\right)  $. \ For all $U$,
we have%
\begin{align*}
\left\langle \psi|\phi\right\rangle  &  =\left\langle P_{\left\vert
\psi\right\rangle },P_{\left\vert \phi\right\rangle }\right\rangle \\
&  =\left\langle U\left[  P_{\left\vert \psi\right\rangle }\right]  ,U\left[
P_{\left\vert \phi\right\rangle }\right]  \right\rangle \\
&  =\left\langle P_{\varphi\left(  U\right)  \left\vert \psi\right\rangle
},P_{\varphi\left(  U\right)  \left\vert \phi\right\rangle }\right\rangle \\
&  =\left\langle \psi\right\vert \varphi\left(  U\right)  ^{\dagger}%
\varphi\left(  U\right)  \left\vert \phi\right\rangle
\end{align*}
where the second line follows from Theorem \ref{invarthm}, and all other lines
from Theorem \ref{isomthm}.
\end{proof}

\subsection{Permanent Definition\label{PERMDEF}}

This section gives a \textit{third} interpretation of the noninteracting-boson
model, which makes clear its connection to the permanent. \ Given an $n\times
n$\ matrix $A=\left(  a_{ij}\right)  \in\mathbb{C}^{n\times n}$, recall that
the permanent is%
\[
\operatorname*{Per}\left(  A\right)  =\sum_{\sigma\in S_{n}}\prod_{i=1}%
^{n}a_{i,\sigma\left(  i\right)  }.
\]
Also, given an $m\times m$\ matrix $V$, let $V_{n,n}$\ be the top-left
$n\times n$\ submatrix of $V$. \ Then the following lemma establishes a direct
connection between $\operatorname*{Per}\left(  V_{n,n}\right)  $ and the
Fock-space inner product defined in Section \ref{POLYDEF}.

\begin{lemma}
\label{perlem1}$\operatorname*{Per}\left(  V_{n,n}\right)  =\left\langle
J_{m,n},V\left[  J_{m,n}\right]  \right\rangle $ for any $m\times
m$\ matrix\ $V$.
\end{lemma}

\begin{proof}
By definition,%
\[
V\left[  J_{m,n}\right]  =\prod_{i=1}^{n}\left(  v_{i1}x_{1}+\cdots
+v_{im}x_{m}\right)  .
\]
Then $\left\langle J_{m,n},V\left[  J_{m,n}\right]  \right\rangle $\ is just
the coefficient of $J_{m,n}=x_{1}\cdots x_{n}$\ in the above polynomial.
\ This coefficient can be calculated as%
\[
\sum_{\sigma\in S_{n}}\prod_{i=1}^{n}v_{i,\sigma\left(  i\right)
}=\operatorname*{Per}\left(  V_{n,n}\right)  .
\]

\end{proof}

Combining Lemma \ref{perlem1}\ with Theorem \ref{invarthm2}, we immediately
obtain the following:

\begin{corollary}
\label{percor}$\operatorname*{Per}\left(  \left(  V^{\dagger}W\right)
_{n,n}\right)  =\left\langle V\left[  J_{m,n}\right]  ,W\left[  J_{m,n}%
\right]  \right\rangle $ for any two matrices $V,W\in\mathbb{C}^{m\times m}$.
\end{corollary}

\begin{proof}%
\[
\operatorname*{Per}\left(  \left(  V^{\dagger}W\right)  _{n,n}\right)
=\left\langle J_{m,n},V^{\dagger}W\left[  J_{m,n}\right]  \right\rangle
=\left\langle V\left[  J_{m,n}\right]  ,W\left[  J_{m,n}\right]  \right\rangle
.
\]

\end{proof}

Now let $U$\ be any $m\times m$\ unitary matrix, and let $S=\left(
s_{1},\ldots,s_{m}\right)  $\ and $T=\left(  t_{1},\ldots,t_{m}\right)  $ be
any two computational basis states (that is, elements of $\Phi_{m,n}$). \ Then
we define an $n\times n$\ matrix $U_{S,T}$\ in the following manner. \ First
form an $m\times n$\ matrix $U_{T}$ by taking $t_{j}$\ copies of the $j^{th}%
$\ column of $U$,\ for each $j\in\left[  m\right]  $. \ Then form the $n\times
n$\ matrix\ $U_{S,T}$ by taking $s_{i}$\ copies of the $i^{th}$\ row of
$U_{T}$, for each $i\in\left[  m\right]  $. \ As an example, suppose%
\[
U=\left(
\begin{array}
[c]{ccc}%
0 & 1 & 0\\
1 & 0 & 0\\
0 & 0 & -1
\end{array}
\right)
\]
and $S=T=\left(  0,1,2\right)  $. \ Then%
\[
U_{S,T}=\left(
\begin{array}
[c]{ccc}%
0 & 0 & 0\\
0 & -1 & -1\\
0 & -1 & -1
\end{array}
\right)  .
\]
Note that if the $s_{i}$'s and $t_{j}$'s are all $0$ or $1$, then $U_{S,T}%
$\ is simply an $n\times n$ submatrix of $U$. \ If some $s_{i}$'s or $t_{j}$'s
are greater than $1$, then $U_{S,T}$\ is like\ a submatrix of $U$, but with
repeated rows and/or columns.

Here is an alternative way to define $U_{S,T}$. \ Given any $S\in\Phi_{m,n}$,
let $I_{S}$\ be a linear substitution of variables, which maps the variables
$x_{1},\ldots,x_{s_{1}}$\ to $x_{1}$, the variables $x_{s_{1}+1}%
,\ldots,x_{s_{1}+s_{2}}$\ to $x_{2}$,\ and so on, so that $I_{S}\left[
x_{1}\cdots x_{n}\right]  =x_{1}^{s_{1}}\cdots x_{m}^{s_{m}}$. \ (If $i>n$,
then $I_{S}\left[  x_{i}\right]  =0$.) \ Then one can check that%
\[
U_{S,T}=\left(  I_{S}^{\dagger}UI_{T}\right)  _{n,n}.
\]
(Note also that $\varphi\left(  I_{S}\right)  \left\vert 1_{n}\right\rangle
=\left\vert S\right\rangle $.)

\begin{theorem}
[Equivalence of All Three Definitions]\label{perust}For all $m\times
m$\ unitaries $U$ and basis states $S,T\in\Phi_{m,n}$,%
\[
\operatorname*{Per}\left(  U_{S,T}\right)  =\left\langle x^{S},U\left[
x^{T}\right]  \right\rangle =\left\langle S|\varphi\left(  U\right)
|T\right\rangle \sqrt{s_{1}!\cdots s_{m}!t_{1}!\cdots t_{m}!}%
\]

\end{theorem}

\begin{proof}
For the first equality, from Corollary \ref{percor} we have%
\begin{align*}
\left\langle x^{S},U\left[  x^{T}\right]  \right\rangle  &  =\left\langle
I_{S}\left[  J_{m,n}\right]  ,UI_{T}\left[  J_{m,n}\right]  \right\rangle \\
&  =\operatorname*{Per}\left(  \left(  I_{S}^{\dagger}UI_{T}\right)
_{n,n}\right) \\
&  =\operatorname*{Per}\left(  U_{S,T}\right)  .
\end{align*}
For the second equality, from Theorem \ref{isomthm} we have%
\begin{align*}
\left\langle S|\varphi\left(  U\right)  |T\right\rangle  &  =\left\langle
P_{\left\vert S\right\rangle },P_{\varphi\left(  U\right)  \left\vert
T\right\rangle }\right\rangle \\
&  =\left\langle P_{\left\vert S\right\rangle },U\left[  P_{\left\vert
T\right\rangle }\right]  \right\rangle \\
&  =\frac{\left\langle x^{S},U\left[  x^{T}\right]  \right\rangle }%
{\sqrt{s_{1}!\cdots s_{m}!t_{1}!\cdots t_{m}!}}.
\end{align*}

\end{proof}

\subsection{Bosonic Complexity Theory\label{BCT}}

Having presented the noninteracting-boson model from three perspectives, we
are finally ready to define \textsc{BosonSampling}, the central computational
problem considered in this work. \ The input to the problem will be an
$m\times n$\ column-orthonormal matrix $A\in\mathcal{U}_{m,n}$.\footnote{Here
we assume each entry of $A$ is represented in binary, so that it has the form
$\left(  x+yi\right)  /2^{p\left(  n\right)  }$, where $x$ and $y$ are
integers and $p$ is some fixed polynomial. \ As a consequence, $A$ might not
be \textit{exactly} column-orthonormal---but as long as $A^{\dag}A$\ is
exponentially close to the identity, $A$\ can easily be \textquotedblleft
corrected\textquotedblright\ to an element of $\mathcal{U}_{m,n}$\ using
Gram-Schmidt orthogonalization. \ Furthermore, it is not hard to show that
every element of $\mathcal{U}_{m,n}$\ can be approximated in this manner.
\ See for example Aaronson\ \cite{aar:bf} for a detailed error analysis.}
\ Given $A$, together with a basis state $S\in\Phi_{m,n}$---that is, a list
$S=\left(  s_{1},\ldots,s_{m}\right)  $\ of nonnegative integers, satisfying
$s_{1}+\cdots+s_{m}=n$---let $A_{S}$\ be the $n\times n$\ matrix obtained by
taking $s_{i}$\ copies of the $i^{th}$\ row of $A$, for all $i\in\left[
m\right]  $. \ Then let $\mathcal{D}_{A}$\ be the probability distribution
over $\Phi_{m,n}$ defined as follows:%
\[
\Pr_{\mathcal{D}_{A}}\left[  S\right]  =\frac{\left\vert \operatorname*{Per}%
\left(  A_{S}\right)  \right\vert ^{2}}{s_{1}!\cdots s_{m}!}.
\]
(Theorem \ref{perust} implies that $\mathcal{D}_{A}$\ is indeed a probability
distribution, for every $A\in\mathcal{U}_{m,n}$.) \ The goal of
\textsc{BosonSampling} is to sample either exactly or approximately from
$\mathcal{D}_{A}$, given $A$ as input.

Of course, we \textit{also} could have defined $\mathcal{D}_{A}$\ as the
distribution over $\Phi_{m,n}$\ obtained by first completing $A$ to any
$m\times m$\ unitary matrix $U$, then measuring the quantum state
$\varphi\left(  U\right)  \left\vert 1_{n}\right\rangle $\ in the
computational basis. \ Or we could have defined $\mathcal{D}_{A}$\ as the
distribution obtained by first applying the linear change of variables $U$ to
the polynomial $x_{1}\cdots x_{n}$ (where again $U$ is any $m\times
m$\ unitary completion of $A$), to obtain a new $m$-variable polynomial%
\[
U\left[  x_{1}\cdots x_{n}\right]  =\sum_{S\in\Phi_{m,n}}\alpha_{S}x^{S},
\]
and then letting%
\[
\Pr_{\mathcal{D}_{A}}\left[  S\right]  =\left\vert \alpha_{S}\right\vert
^{2}s_{1}!\cdots s_{m}!=\frac{\left\vert \left\langle x^{S},U\left[
x_{1}\cdots x_{n}\right]  \right\rangle \right\vert ^{2}}{s_{1}!\cdots s_{m}%
!}.
\]
For most of the paper, though, we will find it most convenient to use the
definition of $\mathcal{D}_{A}$\ in terms of permanents.

Besides the \textsc{BosonSampling}\ problem, we will also need the concept of
an exact or approximate\ \textsc{BosonSampling}\ \textit{oracle}.
\ Intuitively, a \textsc{BosonSampling}\ oracle is simply an oracle
$\mathcal{O}$\ that solves the \textsc{BosonSampling}\ problem: that is,
$\mathcal{O}$ takes as input a matrix $A\in\mathcal{U}_{m,n}$, and outputs a
sample from $\mathcal{D}_{A}$. \ However, there is a subtlety, arising from
the fact that $\mathcal{O}$\ is an oracle for a \textit{sampling} problem.
\ Namely, it is essential that\ $\mathcal{O}$'s \textit{only} source of random
randomness be a string $r\in\left\{  0,1\right\}  ^{\operatorname*{poly}%
\left(  n\right)  }$ that is also given to $\mathcal{O}$\ as input. \ In other
words, if we fix $r$, then $\mathcal{O}\left(  A,r\right)  $\ must be
deterministic, just like a conventional oracle that decides a language. \ Of
course, if $\mathcal{O}$\ were implemented by a classical algorithm, this
requirement would be trivial to satisfy.

More formally:

\begin{definition}
[\textsc{BosonSampling} oracle]\label{bosonoracle}Let $\mathcal{O}$ be an
oracle that takes as input a string $r\in\left\{  0,1\right\}
^{\operatorname*{poly}\left(  n\right)  }$, an $m\times n$\ matrix\ $A\in
\mathcal{U}_{m,n}$, and an error bound $\varepsilon>0$\ encoded as
$0^{1/\varepsilon}$. \ Also, let $\mathcal{D}_{\mathcal{O}}\left(
A,\varepsilon\right)  $\ be the distribution over outputs of $\mathcal{O}$\ if
$A$\ and $\varepsilon$\ are fixed but $r$\ is uniformly random. \ We call
$\mathcal{O}$\ an exact \textsc{BosonSampling}\ oracle if $\mathcal{D}%
_{\mathcal{O}}\left(  A,\varepsilon\right)  =\mathcal{D}_{A}$ for
all\ $A\in\mathcal{U}_{m,n}$. \ Also, we call $\mathcal{O}$\ an approximate
\textsc{BosonSampling}\ oracle if $\left\Vert \mathcal{D}_{\mathcal{O}}\left(
A,\varepsilon\right)  -\mathcal{D}_{A}\right\Vert \leq\varepsilon$\ for all
$A\in\mathcal{U}_{m,n}$\ and $\varepsilon>0$.
\end{definition}

If we like, we can define the complexity class $\mathsf{BosonFP}$, to be the
set of search problems $R=\left(  B_{x}\right)  _{x\in\left\{  0,1\right\}
^{\ast}}$\ that are in $\mathsf{FBPP}^{\mathcal{O}}$\ for every exact
\textsc{BosonSampling}\ oracle\ $\mathcal{O}$. \ We can also define
$\mathsf{BosonFP}_{\varepsilon}$\ to be the set of search problems that are in
$\mathsf{FBPP}^{\mathcal{O}}$\ for every approximate \textsc{BosonSampling}%
\ oracle\ $\mathcal{O}$. \ We then have the following basic inclusions:

\begin{theorem}
\label{infbqp}$\mathsf{FBPP}\subseteq\mathsf{BosonFP}_{\varepsilon
}=\mathsf{BosonFP}\subseteq\mathsf{FBQP}$.
\end{theorem}

\begin{proof}
For $\mathsf{FBPP}\subseteq\mathsf{BosonFP}_{\varepsilon}$, just ignore the
\textsc{BosonSampling} oracle. \ For $\mathsf{BosonFP}_{\varepsilon}%
\subseteq\mathsf{BosonFP}$, note that any exact \textsc{BosonSampling}\ oracle
is also an $\varepsilon$-approximate one for every $\varepsilon$. \ For the
other direction, $\mathsf{BosonFP}\subseteq\mathsf{BosonFP}_{\varepsilon}$,
let $M$\ be a $\mathsf{BosonFP}$ machine, and let $\mathcal{O}$\ be $M$'s
exact \textsc{BosonSampling}\ oracle. \ Since $M$ has to work for
\textit{every} $\mathcal{O}$, we can assume without loss of generality that
$\mathcal{O}$\ is chosen uniformly at random, consistent with the requirement
that $\mathcal{D}_{\mathcal{O}}\left(  A\right)  =\mathcal{D}_{A}$\ for every
$A$. \ We claim that we can simulate $\mathcal{O}$ to sufficient accuracy
using an \textit{approximate} \textsc{BosonSampling}\ oracle. \ To do so, we
simply choose $\varepsilon\ll\delta/p\left(  n\right)  $, where $p\left(
n\right)  $\ is an upper bound on the number of queries to $\mathcal{O}$\ made
by $M$, and $\delta$\ is the desired failure probability of $M$.

For $\mathsf{BosonFP}\subseteq\mathsf{FBQP}$, we use an old observation of
Feynman \cite{feynman:qc}\ and Abrams and Lloyd \cite{al:fermi}: that
fermionic and bosonic systems can be simulated efficiently on a standard
quantum computer. \ In more detail, our quantum computer's state at any time
step will have the form%
\[
\left\vert \psi\right\rangle =\sum_{\left(  s_{1},\ldots,s_{m}\right)  \in
\Phi_{m,n}}\alpha_{s_{1},\ldots,s_{m}}\left\vert s_{1},\ldots,s_{m}%
\right\rangle .
\]
That is, we simply encode each occupation number $0\leq s_{i}\leq n$ in binary
using $\left\lceil \log_{2}n\right\rceil $\ qubits. \ (Thus, the total number
of qubits in our simulation is $m\left\lceil \log_{2}n\right\rceil $.) \ To
initialize, we prepare the state $\left\vert 1_{n}\right\rangle =\left\vert
1,\ldots,1,0,\ldots,0\right\rangle $; to measure, we measure in the
computational basis. \ As for simulating an optical element: recall that such
an element acts nontrivially only on two modes $i$ and $j$, and hence on
$2\left\lceil \log_{2}n\right\rceil $\ qubits. \ So we can describe an optical
element by an\ $O\left(  n^{2}\right)  \times O\left(  n^{2}\right)
$\ unitary matrix $U$---and\ furthermore, we gave an explicit formula
(\ref{gate}) for the entries of $U$. \ It follows immediately, from the
Solovay-Kitaev Theorem (see \cite{nc}), that we can simulate $U$ with error
$\varepsilon$, using $\operatorname*{poly}\left(  n,\log1/\varepsilon\right)
$\ qubit gates. \ Therefore an $\mathsf{FBQP}$\ machine can simulate each call
that a $\mathsf{BosonFP}$\ machine makes to the \textsc{BosonSampling} oracle.
\end{proof}

\section{\label{WARMUPSEC}Efficient Classical Simulation of Linear Optics
Collapses \textsf{PH}}

In this section we prove Theorem \ref{warmup}, our hardness result for exact
\textsc{BosonSampling}. \ First, in Section \ref{BASIC}, we prove that
$\mathsf{P}^{\mathsf{\#P}}\subseteq\mathsf{BPP}^{\mathsf{NP}^{\mathcal{O}}}$,
where $\mathcal{O}$\ is any exact \textsc{BosonSampling}\ oracle. \ In
particular, this implies that, if there exists a polynomial-time classical
algorithm for exact \textsc{BosonSampling}, then $\mathsf{P}^{\mathsf{\#P}%
}=\mathsf{BPP}^{\mathsf{NP}}$\ and hence the polynomial hierarchy collapses to
the third level. \ The proof in Section \ref{BASIC}\ directly exploits the
fact that boson amplitudes are given by the permanents of complex matrices
$X\in\mathbb{C}^{n\times n}$, and that approximating $\operatorname*{Per}%
\left(  X\right)  $\ given such an $X$ is $\mathsf{\#P}$-complete. \ The main
lemma we need to prove is simply that approximating $\left\vert
\operatorname*{Per}\left(  X\right)  \right\vert ^{2}$\ is \textit{also}
$\mathsf{\#P}$-complete. \ Next, in Section \ref{ALTPROOF}, we give a
completely different proof of Theorem \ref{warmup}. \ This proof repurposes
two existing results in quantum computation: the scheme for universal quantum
computing with \textit{adaptive} linear optics due to Knill, Laflamme, and
Milburn \cite{klm}, and the $\mathsf{PostBQP}=\mathsf{PP}$\ theorem of
Aaronson \cite{aar:pp}. \ Finally, in Section \ref{STRONGER}, we observe two
improvements to the basic result.

\subsection{Basic Result\label{BASIC}}

First, we will need a classic result of Stockmeyer \cite{stockmeyer}.

\begin{theorem}
[Stockmeyer \cite{stockmeyer}]\label{approxcount}Given a Boolean function
$f:\left\{  0,1\right\}  ^{n}\rightarrow\left\{  0,1\right\}  $, let%
\[
p=\Pr_{x\in\left\{  0,1\right\}  ^{n}}\left[  f\left(  x\right)  =1\right]
=\frac{1}{2^{n}}\sum_{x\in\left\{  0,1\right\}  ^{n}}f\left(  x\right)  .
\]
Then for all $g\geq1+\frac{1}{\operatorname*{poly}\left(  n\right)  }$, there
exists an $\mathsf{FBPP}^{\mathsf{NP}^{f}}$\ machine that approximates $p$\ to
within a multiplicative factor of $g$.
\end{theorem}

Intuitively, Theorem \ref{approxcount}\ says that a $\mathsf{BPP}%
^{\mathsf{NP}}$\ machine can always estimate the probability $p$ that a
polynomial-time randomized algorithm accepts to within a
$1/\operatorname*{poly}\left(  n\right)  $\ multiplicative factor, even if $p$
is exponentially small. \ Note that Theorem \ref{approxcount} does
\textit{not} generalize to estimating the probability that a quantum algorithm
accepts, since the randomness is \textquotedblleft built in\textquotedblright%
\ to a quantum algorithm, and the $\mathsf{BPP}^{\mathsf{NP}}$\ machine does
not get to choose or control it.

Another interpretation of Theorem \ref{approxcount}\ is that any counting
problem that involves \textit{estimating the sum of }$2^{n}$\textit{
nonnegative real numbers}\footnote{Strictly speaking, Theorem
\ref{approxcount} talks about estimating the sum of $2^{n}$\ \textit{binary}
($\left\{  0,1\right\}  $-valued) numbers, but it is easy to generalize to
arbitrary nonnegative reals.}\ can be approximately solved in $\mathsf{BPP}%
^{\mathsf{NP}}$.

By contrast, if a counting problem involves estimating a sum of \textit{both
positive and negative numbers}---for example, if one wanted to approximate
$\operatorname*{E}_{x\in\left\{  0,1\right\}  ^{n}}\left[  f\left(  x\right)
\right]  $, for some function $f:\left\{  0,1\right\}  ^{n}\rightarrow\left\{
-1,1\right\}  $---then the situation is completely different. \ In that case,
it is easy to show that even multiplicative approximation is $\mathsf{\#P}%
$-hard, and hence unlikely to be in $\mathsf{FBPP}^{\mathsf{NP}}$.

We will show this phenomenon in the special case of the permanent. \ If $X$ is
a non-negative matrix,\ then Jerrum, Sinclair, and Vigoda \cite{jsv}\ famously
showed that one can approximate $\operatorname*{Per}\left(  X\right)  $ to
within multiplicative error $\varepsilon$ in $\operatorname*{poly}\left(
n,1/\varepsilon\right)  $\ time (which improves on Theorem \ref{approxcount}%
\ by getting rid of the $\mathsf{NP}$\ oracle). \ On the other hand, let
$X\in\mathbb{R}^{n\times n}$ be an arbitrary real matrix, with both positive
and negative entries. \ Then we will show that multiplicatively approximating
$\operatorname*{Per}\left(  X\right)  ^{2}=\left\vert \operatorname*{Per}%
\left(  X\right)  \right\vert ^{2}$\ is $\mathsf{\#P}$-hard. \ The reason why
we are interested in $\left\vert \operatorname*{Per}\left(  X\right)
\right\vert ^{2}$, rather than $\operatorname*{Per}\left(  X\right)
$\ itself, is that measurement probabilities in the noninteracting-boson model
are the absolute squares of permanents.

Our starting point is a famous result of Valiant \cite{valiant}:

\begin{theorem}
[Valiant \cite{valiant}]\label{valiantthm}The following problem is
$\mathsf{\#P}$-complete: given a matrix $X\in\left\{  0,1\right\}  ^{n\times
n}$, compute $\operatorname*{Per}\left(  X\right)  $.
\end{theorem}

We now show that $\operatorname*{Per}\left(  X\right)  ^{2}$\ is
$\mathsf{\#P}$-hard\ to approximate.

\begin{theorem}
[Hardness of Approximating $\operatorname*{Per}\left(  X\right)  ^{2}$%
]\label{approxhard}The following problem is $\mathsf{\#P}$-hard, for any
$g\in\left[  1,\operatorname*{poly}\left(  n\right)  \right]  $: given a real
matrix $X\in\mathbb{R}^{n\times n}$, approximate $\operatorname*{Per}\left(
X\right)  ^{2}$ to within a multiplicative factor of $g$.
\end{theorem}

\begin{proof}
Let $\mathcal{O}$\ be an oracle that, given a matrix $M\in\mathbb{R}^{n\times
n}$, outputs a nonnegative real number $\mathcal{O}\left(  M\right)  $\ such
that%
\[
\frac{\operatorname*{Per}\left(  M\right)  ^{2}}{g}\leq\mathcal{O}\left(
M\right)  \leq g\operatorname*{Per}\left(  M\right)  ^{2}.
\]
Also, let $X=\left(  x_{ij}\right)  \in\left\{  0,1\right\}  ^{n\times n}$ be
an input matrix, which we assume for simplicity consists only of $0$s and
$1$s. Then we will show how to compute $\operatorname*{Per}\left(  X\right)
$\ exactly, in polynomial time and using $O\left(  gn^{2}\log n\right)
$\ adaptive queries to $\mathcal{O}$. \ Since $\operatorname*{Per}\left(
X\right)  $\ is $\mathsf{\#P}$-complete\ by Theorem \ref{valiantthm}, this
will immediately imply the lemma.

Since $X$ is non-negative, we can check in polynomial time whether
$\operatorname*{Per}\left(  X\right)  =0$. \ If $\operatorname*{Per}\left(
X\right)  =0$\ we are done, so assume $\operatorname*{Per}\left(  X\right)
\geq1$. \ Then there exists a permutation $\sigma$\ such that $x_{1,\sigma
\left(  1\right)  }=\cdots=x_{n,\sigma\left(  n\right)  }=1$. \ By permuting
the rows and columns, we can assume without loss of generality that
$x_{11}=\cdots=x_{nn}=1$.

Our reduction will use recursion on $n$. \ Let $Y=\left(  y_{ij}\right)  $ be
the bottom-right $\left(  n-1\right)  \times\left(  n-1\right)  $\ submatrix
of $X$. \ Then we will assume inductively that we already know
$\operatorname*{Per}\left(  Y\right)  $. \ We will use that knowledge,
together with $O\left(  gn\log n\right)  $\ queries to $\mathcal{O}$, to find
$\operatorname*{Per}\left(  X\right)  $.

Given a real number $r$, let $X^{\left[  r\right]  }\in\mathbb{R}^{n\times n}%
$\ be a matrix identical to $X$, except that the top-left entry is $x_{11}%
-r$\ instead of $x_{11}$. \ Then it is not hard to see that%
\[
\operatorname*{Per}\left(  X^{\left[  r\right]  }\right)  =\operatorname*{Per}%
\left(  X\right)  -r\operatorname*{Per}\left(  Y\right)  .
\]
Note that $y_{11}=\cdots=y_{\left(  n-1\right)  ,\left(  n-1\right)  }=1$, so
$\operatorname*{Per}\left(  Y\right)  \geq1$. \ Hence there must be a unique
value $r=r^{\ast}$\ such that $\operatorname*{Per}\left(  X^{\left[  r^{\ast
}\right]  }\right)  =0$. \ Furthermore, if we can find that $r^{\ast}$,\ then
we are done, since $\operatorname*{Per}\left(  X\right)  =r^{\ast
}\operatorname*{Per}\left(  Y\right)  $.

To find%
\[
r^{\ast}=\frac{\operatorname*{Per}\left(  X\right)  }{\operatorname*{Per}%
\left(  Y\right)  },
\]
we will use a procedure based on binary search. \ Let $r\left(  0\right)  :=0$
be our \textquotedblleft initial guess\textquotedblright;\ then we will
repeatedly improve this guess to $r\left(  1\right)  $, $r\left(  2\right)  $,
etc. \ The invariant we want to maintain is that%
\[
\mathcal{O}\left(  X^{\left[  r\left(  t+1\right)  \right]  }\right)
\leq\frac{\mathcal{O}\left(  X^{\left[  r\left(  t\right)  \right]  }\right)
}{2}%
\]
for all $t$.

To find $r\left(  t+1\right)  $ starting from $r\left(  t\right)  $: first
observe that%
\begin{align}
\left\vert r\left(  t\right)  -r^{\ast}\right\vert  &  =\frac{\left\vert
r\left(  t\right)  \operatorname*{Per}\left(  Y\right)  -\operatorname*{Per}%
\left(  X\right)  \right\vert }{\operatorname*{Per}\left(  Y\right)
}\label{distbound}\\
&  =\frac{\left\vert \operatorname*{Per}\left(  X^{\left[  r\left(  t\right)
\right]  }\right)  \right\vert }{\operatorname*{Per}\left(  Y\right)
}\nonumber\\
&  \leq\frac{\sqrt{g\cdot\mathcal{O}\left(  X^{\left[  r\left(  t\right)
\right]  }\right)  }}{\operatorname*{Per}\left(  Y\right)  },\nonumber
\end{align}
where the last line follows from $\operatorname*{Per}\left(  M\right)
^{2}/g\leq\mathcal{O}\left(  M\right)  $. \ So setting%
\[
\beta:=\frac{\sqrt{g\cdot\mathcal{O}\left(  X^{\left[  r\left(  t\right)
\right]  }\right)  }}{\operatorname*{Per}\left(  Y\right)  },
\]
we find that $r^{\ast}$\ is somewhere in the interval $I:=\left[  r\left(
t\right)  -\beta,r\left(  t\right)  +\beta\right]  $. \ Divide $I$ into $L$
equal segments (for some $L$ to be determined later), and let $s\left(
1\right)  ,\ldots,s\left(  L\right)  $\ be their left endpoints. \ Then the
procedure is to evaluate $\mathcal{O}\left(  X^{\left[  s\left(  i\right)
\right]  }\right)  $\ for each $i\in\left[  L\right]  $, and set $r\left(
t+1\right)  $\ equal to the $s\left(  i\right)  $\ for which $\mathcal{O}%
\left(  X^{\left[  s\left(  i\right)  \right]  }\right)  $\ is minimized
(breaking ties arbitrarily).

Clearly there exists an $i\in\left[  L\right]  $\ such that $\left\vert
s\left(  i\right)  -r^{\ast}\right\vert \leq\beta/L$---and for that particular
choice of $i$, we have%
\begin{align*}
\mathcal{O}\left(  X^{\left[  s\left(  i\right)  \right]  }\right)   &  \leq
g\operatorname*{Per}\left(  X^{\left[  s\left(  i\right)  \right]  }\right)
^{2}\\
&  =g\left(  \operatorname*{Per}\left(  X\right)  -s\left(  i\right)
\operatorname*{Per}\left(  Y\right)  \right)  ^{2}\\
&  =g\left(  \operatorname*{Per}\left(  X\right)  -\left(  s\left(  i\right)
-r^{\ast}\right)  \operatorname*{Per}\left(  Y\right)  -r^{\ast}%
\operatorname*{Per}\left(  Y\right)  \right)  ^{2}\\
&  =g\left(  s\left(  i\right)  -r^{\ast}\right)  ^{2}\operatorname*{Per}%
\left(  Y\right)  ^{2}\\
&  \leq g\frac{\beta^{2}}{L^{2}}\operatorname*{Per}\left(  Y\right)  ^{2}\\
&  =\frac{g^{2}}{L^{2}}\mathcal{O}\left(  X^{\left[  r\left(  t\right)
\right]  }\right)  .
\end{align*}
Therefore, so long as we choose $L\geq\sqrt{2}g$, we find that%
\[
\mathcal{O}\left(  X^{\left[  r\left(  t+1\right)  \right]  }\right)
\leq\mathcal{O}\left(  X^{\left[  s\left(  i\right)  \right]  }\right)
\leq\frac{\mathcal{O}\left(  X^{\left[  r\left(  t\right)  \right]  }\right)
}{2},
\]
which is what we wanted.

Now observe that%
\[
\mathcal{O}\left(  X^{\left[  r\left(  0\right)  \right]  }\right)
=\mathcal{O}\left(  X\right)  \leq g\operatorname*{Per}\left(  X\right)
^{2}\leq g\left(  n!\right)  ^{2}.
\]
So for some $T=O\left(  n\log n\right)  $,%
\[
\mathcal{O}\left(  X^{\left[  r\left(  T\right)  \right]  }\right)  \leq
\frac{\mathcal{O}\left(  X^{\left[  r\left(  0\right)  \right]  }\right)
}{2^{T}}\leq\frac{g\left(  n!\right)  ^{2}}{2^{T}}\ll\frac{1}{4g}.
\]
By equation (\ref{distbound}), this in turn implies that%
\[
\left\vert r\left(  T\right)  -r^{\ast}\right\vert \leq\frac{\sqrt
{g\cdot\mathcal{O}\left(  X^{\left[  r\left(  T\right)  \right]  }\right)  }%
}{\operatorname*{Per}\left(  Y\right)  }\ll\frac{1}{2\operatorname*{Per}%
\left(  Y\right)  }.
\]
But this means that we can find $r^{\ast}$\ exactly, since $r^{\ast}$ equals a
rational number $\frac{\operatorname*{Per}\left(  X\right)  }%
{\operatorname*{Per}\left(  Y\right)  }$, where\ $\operatorname*{Per}\left(
X\right)  $\ and $\operatorname*{Per}\left(  Y\right)  $\ are both positive
integers and $\operatorname*{Per}\left(  Y\right)  $\ is known.
\end{proof}

Let us remark that one can improve Theorem \ref{approxhard},\ to ensure that
the entries of $X$\ are all at most $\operatorname*{poly}\left(  n\right)
$\ in absolute value. \ We do not pursue that here, since it will not be
needed for our application.

\begin{lemma}
\label{embedlem}Let $X\in\mathbb{C}^{n\times n}$. \ Then for all $m\geq2n$ and
$\varepsilon\leq1/\left\Vert X\right\Vert $, there exists an $m\times
m$\ unitary matrix $U$ that contains $\varepsilon X$ as a submatrix.
\ Furthermore, $U$ can be computed in polynomial time given $X$.
\end{lemma}

\begin{proof}
Let $Y=\varepsilon X$. \ Then it suffices to show how to construct a $2n\times
n$\ matrix $W$ whose columns are orthonormal vectors, and that contains $Y$ as
its top $n\times n$\ submatrix. \ For such a $W$ can easily be completed to an
$m\times n$ matrix whose columns are orthonormal (by filling the bottom
$m-2n$\ rows with zeroes), which can in turn be completed to an $m\times
m$\ unitary matrix in $O\left(  m^{3}\right)  $\ time.

Since $\left\Vert Y\right\Vert \leq\varepsilon\left\Vert X\right\Vert \leq1$,
we have $Y^{\dagger}Y\preceq I$\ in the semidefinite ordering. \ Hence
$I-Y^{\dagger}Y$\ is positive semidefinite. \ So $I-Y^{\dagger}Y$\ has a
Cholesky decomposition $I-Y^{\dagger}Y=Z^{\dagger}Z$, for some $Z\in
\mathbb{C}^{n\times n}$. \ Let us set $W:=%
%TCIMACRO{\QDOVERD{(}{)}{Y}{Z}}%
%BeginExpansion
\genfrac{(}{)}{}{0}{Y}{Z}%
%EndExpansion
$. \ Then $W^{\dagger}W=Y^{\dagger}Y+Z^{\dagger}Z=I$, so the columns of $W$
are orthonormal as desired.
\end{proof}

We are now ready to prove Theorem \ref{warmup}: that $\mathsf{P}%
^{\mathsf{\#P}}\subseteq\mathsf{BPP}^{\mathsf{NP}^{\mathcal{O}}}$ for any
exact \textsc{BosonSampling}\ oracle $\mathcal{O}$.

\begin{proof}
[Proof of Theorem \ref{warmup}]Given a matrix $X\in\mathbb{R}^{n\times n}%
$\ and a parameter $g\in\left[  1+\frac{1}{\operatorname*{poly}\left(
n\right)  },\operatorname*{poly}\left(  n\right)  \right]  $, we know from
Theorem \ref{approxhard}\ that it is $\mathsf{\#P}$-hard\ to approximate
$\operatorname*{Per}\left(  X\right)  ^{2}$\ to within a multiplicative factor
of $g$. \ So to prove the theorem, it suffices to show how to approximate
$\operatorname*{Per}\left(  X\right)  ^{2}$ in $\mathsf{FBPP}^{\mathsf{NP}%
^{\mathcal{O}}}$.

Set $m:=2n$ and $\varepsilon:=1/\left\Vert X\right\Vert \geq
2^{-\operatorname*{poly}\left(  n\right)  }$. \ Then by Lemma \ref{embedlem},
we can efficiently construct an $m\times m$\ unitary matrix $U$ with
$U_{n,n}=\varepsilon X$ as its top-left $n\times n$\ submatrix.\ \ Let $A$ be
the $m\times n$\ column-orthonormal matrix corresponding to the first $n$
columns of $U$. \ Let us feed $A$ as input to $\mathcal{O}$, and consider the
probability $p_{A}$\ that $\mathcal{O}$ outputs $1_{n}$. \ We have%
\begin{align*}
p_{A}  &  =\Pr_{r}\left[  \mathcal{O}\left(  A,r\right)  =1_{n}\right] \\
&  =\left\vert \left\langle 1_{n}|\varphi\left(  U\right)  |1_{n}\right\rangle
\right\vert ^{2}\\
&  =\left\vert \operatorname*{Per}\left(  U_{n,n}\right)  \right\vert ^{2}\\
&  =\varepsilon^{2n}\left\vert \operatorname*{Per}\left(  X\right)
\right\vert ^{2},
\end{align*}
where the third line follows from Theorem \ref{perust}. \ But by Theorem
\ref{approxcount}, we can approximate $p_{A}$\ to within a multiplicative
factor of $g$\ in $\mathsf{FBPP}^{\mathsf{NP}^{\mathcal{O}}}$. \ It follows
that we can approximate $\left\vert \operatorname*{Per}\left(  X\right)
\right\vert ^{2}=\operatorname*{Per}\left(  X\right)  ^{2}$\ in $\mathsf{FBPP}%
^{\mathsf{NP}^{\mathcal{O}}}$\ as well.
\end{proof}

The main fact that we wanted to prove is an immediate corollary of Theorem
\ref{warmup}:

\begin{corollary}
\label{exactcor}Suppose exact \textsc{BosonSampling} can be done in classical
polynomial time.\ \ Then $\mathsf{P}^{\mathsf{\#P}}=\mathsf{BPP}^{\mathsf{NP}%
}$, and hence the polynomial hierarchy collapses to the third level.
\end{corollary}

\begin{proof}
Combining the assumption with Theorem \ref{warmup}, we get that $\mathsf{P}%
^{\mathsf{\#P}}\subseteq\mathsf{BPP}^{\mathsf{NP}}$, which by Toda's Theorem
\cite{toda} implies that $\mathsf{P}^{\mathsf{\#P}}=\mathsf{PH}=\mathsf{\Sigma
}_{\mathsf{3}}^{\mathsf{P}}=\mathsf{BPP}^{\mathsf{NP}}$.
\end{proof}

Likewise, even if exact \textsc{BosonSampling}\ can be done in $\mathsf{BPP}%
^{\mathsf{PH}}$ (that is, using an oracle for some fixed level of the
polynomial hierarchy), we still get that%
\[
\mathsf{P}^{\mathsf{\#P}}\subseteq\mathsf{BPP}^{\mathsf{NP}^{\mathsf{PH}}%
}=\mathsf{BPP}^{\mathsf{PH}}=\mathsf{PH,}%
\]
and hence $\mathsf{PH}$\ collapses.

As another application of Theorem \ref{warmup}, suppose exact
\textsc{BosonSampling}\ can be done in $\mathsf{BPP}^{\mathsf{P{}romiseBQP}}$:
that is, using an oracle for $\mathsf{BQP}$\ \textit{decision} problems.
\ Then we get the containment%
\[
\mathsf{P}^{\mathsf{\#P}}\subseteq\mathsf{BPP}^{\mathsf{NP}^{^{\mathsf{P{}%
romiseBQP}}}}.
\]
Such a containment seems unlikely (though we admit to lacking a strong
intuition here), thereby providing possible evidence for a separation between
$\mathsf{BQP}$\ sampling problems and $\mathsf{BQP}$\ decision problems.

\subsection{Alternate Proof Using KLM\label{ALTPROOF}}

Inspired by recent work of Bremner et al.\ \cite{bjs}, in this section we give
a different proof of Theorem \ref{warmup}. \ This proof makes no use of
permanents or approximate counting; instead, it invokes two previous quantum
computing results---the KLM Theorem \cite{klm} and the $\mathsf{PostBQP}%
=\mathsf{PP}$\ theorem \cite{aar:pp}---as black boxes. \ Compared to the first
proof, the second one has the advantage of being shorter and completely free
of calculations; also, it easily generalizes to many other quantum computing
models, besides noninteracting bosons. \ The disadvantage is that, to those
unfamiliar with \cite{klm,aar:pp}, the second proof gives less intuition about
why Theorem \ref{warmup} is true. \ Also, we do not know how to generalize the
second proof to say anything about the hardness of \textit{approximate}
sampling. \ For that, it seems essential to talk about the \textsc{Permanent}
or some other concrete $\mathsf{\#P}$-complete problem.

Our starting point is the KLM Theorem, which says informally that linear
optics augmented with \textit{adaptive measurements} is universal for quantum
computation. \ A bit more formally, define $\mathsf{BosonP}%
_{\operatorname*{adap}}$ to be the class of languages that are decidable in
$\mathsf{BPP}$\ (that is, classical probabilistic polynomial-time), augmented
with the ability to prepare $k$-photon states (for any $k=\operatorname*{poly}%
\left(  n\right)  $) in any of $m=\operatorname*{poly}\left(  n\right)
$\ modes; apply arbitrary optical elements to pairs of modes; measure the
photon number of any mode at any time; and condition future optical elements
and classical computations on the outcomes of the measurements. \ From Theorem
\ref{infbqp}, it is not hard to see that $\mathsf{BosonP}%
_{\operatorname*{adap}}\subseteq\mathsf{BQP}$. \ The amazing discovery of
Knill et al.\ \cite{klm}\ was that the other direction holds as well:

\begin{theorem}
[KLM Theorem \cite{klm}]\label{klmthm}$\mathsf{BosonP}_{\operatorname*{adap}%
}=\mathsf{BQP}$.
\end{theorem}

In the proof of Theorem \ref{klmthm}, a key step is to consider a model of
linear optics with \textit{postselected }measurements. \ This is similar to
the model with adaptive measurements described above, except that here we
\textit{guess} the outcomes of all the photon-number measurements at the very
beginning, and then only proceed with the computation if the guesses turn out
to be correct. \ In general, the resulting computation will only succeed with
exponentially-small probability, but we know when it does succeed. \ Notice
that, in this model, there is never any need to condition later computational
steps on the outcomes of measurements---since if the computation succeeds,
then we know in advance what all the measurement outcomes are anyway! \ One
consequence is that, without loss of generality, we can postpone all
measurements until the end of the computation.

Along the way to proving Theorem \ref{klmthm}, Knill et al.\ \cite{klm} showed
how to simulate any \textit{postselected} quantum computation using a
\textit{postselected} linear-optics computation.\footnote{Terhal and
DiVincenzo \cite{td} later elaborated on their result, using the term
\textquotedblleft nonadaptive quantum computation\textquotedblright\ (or
$QC_{nad}$) for what we call postselection.} \ To formalize the
\textquotedblleft Postselected KLM Theorem,\textquotedblright\ we now define
the complexity class $\mathsf{PostBosonP}$, which consists of all problems
solvable in polynomial time using linear optics with postselection.

\begin{definition}
[$\mathsf{PostBosonP}$]\label{postbosonpdef}$\mathsf{PostBosonP}$ is the class
of languages $L\subseteq\left\{  0,1\right\}  ^{\ast}$\ for which there exist
deterministic polynomial-time algorithms $\mathcal{V},\mathcal{A},\mathcal{B}%
$\ such that for all inputs $x\in\left\{  0,1\right\}  ^{N}$:

\begin{enumerate}
\item[(i)] The output\ of $\mathcal{V}$ is an $m\times n$\ matrix $V\left(
x\right)  \in\mathcal{U}_{m,n}$\ (for some $m,n=\operatorname*{poly}\left(
N\right)  $), corresponding to a linear-optical network that samples from the
probability distribution $\mathcal{D}_{V\left(  x\right)  }$.

\item[(ii)] $\Pr_{y\sim\mathcal{D}_{V\left(  x\right)  }}\left[
\mathcal{A}\left(  y\right)  \text{ accepts}\right]  >0$.

\item[(iii)] If $x\in L$\ then $\Pr_{y\sim\mathcal{D}_{V\left(  x\right)  }%
}\left[  \mathcal{B}\left(  y\right)  \text{ accepts
%TCIMACRO{\TEXTsymbol{\vert} }%
%BeginExpansion
$\vert$
%EndExpansion
}\mathcal{A}\left(  y\right)  \text{ accepts}\right]  \geq\frac{2}{3}$.

\item[(iv)] If $x\notin L$\ then $\Pr_{y\sim\mathcal{D}_{V\left(  x\right)  }%
}\left[  \mathcal{B}\left(  y\right)  \text{ accepts
%TCIMACRO{\TEXTsymbol{\vert} }%
%BeginExpansion
$\vert$
%EndExpansion
}\mathcal{A}\left(  y\right)  \text{ accepts}\right]  \leq\frac{1}{3}$.
\end{enumerate}
\end{definition}

In our terminology, Knill et al.\ \cite{klm}\ showed that $\mathsf{PostBosonP}%
$\ captures the full power of postselected quantum computation---in other
words, of the class $\mathsf{PostBQP}$\ defined in Section \ref{PRELIM}. \ We
now sketch a proof for completeness.

\begin{theorem}
[Postselected KLM Theorem \cite{klm}]\label{klmthm2}$\mathsf{PostBosonP}%
=\mathsf{PostBQP}$.
\end{theorem}

\begin{proof}
[Proof Sketch]For $\mathsf{PostBosonP}\subseteq\mathsf{PostBQP}$, use the
procedure from Theorem \ref{infbqp}, to create an ordinary quantum circuit $C$
that simulates a given linear-optical network $U$. \ Note that the algorithms
$\mathcal{A}$ and $\mathcal{B}$ from Definition \ref{postbosonpdef}\ can
simply be \textquotedblleft folded\textquotedblright\ into $C$, so that
$\mathcal{A}\left(  y\right)  $\ accepting corresponds to the first qubit of
$C$'s output being measured to be $\left\vert 1\right\rangle $, and
$\mathcal{B}\left(  y\right)  $\ accepting corresponds to the second qubit of
$C$'s output being measured to be $\left\vert 1\right\rangle $.

The more interesting direction is $\mathsf{PostBQP}\subseteq
\mathsf{PostBosonP}$. \ To simulate $\mathsf{BQP}$\ in $\mathsf{PostBosonP}$,
the basic idea of KLM is to use \textquotedblleft nondeterministic
gates,\textquotedblright\ which consist of sequences of beamsplitters and
phaseshifters followed by postselected photon-number measurements.
\ \textit{If} the measurements return a particular outcome, \textit{then} the
effect of the beamsplitters and phaseshifters is to implement (perfectly) a
$2$-qubit gate that is known to be universal for standard quantum computation.
\ We refer the reader to \cite{klm}\ for the details of how such gates are
constructed; for now, assume we have them. \ Then for any $\mathsf{BQP}%
$\ machine $M$, it is easy to create a $\mathsf{PostBosonP}$ machine
$M^{\prime}$ that simulates $M$. \ But once we have $\mathsf{BQP}$, we also
get $\mathsf{PostBQP}$\ essentially \textquotedblleft free of
charge.\textquotedblright\ \ This is because the simulating machine
$M^{\prime}$ can postselect, not only on its nondeterministic gates working
correctly, but also (say) on $M$ reaching a final configuration whose first
qubit is $\left\vert 1\right\rangle $.
\end{proof}

We can now complete our alternative proof of Theorem \ref{warmup}, that
$\mathsf{P}^{\mathsf{\#P}}\subseteq\mathsf{BPP}^{\mathsf{NP}^{\mathcal{O}}}$
for any exact \textsc{BosonSampling}\ oracle $\mathcal{O}$.

\begin{proof}
[Proof of Theorem \ref{warmup}]Let $\mathcal{O}$\ be an exact
\textsc{BosonSampling}\ oracle. \ Then we claim that $\mathsf{PostBosonP}%
\subseteq\mathsf{PostBPP}^{\mathcal{O}}$. \ To see this, let $\mathcal{V}%
,\mathcal{A},\mathcal{B}$\ be the polynomial-time Turing machines from
Definition \ref{postbosonpdef}.\ \ Then we can create a $\mathsf{PostBPP}%
^{\mathcal{O}}$ machine that, given an input $x$ and random string $r$:

\begin{itemize}
\item[(i)] \textquotedblleft Succeeds\textquotedblright\ if $\mathcal{A}%
\left(  \mathcal{O}\left(  V\left(  x\right)  ,r\right)  \right)  $\ accepts,
and \textquotedblleft fails\textquotedblright\ otherwise.

\item[(ii)] Conditioned on succeeding, accepts if $\mathcal{B}\left(
\mathcal{O}\left(  V\left(  x\right)  ,r\right)  \right)  $\ accepts\ and
rejects otherwise.
\end{itemize}

Then%
\[
\mathsf{PP}=\mathsf{PostBQP}=\mathsf{PostBosonP}\subseteq\mathsf{PostBPP}%
^{\mathcal{O}}\subseteq\mathsf{BPP}^{\mathsf{NP}^{\mathcal{O}}},
\]
where the first equality comes from Theorem \ref{postbqpthm}\ and the second
from Theorem \ref{klmthm2}. \ Therefore $\mathsf{P}^{\mathsf{\#P}}%
=\mathsf{P}^{\mathsf{PP}}$\ is contained in $\mathsf{BPP}^{\mathsf{NP}%
^{\mathcal{O}}}$\ as well.
\end{proof}

\subsection{Strengthening the Result\label{STRONGER}}

In this section, we make two simple but interesting improvements to Theorem
\ref{warmup}.

The first improvement is this: instead of considering a whole collection of
distributions, we can give a \textit{fixed} distribution $\mathcal{D}_{n}$
(depending only on the input size $n$) that can be sampled by a boson
computer, but that cannot be efficiently sampled classically unless the
polynomial hierarchy collapses. \ This $\mathcal{D}_{n}$\ will effectively be
a \textquotedblleft complete distribution\textquotedblright\ for the
noninteracting-boson model under nondeterministic reductions. \ Let us discuss
how to construct such a $\mathcal{D}_{n}$, using the approach of Section
\ref{ALTPROOF}.

Let $p\left(  n\right)  $ be some fixed polynomial (say $n^{2}$), and let
$\mathcal{C}$\ be the set of all quantum circuits on $n$\ qubits\ with at most
$p\left(  n\right)  $\ gates (over some finite universal basis, such as
$\left\{  \text{\textsc{Hadamard}},\text{\textsc{Toffoli}}\right\}  $
\cite{shi:gate}). \ Then consider the following $\mathsf{PostBQP}$\ algorithm
$\mathcal{A}$, which takes as input a description of a circuit $C^{\ast}%
\in\mathcal{C}$. \ First, generate a uniform superposition%
\[
\left\vert \mathcal{C}\right\rangle =\frac{1}{\sqrt{\left\vert \mathcal{C}%
\right\vert }}\sum_{C\in\mathcal{C}}\left\vert C\right\rangle
\]
over descriptions of all circuits $C\in\mathcal{C}$. \ Then measure
$\left\vert \mathcal{C}\right\rangle $\ in the standard basis, and postselect
on the outcome being $\left\vert C^{\ast}\right\rangle $. \ Finally, assuming
$\left\vert C^{\ast}\right\rangle $\ was obtained, take some fixed universal
circuit $U$ with the property that%
\[
\Pr\left[  U\left(  \left\vert C\right\rangle \right)  \text{ accepts}\right]
\approx\Pr\left[  C\left(  0^{n}\right)  \text{ accepts}\right]
\]
for all $C\in\mathcal{C}$, and run $U$\ on input $\left\vert C^{\ast
}\right\rangle $. \ Now, since $\mathsf{PostBQP}=\mathsf{PostBosonP}$ by
Theorem \ref{klmthm2}, it is clear that $\mathcal{A}$\ can be
\textquotedblleft compiled\textquotedblright\ into a postselected
linear-optical network $\mathcal{A}^{\prime}$. \ Let $\mathcal{D}%
_{\mathcal{A}^{\prime}}$\ be the probability distribution sampled by
$\mathcal{A}^{\prime}$ if we ignore the postselection steps. \ Then
$\mathcal{D}_{\mathcal{A}^{\prime}}$\ is our desired universal distribution
$\mathcal{D}_{n}$.

More concretely, we claim that, if $\mathcal{D}_{n}$\ can be sampled in
$\mathsf{FBPP}$, then $\mathsf{P}^{\mathsf{\#P}}=\mathsf{PH}=\mathsf{BPP}%
^{\mathsf{NP}}$. \ To see this, let $\mathcal{O}\left(  r\right)  $ be a
polynomial-time classical algorithm that outputs a sample from $\mathcal{D}%
_{n}$, given as input a random string $r\in\left\{  0,1\right\}
^{\operatorname*{poly}\left(  n\right)  }$. \ Then, as in the proof of Theorem
\ref{warmup} in Section \ref{ALTPROOF}, we have $\mathsf{PostBosonP}%
\subseteq\mathsf{PostBPP}$. \ For let $\mathcal{V},\mathcal{A},\mathcal{B}%
$\ be the polynomial-time algorithms from Definition \ref{postbosonpdef}%
.\ \ Then we can create a $\mathsf{PostBPP}$ machine that, given an input $x$
and random string $r$:

\begin{enumerate}
\item[(1)] Postselects on $\mathcal{O}\left(  r\right)  $\ containing an
encoding of the linear-optical network $V\left(  x\right)  $.

\item[(2)] Assuming $\left\vert V\left(  x\right)  \right\rangle $\ is
observed, simulates the $\mathsf{PostBosonP}$\ algorithm: that is,
\textquotedblleft succeeds\textquotedblright\ if $\mathcal{A}\left(
\mathcal{O}\left(  r\right)  \right)  $\ accepts and fails otherwise, and
\textquotedblleft accepts\textquotedblright\ if $\mathcal{B}\left(
\mathcal{O}\left(  r\right)  \right)  $\ accepts and rejects otherwise.
\end{enumerate}

Our second improvement to Theorem \ref{warmup} weakens the physical resource
requirements needed to sample from a hard distribution. \ Recall that we
assumed our boson computer began in the \textquotedblleft standard initial
state\textquotedblright\textit{ }$\left\vert 1_{n}\right\rangle :=\left\vert
1,\ldots,1,0,\ldots,0\right\rangle $, in which the first $n$ modes were
occupied by a single boson each. \ Unfortunately, in the optical setting, it
is notoriously difficult to produce a single photon on demand (see Section
\ref{EXPER} for more about this). \ Using a standard laser, it is much easier
to produce so-called \textit{coherent states}, which have the form%
\[
\left\vert \alpha\right\rangle :=e^{-\left\vert \alpha\right\vert ^{2}/2}%
\sum_{n=0}^{\infty}\frac{\alpha^{n}}{\sqrt{n!}}\left\vert n\right\rangle
\]
for some complex number $\alpha$. \ (Here $\left\vert n\right\rangle
$\ represents a state of $n$ photons.) \ However, we now observe that the
KLM-based proof of Theorem \ref{warmup} goes through almost without change, if
the inputs are coherent states rather than individual photons. \ The reason is
that, in the $\mathsf{PostBosonP}$\ model, we can first prepare a coherent
state (say $\left\vert \alpha=1\right\rangle $), then measure it and
\textit{postselect} on getting a single photon. \ In this way, we can use
postselection to generate the standard initial state $\left\vert
1_{n}\right\rangle $, then run the rest of the computation as before.

Summarizing the improvements:

\begin{theorem}
\label{warmup2}There exists a family of distributions $\left\{  \mathcal{D}%
_{n}\right\}  _{n\geq1}$, depending only on $n$, such that:

\begin{enumerate}
\item[(i)] For all $n$, a boson computer with coherent-state inputs can sample
from $\mathcal{D}_{n}$\ in $\operatorname*{poly}\left(  n\right)  $\ time.

\item[(ii)] Let $\mathcal{O}$ be any oracle that takes as input a random
string $r$ (which $\mathcal{O}$\ uses as its only source of randomness)
together with $n$, and that outputs a sample $\mathcal{O}_{n}\left(  r\right)
$ from $\mathcal{D}_{n}$.\ \ Then $\mathsf{P}^{\mathsf{\#P}}\subseteq
\mathsf{BPP}^{\mathsf{NP}^{\mathcal{O}}}$.
\end{enumerate}
\end{theorem}

\section{Main Result\label{MAIN}}

We now move on to prove our main result: that even \textit{approximate}
classical simulation of boson computations would have surprising complexity consequences.

\subsection{Truncations of Haar-Random Unitaries\label{EMBEDDING}}

In this section we prove a statement we will need from random matrix theory,
which seems new and might be of independent interest. \ Namely: \textit{any
}$m^{1/6}\times m^{1/6}$\textit{ submatrix of an }$m\times m$\textit{
Haar-random unitary matrix is close, in variation distance, to a matrix of
i.i.d.\ Gaussians.} \ It is easy to see that any individual \textit{entry} of
a Haar unitary matrix is approximately Gaussian. \ Thus, our result just says
that any small enough \textit{set} of entries is approximately
independent---and that here, \textquotedblleft small enough\textquotedblright%
\ can mean not only a constant number of entries, but even $m^{\Omega\left(
1\right)  }$\ of them. \ This is not surprising: it simply means that one
needs to examine a significant fraction of the entries before one
\textquotedblleft notices\textquotedblright\ the unitarity constraint.

Given $m\geq n$, recall that $\mathcal{U}_{m,n}$\ is the set of $m\times
n$\ complex matrices whose columns are orthonormal vectors, and $\mathcal{H}%
_{m,n}$\ is the Haar measure over $\mathcal{U}_{m,n}$. \ Define $\mathcal{S}%
_{m,n}$\ to be the distribution over $n\times n$\ matrices obtained by first
drawing a unitary $U$\ from $\mathcal{H}_{m,m}$, and then outputting $\sqrt
{m}U_{n,n}$\ where $U_{n,n}$\ is the top-left $n\times n$\ submatrix of $U$.
\ In other words, $\mathcal{S}_{m,n}$\ is the distribution over $n\times
n$\ truncations of $m\times m$\ Haar unitary matrices,\ where the entries have
been scaled up by a factor of $\sqrt{m}$ so that they have mean $0$ and
variance $1$. \ Also, recall that $\mathcal{G}^{n\times n}$\ is the
probability distribution over $n\times n$\ complex matrices whose entries are
independent Gaussians with mean $0$ and variance $1$. \ Then our main result
states that $\mathcal{S}_{m,n}$\ is close in variation distance to
$\mathcal{G}^{n\times n}$:

\begin{theorem}
\label{truncthm}Let $m\geq\frac{n^{5}}{\delta}\log^{2}\frac{n}{\delta}$, for
any $\delta>0$. \ Then $\left\Vert \mathcal{S}_{m,n}-\mathcal{G}^{n\times
n}\right\Vert =O\left(  \delta\right)  $.
\end{theorem}

The bound $m\geq\frac{n^{5}}{\delta}\log^{2}\frac{n}{\delta}$\ is almost
certainly not tight; we suspect that it can be improved (for example) to
$m=O\left(  n^{2}/\delta\right)  $. \ For our purposes, however, what is
important is simply that $m$\ is polynomial in $n$\ and $1/\delta$.

Let $p_{G},p_{S}:\mathbb{C}^{n\times n}\rightarrow\mathbb{R}^{+}$\ be the
probability density functions of $\mathcal{G}^{n\times n}$\ and $\mathcal{S}%
_{m,n}$\ respectively (for convenience, we drop the subscripts $m$\ and $n$).
\ Then for our application, we will actually need the following stronger
version of Theorem \ref{truncthm}:

\begin{theorem}
[Haar-Unitary Hiding Theorem]\label{truncthm2}Let $m\geq\frac{n^{5}}{\delta
}\log^{2}\frac{n}{\delta}$. \ Then%
\[
p_{S}\left(  X\right)  \leq\left(  1+O\left(  \delta\right)  \right)
p_{G}\left(  X\right)
\]
for all $X\in\mathbb{C}^{n\times n}$.
\end{theorem}

Fortunately, Theorem \ref{truncthm2} will follow fairly easily from our proof
of Theorem \ref{truncthm}.

Surprisingly, Theorems \ref{truncthm} and \ref{truncthm2}\ do not seem to have
appeared in the random matrix theory literature, although truncations of Haar
unitary matrices have been studied in detail. \ In particular, Petz and
R\'{e}ffy \cite{petzreffy} showed that the truncated Haar-unitary distribution
$\mathcal{S}_{m,n}$\ converges to the Gaussian distribution, when $n$\ is
fixed and $m\rightarrow\infty$. \ (Mastrodonato and Tumulka \cite{mastrodon}
later gave an elementary proof of this fact.) \ In a followup paper, Petz and
R\'{e}ffy\ \cite{petzreffy2}\ proved a large deviation bound for the empirical
eigenvalue density of matrices drawn from $\mathcal{S}_{m,n}$ (see also
R\'{e}ffy's PhD thesis\ \cite{reffy}). \ We will use some observations from
those papers, especially an explicit formula in \cite{reffy}\ for the
probability density function of $\mathcal{S}_{m,n}$.

We now give an overview of the proof of Theorem \ref{truncthm}. \ Our goal is
to prove that%
\[
\Delta\left(  p_{G},p_{S}\right)  :=\int_{X\in\mathbb{C}^{n\times n}%
}\left\vert p_{G}\left(  X\right)  -p_{S}\left(  X\right)  \right\vert dX
\]
is small,\ where the integral (like all others in this section) is with
respect to the Lebesgue measure over the entries of $X$.

The first crucial observation is that the probability distributions
$\mathcal{G}^{n\times n}$\ and $\mathcal{S}_{m,n}$\ are both invariant under
left-multiplication or right-multiplication by a unitary matrix. \ It follows
that $p_{G}\left(  X\right)  $\ and $p_{S}\left(  X\right)  $\ both depend
only on the \textit{list of singular values} of $X$. \ For we can always write
$X=\left(  x_{ij}\right)  $ as $UDV$, where $U,V$\ are unitary and $D=\left(
d_{ij}\right)  $\ is a diagonal matrix of singular values; then $p_{G}\left(
X\right)  =p_{G}\left(  D\right)  $\ and $p_{S}\left(  X\right)  =p_{S}\left(
D\right)  $. \ Let $\lambda_{i}:=d_{ii}^{2}$\ be the square of the $i^{th}%
$\ singular value of $X$.\ \ Then from the identity%
\begin{equation}
\sum_{i,j\in\left[  n\right]  }\left\vert x_{ij}\right\vert ^{2}=\sum
_{i\in\left[  n\right]  }\lambda_{i}, \label{frobenius}%
\end{equation}
we get the following formula for $p_{G}$:%
\[
p_{G}\left(  X\right)  =\prod_{i,j\in\left[  n\right]  }\frac{1}{\pi
}e^{-\left\vert x_{ij}\right\vert ^{2}}=\frac{1}{\pi^{n^{2}}}\prod
_{i\in\left[  n\right]  }e^{-\lambda_{i}}.
\]
Also, R\'{e}ffy \cite[p. 61]{reffy} has shown that, provided $m\geq2n$, we
have%
\begin{equation}
p_{S}\left(  X\right)  =c_{m,n}\prod_{i\in\left[  n\right]  }\left(
1-\frac{\lambda_{i}}{m}\right)  ^{m-2n}I_{\lambda_{i}\leq m} \label{reffy}%
\end{equation}
for some constant $c_{m,n}$, where $I_{\lambda_{i}\leq m}$\ equals $1$ if
$\lambda_{i}\leq m$\ and $0$\ otherwise. \ Here and throughout, the
$\lambda_{i}$'s should be understood as functions $\lambda_{i}\left(
X\right)  $\ of $X$.

Let $\lambda_{\max}:=\max_{i}\lambda_{i}$\ be the greatest squared spectral
value of $X$. Then we can divide the space $\mathbb{C}^{n\times n}$\ of
matrices into two parts: the \textit{head} $R_{\operatorname*{head}}$,
consisting of matrices $X$\ such that $\lambda_{\max}\leq k$, and the
\textit{tail} $R_{\operatorname*{tail}}$, consisting of matrices $X$\ such
that $\lambda_{\max}>k$,\ for a value $k\leq\frac{m}{2n^{2}}$\ that we will
set later. \ At a high level, our strategy for upper-bounding $\Delta\left(
p_{G},p_{S}\right)  $\ will be to show that the head distributions are close
and the tail distributions are small. \ More formally, define%
\begin{align*}
g_{\operatorname*{head}}  &  :=\int_{X\in R_{\operatorname*{head}}}%
p_{G}\left(  X\right)  dX,\\
s_{\operatorname*{head}}  &  :=\int_{X\in R_{\operatorname*{head}}}%
p_{S}\left(  X\right)  dX,\\
\Delta_{\operatorname*{head}}  &  :=\int_{X\in R_{\operatorname*{head}}%
}\left\vert p_{G}\left(  X\right)  -p_{S}\left(  X\right)  \right\vert dX,
\end{align*}
and define $g_{\operatorname*{tail}}$, $s_{\operatorname*{tail}}$, and
$\Delta_{\operatorname*{tail}}$\ similarly with integrals over
$R_{\operatorname*{tail}}$. \ Note that $g_{\operatorname*{head}%
}+g_{\operatorname*{tail}}=s_{\operatorname*{head}}+s_{\operatorname*{tail}%
}=1$ by normalization. \ Also, by the triangle inequality,%
\[
\Delta\left(  p_{G},p_{S}\right)  =\Delta_{\operatorname*{head}}%
+\Delta_{\operatorname*{tail}}\leq\Delta_{\operatorname*{head}}%
+g_{\operatorname*{tail}}+s_{\operatorname*{tail}}.
\]
So to upper-bound $\Delta\left(  p_{G},p_{S}\right)  $, it suffices to
upper-bound $g_{\operatorname*{tail}}$, $s_{\operatorname*{tail}}$, and
$\Delta_{\operatorname*{head}}$\ separately, which we now proceed to do in
that order.

\begin{lemma}
\label{gtail}$g_{\operatorname*{tail}}\leq n^{2}e^{-k/n^{2}}.$
\end{lemma}

\begin{proof}
We have%
\begin{align*}
g_{\operatorname*{tail}}  &  =\Pr_{X\sim\mathcal{G}^{n\times n}}\left[
\lambda_{\max}>k\right] \\
&  \leq\Pr_{X\sim\mathcal{G}^{n\times n}}\left[
%TCIMACRO{\tsum \nolimits_{i,j\in\left[  n\right]  }}%
%BeginExpansion
{\textstyle\sum\nolimits_{i,j\in\left[  n\right]  }}
%EndExpansion
\left\vert x_{ij}\right\vert ^{2}>k\right] \\
&  \leq\sum_{i,j\in\left[  n\right]  }\Pr_{X\sim\mathcal{G}^{n\times n}%
}\left[  \left\vert x_{ij}\right\vert ^{2}>\frac{k}{n^{2}}\right] \\
&  =n^{2}e^{-k/n^{2}},
\end{align*}
where the second line uses the identity (\ref{frobenius}) and the third line
uses the union bound.
\end{proof}

\begin{lemma}
\label{stail}$s_{\operatorname*{tail}}\leq n^{2}e^{-k/(2n^{2})}.$
\end{lemma}

\begin{proof}
Recall that $\mathcal{H}_{m,m}$\ is the Haar measure over $m\times m$\ unitary
matrices. \ Then for a single entry (say $u_{11}$)\ of a matrix $U=\left(
u_{ij}\right)  $ drawn from $\mathcal{H}_{m,m}$,%
\[
\Pr_{U\sim\mathcal{H}_{m.m}}\left[  \left\vert u_{11}\right\vert ^{2}\geq
r\right]  =\left(  1-r\right)  ^{m-1}%
\]
for all $r\in\left[  0,1\right]  $, which can be calculated from the density
function given by R\'{e}ffy \cite{reffy}\ for the case $n=1$. \ So as in Lemma
\ref{gtail},%
\begin{align*}
s_{\operatorname*{tail}}  &  =\Pr_{X\sim\mathcal{S}_{m,n}}\left[
\lambda_{\max}>k\right] \\
&  \leq\Pr_{X\sim\mathcal{S}_{m,n}}\left[
%TCIMACRO{\tsum \nolimits_{i,j\in\left[  n\right]  }}%
%BeginExpansion
{\textstyle\sum\nolimits_{i,j\in\left[  n\right]  }}
%EndExpansion
\left\vert x_{ij}\right\vert ^{2}>k\right] \\
&  \leq\sum_{i,j\in\left[  n\right]  }\Pr_{X\sim\mathcal{S}_{m,n}}\left[
\left\vert x_{ij}\right\vert ^{2}>\frac{k}{n^{2}}\right] \\
&  =n^{2}\Pr_{U\sim\mathcal{H}_{m,m}}\left[  \left\vert u_{11}\right\vert
^{2}>\frac{k}{mn^{2}}\right] \\
&  =n^{2}\left(  1-\frac{k}{mn^{2}}\right)  ^{m-1}\\
&  <n^{2}e^{-k\left(  1-1/m\right)  /n^{2}}\\
&  <n^{2}e^{-k/(2n^{2})}.
\end{align*}

\end{proof}

The rest of the proof is devoted to upper-bounding $\Delta
_{\operatorname*{head}}$, the distance\ between the two head distributions.
\ Recall that R\'{e}ffy's formula for the density function $p_{S}\left(
X\right)  $ (equation (\ref{reffy})) involved a multiplicative constant
$c_{m,n}$. \ Since it is difficult to compute the value of $c_{m,n}%
$\ explicitly, we will instead define%
\[
\zeta:=\frac{\left(  1/\pi\right)  ^{n^{2}}}{c_{m,n}},
\]
and consider the scaled density function%
\[
\widetilde{p}_{S}\left(  X\right)  :=\zeta\cdot p_{S}\left(  X\right)
=\frac{1}{\pi^{n^{2}}}\prod_{i\in\left[  n\right]  }\left(  1-\frac
{\lambda_{i}}{m}\right)  ^{m-2n}I_{\lambda_{i}\leq m}.
\]
We will first show that $p_{G}$\ and $\widetilde{p}_{S}$\ are close on
$R_{\operatorname*{head}}$. \ We will then deduce from that result, together
with the fact that $g_{\operatorname*{tail}}$\ and $s_{\operatorname*{tail}}%
$\ are small, that $p_{G}$\ and $p_{S}$\ must be close on
$R_{\operatorname*{head}}$, which is what we wanted to show. \ Strangely,
nowhere in this argument do we ever bound $\zeta$\ directly. \ \textit{After}
proving Theorem \ref{truncthm}, however, we will then need to go back and show
that $\zeta$\ is close to $1$, on the way to proving Theorem \ref{truncthm2}.

Let%
\begin{equation}
\widetilde{\Delta}_{\operatorname*{head}}:=\int_{X\in R_{\operatorname*{head}%
}}\left\vert p_{G}\left(  X\right)  -\widetilde{p}_{S}\left(  X\right)
\right\vert dX. \label{dheadeq}%
\end{equation}
Then our first claim is the following.

\begin{lemma}
\label{dhead}$\widetilde{\Delta}_{\operatorname*{head}}\leq\frac{4nk\left(
n+k\right)  }{m}.$
\end{lemma}

\begin{proof}
As a first observation, when we restrict to $R_{\operatorname*{head}}$, we
have $\lambda_{i}\leq k\leq\frac{m}{2n^{2}}<m$\ for all $i\in\left[  n\right]
$ by assumption. \ So we can simplify the expression for $\widetilde{p}%
_{S}\left(  X\right)  $\ by removing the indicator variable $I_{\lambda
_{i}\leq m}$:%
\[
\widetilde{p}_{S}\left(  X\right)  =\frac{1}{\pi^{n^{2}}}\prod_{i\in\left[
n\right]  }\left(  1-\frac{\lambda_{i}}{m}\right)  ^{m-2n}.
\]
Now let us rewrite equation (\ref{dheadeq}) in the form%
\[
\widetilde{\Delta}_{\operatorname*{head}}=\int_{X\in R_{\operatorname*{head}}%
}p_{G}\left(  X\right)  \left\vert 1-\frac{\widetilde{p}_{S}\left(  X\right)
}{p_{G}\left(  X\right)  }\right\vert dX.
\]
Then plugging in the expressions for $\widetilde{p}_{S}\left(  X\right)
$\ and $p_{G}\left(  X\right)  $ respectively gives the ratio%
\begin{align*}
\frac{\widetilde{p}_{S}\left(  X\right)  }{p_{G}\left(  X\right)  }  &
=\frac{\pi^{-n^{2}}\prod_{i\in\left[  n\right]  }\left(  1-\lambda
_{i}/m\right)  ^{m-2n}}{\pi^{-n^{2}}\prod_{i\in\left[  n\right]  }%
e^{-\lambda_{i}}}\\
&  =\exp\left(  \sum_{i\in\left[  n\right]  }f\left(  \lambda_{i}\right)
\right)  ,
\end{align*}
where%
\begin{align*}
f\left(  \lambda_{i}\right)   &  =\ln\frac{\left(  1-\lambda_{i}/m\right)
^{m-2n}}{e^{-\lambda_{i}}}\\
&  =\lambda_{i}-\left(  m-2n\right)  \left(  -\ln\left(  1-\lambda
_{i}/m\right)  \right)  .
\end{align*}
Since $0\leq\lambda_{i}<m$, we may use the Taylor expansion%
\[
-\ln\left(  1-\lambda_{i}/m\right)  =\frac{\lambda_{i}}{m}+\frac{1}{2}%
\frac{\lambda_{i}^{2}}{m^{2}}+\frac{1}{3}\frac{\lambda_{i}^{3}}{m^{3}}+\cdots
\]
So we can upper-bound $f\left(  \lambda_{i}\right)  $\ by%
\begin{align*}
f\left(  \lambda_{i}\right)   &  \leq\lambda_{i}-\left(  m-2n\right)
\frac{\lambda_{i}}{m}\\
&  =\frac{2n\lambda_{i}}{m}\\
&  \leq\frac{2nk}{m},
\end{align*}
and can \textit{lower}-bound $f\left(  \lambda_{i}\right)  $ by%
\begin{align*}
f\left(  \lambda_{i}\right)   &  \geq\lambda_{i}-\left(  m-2n\right)  \left(
\frac{\lambda_{i}}{m}+\frac{1}{2}\frac{\lambda_{i}^{2}}{m^{2}}+\frac{1}%
{3}\frac{\lambda_{i}^{3}}{m^{3}}+\cdots\right) \\
&  >\lambda_{i}-\left(  m-2n\right)  \left(  \frac{\lambda_{i}}{m}%
+\frac{\lambda_{i}^{2}}{m^{2}}+\frac{\lambda_{i}^{3}}{m^{3}}+\cdots\right) \\
&  =\lambda_{i}-\frac{\left(  m-2n\right)  \lambda_{i}}{m\left(  1-\lambda
_{i}/m\right)  }\\
&  >\lambda_{i}-\frac{\lambda_{i}}{1-\lambda_{i}/m}\\
&  >-\frac{\lambda_{i}^{2}}{m-\lambda_{i}}\\
&  \geq-\frac{2k^{2}}{m}.
\end{align*}
Here the last line used the fact that $\lambda_{i}\leq k\leq\frac{m}{2n^{2}%
}<\frac{m}{2}$, since $X\in R_{\operatorname*{head}}$. \ It follows that%
\[
-\frac{2nk^{2}}{m}\leq\sum_{i\in\left[  n\right]  }f\left(  \lambda
_{i}\right)  \leq\frac{2n^{2}k}{m}.
\]
So%
\begin{align*}
\left\vert 1-\frac{\widetilde{p}_{S}\left(  X\right)  }{p_{G}\left(  X\right)
}\right\vert  &  =\left\vert 1-\exp\left(  \sum_{i\in\left[  n\right]
}f\left(  \lambda_{i}\right)  \right)  \right\vert \\
&  \leq\max\left\{  1-\exp\left(  -\frac{2nk^{2}}{m}\right)  ,\exp\left(
\frac{2n^{2}k}{m}\right)  -1\right\} \\
&  \leq\max\left\{  \frac{2nk^{2}}{m},\frac{4n^{2}k}{m}\right\} \\
&  \leq\frac{4nk\left(  n+k\right)  }{m}%
\end{align*}
where the last line used the fact that $e^{\delta}-1<2\delta$\ for all
$\delta\leq1$.

To conclude,%
\begin{align*}
\widetilde{\Delta}_{\operatorname*{head}}  &  \leq\int_{X\in
R_{\operatorname*{head}}}p_{G}\left(  X\right)  \left[  \frac{4nk\left(
n+k\right)  }{m}\right]  dX\\
&  \leq\frac{4nk\left(  n+k\right)  }{m}.
\end{align*}

\end{proof}

Combining Lemmas \ref{gtail}, \ref{stail}, \ref{dhead}, and \ref{alpha1}, and
making repeated use of the triangle inequality, we find that%
\begin{align*}
\Delta_{\operatorname*{head}}  &  =\int_{X\in R_{\operatorname*{head}}%
}\left\vert p_{G}\left(  X\right)  -p_{S}\left(  X\right)  \right\vert dX\\
&  \leq\widetilde{\Delta}_{\operatorname*{head}}+\int_{X\in
R_{\operatorname*{head}}}\left\vert \widetilde{p}_{S}\left(  X\right)
-p_{S}\left(  X\right)  \right\vert dX\\
&  =\widetilde{\Delta}_{\operatorname*{head}}+\left\vert \zeta
s_{\operatorname*{head}}-s_{\operatorname*{head}}\right\vert \\
&  \leq\widetilde{\Delta}_{\operatorname*{head}}+\left\vert \zeta
s_{\operatorname*{head}}-g_{\operatorname*{head}}\right\vert +\left\vert
g_{\operatorname*{head}}-1\right\vert +\left\vert 1-s_{\operatorname*{head}%
}\right\vert \\
&  \leq2\widetilde{\Delta}_{\operatorname*{head}}+g_{\operatorname*{tail}%
}+s_{\operatorname*{tail}}\\
&  \leq\frac{8nk\left(  n+k\right)  }{m}+n^{2}e^{-k/n^{2}}+n^{2}%
e^{-k/(2n^{2})}.
\end{align*}
Therefore%
\begin{align*}
\Delta\left(  p_{G},p_{S}\right)   &  \leq\Delta_{\operatorname*{head}%
}+g_{\operatorname*{tail}}+s_{\operatorname*{tail}}\\
&  \leq\frac{8nk\left(  n+k\right)  }{m}+2n^{2}e^{-k/n^{2}}+2n^{2}%
e^{-k/(2n^{2})}.
\end{align*}
Recalling that $m\geq\frac{n^{5}}{\delta}\log^{2}\frac{n}{\delta}$, let us now
make the choice $k:=6n^{2}\log\frac{n}{\delta}$. \ Then the constraint
$k\leq\frac{m}{2n^{2}}$\ is satisfied, and furthermore $\Delta\left(
p_{G},p_{S}\right)  =O\left(  \delta\right)  $. \ This completes the proof of
Theorem \ref{truncthm}.

The above derivation \textquotedblleft implicitly\textquotedblright\ showed
that $\zeta$\ is close to $1$. \ As a first step toward proving Theorem
\ref{truncthm2}, let\ us now make the bound on $\zeta$\ explicit.

\begin{lemma}
\label{alpha1}$\left\vert \zeta-1\right\vert =O\left(  \delta\right)  .$
\end{lemma}

\begin{proof}
We have%
\begin{align*}
\left\vert \zeta s_{\operatorname*{head}}-s_{\operatorname*{head}}\right\vert
&  \leq\left\vert \zeta s_{\operatorname*{head}}-g_{\operatorname*{head}%
}\right\vert +\left\vert g_{\operatorname*{head}}-1\right\vert +\left\vert
1-s_{\operatorname*{head}}\right\vert \\
&  =\widetilde{\Delta}_{\operatorname*{head}}+g_{\operatorname*{tail}%
}+s_{\operatorname*{tail}}\\
&  \leq\frac{4nk\left(  n+k\right)  }{m}+n^{2}e^{-k/n^{2}}+n^{2}%
e^{-k/(2n^{2})}%
\end{align*}
and%
\[
s_{\operatorname*{head}}=1-s_{\operatorname*{tail}}\geq1-n^{2}e^{-k/(2n^{2}%
)}.
\]
As before, recall that $m\geq\frac{n^{5}}{\delta}\log^{2}\frac{n}{\delta}%
$\ and set $k:=6n^{2}\log\frac{n}{\delta}$. \ Then%
\begin{align*}
\left\vert \zeta-1\right\vert  &  =\frac{\left\vert \zeta
s_{\operatorname*{head}}-s_{\operatorname*{head}}\right\vert }%
{s_{\operatorname*{head}}}\\
&  \leq\frac{4nk\left(  n+k\right)  /m+n^{2}e^{-k/n^{2}}+n^{2}e^{-k/(2n^{2})}%
}{1-n^{2}e^{-k/(2n^{2})}}\\
&  =O\left(  \delta\right)  .
\end{align*}

\end{proof}

We can now prove Theorem \ref{truncthm2}, that $p_{S}\left(  X\right)
\leq\left(  1+O\left(  \delta\right)  \right)  p_{G}\left(  X\right)  $ for
all $X\in\mathbb{C}^{n\times n}$.

\begin{proof}
[Proof of Theorem \ref{truncthm2}]Our goal is to upper-bound%
\[
C:=\max_{X\in\mathbb{C}^{n\times n}}\frac{p_{S}\left(  X\right)  }%
{p_{G}\left(  X\right)  }.
\]
Using the notation of Lemma \ref{dhead}, we can rewrite $C$ as%
\[
\frac{1}{\zeta}\max_{X\in\mathbb{C}^{n\times n}}\frac{\widetilde{p}_{S}\left(
X\right)  }{p_{G}\left(  X\right)  }=\frac{1}{\zeta}\max_{\lambda_{1}%
,\ldots,\lambda_{n}\geq0}\exp\left(  \sum_{i\in\left[  n\right]  }f\left(
\lambda_{i}\right)  \right)  ,
\]
where%
\[
f\left(  \lambda_{i}\right)  :=\lambda_{i}+\left(  m-2n\right)  \ln\left(
1-\lambda_{i}/m\right)  .
\]
By elementary calculus, the function $f\left(  \lambda\right)  $\ achieves its
maximum at $\lambda=2n$; note that this is a valid maximum since $m\geq2n$.
\ Setting $\lambda_{i}=2n$\ for all $i$ then yields%
\begin{align*}
C  &  =\frac{1}{\zeta}\exp\left(  2n^{2}+n\left(  m-2n\right)  \ln\left(
1-\frac{2n}{m}\right)  \right) \\
&  =\frac{1}{\zeta}e^{2n^{2}}\left(  1-\frac{2n}{m}\right)  ^{n\left(
m-2n\right)  }\\
&  <\frac{1}{\zeta}e^{2n^{2}}e^{-2n^{2}\left(  m-2n\right)  /m}\\
&  =\frac{1}{\zeta}e^{4n^{3}/m}\\
&  \leq\frac{1}{1-O\left(  \delta\right)  }\left(  1+O\left(  \delta\right)
\right) \\
&  =1+O\left(  \delta\right)  .
\end{align*}
Here the second-to-last line used Lemma \ref{alpha1}, together with the fact
that $m\gg\frac{4n^{3}}{\delta}$.
\end{proof}

\subsection{\label{MAINRESULT}Hardness of Approximate \textsc{BosonSampling}}

Having proved Theorem \ref{truncthm2}, we are finally ready to prove the main
result of the paper: that $\left\vert \text{\textsc{GPE}}\right\vert _{\pm
}^{2}\in\mathsf{FBPP}^{\mathsf{NP}^{\mathcal{O}}}$, where $\mathcal{O}$\ is
any approximate \textsc{BosonSampling}\ oracle. \ In other words, if there is
a fast classical algorithm for approximate \textsc{BosonSampling}, then there
is also a\ $\mathsf{BPP}^{\mathsf{NP}}$\ algorithm to\ estimate $\left\vert
\operatorname*{Per}\left(  X\right)  \right\vert ^{2}$, with high probability
for a Gaussian random matrix $X\sim\mathcal{G}^{n\times n}$.

We first need a technical lemma, which formalizes the well-known concept of
rejection sampling.

\begin{lemma}
[Rejection Sampling]\label{rejsample}Let\ $\mathcal{D}=\left\{  p_{x}\right\}
$ and $\mathcal{E}=\left\{  q_{x}\right\}  $\ be any two distributions over a
finite set $S$. \ Suppose that there exists a polynomial-time algorithm to
compute $\zeta q_{x}/p_{x}$\ given $x\in S$, where $\zeta$\ is some constant
independent of $x$ such that $\left\vert \zeta-1\right\vert \leq\delta$.
\ Suppose also that $q_{x}/p_{x}\leq1+\delta$\ for all $x\in S$. \ Then there
exists a $\mathsf{BPP}$ algorithm $\mathcal{R}$\ that takes a sample
$x\sim\mathcal{D}$\ as input, and either accepts or rejects. $\mathcal{R}%
$\ has the following properties:

\begin{enumerate}
\item[(i)] Conditioned on $\mathcal{R}$\ accepting, $x$\ is distributed
according to $\mathcal{E}$.

\item[(ii)] The probability that $\mathcal{R}$\ rejects (over both its
internal randomness and $x\sim\mathcal{D}$) is $O\left(  \delta\right)  $.
\end{enumerate}
\end{lemma}

\begin{proof}
$\mathcal{R}$\ works as follows: first compute $\zeta q_{x}/p_{x}$;\ then
accept with probability $\frac{\zeta q_{x}/p_{x}}{\left(  1+\delta\right)
^{2}}\leq1$. \ Property (i) is immediate. \ For property (ii),%
\begin{align*}
\Pr\left[  \mathcal{R}~\text{rejects}\right]   &  =\sum_{x\in S}p_{x}\left(
1-\frac{\zeta q_{x}/p_{x}}{\left(  1+\delta\right)  ^{2}}\right) \\
&  =\sum_{x\in S}\left(  p_{x}-\frac{\zeta q_{x}}{\left(  1+\delta\right)
^{2}}\right) \\
&  =1-\frac{\zeta}{\left(  1+\delta\right)  ^{2}}\\
&  =O\left(  \delta\right)  .
\end{align*}

\end{proof}

By combining Lemma \ref{rejsample} with Theorem \ref{truncthm2}, we now show
how it is possible to \textquotedblleft hide\textquotedblright\ a matrix
$X\sim\mathcal{G}^{n\times n}$\ of i.i.d.\ Gaussians as a random $n\times
n$\ submatrix\ of a Haar-random $m\times n$\ column-orthonormal matrix $A$,
provided $m=\Omega\left(  n^{5}\log^{2}n\right)  $. \ Our hiding procedure
does not involve any distortion of $X$. \ We believe that the hiding procedure
could be implemented in $\mathsf{BPP}$; however, we will show only that it can
be implemented in $\mathsf{BPP}^{\mathsf{NP}}$, since that is easier and
suffices for our application.

\begin{lemma}
[Hiding Lemma]\label{hidelemma}Let $m\geq\frac{n^{5}}{\delta}\log^{2}\frac
{n}{\delta}$ for some $\delta>0$. \ Then there exists a\ $\mathsf{BPP}%
^{\mathsf{NP}}$\ algorithm\ $\mathcal{A}$\ that takes as input a matrix
$X\sim\mathcal{G}^{n\times n}$, that \textquotedblleft
succeeds\textquotedblright\ with probability $1-O\left(  \delta\right)
$\ over $X$, and that, conditioned on succeeding, samples a\ matrix
$A\in\mathcal{U}_{m,n}$\ from a probability distribution $\mathcal{D}_{X}%
$,\ such that the following properties hold:

\begin{enumerate}
\item[(i)] $X/\sqrt{m}$ occurs as a uniformly-random $n\times n$\ submatrix of
$A\sim\mathcal{D}_{X}$, for every $X$\ such that $\Pr\left[  \mathcal{A}%
\left(  X\right)  ~\text{succeeds}\right]  >0$.

\item[(ii)] The distribution over $A\in\mathbb{C}^{m\times n}$\ induced by
drawing $X\sim\mathcal{G}^{n\times n}$, running $\mathcal{A}\left(  X\right)
$, and conditioning on $\mathcal{A}\left(  X\right)  $\ succeeding is simply
$\mathcal{H}_{m,n}$ (the Haar measure over $m\times n$\ column-orthonormal matrices).
\end{enumerate}
\end{lemma}

\begin{proof}
Given a sample $X\sim\mathcal{G}^{n\times n}$, the first step is to
\textquotedblleft convert\textquotedblright\ $X$\ into a sample from the
truncated Haar measure $\mathcal{S}_{m,n}$. \ To do so, we use the rejection
sampling procedure\ from Lemma \ref{rejsample}. \ By Theorem \ref{truncthm2},
we have $p_{S}\left(  X\right)  /p_{G}\left(  X\right)  \leq1+O\left(
\delta\right)  $\ for all $X\in\mathbb{C}^{n\times n}$, where $p_{S}$\ and
$p_{G}$\ are the probability density functions of $\mathcal{S}_{m,n}$\ and
$\mathcal{G}^{n\times n}$\ respectively. \ Also, letting $\zeta:=\left(
1/\pi\right)  ^{n^{2}}/c_{m,n}$\ be the constant from Section \ref{EMBEDDING},
we have%
\[
\frac{\zeta\cdot p_{S}\left(  X\right)  }{p_{G}\left(  X\right)  }%
=\frac{\widetilde{p}_{S}\left(  X\right)  }{p_{G}\left(  X\right)  }%
=\frac{\prod_{i\in\left[  n\right]  }\left(  1-\lambda_{i}/m\right)  ^{m-2n}%
}{\prod_{i\in\left[  n\right]  }e^{-\lambda_{i}}},
\]
which is clearly computable in polynomial time (to any desired precision)
given $X$. \ Finally, we saw from Lemma \ref{alpha1}\ that $\left\vert
\zeta-1\right\vert =O\left(  \delta\right)  $.

So by Lemma \ref{rejsample}, the rejection sampling procedure $\mathcal{R}%
$\ has the following properties:

\begin{enumerate}
\item[(1)] $\mathcal{R}$\ can be implemented in $\mathsf{BPP}$.

\item[(2)] $\mathcal{R}$ rejects with probability $O\left(  \delta\right)  $.

\item[(3)] Conditioned on $\mathcal{R}$ accepting, we have $X\sim
\mathcal{S}_{m,n}$.
\end{enumerate}

Now suppose $\mathcal{R}$\ accepts, and let $X^{\prime}:=X/\sqrt{m}$. \ Then
our problem reduces to embedding $X^{\prime}$\ as a random submatrix of a
sample $A$\ from $\mathcal{H}_{m,n}$. \ We do this as follows. \ Given a
matrix $A\in\mathcal{U}_{m,n}$, let $E_{X}\left(  A\right)  $\ be the event
that $X^{\prime}$\ occurs as an $n\times n$\ submatrix of $A$. \ Then let
$\mathcal{D}_{X}$\ be the distribution over $A\in\mathcal{U}_{m,n}$\ obtained
by first sampling $A$\ from $\mathcal{H}_{m,n}$, and then conditioning on
$E_{X}\left(  A\right)  $ holding. \ Note that $\mathcal{D}_{X}$\ is
well-defined, since for every $X$ in the support of $\mathcal{S}_{m,n}$, there
is \textit{some} $A\in\mathcal{U}_{m,n}$\ satisfying $E_{X}\left(  A\right)  $.

We now check that $\mathcal{D}_{X}$\ satisfies properties (i) and (ii). \ For
(i), every element in the support of $\mathcal{D}_{X}$\ contains $X^{\prime}$
as a submatrix by definition, and by symmetry, this $X^{\prime}$ occurs at a
uniformly-random location. \ For (ii), notice that we could equally well have
sampled $A\sim\mathcal{D}_{X}$\ by first sampling $X\sim\mathcal{S}_{m,n}$,
then placing $X^{\prime}$\ at a uniformly-random location within $A$, and
finally \textquotedblleft filling in\textquotedblright\ the remaining $\left(
m-n\right)  \times n$\ block of $A$\ by drawing it from $\mathcal{H}_{m,n}%
$\ conditioned on $X^{\prime}$. \ From this perspective, however, it is clear
that $A$\ is Haar-random, since $\mathcal{S}_{m,n}$\ was just a truncation of
$\mathcal{H}_{m,n}$ to begin with.

The last thing we need to show is that, given $X$ as input, we can sample from
$\mathcal{D}_{X}$ in $\mathsf{BPP}^{\mathsf{NP}}$. \ As a first step, we can
certainly sample from $\mathcal{H}_{m,n}$\ in $\mathsf{BPP}$. \ To do so, for
example, we can first generate a matrix $A\sim\mathcal{G}^{m\times n}$\ of
independent Gaussians,\ and then apply the Gram-Schmidt orthogonalization
procedure to $A$. \ Now, given a $\mathsf{BPP}$ algorithm that samples
$A\sim\mathcal{H}_{m,n}$, the remaining task is to condition on the event
$E_{X}\left(  A\right)  $. \ Given $X$\ and $A$, it is easy to check whether
$E_{X}\left(  A\right)  $\ holds. \ But this means that we can sample from the
conditional distribution $\mathcal{D}_{X}$\ in the complexity class
$\mathsf{PostBPP}$.

Composing a $\mathsf{BPP}$\ algorithm with a $\mathsf{PostBPP}$\ one yields an
algorithm that runs in $\mathsf{BP}\cdot\mathsf{PostBPP}\subseteq
\mathsf{BPP}^{\mathsf{NP}}$.
\end{proof}

The final step is to prove that, \textit{if} we had an oracle $\mathcal{O}%
$\ for approximate \textsc{BosonSampling}, then by using $\mathcal{O}$\ in
conjunction with the hiding procedure from Lemma \ref{hidelemma},\ we could
estimate $\left\vert \operatorname*{Per}\left(  X\right)  \right\vert ^{2}%
$\ in $\mathsf{BPP}^{\mathsf{NP}}$, where $X\sim\mathcal{G}^{n\times n}$\ is a
Gaussian input matrix.

To prove this theorem, we need to recall some definitions from previous
sections. \ The set of tuples $S=\left(  s_{1},\ldots,s_{m}\right)  $
satisfying $s_{1},\ldots,s_{m}\geq0$\ and $s_{1}+\cdots+s_{m}=n$ is denoted
$\Phi_{m,n}$. \ Given a matrix $A\in\mathcal{U}_{m,n}$, we denote by
$\mathcal{D}_{A}$\ the distribution over $\Phi_{m,n}$\ where each $S$\ occurs
with probability%
\[
\Pr_{\mathcal{D}_{A}}\left[  S\right]  =\frac{\left\vert \operatorname*{Per}%
\left(  A_{S}\right)  \right\vert ^{2}}{s_{1}!\cdots s_{m}!}.
\]
Also, recall that in the $\left\vert \text{\textsc{GPE}}\right\vert _{\pm}%
^{2}$\ problem, we are given an input of the form $\left\langle
X,0^{1/\varepsilon},0^{1/\delta}\right\rangle $, where $X$\ is an $n\times n$
matrix drawn from the Gaussian distribution\ $\mathcal{G}^{n\times n}$.\ \ The
goal is to approximate $\left\vert \operatorname*{Per}\left(  X\right)
\right\vert ^{2}$\ to within an additive error $\varepsilon\cdot n!$, with
probability at least $1-\delta$\ over $X$.

We now prove Theorem \ref{mainresult}, our main result. \ Let us restate the
theorem for convenience:

\begin{quotation}
\textit{Let }$\mathcal{O}$\textit{\ be any approximate }\textsc{BosonSampling}%
\textit{\ oracle. \ Then }$\left\vert \text{\textsc{GPE}}\right\vert _{\pm
}^{2}\in\mathsf{FBPP}^{\mathsf{NP}^{\mathcal{O}}}$.
\end{quotation}

\begin{proof}
[Proof of Theorem \ref{mainresult}]Let $X\sim\mathcal{G}^{n\times n}$ be an
input matrix, and let $\varepsilon,\delta>0$\ be error parameters. \ Then we
need to show how to approximate $\left\vert \operatorname*{Per}\left(
X\right)  \right\vert ^{2}$ to within an additive error $\varepsilon\cdot
n!$,\ with probability at least $1-\delta$\ over $X$, in the complexity class
$\mathsf{FBPP}^{\mathsf{NP}^{\mathcal{O}}}$. \ The running time should be
polynomial in $n$,\ $1/\varepsilon$, and $1/\delta$.

Let $m:=\frac{K}{\delta}n^{5}\log^{2}\frac{n}{\delta}$, where $K$ is a
suitably large constant. \ Also, let $X^{\prime}:=X/\sqrt{m}$\ be a scaled
version of $X$. \ Then we can state our problem equivalently as follows:
approximate%
\[
\left\vert \operatorname*{Per}\left(  X^{\prime}\right)  \right\vert
^{2}=\frac{\left\vert \operatorname*{Per}\left(  X\right)  \right\vert ^{2}%
}{m^{n}}%
\]
to within an additive error $\varepsilon\cdot n!/m^{n}$.

As a first step, Lemma \ref{hidelemma} says that in $\mathsf{BPP}%
^{\mathsf{NP}}$, and with high probability over $X^{\prime}$, we can generate
a matrix $A\in\mathcal{U}^{m\times n}$ that is exactly Haar-random, and that
contains $X^{\prime}$\ as a random $n\times n$\ submatrix. \ So certainly we
can generate such an $A$\ in $\mathsf{FBPP}^{\mathsf{NP}^{\mathcal{O}}}%
$\ (indeed, without using the oracle $\mathcal{O}$). \ Provided we chose $K$
sufficiently large, this procedure will succeed with probability at least
(say) $1-\delta/4$.

Set $\beta:=\varepsilon\delta/24$. \ Suppose we feed $\left\langle
A,0^{1/\beta},r\right\rangle $\ to the approximate \textsc{BosonSampling}%
\ oracle $\mathcal{O}$, where $r\in\left\{  0,1\right\}
^{\operatorname*{poly}\left(  m\right)  }$\ is a random string. \ Then by
definition, as $r$ is varied, $\mathcal{O}$\ returns a sample from a
probability distribution $\mathcal{D}_{A}^{\prime}$\ such that $\left\Vert
\mathcal{D}_{A}-\mathcal{D}_{A}^{\prime}\right\Vert \leq\beta$.

Let $p_{S}:=\Pr_{\mathcal{D}_{A}}\left[  S\right]  $ and $q_{S}:=\Pr
_{\mathcal{D}_{A}^{\prime}}\left[  S\right]  $ for all $S\in\Phi_{m,n}$.
\ Also, let $W\subset\left[  m\right]  $\ be the subset of $n$\ rows of $A$ in
which $X^{\prime}$\ occurs as a submatrix. \ Then we will be particularly
interested in the basis state $S^{\ast}=\left(  s_{1},\ldots,s_{m}\right)  $,
which is defined by $s_{i}=1$\ if $i\in W$\ and $s_{i}=0$\ otherwise. \ Notice
that%
\[
p_{S^{\ast}}=\frac{\left\vert \operatorname*{Per}\left(  A_{S^{\ast}}\right)
\right\vert ^{2}}{s_{1}!\cdots s_{m}!}=\left\vert \operatorname*{Per}\left(
X^{\prime}\right)  \right\vert ^{2},
\]
and that%
\[
q_{S^{\ast}}=\Pr_{\mathcal{D}_{A}^{\prime}}\left[  S^{\ast}\right]  =\Pr
_{r\in\left\{  0,1\right\}  ^{\operatorname*{poly}\left(  m\right)  }}\left[
\mathcal{O}\left(  A,0^{1/\beta},r\right)  =S^{\ast}\right]  .
\]
In other words: $p_{S^{\ast}}$\ encodes the squared permanent that we are
trying to approximate, while $q_{S^{\ast}}$\ can be approximated in
$\mathsf{FBPP}^{\mathsf{NP}^{\mathcal{O}}}$\ using Stockmeyer's approximate
counting method (Theorem \ref{approxcount}). \ Therefore, to show that with
high probability we can approximate $p_{S^{\ast}}$\ in $\mathsf{FBPP}%
^{\mathsf{NP}^{\mathcal{O}}}$, it suffices to show that $p_{S^{\ast}}$\ and
$q_{S^{\ast}}$\ are close with high probability over $X$ and $A$.

Call a basis state $S\in\Phi_{m,n}$\ \textit{collision-free} if each $s_{i}%
$\ is either $0$ or $1$. \ Let $G_{m,n}$\ be the set of collision-free $S$'s,
and notice that $S^{\ast}\in G_{m,n}$. \ From now on, we will find it
convenient to restrict attention to $G_{m,n}$.

Let $\Delta_{S}:=\left\vert p_{S}-q_{S}\right\vert $, so that%
\[
\left\Vert \mathcal{D}_{A}-\mathcal{D}_{A}^{\prime}\right\Vert =\frac{1}%
{2}\sum_{S\in\Phi_{m,n}}\Delta_{S}.
\]
Then%
\begin{align*}
\operatorname*{E}_{S\in G_{m,n}}\left[  \Delta_{S}\right]   &  \leq\frac
{\sum_{S\in\Phi_{m,n}}\Delta_{S}}{\left\vert G_{m,n}\right\vert }\\
&  =\frac{2\left\Vert \mathcal{D}_{A}-\mathcal{D}_{A}^{\prime}\right\Vert
}{\left\vert G_{m,n}\right\vert }\\
&  \leq\frac{2\beta}{\binom{m}{n}}\\
&  <3\beta\cdot\frac{n!}{m^{n}},
\end{align*}
where the last line used the fact that $m=\omega\left(  n^{2}\right)  $. \ So
by Markov's inequality, for all $k>1$,%
\[
\Pr_{S\in G_{m,n}}\left[  \Delta_{S}>3\beta k\cdot\frac{n!}{m^{n}}\right]
<\frac{1}{k}.
\]
In particular, if we set $k:=4/\delta$ and notice that $4\beta k=12\beta
/\delta=\varepsilon/2$,%
\[
\Pr_{S\in G_{m,n}}\left[  \Delta_{S}>\frac{\varepsilon}{2}\cdot\frac{n!}%
{m^{n}}\right]  <\frac{\delta}{4}.
\]

Of course, our goal is to upper-bound $\Delta_{S^{\ast}}$, not $\Delta_{S}%
$\ for a randomly-chosen $S\in G_{m,n}$. \ However, a crucial observation is
that, from the perspective of $\mathcal{O}$---which sees only $A$, and not
$S^{\ast}$\ or $X^{\prime}$---the distribution over possible values of
$S^{\ast}$\ is simply the uniform one. \ To see this, notice that instead of
sampling $X$ \textit{and then} $A$ (as in Lemma \ref{hidelemma}), we could
have equally well generated the pair $\left\langle X,A\right\rangle $\ by
first sampling $A$\ from the Haar measure $\mathcal{H}_{m,n}$, and then
setting $X:=\sqrt{m}A_{S^{\ast}}$, for $S^{\ast}$\ chosen uniformly from
$G_{m,n}$. \ It follows that seeing $A$\ gives $\mathcal{O}$\ no information
whatsoever about the identity of $S^{\ast}$. \ So even if $\mathcal{O}$\ is
trying adversarially to maximize $\Delta_{S^{\ast}}$, we still have%
\[
\Pr_{X,A}\left[  \Delta_{S^{\ast}}>\frac{\varepsilon}{2}\cdot\frac{n!}{m^{n}%
}\right]  <\frac{\delta}{4}.
\]

Now suppose we use Stockmeyer's algorithm to approximate $q_{S^{\ast}}$\ in
$\mathsf{FBPP}^{\mathsf{NP}^{\mathcal{O}}}$. \ Then by Theorem
\ref{approxcount}, for all $\alpha>0$, we can obtain an estimate
$\widetilde{q}_{S^{\ast}}$\ such that%
\[
\Pr\left[  \left\vert \widetilde{q}_{S^{\ast}}-q_{S^{\ast}}\right\vert
>\alpha\cdot q_{S^{\ast}}\right]  <\frac{1}{2^{m}},
\]
in time polynomial in $m$\ and $1/\alpha$. \ Note that%
\[
\operatorname*{E}_{S\in G_{m,n}}\left[  q_{S}\right]  \leq\frac{1}{\left\vert
G_{m,n}\right\vert }=\frac{1}{\binom{m}{n}}<2\frac{n!}{m^{n}},
\]
so%
\[
\Pr_{S\in G_{m,n}}\left[  q_{S}>2k\cdot\frac{n!}{m^{n}}\right]  <\frac{1}{k}%
\]
for all $k>1$\ by Markov's inequality, so%
\[
\Pr_{X,A}\left[  q_{S^{\ast}}>2k\cdot\frac{n!}{m^{n}}\right]  <\frac{1}{k}%
\]
by the same symmetry principle used previously for $\Delta_{S^{\ast}}$.

Let us now make the choice $\alpha:=\varepsilon\delta/16$\ and $k:=4/\delta$.
\ Then putting everything together and applying the union bound,%
\begin{align*}
\Pr\left[  \left\vert \widetilde{q}_{S^{\ast}}-p_{S^{\ast}}\right\vert
>\varepsilon\cdot\frac{n!}{m^{n}}\right]   &  \leq\Pr\left[  \left\vert
\widetilde{q}_{S^{\ast}}-q_{S^{\ast}}\right\vert >\frac{\varepsilon}{2}%
\cdot\frac{n!}{m^{n}}\right]  +\Pr\left[  \left\vert q_{S^{\ast}}-p_{S^{\ast}%
}\right\vert >\frac{\varepsilon}{2}\cdot\frac{n!}{m^{n}}\right] \\
&  \leq\Pr\left[  q_{S^{\ast}}>2k\cdot\frac{n!}{m^{n}}\right]  +\Pr\left[
\left\vert \widetilde{q}_{S^{\ast}}-q_{S^{\ast}}\right\vert >\alpha\cdot
q_{S^{\ast}}\right]  +\Pr\left[  \Delta_{S^{\ast}}>\frac{\varepsilon}{2}%
\cdot\frac{n!}{m^{n}}\right] \\
&  <\frac{1}{k}+\frac{1}{2^{m}}+\frac{\delta}{4}\\
&  =\frac{\delta}{2}+\frac{1}{2^{m}},
\end{align*}
where the probabilities are over $X$\ and $A$ as well as the internal
randomness used by the approximate counting procedure. \ So, including the
probability that the algorithm $\mathcal{A}$\ from Lemma \ref{hidelemma}%
\ fails, the total probability that our $\mathsf{FBPP}^{\mathsf{NP}%
^{\mathcal{O}}}$\ machine fails to output a good enough approximation to
$p_{S^{\ast}}=\left\vert \operatorname*{Per}\left(  X^{\prime}\right)
\right\vert ^{2}$\ is at most%
\[
\frac{\delta}{4}+\left(  \frac{\delta}{2}+\frac{1}{2^{m}}\right)  <\delta,
\]
as desired. \ This completes the proof.
\end{proof}

\subsection{Implications\label{IMPLIC}}

In this section, we harvest\ some implications of Theorem \ref{mainresult} for
quantum complexity theory. \ First, if a fast classical algorithm for
\textsc{BosonSampling} exists, then it would have a surprising consequence for
the classical complexity of the $\left\vert \text{\textsc{GPE}}\right\vert
_{\pm}^{2}$\ problem.

\begin{corollary}
\label{cor0}Suppose \textsc{BosonSampling}$\in\mathsf{SampP}$.\ \ Then
$\left\vert \text{\textsc{GPE}}\right\vert _{\pm}^{2}\in\mathsf{FBPP}%
^{\mathsf{NP}}$. \ Indeed, even if \textsc{BosonSampling}$\in\mathsf{SampP}%
^{\mathsf{PH}}$, then $\left\vert \text{\textsc{GPE}}\right\vert _{\pm}^{2}%
\in\mathsf{FBPP}^{\mathsf{PH}}$.
\end{corollary}

However, we would also like evidence that a boson computer can solve
\textit{search} problems that are intractable classically. \ Fortunately, by
using Theorem \ref{samprel}---the \textquotedblleft Sampling/Searching
Equivalence Theorem\textquotedblright---we can obtain such evidence in a
completely automatic\ way. \ In particular, combining Corollary \ref{cor0}
with Theorem \ref{samprel}\ yields the following conclusion.

\begin{corollary}
\label{corb0}There exists a search problem $R\in\mathsf{BosonFP}\ $such that
$\left\vert \text{\textsc{GPE}}\right\vert _{\pm}^{2}\in\mathsf{FBPP}%
^{\mathsf{NP}^{\mathcal{O}}}$ for all computable oracles $\mathcal{O}$\ that
solve $R$. \ So in particular, if $\mathsf{BosonFP}\subseteq\mathsf{FBPP}$
(that is, all search problems solvable by a boson computer are also solvable
classically), then $\left\vert \text{\textsc{GPE}}\right\vert _{\pm}^{2}%
\in\mathsf{FBPP}^{\mathsf{NP}}$.
\end{corollary}

Recall from Theorem \ref{infbqp}\ that $\mathsf{BosonFP}\subseteq
\mathsf{FBQP}$: that is, linear-optics computers can be simulated efficiently
by \textquotedblleft ordinary\textquotedblright\ quantum computers. \ Thus,
Corollary \ref{corb0}\ implies in particular that, if $\mathsf{FBPP=FBQP}$,
then $\left\vert \text{\textsc{GPE}}\right\vert _{\pm}^{2}\in\mathsf{FBPP}%
^{\mathsf{NP}}$. \ Or in other words: if $\left\vert \text{\textsc{GPE}%
}\right\vert _{\pm}^{2}$\ is $\mathsf{\#P}$-hard, then $\mathsf{FBPP}$ cannot
equal $\mathsf{FBQP}$, unless $\mathsf{P}^{\#\mathsf{P}}=\mathsf{BPP}%
^{\mathsf{NP}}$\ and the polynomial hierarchy collapses. \ This would arguably
be our strongest evidence to date against the Extended Church-Turing Thesis.

In Sections \ref{GPE2GPE}, \ref{DISTPER}, and \ref{HARDPER}, we initiate a
program aimed at proving $\left\vert \text{\textsc{GPE}}\right\vert _{\pm}%
^{2}$\ is $\mathsf{\#P}$-hard.

\section{Experimental Prospects\label{EXPER}}

Our main goal in this paper was to define and study a \textit{theoretical}
model of quantum computing with noninteracting bosons. \ There are several
ways to motivate this model other than practical realizability: for
example,\ it abstracts a basic class of physical systems, it leads to
interesting new complexity classes between $\mathsf{BPP}$\ and $\mathsf{BQP}%
$,\ and it helped us provide evidence that quantum mechanics \textit{in
general} is hard to simulate classically. \ (In other words, even if we only
cared about \textquotedblleft standard\textquotedblright\ quantum computing,
we would not know how to prove results like Theorem \ref{mainresult}\ without
using linear optics as a proof tool.)

Clearly, though, a major motivation for our results is that they raise the
possibility of actually \textit{building} a scalable linear-optics computer,
and using it to solve the \textsc{BosonSampling}\ problem. \ By doing this,
one could hope to give evidence that nontrivial quantum computation is
possible, \textit{without} having to solve all the technological problems of
building a universal quantum computer. \ In other words, one could see our
results as suggesting a new path to testing the Extended Church-Turing Thesis,
which might be more experimentally accessible than alternative paths.

A serious discussion of implementation issues\ is outside the scope of this
paper. \ Here, though, we offer some preliminary observations\ that emerged
from our discussions with quantum optics experts. \ These observations concern
both the challenges of performing a \textsc{BosonSampling}\ experiment, and
the implications of such an experiment for complexity theory.

\subsection{The Generalized Hong-Ou-Mandel Dip\label{EPROPOSAL}}

From a physics standpoint, the experiment that we are asking for is
essentially a generalization of the \textit{Hong-Ou-Mandel dip} \cite{hom} to
three or more photons. \ The Hong-Ou-Mandel dip (see Figure \ref{homfig}) is a
well-known effect in quantum optics whereby two identical photons, which were
initially in different modes, become \textit{correlated} after passing through
a beamsplitter that applies the Hadamard transformation.%
%TCIMACRO{\FRAME{ftbpFU}{2.1707in}{2.0583in}{0pt}{\Qcb{The Hong-Ou-Mandel
%dip.}}{\Qlb{homfig}}{hom.eps}{\special{ language "Scientific Word";
%type "GRAPHIC";  maintain-aspect-ratio TRUE;  display "USEDEF";
%valid_file "F";  width 2.1707in;  height 2.0583in;  depth 0pt;
%original-width 4.2886in;  original-height 4.0629in;  cropleft "0";
%filename 'hom.eps';file-properties "XNPEU";}} }%
%BeginExpansion
\begin{figure}[ptb]%
\centering
\includegraphics[
height=2.0583in,
width=2.1707in
]%
{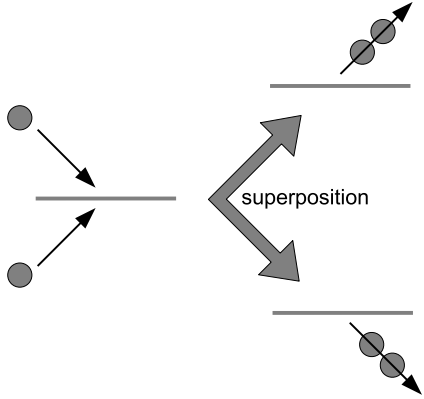}%
\caption{The Hong-Ou-Mandel dip.}%
\label{homfig}%
\end{figure}
%EndExpansion
More formally, the basis state $\left\vert 1,1\right\rangle $\ evolves to%
\[
\frac{\left\vert 2,0\right\rangle -\left\vert 0,2\right\rangle }{\sqrt{2}},
\]
so that a subsequent measurement reveals either both photons in the first mode
or else both photons in the second mode. \ This behavior is exactly what one
would predict from the model in Section \ref{MODEL}, in which $n$-photon
transition amplitudes are given by the permanents of $n\times n$ matrices.
\ More concretely, the amplitude of the basis state $\left\vert
1,1\right\rangle $\ \textquotedblleft dips\textquotedblright\ to $0$ because%
\[
\operatorname*{Per}\left(
\begin{array}
[c]{cc}%
\frac{1}{\sqrt{2}} & \frac{1}{\sqrt{2}}\\
\frac{1}{\sqrt{2}} & -\frac{1}{\sqrt{2}}%
\end{array}
\right)  =0,
\]
and hence there is destructive interference between the two paths mapping
$\left\vert 1,1\right\rangle $\ to itself.

Our challenge to experimentalists is to \textit{confirm directly that the
quantum-mechanical formula for }$n$\textit{-boson transition amplitudes in
terms of }$n\times n$\textit{\ permanents given in Section \ref{PERMDEF}%
,\ namely}%
\begin{equation}
\left\langle S|\varphi\left(  U\right)  |T\right\rangle =\frac
{\operatorname*{Per}\left(  U_{S,T}\right)  }{\sqrt{s_{1}!\cdots s_{m}%
!t_{1}!\cdots t_{m}!}}, \label{perform}%
\end{equation}
\textit{continues to hold for large values of }$n$\textit{. \ }In other words,
demonstrate a Hong-Ou-Mandel interference pattern involving as many identical
bosons as possible (though even $3$ or $4$ bosons would be of interest here).

The point of such an experiment would be to produce evidence that a
linear-optical network can indeed solve the \textsc{BosonSampling}\ problem in
a scalable way---and that therefore, no polynomial-time classical algorithm
can sample the observed distribution over photon numbers (modulo our
conjectures about the computational complexity of the permanent).

Admittedly, since complexity theory deals only with \textit{asymptotic}
statements, no finite experiment can answer the relevant questions
definitively. \ That is, even if formula (\ref{perform})\ were confirmed in
the case of $30$ identical bosons, a true-believer in the Extended
Church-Turing Thesis could always maintain that the formula would break down
for $31$ bosons, and so on. \ Thus, the goal here is simply to collect enough
evidence, for large enough $n$, that the ECT becomes less tenable as a
scientific hypothesis.

Of course, one should not choose $n$ \textit{so} large that a classical
computer cannot even efficiently \textit{verify} that the formula
(\ref{perform})\ holds! \ It is important to understand this difference
between the \textsc{BosonSampling}\ problem on the one hand, and $\mathsf{NP}%
$\ problems\ such as \textsc{Factoring} on the other. \ Unlike with
\textsc{Factoring}, there does not seem to be any \textit{witness} for
\textsc{BosonSampling} that a classical computer can efficiently verify, much
less a witness that a boson computer can produce.\footnote{Indeed, given
a\ matrix $X\in\mathbb{C}^{n\times n}$, there \textit{cannot} in general be an
$\mathsf{NP}$\ witness proving the value of $\operatorname*{Per}\left(
X\right)  $, unless $\mathsf{P}^{\#\mathsf{P}}=\mathsf{P}^{\mathsf{NP}}$ and
the polynomial hierarchy collapses. \ On the other hand, this argument does
not rule out an \textit{interactive protocol} with a $\mathsf{BPP}$\ verifier
and a \textsc{BosonSampling}\ prover. \ Whether any such protocol exists for
verifying statements not in $\mathsf{BPP}$ is an extremely interesting open
problem.} \ This means that, when $n$ is very large (say, more than $100$),
even if a linear-optics device is correctly solving \textsc{BosonSampling},
there might be no feasible way to prove this without presupposing the truth of
the physical laws being tested! \ Thus, for experimental purposes, the most
useful values of $n$ are presumably those for which a classical computer has
some difficulty computing an $n\times n$\ permanent, but can nevertheless do
so in order to confirm the results. \ We estimate this range as $10\leq
n\leq50$.

But how exactly should one verify formula (\ref{perform})? \ One approach
would be to perform full quantum state tomography on the output state of a
linear-optical network, or at least to characterize the distribution over
photon numbers. \ However, this approach would require a number of
experimental runs that grows exponentially with $n$, and is probably not needed.

Instead, given a system with $n$ identical photons and $m\geq n$\ modes, one
could do something like the following:

\begin{enumerate}
\item[(1)] Prepare the \textquotedblleft standard initial
state\textquotedblright\ $\left\vert 1_{n}\right\rangle $, in which modes
$1,\ldots,n$ are occupied with a single photon each and modes $n+1,\ldots
,m$\ are unoccupied.

\item[(2)] By passing the photons through a suitable network of beamsplitters
and phaseshifters, apply an $m\times m$\ mode-mixing unitary transformation
$U$. \ This maps the state $\left\vert 1_{n}\right\rangle $\ to $\varphi
\left(  U\right)  \left\vert 1_{n}\right\rangle $, where $\varphi\left(
U\right)  $\ is the induced action of $U$ on $n$-photon states.

\item[(3)] For each mode $i\in\left[  m\right]  $, measure the number of
photons $s_{i}$\ in the $i^{th}$ mode. \ This collapses the state
$\varphi\left(  U\right)  \left\vert 1_{n}\right\rangle $\ to some $\left\vert
S\right\rangle =\left\vert s_{1},\ldots,s_{m}\right\rangle $, where
$s_{1},\ldots,s_{m}$\ are nonnegative integers summing to $n$.

\item[(4)] Using a classical computer, calculate $\left\vert
\operatorname*{Per}\left(  U_{1_{n},S}\right)  \right\vert ^{2}/s_{1}!\cdots
s_{m}!$, the theoretical probability of observing the basis state $\left\vert
S\right\rangle $.

\item[(5)] Repeat steps (1) to (4), for a number of repetitions that scales
polynomially with $n$\ and $m$.

\item[(6)] Plot the empirical frequency of $\left\vert \operatorname*{Per}%
\left(  U_{1_{n},S}\right)  \right\vert ^{2}/s_{1}!\cdots s_{m}!>x$ for all
$x\in\left[  0,1\right]  $, with particular focus on the range $x\approx
1/\binom{m+n-1}{n}$. \ Check for agreement with the frequencies predicted by
quantum mechanics (which can again be calculated using a classical computer,
either deterministically or via Monte Carlo simulation).
\end{enumerate}

The procedure above does not prove that the final state is $\varphi\left(
U\right)  \left\vert 1_{n}\right\rangle $. \ However, it at least checks that
the basis states $\left\vert S\right\rangle $\ with large values of
$\left\vert \operatorname*{Per}\left(  U_{1_{n},S}\right)  \right\vert ^{2}$
are more likely to be observed than those with small values of $\left\vert
\operatorname*{Per}\left(  U_{1_{n},S}\right)  \right\vert ^{2}$, in the
manner predicted by formula (\ref{perform}).

\subsection{Physical Resource Requirements\label{ISSUES}}

We now make some miscellaneous remarks about the physical resource
requirements for our experiment.\bigskip

\textbf{Platform.} \ The obvious\ platform for our proposed experiment is
linear optics. \ However, one could also do the experiment (for example) in a
solid-state system, using bosonic excitations. \ What is essential is just
that the excitations behave as \textit{indistinguishable }bosons when they are
far apart. \ In other words, the amplitude for $n$ excitations to transition
from one basis state to another must be given by the permanent of an $n\times
n$\ matrix of transition amplitudes for the individual excitations. \ On the
other hand, the more general formula (\ref{perform}) need not hold; that is,
it is acceptable for the bosonic approximation to break down for processes
that involve multiple excitations in the same mode. \ (The reason is that the
events that most interest us do not involve collisions anyway.)\bigskip

\textbf{Initial state.} \ In our experiment, the initial state would ideally
consist of at most \textit{one photon per mode}: that is, single-photon Fock
states. \ This is already a nontrivial requirement, since a standard laser
outputs not Fock states but \textit{coherent states}, which have the form%
\[
\left\vert \alpha\right\rangle =e^{-\left\vert \alpha\right\vert ^{2}/2}%
\sum_{n=0}^{\infty}\frac{\alpha^{n}}{\sqrt{n!}}\left\vert n\right\rangle
\]
for some $\alpha\in\mathbb{C}$. \ (In other words, sometimes there are zero
photons, sometimes one, sometimes two, etc., with the number of photons
following a Poisson distribution.) \ Fortunately, the task of building
reliable single-photon sources is an extremely well-known one in quantum
optics \cite{lounis}, and the technology to generate single-photon Fock states
has been steadily improving over the past decade.

Still, one can ask whether any analogue of our computational hardness results
goes through, if the inputs are coherent states rather than Fock states.
\ Bartlett and Sanders \cite{bartlettsanders} have shown that, if the inputs
to a linear-optical network are coherent states, \textit{and} the measurements
are so-called \textit{homodyne}\ (or more generally \textit{Gaussian})
measurements, then the probability distribution over measurement outcomes can
be sampled in classical polynomial time. \ Intuitively, in this case the
photons behave like classical waves, so there is no possibility of a
superpolynomial quantum speedup.

On the other hand, if we have coherent-state inputs and measurements in the
\textit{photon-number} basis, then the situation is more complicated. \ As
pointed out in Section \ref{STRONGER}, in this case Theorem \ref{warmup} still
holds: using postselection, one can prove that \textit{exact} classical
simulation of the linear-optics experiment would collapse the polynomial
hierarchy. \ However, we do not know whether \textit{approximate} classical
simulation would already have surprising complexity consequences in this
case.\bigskip

\textbf{Measurements.} \ For our experiment, it is desirable to have an array
of $m$\ photodetectors, which reliably measure the number of photons $s_{i}$
in each mode $i\in\left[  m\right]  $. \ However, it would also suffice to use
detectors that only measure whether each $s_{i}$ is zero or nonzero. \ This is
because our hardness results talk only about basis states $\left\vert
S\right\rangle =\left\vert s_{1},\ldots,s_{m}\right\rangle $\ that are
\textit{collision-free}, meaning that $s_{i}\in\left\{  0,1\right\}  $\ for
all $i\in\left[  m\right]  $. \ \ Thus, one could simply postselect on the
runs in which exactly $n$ of the $m$ detectors record a photon, in which case
one knows that $s_{i}=1$\ for the corresponding modes $i$, while $s_{i}=0$ for
the remaining $m-n$\ modes. \ (In Appendix \ref{BIRTHDAY}, we will prove a
\textquotedblleft Boson Birthday Bound,\textquotedblright\ which shows that as
long as $m$ is sufficiently large and the mode-mixing unitary $U$ is
Haar-random, this postselection step succeeds with probability close to $1$.
\ Intuitively, if $m$ is large enough, then collision-free basis states are
the overwhelming majority.)

What might \textit{not} suffice are so-called Gaussian\ measurements. \ As
mentioned earlier, if the measurements are Gaussian \textit{and} the inputs
are coherent states, then Bartlett and Sanders \cite{bartlettsanders} showed
that no superpolynomial quantum speedup is possible. \ We do not know what the
situation is if the measurements are Gaussian and the inputs are single-photon
Fock states.

Like single-photon sources, photodetectors have improved dramatically over the
past decade, but of course no detector will be 100\% efficient.\footnote{Here
the \textquotedblleft efficiency\textquotedblright\ of a photodetector refers
to the probability of its detecting a photon that is present.} \ As we discuss
later, the higher the photodetector efficiencies, the less need there is for
\textit{postselection}, and therefore, the more easily one can scale to larger
numbers of photons.\bigskip

\textbf{Number of photons }$n$\textbf{.} \ An obvious question is how many
photons are needed for our experiment. \ The short answer is simply
\textquotedblleft the more, the better!\textquotedblright\ \ The goal of the
experiment is to confirm that, for \textit{every} positive integer $n$, the
transition amplitudes for $n$ identical bosons are given by $n\times
n$\ permanents, as quantum mechanics predicts. \ So the larger the $n$, the
stronger the evidence for this claim, and the greater the strain on any
competing interpretation.

At present, it seems fair to say that our experiment has already been done for
$n=2$\ (this is the Hong-Ou-Mandel dip \cite{hom}). \ However, we are not
aware of any experiment directly testing formula (\ref{perform})\ even for
$n=3$. \ Experimentalists we consulted expressed the view that this is mostly
just a matter of insufficient motivation before now, and that the $n=3$\ and
even $n=4$\ cases ought to be feasible with current technology.

Of course, the most interesting regime for computer science is the one where
$n$ is large enough that a classical computer would have difficulty computing
an $n\times n$ permanent. \ The best known classical algorithm for the
permanent, \textit{Ryser's algorithm}, uses about $2^{n+1}n^{2}$%
\ floating-point operations. \ If $n=10$, then this is about $200,000$
operations; if $n=20$, it is about $800$\ million; if $n=30$, it is about
$2$\ trillion. \ In any of these cases, it would be exciting to perform a
linear-optics experiment that \textquotedblleft
almost-instantly\textquotedblright\ sampled from a distribution in which the
probabilities were given by $n\times n$\ permanents.\bigskip

\textbf{Number of modes }$m$\textbf{.} \ Another important question is how
many \textit{modes} are needed for our experiment. \ We showed in Theorem
\ref{mainresult}\ that it suffices to use $m=O\left(  \frac{1}{\delta}%
n^{5}\log^{2}\frac{n}{\delta}\right)  $ modes, which is polynomial in $n$
but\ impractical. \ We strongly believe that an improved analysis could yield
$m=O\left(  n^{2}\right)  $. \ On the other hand, by the birthday paradox, we
cannot have \textit{fewer} than $m=\Omega\left(  n^{2}\right)  $\ modes, if we
want the state $\varphi\left(  U\right)  \left\vert 1_{n}\right\rangle $\ to
be dominated by \textit{collision-free}\ photon configurations\ (meaning those
containing at most one\ photon per mode).

Unfortunately, a quadratic number of modes might still be difficult to arrange
in practice. \ So the question arises: what would happen if we ran our
experiment with a \textit{linear} number of modes, $m=O\left(  n\right)  $?
\ In that case, almost every basis state would contain collisions, so our
formal argument for the classical hardness of approximate
\textsc{BosonSampling}, based on Conjectures \ref{pacc}\ and \ref{pgc}, would
no longer apply. \ On the other hand, we suspect it would still be
\textit{true} that sampling is classically hard! \ Giving a formal argument
for the hardness of approximate \textsc{BosonSampling}, with $n$\ photons and
$m=O\left(  n\right)  $\ modes, is an important technical challenge that we leave.

In the meantime, if the goal of one's experiment is just to verify that the
permanent formula (\ref{perform}) remains correct for large values of $n$,
then large numbers of photon collisions are presumably acceptable. \ In this
case, it should suffice to set $m\approx n$, or possibly even $m\ll n$ (though
note that it is easy to give a classical simulation algorithm that runs in
$n^{O\left(  m\right)  }$\ time).\bigskip

\textbf{Choice of unitary transformation }$U$\textbf{.} \ One could look for
an $n$-photon Hong-Ou-Mandel dip using \textit{any} unitary transformation $U$
that produces nontrivial interference among $n$ of the $m$\ modes. \ However,
some choices of $U$ are more interesting than others. \ The prescription
suggested by our results is to choose $U$\ \textit{randomly}, according to the
Haar measure over $m\times m$\ unitaries. \ Once $U$ is chosen, one can then
\textquotedblleft hardwire\textquotedblright\ a network of beamsplitters and
phaseshifters that produces $U$.

There are at least three reasons why using a Haar-random $U$ seems like a good idea:

\begin{enumerate}
\item[(1)] Theorem \ref{truncthm} showed that\ any sufficiently small
submatrix of a Haar-random unitary matrix $U$ is close to a matrix of
i.i.d.\ Gaussians. \ This extremely useful fact is what let us prove Theorem
\ref{mainresult}, which relates the hardness of approximate
\textsc{BosonSampling}\ to the hardness of more \textquotedblleft
natural\textquotedblright\ problems that have nothing to do with unitary matrices.

\item[(2)] Setting aside our results, the Haar measure is the unique
rotationally-invariant measure over unitaries. \ This makes it an obvious
choice, if the goal is to avoid any \textquotedblleft special
structure\textquotedblright\ that might make the \textsc{BosonSampling}%
\ problem easy.

\item[(3)] In the linear-optics model, one simple way to apply a Haar-random
$m\times m$\ unitary matrix $U$ is via a network of $\operatorname*{poly}%
\left(  m\right)  $ \textit{randomly-chosen} beamsplitters and
phaseshifters.\bigskip
\end{enumerate}

\textbf{Optical elements.} \ One might worry about the \textit{number} of
beamsplitters and phaseshifters needed to implement an arbitrary $m\times
m$\ unitary transformation $U$, or a Haar-random $U$ in particular. \ And
indeed, the upper bound of Reck et al.\ \cite{rzbb}\ (Lemma \ref{decompose})
shows only that $O\left(  m^{2}\right)  $\ beamsplitters and phaseshifters
suffice to implement any unitary, and this is easily seen to be tight by a
dimension argument. Unfortunately, a network of $\thicksim m^{2}$\ optical
elements might already strain the limits of practicality, especially if $m$
has been chosen to be quadratically larger than $n$.

Happily, Section \ref{DEPTH} will show how to reduce the number of optical
elements from $O\left(  m^{2}\right)  $\ to $O\left(  mn\right)  $, by
exploiting a simple observation: namely, \textit{we only care about the
optical network's behavior on the first }$n$\textit{ modes, since the standard
initial state }$\left\vert 1_{n}\right\rangle $\textit{ has no photons in the
remaining }$m-n$\textit{\ modes anyway.} \ Section \ref{DEPTH}\ will also show
how to \textquotedblleft parallelize\textquotedblright\ the resulting optical
network, so that the $O\left(  mn\right)  $\ beamsplitters and phaseshifters
are arranged into only\ $O\left(  n\log m\right)  $ layers.

Whether one can parallelize linear-optics computations still further, and
whether one can sample from hard distributions using even fewer optical
elements (say, $O\left(  m\log m\right)  $), are interesting topics for future
work.\bigskip

\textbf{Error.} \ There are many sources of error in our experiment;
understanding and controlling the errors is perhaps the central challenge an
experimentalist will face. \ At the most obvious level:

\begin{enumerate}
\item[(1)] Generation of single-photon Fock states will not be perfectly reliable.

\item[(2)] The beamsplitters and phaseshifters will not induce exactly the
desired unitary transformations.

\item[(3)] Each photon will have some probability of \textquotedblleft getting
lost along the way.\textquotedblright

\item[(4)] The photodetectors will not have perfect efficiency.

\item[(5)] If the lengths of the optical fibers are not well-calibrated, or
the single-photon sources are not synchronized, or there is vibration, etc.,
then the photons will generally arrive at the photodetectors at different times.
\end{enumerate}

If (5) occurs, then the photons effectively become \textit{distinguishable},
and the amplitudes will no longer correspond to $n\times n$\ permanents. \ So
then how well-synchronized do the photons need to be? \ To answer this
question, recall that each photon is actually a Gaussian wavepacket in the
position basis, rather than a localized point. \ For formula (\ref{perform})
to hold, what is necessary is that the photons arrive at the photodetectors
within a short enough time interval that their wavepackets have large pairwise overlaps.

The fundamental worry is that, as we increase the number of photons $n$, the
probability of a successful run of the experiment might decrease like $c^{-n}%
$. \ In practice, experimentalists usually deal with such behavior by
\textit{postselecting} on the successful runs. \ In our context, that could
mean (for example) that we only count the runs in which $n$ detectors register
a photon simultaneously, even if such runs are exponentially unlikely. \ We
expect that any realistic implementation of our experiment would involve at
least some postselection. \ However, if the eventual goal is to scale to large
values of $n$, then any need to postselect on an event with probability
$c^{-n}$\ presents an obvious barrier. \ Indeed, from an asymptotic
perspective, this sort of postselection defeats the entire purpose of using a
quantum computer rather than a classical computer.

For this reason, while even a heavily-postselected Hong-Ou-Mandel dip with
(say) $n=3$, $4$, or $5$\ photons would be interesting, our real hope is that
it will ultimately be possible to scale our experiment to interestingly large
values of $n$, while maintaining a total error that is closer to $0$ than to
$1$. \ However, supposing this turns out to be possible, one can still ask:
\textit{how close to }$0$\textit{ does the error need to be?}

Unfortunately, just like with the question of how many photons are needed, it
is difficult to give a direct answer, because of the reliance of our results
on asymptotics. \ What Theorem \ref{mainresult}\ shows is that, \textit{if}
one can scale the \textsc{BosonSampling}\ experiment to $n$ photons and error
$\delta$\ in total variation distance, using an amount of \textquotedblleft
experimental effort\textquotedblright\ that scales polynomially with both
$n$\ and $1/\delta$, then modulo our complexity conjectures, the Extended
Church-Turing Thesis is false. \ The trouble is that no finite experiment can
ever prove (or disprove) the claim that scaling to $n$ photons and error
$\delta$\ takes $\operatorname*{poly}\left(  n,1/\delta\right)  $%
\ experimental effort. \ One can, however, build a circumstantial case for
this claim---by increasing $n$, decreasing $\delta$, and making it clear that,
with reasonable effort, one could have increased $n$ and decreased $\delta
$\ still further.

One challenge we leave is to prove a computational hardness result that works
for a \textit{fixed} (say, constant) error $\delta$, rather than treating
$1/\delta$\ as an input parameter to the sampling algorithm along with $n$.
\ A second challenge is whether any nontrivial error-correction is possible
within the noninteracting-boson model. \ In standard quantum computing, the
famous Threshold Theorem \cite{ab,klz} asserts that there exists a constant
$\tau>0$ such that, even if each qubit fails with independent probability
$\tau$ at each time step, one can still \textquotedblleft correct errors
faster than they happen,\textquotedblright\ and thereby perform an arbitrarily
long quantum computation. \ In principle, the Threshold Theorem could be
applied to our experiment, to deal with all the sources of error listed above.
\ The issue is that, if we have the physical resources available for
fault-tolerant quantum computing, then perhaps we ought to forget about
\textsc{BosonSampling}, and simply run a universal quantum computation! \ What
we want, ideally, is a way to reduce the error in our experiment,
\textit{without} giving up on the implementation advantages that make the
experiment attractive in the first place.

\subsection{Reducing the Size and Depth of Optical Networks\label{DEPTH}}

In this section, we discuss how best to realize an $m\times m$\ unitary
transformation $U$, acting on the initial state $\left\vert 1_{n}\right\rangle
$, as a product of beamsplitters and phaseshifters. \ If we implement $U$\ in
the \textquotedblleft obvious\textquotedblright\ way---by appealing to Lemma
\ref{decompose}---then the number of optical elements and the depth will both
be $O\left(  m^{2}\right)  $. \ However, we can obtain a significant
improvement by noticing that our goal is just to apply \textit{some} unitary
transformation $\widetilde{U}$\ such that $\varphi(\widetilde{U})\left\vert
1_{n}\right\rangle =\varphi\left(  U\right)  \left\vert 1_{n}\right\rangle $:
we do not care about the behavior on $\widetilde{U}$\ on inputs other than
$\left\vert 1_{n}\right\rangle $. \ This yields a network in which the number
of optical elements and the depth are both $O\left(  mn\right)  $.

The following theorem shows that we can reduce the depth further, to $O\left(
n\log m\right)  $, by exploiting parallelization.

\begin{theorem}
[Parallelization of Linear-Optics Circuits]\label{depthreduce}Given any
$m\times m$\ unitary operation $U$, one can map the initial state $\left\vert
1_{n}\right\rangle $\ to $\varphi\left(  U\right)  \left\vert 1_{n}%
\right\rangle $\ using a linear-optical network of depth $O\left(  n\log
m\right)  $, consisting of $O\left(  mn\right)  $\ beamsplitters and phaseshifters.
\end{theorem}

\begin{proof}
We will consider a linear-optics system with $m+n$\ modes. \ Let%
\[
V=\left(
\begin{array}
[c]{cc}%
U & 0\\
0 & I
\end{array}
\right)
\]
be a unitary transformation that acts as $U$\ on the first $m$ modes, and as
the identity on the remaining $n$ modes. \ Then our goal will be to map
$\left\vert 1_{n}\right\rangle $ to $\varphi\left(  V\right)  \left\vert
1_{n}\right\rangle $.

Let $\left\vert e_{i}\right\rangle $\ be the basis state that consists of a
single photon in mode $i$, and no photons in the remaining $m+n-1$\ modes.
\ Also, let $\left\vert \psi_{i}\right\rangle =V\left\vert e_{i}\right\rangle
$. \ Then it clearly suffices to implement \textit{some} unitary
transformation $\widetilde{V}$\ that maps $\left\vert e_{i}\right\rangle $\ to
$\left\vert \psi_{i}\right\rangle $\ for all $i\in\left[  n\right]  $---for
then $\varphi(\widetilde{V})\left\vert 1_{n}\right\rangle =\varphi\left(
V\right)  \left\vert 1_{n}\right\rangle $ by linearity.

Our first claim is that, for each $i\in\left[  n\right]  $
\textit{individually}, there exists a unitary transformation $V_{i}$\ that
maps $\left\vert e_{i}\right\rangle $\ to $\left\vert \psi_{i}\right\rangle $,
and that can be implemented by a linear-optical network of depth $\log
_{2}m+O\left(  1\right)  $ with $O\left(  m\right)  $\ optical elements. \ To
implement $V_{i}$, we use a binary doubling strategy: first map $\left\vert
e_{i}\right\rangle $\ to a superposition of the first two modes,%
\[
\left\vert z_{1}\right\rangle =\alpha_{1}\left\vert e_{1}\right\rangle
+\alpha_{2}\left\vert e_{2}\right\rangle .
\]
Then, by using two beamsplitters in parallel, map the above state $\left\vert
z_{1}\right\rangle $\ to a superposition of the first four modes,%
\[
\left\vert z_{2}\right\rangle =\alpha_{1}\left\vert e_{1}\right\rangle
+\alpha_{2}\left\vert e_{2}\right\rangle +\alpha_{3}\left\vert e_{3}%
\right\rangle +\alpha_{4}\left\vert e_{4}\right\rangle .
\]
Next, by using four beamsplitters in parallel, map $\left\vert z_{2}%
\right\rangle $ to a superposition $\left\vert z_{3}\right\rangle $\ of the
first eight modes, and so on until $\left\vert \psi_{i}\right\rangle $\ is
reached. \ It is clear that the total depth required is $\log_{2}m+O\left(
1\right)  $, while the number of optical elements required is $O\left(
m\right)  $. \ This proves the claim.

Now let $S_{i}$\ be a unitary transformation that swaps modes $i$ and $m+i$,
and that acts as the identity on the remaining $m+n-2$\ modes. \ Then we will
implement $\widetilde{V}$ as follows:%
\[
\widetilde{V}=V_{n}S_{n}V_{n}^{\dagger}\cdot\cdots\cdot V_{2}S_{2}%
V_{2}^{\dagger}\cdot V_{1}S_{1}V_{1}^{\dagger}\cdot S_{n}\cdots S_{1}.
\]
In other words: first swap modes $1,\ldots,n$\ with modes $m+1,\ldots,m+n$.
\ Then, for all $i:=1$\ to $n$, apply $V_{i}S_{i}V_{i}^{\dagger}$.

Since each $S_{i}$ involves only one optical element, while each $V_{i}$ and
$V_{i}^{\dagger}$ involves $O\left(  m\right)  $\ optical elements\ and
$O\left(  \log m\right)  $\ depth, it is clear that we can implement
$\widetilde{V}$\ using a linear-optical network of depth $O\left(  n\log
m\right)  $\ with $O\left(  mn\right)  $\ optical elements.

To prove the theorem, we need to verify that $\widetilde{V}\left\vert
e_{i}\right\rangle =\left\vert \psi_{i}\right\rangle $ for all $i\in\left[
n\right]  $. \ We do so in three steps. \ First, notice that for all
$i\in\left[  n\right]  $,%
\begin{align*}
V_{i}S_{i}V_{i}^{\dagger}\left(  S_{i}\left\vert e_{i}\right\rangle \right)
&  =V_{i}S_{i}V_{i}^{\dagger}\left\vert e_{m+i}\right\rangle \\
&  =V_{i}S_{i}\left\vert e_{m+i}\right\rangle \\
&  =V_{i}\left\vert e_{i}\right\rangle \\
&  =\left\vert \psi_{i}\right\rangle .
\end{align*}
where the second line follows since $V_{i}^{\dagger}$\ acts only on the first
$m$ modes.

Second, for all $i,j\in\left[  n\right]  $\ with $i\neq j$,%
\[
V_{j}S_{j}V_{j}^{\dagger}\left\vert e_{m+i}\right\rangle =\left\vert
e_{m+i}\right\rangle ,
\]
since $V_{j}$\ and $S_{j}$\ both act as the identity on $\left\vert
e_{m+i}\right\rangle $.

Third, notice that $\left\langle \psi_{i}|\psi_{j}\right\rangle =0$ for all
$i\neq j$, since $\left\vert \psi_{i}\right\rangle $\ and $\left\vert \psi
_{j}\right\rangle $\ correspond to two different columns of the unitary matrix
$U$. \ Since unitaries preserve inner product, this means that $V_{j}%
^{\dagger}\left\vert \psi_{i}\right\rangle $\ is also orthogonal to
$V_{j}^{\dagger}\left\vert \psi_{j}\right\rangle =V_{j}^{\dagger}%
V_{j}\left\vert e_{j}\right\rangle =\left\vert e_{j}\right\rangle $: in other
words, the state $V_{j}^{\dagger}\left\vert \psi_{i}\right\rangle $\ has no
support on the $j^{th}$\ mode.\ \ It follows that $S_{j}$\ acts as the
identity on $V_{j}^{\dagger}\left\vert \psi_{i}\right\rangle $---and
therefore, for all $i,j\in\left[  n\right]  $\ with $i\neq j$, we have%
\[
V_{j}S_{j}V_{j}^{\dagger}\left\vert \psi_{i}\right\rangle =V_{j}V_{j}%
^{\dagger}\left\vert \psi_{i}\right\rangle =\left\vert \psi_{i}\right\rangle
.
\]
Summarizing, we find that for all $i\in\left[  n\right]  $:

\begin{itemize}
\item $V_{i}S_{i}V_{i}^{\dagger}$\ maps $\left\vert e_{m+i}\right\rangle $\ to
$\left\vert \psi_{i}\right\rangle $.

\item $V_{j}S_{j}V_{j}^{\dagger}$\ maps $\left\vert e_{m+i}\right\rangle $\ to
itself for all $j<i$.

\item $V_{j}S_{j}V_{j}^{\dagger}$\ maps $\left\vert \psi_{i}\right\rangle
$\ to itself for all $j>i$.
\end{itemize}

We conclude that $\widetilde{V}\left\vert e_{i}\right\rangle =V_{i}S_{i}%
V_{i}^{\dagger}\left\vert e_{m+i}\right\rangle =\left\vert \psi_{i}%
\right\rangle $ for all $i\in\left[  n\right]  $. \ This proves the theorem.
\end{proof}

\section{Reducing \textsc{GPE}$_{\times}$\ to $\left\vert \text{\textsc{GPE}%
}\right\vert _{\pm}^{2}$\label{GPE2GPE}}

The goal of this section is to prove Theorem \ref{decompthm}: that, assuming
Conjecture \ref{pacc} (the Permanent Anti-Concentration Conjecture), the
\textsc{GPE}$_{\times}$\ and $\left\vert \text{\textsc{GPE}}\right\vert _{\pm
}^{2}$\ problems are polynomial-time equivalent. \ Or in words: if we can
additively estimate $\left\vert \operatorname*{Per}\left(  X\right)
\right\vert ^{2}$\ with high probability over a Gaussian matrix $X\sim
\mathcal{G}^{n\times n}$, then we can also multiplicatively estimate
$\operatorname*{Per}\left(  X\right)  $\ with high probability over a Gaussian
matrix $X$.

Given as input a matrix $X\thicksim\mathcal{N}\left(  0,1\right)
_{\mathbb{C}}^{n\times n}$ of i.i.d.\ Gaussians, together with error bounds
$\varepsilon,\delta>0$, recall that the \textsc{GPE}$_{\times}$\ problem
(Problem \ref{gpe*}) asks us to estimate $\operatorname*{Per}\left(  X\right)
$ to within error $\pm\varepsilon\cdot\left\vert \operatorname*{Per}\left(
X\right)  \right\vert $, with probability at least $1-\delta$\ over $X$, in
$\operatorname*{poly}\left(  n,1/\varepsilon,1/\delta\right)  $\ time.
\ Meanwhile, the $\left\vert \text{\textsc{GPE}}\right\vert _{\pm}^{2}%
$\ problem\ (Problem \ref{gpe2+}) asks us to estimate $\left\vert
\operatorname*{Per}\left(  X\right)  \right\vert ^{2}$ to within error
$\pm\varepsilon\cdot n!$, with probability at least $1-\delta$\ over $X$, in
$\operatorname*{poly}\left(  n,1/\varepsilon,1/\delta\right)  $\ time. \ It is
easy to give a reduction from $\left\vert \text{\textsc{GPE}}\right\vert
_{\pm}^{2}$\ to \textsc{GPE}$_{\times}$. \ The hard direction, and the one
that requires Conjecture \ref{pacc}, is to reduce \textsc{GPE}$_{\times}$\ to
$\left\vert \text{\textsc{GPE}}\right\vert _{\pm}^{2}$.

While technical, this reduction is essential for establishing the connection
we want between

\begin{enumerate}
\item[(1)] Theorem \ref{mainresult} (our main result), which relates the
classical hardness of \textsc{BosonSampling} to $\left\vert \text{\textsc{GPE}%
}\right\vert _{\pm}^{2}$, and

\item[(2)] Conjecture \ref{pgc}\ (the Permanent-of-Gaussians Conjecture),
which asserts that the \textsc{Gaussian Permanent Estimation} problem is
$\mathsf{\#P}$-hard, in the more \textquotedblleft natural\textquotedblright%
\ setting of multiplicative rather than additive estimation, and
$\operatorname*{Per}\left(  X\right)  $\ rather than $\left\vert
\operatorname*{Per}\left(  X\right)  \right\vert ^{2}$.
\end{enumerate}

Besides \textsc{GPE}$_{\times}$\ and\ $\left\vert \text{\textsc{GPE}%
}\right\vert _{\pm}^{2}$, one could of course also define two
\textquotedblleft hybrid\textquotedblright\ problems: \textsc{GPE}$_{\pm}%
$\ (additive estimation of $\operatorname*{Per}\left(  X\right)  $), and
$\left\vert \text{\textsc{GPE}}\right\vert _{\times}^{2}$\ (multiplicative
estimation of $\left\vert \operatorname*{Per}\left(  X\right)  \right\vert
^{2}$). \ Mercifully, we will not need to make explicit use of these hybrid
problems. Indeed, assuming Conjecture \ref{pacc}, they will simply become
equivalent to \textsc{GPE}$_{\times}$\ and $\left\vert \text{\textsc{GPE}%
}\right\vert _{\pm}^{2}$\ as a byproduct.

Let us start by proving the easy direction of the equivalence between
\textsc{GPE}$_{\times}$\ and\ $\left\vert \text{\textsc{GPE}}\right\vert
_{\pm}^{2}$. \ This direction does not rely on any unproved conjectures.

\begin{lemma}
\label{implications}$\left\vert \text{\textsc{GPE}}\right\vert _{\pm}^{2}$ is
polynomial-time reducible to \textsc{GPE}$_{\times}$.
\end{lemma}

\begin{proof}
Suppose we have a polynomial-time algorithm $M$ that, given $\left\langle
X,0^{1/\varepsilon},0^{1/\delta}\right\rangle $, outputs a good multiplicative
approximation to $\operatorname*{Per}\left(  X\right)  $---that is, a $z$ such
that
\[
\left\vert z-\operatorname*{Per}\left(  X\right)  \right\vert \leq
\varepsilon\left\vert \operatorname*{Per}\left(  X\right)  \right\vert
\]
---with probability at least $1-\delta$\ over $X\sim\mathcal{G}^{n\times n}$.
\ Then certainly $\left\vert z\right\vert ^{2}$\ is a good multiplicative
approximation to $\left\vert \operatorname*{Per}\left(  X\right)  \right\vert
^{2}$:%
\begin{align*}
\left\vert \left\vert z\right\vert ^{2}-\left\vert \operatorname*{Per}\left(
X\right)  \right\vert ^{2}\right\vert  &  =\left\vert \left\vert z\right\vert
-\left\vert \operatorname*{Per}\left(  X\right)  \right\vert \right\vert
\left(  \left\vert z\right\vert +\left\vert \operatorname*{Per}\left(
X\right)  \right\vert \right) \\
&  \leq\varepsilon\left(  2+\varepsilon\right)  \left\vert \operatorname*{Per}%
\left(  X\right)  \right\vert ^{2}\\
&  \leq3\varepsilon\left\vert \operatorname*{Per}\left(  X\right)  \right\vert
^{2}.
\end{align*}
We claim that $\left\vert z\right\vert ^{2}$\ is also a good \textit{additive}
approximation to $\left\vert \operatorname*{Per}\left(  X\right)  \right\vert
^{2}$, with high probability over $X$. \ For by Markov's inequality,%
\[
\Pr_{X}\left[  \left\vert \operatorname*{Per}\left(  X\right)  \right\vert
^{2}>k\cdot n!\right]  <\frac{1}{k}.
\]
So by the union bound,%
\begin{align*}
\Pr_{X}\left[  \left\vert \left\vert z\right\vert ^{2}-\left\vert
\operatorname*{Per}\left(  X\right)  \right\vert ^{2}\right\vert >\varepsilon
k\cdot n!\right]   &  \leq\Pr_{X}\left[  \left\vert \left\vert z\right\vert
^{2}-\left\vert \operatorname*{Per}\left(  X\right)  \right\vert
^{2}\right\vert >3\varepsilon\left\vert \operatorname*{Per}\left(  X\right)
\right\vert ^{2}\right]  +\Pr_{X}\left[  3\varepsilon\left\vert
\operatorname*{Per}\left(  X\right)  \right\vert ^{2}>\varepsilon k\cdot
n!\right] \\
&  \leq\delta+\frac{3}{k}.
\end{align*}
Thus, we can achieve any desired additive error bounds $\left(  \varepsilon
^{\prime},\delta^{\prime}\right)  $ by (for example) setting $\varepsilon
:=\varepsilon^{\prime}\delta^{\prime}/6$, $\delta:=\delta^{\prime}/2$, and
$k:=6/\delta^{\prime}$, so that $\varepsilon k=\varepsilon^{\prime}$\ and
$\delta+\frac{3}{k}\leq\delta^{\prime}$. \ Clearly this increases $M$'s
running time by at most a polynomial factor.
\end{proof}

We now prove that, assuming the Permanent Anti-Concentration Conjecture,
approximating $\left\vert \operatorname*{Per}\left(  X\right)  \right\vert
^{2}$\ for a Gaussian random matrix $X\sim\mathcal{G}^{n\times n}$\ is as hard
as approximating $\operatorname*{Per}\left(  X\right)  $ itself. \ This result
can be seen as an average-case analogue of Theorem \ref{approxhard}.\ \ To
prove it, we need to give a reduction that estimates the phase
$\operatorname*{Per}\left(  X\right)  /\left\vert \operatorname*{Per}\left(
X\right)  \right\vert $\ of a permanent $\operatorname*{Per}\left(  X\right)
$, given only the ability to estimate $\left\vert \operatorname*{Per}\left(
X\right)  \right\vert $ (for most Gaussian matrices $X$). \ As in the proof of
Theorem \ref{approxhard},\ our reduction proceeds by induction on $n$: we
assume the ability to estimate $\operatorname*{Per}\left(  Y\right)  $\ for a
certain $\left(  n-1\right)  \times\left(  n-1\right)  $\ submatrix $Y$ of
$X$, and then use that (together with estimates of $\left\vert
\operatorname*{Per}\left(  X^{\prime}\right)  \right\vert $\ for various
$n\times n$\ matrices $X^{\prime}$) to estimate $\operatorname*{Per}\left(
X\right)  $. \ Unfortunately, the reduction and its analysis are more
complicated than in Theorem \ref{approxhard},\ since in this case, we can only
assume that our oracle\ estimates $\left\vert \operatorname*{Per}\left(
X\right)  \right\vert ^{2}$\ with high probability if $X$ \textquotedblleft
looks like\textquotedblright\ a Gaussian matrix. \ This rules out the adaptive
reduction of Theorem \ref{approxhard}, which even starting with a Gaussian
matrix $X$, would vary the top-left entry so as to produce new matrices
$X^{\prime}$\ that look nothing like Gaussian matrices. \ Instead, we will use
a \textit{nonadaptive} reduction, which in turn necessitates a more delicate
error analysis, as well as an appeal to Conjecture \ref{pacc}.

To do the error analysis, we first need a technical lemma about the numerical
stability of \textit{triangulation}. \ By triangulation, we simply mean a
procedure that determines a point $x\in\mathbb{R}^{d}$, given the Euclidean
distances $\Delta\left(  x,y_{i}\right)  $\ between $x$\ and $d+1$\ fixed
points $y_{1},\ldots,y_{d+1}\in\mathbb{R}^{d}$\ that are in general position.
\ So for example, the $d=3$\ case corresponds to how a GPS receiver would
calculate its position given its distances to four satellites. \ We will be
interested in the $d=2$\ case, which corresponds to calculating an unknown
complex number $x=\operatorname*{Per}\left(  X\right)  \in\mathbb{C}$\ given
the squared Euclidean distances $\left\vert x-y_{1}\right\vert ^{2},\left\vert
x-y_{2}\right\vert ^{2},\left\vert x-y_{3}\right\vert ^{2}$, for some
$y_{1},y_{2},y_{3}\in\mathbb{C}$ that are in general position. \ The question
that interests us is this:

\begin{quotation}
\noindent\textit{Suppose our estimates of the squared distances }$\left\vert
x-y_{1}\right\vert ^{2},\left\vert x-y_{2}\right\vert ^{2},\left\vert
x-y_{3}\right\vert ^{2}$\textit{\ are noisy, and our estimates of the points
}$y_{1},y_{2},y_{3}$\textit{\ are also noisy. \ How much noise does that
induce in our resulting estimate of }$x$\textit{?}
\end{quotation}

The following lemma answers that question, in the special case where $y_{1}%
=0$, $y_{2}=w$, $y_{3}=iw$ for some complex number $w$.

\begin{lemma}
[Stability of Triangulation]\label{triang}Let $z=re^{i\theta}\in\mathbb{C}$ be
a hidden complex number that we are trying to estimate, and\ let $w=ce^{i\tau
}\in\mathbb{C}$\ be a second \textquotedblleft reference\textquotedblright%
\ number ($\,r,c>0$, $\theta,\tau\in\left(  -\pi,\pi\right]  $). \ For some
known constant $\lambda>0$, let%
\begin{align*}
R  &  :=\left\vert z\right\vert ^{2}=r^{2},\\
S  &  :=\left\vert z-\lambda w\right\vert ^{2}=r^{2}+\lambda^{2}c^{2}-2\lambda
rc\cos\left(  \theta-\tau\right)  ,\\
T  &  :=\left\vert z-i\lambda w\right\vert ^{2}=r^{2}+\lambda^{2}%
c^{2}-2\lambda rc\sin\left(  \theta-\tau\right)  ,\\
C  &  :=\left\vert w\right\vert ^{2}=c^{2}.
\end{align*}
Suppose we are given approximations $\widetilde{R},\widetilde{S}%
,\widetilde{T},\widetilde{C},\widetilde{\tau}$\ to $R,S,T,C,\tau
$\ respectively, such that%
\begin{align*}
\left\vert \widetilde{R}-R\right\vert ,\left\vert \widetilde{S}-S\right\vert
,\left\vert \widetilde{T}-T\right\vert  &  <\varepsilon\lambda^{2}C,\\
\left\vert \widetilde{C}-C\right\vert  &  <\varepsilon C.
\end{align*}
Suppose also that $\varepsilon\leq\frac{1}{10}\min\left\{  1,\frac{R}%
{\lambda^{2}C}\right\}  $. \ Then the approximation%
\[
\widetilde{\theta}:=\widetilde{\tau}+\operatorname*{sgn}\left(  \widetilde{R}%
+\widetilde{C}-\widetilde{T}\right)  \arccos\left(  \frac{\widetilde{R}%
+\widetilde{C}-\widetilde{S}}{2\sqrt{\widetilde{R}\widetilde{C}}}\right)
\bigskip
\]
satisfies%
\[
\left\vert \widetilde{\theta}-\theta\right\vert \operatorname{mod}2\pi
\leq\left\vert \widetilde{\tau}-\tau\right\vert +1.37\sqrt{\varepsilon}\left(
\lambda\sqrt{\frac{C}{R}}+1\right)  .
\]

\end{lemma}

\begin{proof}
Without loss of generality, we can set $\lambda:=1$; the result for general
$\lambda>0$\ then follows by replacing $w$\ with $\lambda w$ and $C$\ with
$\lambda^{2}C$.

Let $\alpha:=R/C$, $\beta:=S/C$, and $\gamma:=T/C$, and note that $\alpha
\geq2\varepsilon$. \ Observe that%
\begin{align*}
\cos\left(  \theta-\tau\right)   &  =\frac{R+C-S}{2\sqrt{RC}}=\frac
{\alpha+1-\beta}{2\sqrt{\alpha}},\\
\sin\left(  \theta-\tau\right)   &  =\frac{R+C-T}{2\sqrt{RC}}=\frac
{\alpha+1-\gamma}{2\sqrt{\alpha}}.
\end{align*}
So we can write%
\[
\theta=\tau+b\arccos\left(  \frac{\alpha+1-\beta}{2\sqrt{\alpha}}\right)
\]
where $b\in\left\{  -1,1\right\}  $\ is a sign term given by%
\[
b:=\operatorname*{sgn}\left(  \theta-\tau\right)  =\operatorname*{sgn}\left(
\sin\left(  \theta-\tau\right)  \right)  =\operatorname*{sgn}\left(
\alpha+1-\gamma\right)  .
\]
Now let $\widetilde{\alpha}:=\widetilde{R}/C$, $\widetilde{\beta
}:=\widetilde{S}/C$, $\widetilde{\gamma}:=\widetilde{T}/C$, and $\chi
:=\widetilde{C}/C$. \ Note that $\left\vert \widetilde{\alpha}-\alpha
\right\vert ,\left\vert \widetilde{\beta}-\beta\right\vert ,\left\vert
\widetilde{\gamma}-\gamma\right\vert ,\left\vert \chi-1\right\vert
<\varepsilon$. \ Let%
\begin{align*}
\widetilde{b}  &  :=\operatorname*{sgn}\left(  \widetilde{\alpha}%
+\chi-\widetilde{\gamma}\right)  ,\\
\widetilde{\theta}  &  :=\widetilde{\tau}+\widetilde{b}\arccos\left(
\frac{\widetilde{\alpha}+\chi-\widetilde{\beta}}{2\sqrt{\widetilde{\alpha}%
\chi}}\right)  .
\end{align*}

We now consider two cases. \ First suppose $\left\vert \alpha+1-\gamma
\right\vert \leq3\varepsilon$. \ Then $\left\vert 2\sqrt{\alpha}\sin\left(
\theta-\tau\right)  \right\vert \leq3\varepsilon$, which implies%
\[
\sin^{2}\left(  \theta-\tau\right)  \leq\frac{9\varepsilon^{2}}{4\alpha}.
\]
Likewise, we have%
\begin{align*}
\left\vert 2\sqrt{\widetilde{\alpha}\lambda}\sin\left(  \widetilde{\theta
}-\widetilde{\tau}\right)  \right\vert  &  =\left\vert \widetilde{\alpha}%
+\chi-\widetilde{\gamma}\right\vert \\
&  \leq\left\vert \alpha+1-\gamma\right\vert +\left\vert \widetilde{\alpha
}-\alpha\right\vert +\left\vert \chi-1\right\vert +\left\vert
\widetilde{\gamma}-\gamma\right\vert \\
&  \leq6\varepsilon
\end{align*}
and hence%
\[
\sin^{2}\left(  \widetilde{\theta}-\widetilde{\tau}\right)  \leq\frac{\left(
6\varepsilon\right)  ^{2}}{\left(  2\sqrt{\widetilde{\alpha}\chi}\right)
^{2}}\leq\frac{9\varepsilon^{2}}{\left(  \alpha-\varepsilon\right)  \left(
1-\varepsilon\right)  }.
\]
So if we write%
\begin{align*}
\theta &  =\tau+b\arccos\left(  \cos\left(  \theta-\tau\right)  \right)  ,\\
\widetilde{\theta}  &  =\widetilde{\tau}+\widetilde{b}\arccos\left(
\cos\left(  \widetilde{\theta}-\widetilde{\tau}\right)  \right)  ,
\end{align*}
we find that%
\begin{align*}
\left\vert \widetilde{\theta}-\theta\right\vert -\left\vert \widetilde{\tau
}-\tau\right\vert  &  \leq\left\vert \arccos\left(  \cos\left(  \theta
-\tau\right)  \right)  \right\vert +\left\vert \arccos\left(  \cos\left(
\widetilde{\theta}-\widetilde{\tau}\right)  \right)  \right\vert \\
&  \leq\arccos\sqrt{1-\frac{9\varepsilon^{2}}{4\alpha}}+\arccos\sqrt
{1-\frac{9\varepsilon^{2}}{\left(  \alpha-\varepsilon\right)  \left(
1-\varepsilon\right)  }}\\
&  =\arcsin\frac{3\varepsilon}{2\sqrt{\alpha}}+\arcsin\frac{3\varepsilon
}{\sqrt{\left(  \alpha-\varepsilon\right)  \left(  1-\varepsilon\right)  }}\\
&  \leq1.1\left(  \frac{3\varepsilon}{2\sqrt{\alpha}}+\frac{3\varepsilon
}{\sqrt{\left(  \alpha-\varepsilon\right)  \left(  1-\varepsilon\right)  }%
}\right) \\
&  \leq5.32\frac{\varepsilon}{\sqrt{\alpha}}.
\end{align*}
Here the last two lines used the fact that $\varepsilon\leq\frac{1}{10}%
\min\left\{  1,\alpha\right\}  $, together with the inequality $\arcsin
x\leq1.1x$\ for small enough $x$.

Next suppose $\left\vert \alpha+1-\gamma\right\vert >3\varepsilon$. \ Then by
the triangle inequality,%
\[
\left\vert \left\vert \widetilde{\alpha}+\lambda-\widetilde{\gamma}\right\vert
-\left\vert \alpha+1-\gamma\right\vert \right\vert \leq\left\vert
\widetilde{\alpha}-\alpha\right\vert +\left\vert \widetilde{\gamma}%
-\gamma\right\vert +\left\vert \chi-1\right\vert \leq3\varepsilon,
\]
which implies that\ $\operatorname*{sgn}\left(  \widetilde{\alpha}%
+\chi-\widetilde{\gamma}\right)  =\operatorname*{sgn}\left(  \alpha
+1-\gamma\right)  $ and hence $\widetilde{b}=b$. \ So%
\begin{align*}
\left\vert \widetilde{\theta}-\theta\right\vert -\left\vert \widetilde{\tau
}-\tau\right\vert  &  \leq\left\vert \arccos\left(  \frac{\widetilde{\alpha
}+\chi-\widetilde{\beta}}{2\sqrt{\widetilde{\alpha}\chi}}\right)
-\arccos\left(  \frac{\alpha+1-\beta}{2\sqrt{\alpha}}\right)  \right\vert \\
&  \leq\arccos\left(  \frac{\alpha+1-\beta-3\varepsilon}{2\sqrt
{\widetilde{\alpha}\chi}}\right)  -\arccos\left(  \frac{\alpha+1-\beta}%
{2\sqrt{\alpha}}\right) \\
&  \leq\frac{3}{2}\sqrt{\frac{3\varepsilon}{2\sqrt{\widetilde{\alpha}\chi}%
}+\left\vert \alpha+1-\beta\right\vert \left\vert \frac{1}{2\sqrt{\alpha}%
}-\frac{1}{2\sqrt{\widetilde{\alpha}\chi}}\right\vert }\\
&  \leq\frac{3}{2}\sqrt{\frac{3\varepsilon}{2\sqrt{\left(  \alpha
-\varepsilon\right)  \left(  1-\varepsilon\right)  }}+2\sqrt{\alpha}\left\vert
\frac{1}{2\sqrt{\alpha}}-\frac{1}{2\sqrt{\widetilde{\alpha}\chi}}\right\vert
}\\
&  \leq\frac{3}{2}\sqrt{\frac{3\varepsilon}{2\sqrt{\left(  0.9\alpha\right)
\left(  0.9\right)  }}+\left\vert 1-\sqrt{\frac{\alpha}{\widetilde{\alpha}%
\chi}}\right\vert }\\
&  \leq\frac{3}{2}\sqrt{\frac{5\varepsilon}{3\sqrt{\alpha}}+\left(
\sqrt{\frac{\alpha}{\left(  \alpha-\varepsilon\right)  \left(  1-\varepsilon
\right)  }}-1\right)  }\\
&  \leq\frac{3}{2}\sqrt{\frac{5\varepsilon}{3\sqrt{\alpha}}+\frac{1}{2}\left(
\frac{\alpha}{\left(  \alpha-\varepsilon\right)  \left(  1-\varepsilon\right)
}-1\right)  }\\
&  \leq\frac{3}{2}\sqrt{\varepsilon}\sqrt{\frac{5}{3\sqrt{\alpha}}%
+\frac{\left(  1+\alpha\right)  }{2\left(  0.9\alpha\right)  \left(
0.9\right)  }}\\
&  \leq\frac{3}{2}\sqrt{\varepsilon}\sqrt{\frac{5}{6}}\sqrt{\frac{1}{\alpha
}+\frac{2}{\sqrt{\alpha}}+1}\\
&  \leq1.37\sqrt{\varepsilon}\left(  \frac{1}{\sqrt{\alpha}}+1\right)  .
\end{align*}
Here the second line used the monotonicity of the $\arccos$\ function, the
third line used the inequality%
\[
\arccos\left(  x-\varepsilon\right)  -\arccos x\leq1.5\sqrt{\varepsilon}%
\]
for $\varepsilon\leq\frac{1}{2}$, and the fifth and ninth lines used the fact
that $\varepsilon\leq\min\left\{  \frac{1}{10},\frac{\alpha}{10}\right\}  $.
\ Combining the two cases, we have%
\[
\left\vert \widetilde{\theta}-\theta\right\vert \leq\left\vert \widetilde{\tau
}-\tau\right\vert +\max\left\{  5.32\frac{\varepsilon}{\sqrt{\alpha}%
},1.37\sqrt{\varepsilon}\left(  \frac{1}{\sqrt{\alpha}}+1\right)  \right\}  .
\]
Using the fact that $\varepsilon\leq\min\left\{  \frac{1}{10},\frac{\alpha
}{10}\right\}  $, one can check that the second item in the maximum is always
greater. \ Therefore%
\[
\left\vert \widetilde{\theta}-\theta\right\vert \leq\left\vert \widetilde{\tau
}-\tau\right\vert +1.37\sqrt{\varepsilon}\left(  \frac{1}{\sqrt{\alpha}%
}+1\right)  =\left\vert \widetilde{\tau}-\tau\right\vert +1.37\sqrt
{\varepsilon}\left(  \sqrt{\frac{C}{R}}+1\right)
\]
as claimed.
\end{proof}

We will also need a lemma about the autocorrelation of the Gaussian
distribution, which will be reused in Section \ref{HARDPER}.

\begin{lemma}
[Autocorrelation of Gaussian Distribution]\label{autocor}Consider the
distributions%
\begin{align*}
\mathcal{D}_{1}  &  =\mathcal{N}\left(  0,\left(  1-\varepsilon\right)
^{2}\right)  _{\mathbb{C}}^{N},\\
\mathcal{D}_{2}  &  =\prod_{i=1}^{N}\mathcal{N}\left(  v_{i},1\right)
_{\mathbb{C}}%
\end{align*}
for some vector $v\in\mathbb{C}^{N}$. \ We have%
\begin{align*}
\left\Vert \mathcal{D}_{1}-\mathcal{G}^{N}\right\Vert  &  \leq2N\varepsilon,\\
\left\Vert \mathcal{D}_{2}-\mathcal{G}^{N}\right\Vert  &  \leq\left\Vert
v\right\Vert _{2}.
\end{align*}

\end{lemma}

\begin{proof}
It will be helpful to think of each complex coordinate as two real
coordinates, in which case $\mathcal{G}^{N}=\mathcal{N}\left(  0,1/2\right)
_{\mathbb{R}}^{2N}$ and $v$ is a vector in $\mathbb{R}^{2N}$.

For the first part, we have%
\begin{align*}
\left\Vert \mathcal{D}_{1}-\mathcal{G}^{N}\right\Vert  &  \leq2N\left\Vert
\mathcal{N}\left(  0,\frac{\left(  1-\varepsilon\right)  ^{2}}{2}\right)
_{\mathbb{R}}-\mathcal{N}\left(  0,\frac{1}{2}\right)  _{\mathbb{R}%
}\right\Vert \\
&  =\frac{N}{\sqrt{\pi}}\int_{-\infty}^{\infty}\left\vert e^{-x^{2}/\left(
1-\varepsilon\right)  ^{2}}-e^{-x^{2}}\right\vert dx\\
&  \leq2N\varepsilon
\end{align*}
where the first line follows from the triangle inequality and the last line
from straightforward estimates.

For the second part, by the rotational invariance of the Gaussian
distribution, the variation distance is unaffected if we replace $v$\ by any
other vector with the same $2$-norm. \ So let $v:=\left(  \ell,0,\ldots
,0\right)  $ where $\ell=\left\Vert v\right\Vert _{2}$. \ Then%
\begin{align*}
\left\Vert \mathcal{D}_{2}-\mathcal{G}^{N}\right\Vert  &  =\frac{1}{2}%
\int_{x_{1},\ldots,x_{2N}=-\infty}^{\infty}\left\vert \frac{e^{-\left(
x_{1}-\ell\right)  ^{2}}}{\sqrt{\pi}}\frac{e^{-x_{2}^{2}}}{\sqrt{\pi}}%
\cdots\frac{e^{-x_{2N}^{2}}}{\sqrt{\pi}}-\frac{e^{-x_{1}^{2}}}{\sqrt{\pi}%
}\frac{e^{-x_{2}^{2}}}{\sqrt{\pi}}\cdots\frac{e^{-x_{2N}^{2}}}{\sqrt{\pi}%
}\right\vert dx_{1}\cdots dx_{2N}\\
&  =\frac{1}{2\sqrt{\pi}}\int_{-\infty}^{\infty}\left\vert e^{-\left(
x-\ell\right)  ^{2}}-e^{-x^{2}}\right\vert dx\\
&  \leq\ell,
\end{align*}
where the last line follows from straightforward estimates.
\end{proof}

Using Lemmas \ref{triang} and \ref{autocor}, we can now complete the proof of
Theorem \ref{decompthm}: that assuming Conjecture \ref{pacc}\ (the Permanent
Anti-Concentration Conjecture), the \textsc{GPE}$_{\times}$\ and $\left\vert
\text{\textsc{GPE}}\right\vert _{\pm}^{2}$\ problems are polynomial-time equivalent.

\begin{proof}
[Proof of Theorem \ref{decompthm}]Lemma \ref{implications}\ already gave an
unconditional reduction from $\left\vert \text{\textsc{GPE}}\right\vert _{\pm
}^{2}$ to \textsc{GPE}$_{\times}$. \ So it suffices to give a reduction from
\textsc{GPE}$_{\times}$\ to $\left\vert \text{\textsc{GPE}}\right\vert _{\pm
}^{2}$, assuming the Permanent Anti-Concentration Conjecture.

Throughout the proof, we will fix an $N\times N$ input matrix $X=\left(
x_{ij}\right)  \in\mathbb{C}^{N\times N}$, which we think of as sampled from
the Gaussian distribution $\mathcal{G}^{N\times N}$. \ Probabilities will
always be with respect to $X\sim\mathcal{G}^{N\times N}$. \ For convenience,
we will often assume that \textquotedblleft bad events\textquotedblright%
\ (i.e., estimates of various quantities outside the desired error bounds)
simply do not occur; then, at the end, we will use the union bound to show
that the assumption was justified.

The \textsc{GPE}$_{\times}$ problem can be stated as follows.\ \ Given\ the
input $\left\langle X,0^{1/\varepsilon},0^{1/\delta}\right\rangle $ for some
$\varepsilon,\delta>0$, output a complex number $z\in\mathbb{C}$ such that%
\[
\left\vert z-\operatorname*{Per}\left(  X\right)  \right\vert \leq
\varepsilon\left\vert \operatorname*{Per}\left(  X\right)  \right\vert ,
\]
with success probability at least $1-\delta$ over $X$, in time
$\operatorname*{poly}\left(  N,1/\varepsilon,1/\delta\right)  $.

Let $\mathcal{O}$\ be an oracle that solves $\left\vert \text{\textsc{GPE}%
}\right\vert _{\pm}^{2}$. \ That is, given an input $\left\langle
A,0^{1/\epsilon},0^{1/\Delta}\right\rangle $\ where $A$\ is an $n\times n$
complex matrix, $\mathcal{O}$ outputs a nonnegative real number $\mathcal{O}%
\left(  \left\langle A,0^{1/\epsilon},0^{1/\Delta}\right\rangle \right)
$\ such that%
\[
\Pr_{A\sim\mathcal{G}^{n\times n}}\left[  \left\vert \mathcal{O}\left(
\left\langle A,0^{1/\epsilon},0^{1/\Delta}\right\rangle \right)  -\left\vert
\operatorname*{Per}\left(  A\right)  \right\vert ^{2}\right\vert \leq
\epsilon\left\vert \operatorname*{Per}\left(  A\right)  \right\vert
^{2}\right]  \geq1-\Delta.
\]
Then assuming Conjecture \ref{pacc}, we will show how to solve the
\textsc{GPE}$_{\times}$ instance $\left\langle X,0^{1/\varepsilon}%
,0^{1/\delta}\right\rangle $\ in time $\operatorname*{poly}\left(
N,1/\varepsilon,1/\delta\right)  $, with the help of $3N$\ nonadaptive queries
to $\mathcal{O}$.

Let $R=\left\vert \operatorname*{Per}\left(  X\right)  \right\vert ^{2}$.
\ Then by simply calling $\mathcal{O}$ on the input matrix $X$, we can obtain
a good approximation $\widetilde{R}$\ to $R$, such that (say) $\left\vert
\widetilde{R}-R\right\vert \leq\varepsilon R/10$. \ Therefore, our problem
reduces to estimating the \textit{phase} $\theta=\operatorname*{Per}\left(
X\right)  /\left\vert \operatorname*{Per}\left(  X\right)  \right\vert $. \ In
other words, we need to give a procedure that returns an approximation
$\widetilde{\theta}$\ to $\theta$\ such that (say) $\left\vert
\widetilde{\theta}-\theta\right\vert \leq0.9\varepsilon$, and does so with
high probability. \ (Here and throughout, it is understood that all
differences between angles are $\operatorname{mod}2\pi$.)

For all $n\in\left[  N\right]  $, let $X_{n}$\ be the bottom-right $n\times
n$\ submatrix of $X$ (thus $X_{N}=X$). \ A crucial observation is that, since
$X$ is a sample from $\mathcal{G}^{N\times N}$, each $X_{n}$ can be thought of
as a sample from $\mathcal{G}^{n\times n}$.

As in Theorem \ref{approxhard}, given a complex number $w$ and a matrix
$A=\left(  a_{ij}\right)  $, let $A^{\left[  w\right]  }$\ be the matrix that
is identical to $A$, except that its top-left entry equals $a_{11}-w$\ instead
of $a_{11}$. \ Then for any $n$ and $w$, we can think of the matrix
$X_{n}^{\left[  w\right]  }$\ as having been drawn from a distribution
$\mathcal{D}_{n}^{\left[  w\right]  }$ that is identical to $\mathcal{G}%
^{n\times n}$, except that the top-left entry is distributed according to
$\mathcal{N}\left(  -w,1\right)  _{\mathbb{C}}$\ rather than $\mathcal{G}$.
\ Recall that by Lemma \ref{autocor}, the variation distance between
$\mathcal{D}_{n}^{\left[  w\right]  }$\ and $\mathcal{G}^{n\times n}%
$\ satisfies%
\[
\left\Vert \mathcal{D}_{n}^{\left[  w\right]  }-\mathcal{G}^{n\times
n}\right\Vert \leq\left\vert w\right\vert .
\]

Let $\lambda>0$ be a parameter to be determined later. \ Then for each
$n\in\left[  N\right]  $, we will be interested in two specific $n\times
n$\ matrices besides $X_{n}$, namely $X_{n}^{\left[  \lambda\right]  }$\ and
$X_{n}^{\left[  i\lambda\right]  }$. \ Similarly to Theorem \ref{approxhard},
our reduction will be based on the identities%
\begin{align*}
\operatorname*{Per}\left(  X_{n}^{\left[  \lambda\right]  }\right)   &
=\operatorname*{Per}\left(  X_{n}\right)  -\lambda\operatorname*{Per}\left(
X_{n-1}\right)  ,\\
\operatorname*{Per}\left(  X_{n}^{\left[  i\lambda\right]  }\right)   &
=\operatorname*{Per}\left(  X_{n}\right)  -i\lambda\operatorname*{Per}\left(
X_{n-1}\right)  .
\end{align*}
More concretely, let%
\begin{align*}
R_{n}  &  :=\left\vert \operatorname*{Per}\left(  X_{n}\right)  \right\vert
^{2},\\
\theta_{n}  &  :=\frac{\operatorname*{Per}\left(  X_{n}\right)  }{\left\vert
\operatorname*{Per}\left(  X_{n}\right)  \right\vert },\\
S_{n}  &  :=\left\vert \operatorname*{Per}\left(  X_{n}^{\left[
\lambda\right]  }\right)  \right\vert ^{2}=\left\vert \operatorname*{Per}%
\left(  X_{n}\right)  -\lambda\operatorname*{Per}\left(  X_{n-1}\right)
\right\vert ^{2},\\
T_{n}  &  :=\left\vert \operatorname*{Per}\left(  X_{n}^{\left[
i\lambda\right]  }\right)  \right\vert ^{2}=\left\vert \operatorname*{Per}%
\left(  X_{n}\right)  -i\lambda\operatorname*{Per}\left(  X_{n-1}\right)
\right\vert ^{2}.
\end{align*}
Then some simple algebra---identical to what appeared in Lemma \ref{triang}%
---yields the identity%
\[
\theta_{n}=\theta_{n-1}+\operatorname*{sgn}\left(  R_{n}+R_{n-1}-T_{n}\right)
\arccos\left(  \frac{R_{n}+R_{n-1}-S_{n}}{2\sqrt{R_{n}R_{n-1}}}\right)
\bigskip
\]
for all $n\geq2$. \ \textquotedblleft Unravelling\textquotedblright\ this
recursive identity, we obtain a useful formula for $\theta=\theta_{N}%
=\frac{\operatorname*{Per}\left(  X\right)  }{\left\vert \operatorname*{Per}%
\left(  X\right)  \right\vert }$:%
\[
\theta=\frac{x_{NN}}{\left\vert x_{NN}\right\vert }+\sum_{n=2}^{N}\xi_{n}%
\]
where%
\[
\xi_{n}:=\operatorname*{sgn}\left(  R_{n}+R_{n-1}-T_{n}\right)  \arccos\left(
\frac{R_{n}+R_{n-1}-S_{n}}{2\sqrt{R_{n}R_{n-1}}}\right)  \bigskip.
\]
Our procedure to approximate $\theta$\ will simply consist of evaluating the
above expression for all $n\geq2$, but using estimates $\widetilde{R}%
_{n},\widetilde{S}_{n},\widetilde{T}_{n}$\ produced by the oracle
$\mathcal{O}$\ in place of the true values $R_{n},S_{n},T_{n}$.

In more detail, let $\widetilde{R}_{1}:=\left\vert x_{NN}\right\vert ^{2}$,
and for all $n\geq2$, let%
\begin{align*}
\widetilde{R}_{n}  &  :=\mathcal{O}\left(  \left\langle X_{n},0^{1/\epsilon
},0^{1/\Delta}\right\rangle \right)  ,\\
\widetilde{S}_{n}  &  :=\mathcal{O}\left(  \left\langle X_{n}^{\left[
\lambda\right]  },0^{1/\epsilon},0^{1/\Delta}\right\rangle \right)  ,\\
\widetilde{T}_{n}  &  :=\mathcal{O}\left(  \left\langle X_{n}^{\left[
i\lambda\right]  },0^{1/\epsilon},0^{1/\Delta}\right\rangle \right)  ,
\end{align*}
where $\epsilon,\Delta>1/\operatorname*{poly}\left(  N\right)  $ are
parameters to be determined later. \ Then our procedure for approximating
$\theta$\ is to return%
\[
\widetilde{\theta}:=\frac{x_{NN}}{\left\vert x_{NN}\right\vert }+\sum
_{n=2}^{N}\widetilde{\xi}_{n},
\]
where%
\[
\widetilde{\xi}_{n}:=\operatorname*{sgn}\left(  \widetilde{R}_{n}%
+\widetilde{R}_{n-1}-\widetilde{T}_{n}\right)  \arccos\left(  \frac
{\widetilde{R}_{n}+\widetilde{R}_{n-1}-\widetilde{S}_{n}}{2\sqrt
{\widetilde{R}_{n}\widetilde{R}_{n-1}}}\right)  \bigskip.
\]
Clearly this procedure runs in polynomial time and makes at most
$3N$\ nonadaptive calls to $\mathcal{O}$.

We now upper-bound the error $\left\vert \widetilde{\theta}-\theta\right\vert
$\ incurred\ in the approximation. \ Since%
\[
\left\vert \widetilde{\theta}-\theta\right\vert \leq\sum_{n=2}^{N}\left\vert
\widetilde{\xi}_{n}-\xi_{n}\right\vert ,
\]
it suffices to upper-bound $\left\vert \widetilde{\xi}_{n}-\xi_{n}\right\vert
$\ for each $n$. \ By the definition of $\mathcal{O}$, for all $n\in\left[
N\right]  $\ we have%
\begin{align*}
\Pr\left[  \left\vert \widetilde{R}_{n}-R_{n}\right\vert \leq\epsilon
R_{n}\right]   &  \geq1-\Delta,\\
\Pr\left[  \left\vert \widetilde{S}_{n}-S_{n}\right\vert \leq\epsilon
S_{n}\right]   &  \geq1-\Delta-\left\Vert \mathcal{D}_{n}^{\left[
\lambda\right]  }-\mathcal{G}^{n\times n}\right\Vert \\
&  \geq1-\Delta-\lambda,\\
\Pr\left[  \left\vert \widetilde{T}_{n}-T_{n}\right\vert \leq\epsilon
T_{n}\right]   &  \geq1-\Delta-\left\Vert \mathcal{D}_{n}^{\left[
i\lambda\right]  }-\mathcal{G}^{n\times n}\right\Vert \\
&  \geq1-\Delta-\lambda.
\end{align*}
Also, let $p\left(  n,1/\beta\right)  $ be a polynomial such that%
\[
\Pr_{A\sim\mathcal{G}^{n\times n}}\left[  \left\vert \operatorname*{Per}%
\left(  A\right)  \right\vert ^{2}\geq\frac{n!}{p\left(  n,1/\beta\right)
}\right]  \geq1-\beta
\]
for all $n$ and $\beta>0$; such a $p$ is guaranteed to exist by Conjecture
\ref{pacc}. \ It will later be convenient to assume $p$ is monotone. \ Then%
\[
\Pr\left[  R_{n}\geq\frac{n!}{p\left(  n,1/\beta\right)  }\right]  \geq1-\beta
\]
In the other direction, for all $0<\kappa<1$ Markov's inequality gives us%
\begin{align*}
\Pr\left[  R_{n}\leq\frac{n!}{\kappa}\right]   &  \geq1-\kappa,\\
\Pr\left[  S_{n}\leq\frac{n!}{\kappa}\right]   &  \geq1-\kappa-\lambda,\\
\Pr\left[  T_{n}\leq\frac{n!}{\kappa}\right]   &  \geq1-\kappa-\lambda,
\end{align*}
where we have again used the fact that $S_{n},T_{n}$\ are random variables
with variation distance at most $\lambda$\ from $R_{n}$. \ Now think of
$\beta,\kappa>1/\operatorname*{poly}\left(  N\right)  $ as parameters to be
determined later, and suppose that \textit{all seven} of the events listed
above hold, for all $n\in\left[  N\right]  $. \ In that case,%
\begin{align*}
\left\vert \widetilde{R}_{n}-R_{n}\right\vert  &  \leq\epsilon R_{n}\\
&  \leq\epsilon\frac{n!}{\kappa}\\
&  =\epsilon\frac{R_{n-1}n}{\kappa}\frac{\left(  n-1\right)  !}{R_{n-1}}\\
&  \leq\epsilon\frac{R_{n-1}n}{\kappa}p\left(  n-1,1/\beta\right) \\
&  \leq\epsilon\frac{R_{n-1}N}{\kappa}p\left(  N,1/\beta\right) \\
&  =\frac{\epsilon N\cdot p\left(  N,1/\beta\right)  }{\kappa\lambda^{2}%
}\lambda^{2}R_{n-1}%
\end{align*}
and likewise%
\[
\left\vert \widetilde{S}_{n}-S_{n}\right\vert ,\left\vert \widetilde{T}%
_{n}-T_{n}\right\vert \leq\frac{\epsilon N\cdot p\left(  N,1/\beta\right)
}{\kappa\lambda^{2}}\lambda^{2}R_{n-1}.
\]
Plugging the above bounds into Lemma \ref{triang}, we find that, if there are
no \textquotedblleft bad events,\textquotedblright\ then noisy triangulation
returns an estimate $\widetilde{\xi}_{n}$\ of $\xi_{n}$\ such that%
\begin{align*}
\left\vert \widetilde{\xi}_{n}-\xi_{n}\right\vert  &  \leq1.37\sqrt
{\frac{\epsilon N\cdot p\left(  N,1/\beta\right)  }{\kappa\lambda^{2}}}\left(
\lambda\sqrt{\frac{R_{n-1}}{R_{n}}}+1\right) \\
&  \leq1.37\sqrt{\frac{\epsilon N\cdot p\left(  N,1/\beta\right)  }%
{\kappa\lambda^{2}}}\left(  \lambda\sqrt{\frac{\left(  n-1\right)  !/\kappa
}{n!/p\left(  N,1/\beta\right)  }}+1\right) \\
&  \leq1.37\sqrt{\epsilon}\left(  \frac{p\left(  N,1/\beta\right)  }{\kappa
}+\frac{\sqrt{N}\sqrt{p\left(  N,1/\beta\right)  }}{\lambda\sqrt{\kappa}%
}\right)  .
\end{align*}
We now upper-bound the probability of a bad event. \ Taking the union bound
over all $n\in\left[  N\right]  $\ and all seven possible bad events, we find
that the total probability that the procedure fails is at most%
\[
p_{\operatorname*{FAIL}}:=\left(  3\Delta+3\kappa+4\lambda+\beta\right)  N.
\]
Thus, let us now make the choices $\Delta,\kappa:=\frac{\delta}{12N}$,
$\lambda:=\frac{\delta}{16N}$, and $\beta:=\frac{\delta}{4N}$, so that
$p_{\operatorname*{FAIL}}\leq\delta$ as desired. \ Let us also make the choice%
\[
\epsilon:=\frac{\varepsilon^{2}\delta^{3}}{7120N^{6}p\left(  N,4N/\delta
\right)  ^{2}}.
\]
Then%
\begin{align*}
\left\vert \widetilde{\theta}-\theta\right\vert  &  \leq\sum_{n=2}%
^{N}\left\vert \widetilde{\xi}_{n}-\xi_{n}\right\vert \\
&  \leq1.37\sqrt{\epsilon}\left(  \frac{12N\cdot p\left(  N,4N/\delta\right)
}{\delta}+\frac{32\sqrt{3}N^{2}\sqrt{p\left(  N,4N/\delta\right)  }}%
{\delta^{3/2}}\right)  N\\
&  \leq\frac{9\varepsilon}{10}%
\end{align*}
as desired. \ Furthermore, if none of the bad events happen, then we get
\textquotedblleft for free\textquotedblright\ that%
\[
\left\vert \widetilde{R}-R\right\vert =\left\vert \widetilde{R}_{N}%
-R_{N}\right\vert \leq\epsilon R_{N}\leq\frac{\varepsilon R}{10}.
\]
So letting $r:=\sqrt{R}$\ and $\widetilde{r}:=\sqrt{\widetilde{R}}$, by the
triangle inequality we have%
\begin{align*}
\left\vert \widetilde{r}e^{i\widetilde{\theta}}-re^{i\theta}\right\vert  &
\leq\left\vert \widetilde{r}-r\right\vert +r\sqrt{2-2\cos\left(
\widetilde{\theta}-\theta\right)  }\\
&  \leq\frac{\left\vert \widetilde{R}-R\right\vert }{\widetilde{r}%
+r}+r\left\vert \widetilde{\theta}-\theta\right\vert \\
&  \leq\frac{\varepsilon R}{10r}+r\frac{9\varepsilon}{10}\\
&  =\varepsilon r\\
&  =\varepsilon\left\vert \operatorname*{Per}\left(  X\right)  \right\vert ,
\end{align*}
and hence we have successfully approximated $\operatorname*{Per}\left(
X\right)  =re^{i\theta}$.
\end{proof}

\section{The Distribution of Gaussian Permanents\label{DISTPER}}

In this section, we seek an understanding of the distribution over
$\operatorname*{Per}\left(  X\right)  $, where $X\sim\mathcal{G}^{n\times n}%
$\ is a matrix of i.i.d.\ Gaussians. \ Here, recall that $\mathcal{G}%
=\mathcal{N}\left(  0,1\right)  _{\mathbb{C}}$\ is the standard complex normal
distribution, though one suspects that most issues would be similar with
$\mathcal{N}\left(  0,1\right)  _{\mathbb{R}}$, or possibly even the uniform
distribution over $\left\{  -1,1\right\}  $. \ As explained in Section
\ref{APPROXCASE}, the reason why we focus on the complex Gaussian ensemble
$\mathcal{G}^{n\times n}$\ is simply that, as shown by Theorem \ref{truncthm},
the Gaussian ensemble arises naturally when we consider truncations of
Haar-random unitary matrices.

Our goal is to give evidence in favor of Conjecture \ref{pacc}, the
\textit{Permanent Anti-Concentration Conjecture} (PACC). \ This is the
conjecture that, if $X\sim\mathcal{G}^{n\times n}$\ is Gaussian, then
$\operatorname*{Per}\left(  X\right)  $ is \textquotedblleft not \textit{too}
concentrated around $0$\textquotedblright: a $1-1/\operatorname*{poly}\left(
n\right)  $\ fraction of its probability mass is greater than $\sqrt
{n!}/\operatorname*{poly}\left(  n\right)  $\ in absolute value, $\sqrt{n!}%
$\ being the standard deviation. More formally, there exists a polynomial $p$
such that for all $n$ and $\delta>0$,%
\[
\Pr_{X\sim\mathcal{G}^{n\times n}}\left[  \left\vert \operatorname*{Per}%
\left(  X\right)  \right\vert <\frac{\sqrt{n!}}{p\left(  n,1/\delta\right)
}\right]  <\delta.
\]
An equivalent formulation is that there exist constants $C,D$\ and $\beta
>0$\ such that for all $n$ and $\varepsilon>0$,%
\[
\Pr_{X\sim\mathcal{G}^{n\times n}}\left[  \left\vert \operatorname*{Per}%
\left(  X\right)  \right\vert <\varepsilon\sqrt{n!}\right]  <Cn^{D}%
\varepsilon^{\beta}.
\]

Conjecture \ref{pacc} has two applications to strengthening the conclusions of
this paper. \ First, it lets us \textit{multiplicatively} estimate
$\operatorname*{Per}\left(  X\right)  $ (that is, solve the \textsc{GPE}%
$_{\times}$\ problem), assuming only that we can \textit{additively} estimate
$\operatorname*{Per}\left(  X\right)  $ (that is, solve the \textsc{GPE}%
$_{\pm}$\ problem). \ Indeed, if Conjecture \ref{pacc}\ holds, then as pointed
out in Lemma \ref{implications}, additive and multiplicative estimation become
equivalent for this problem. \ Second, as shown by Theorem \ref{decompthm},
Conjecture \ref{pacc} lets us estimate $\operatorname*{Per}\left(  X\right)
$\ itself, assuming we can estimate\ $\left\vert \operatorname*{Per}\left(
X\right)  \right\vert ^{2}$. \ The bottom line is that, if Conjecture
\ref{pacc}\ holds, then we can base our conclusions about the hardness of
approximate \textsc{BosonSampling}\ on the natural conjecture that
\textsc{GPE}$_{\times}$\ is $\mathsf{\#P}$-hard, rather than the
relatively-contrived conjecture that $\left\vert \text{\textsc{GPE}%
}\right\vert _{\pm}^{2}$\ is $\mathsf{\#P}$-hard.

At a less formal level, we believe proving Conjecture \ref{pacc} might also
provide intuition essential to proving the \textquotedblleft
bigger\textquotedblright\ conjecture, that these problems are $\mathsf{\#P}%
$-hard\ in the first place.

The closest result to Conjecture \ref{pacc}\ that we know of comes from a 2009
paper of Tao and Vu \cite{taovu:perm}. \ These authors show the following:

\begin{theorem}
[Tao-Vu \cite{taovu:perm}]\label{taovuthm}For all $\varepsilon>0$ and
sufficiently large $n$,%
\[
\Pr_{X\in\left\{  -1,1\right\}  ^{n\times n}}\left[  \left\vert
\operatorname*{Per}\left(  X\right)  \right\vert <\frac{\sqrt{n!}%
}{n^{\varepsilon n}}\right]  <\frac{1}{n^{0.1}}.
\]

\end{theorem}

Alas, Theorem \ref{taovuthm} falls short of what we need in two respects.
\ First, it only upper-bounds the probability that $\left\vert
\operatorname*{Per}\left(  X\right)  \right\vert <\sqrt{n!}/n^{\varepsilon n}%
$, whereas we need to upper-bound the probability that $\left\vert
\operatorname*{Per}\left(  X\right)  \right\vert <\sqrt{n!}%
/\operatorname*{poly}\left(  n\right)  $. \ Second, the upper bound obtained
is $1/n^{0.1}$ (and Tao and Vu say that their technique seems to hit a barrier
at $1/\sqrt{n}$), whereas we need an upper bound of $1/\operatorname*{poly}%
\left(  n\right)  $. \ A more minor problem is that Theorem \ref{taovuthm}
applies to Bernoulli random matrices, not Gaussian ones. \ Of course, the
Gaussian case might be \textit{easier} rather than harder.

In the rest of the section, we will give three pieces of evidence for
Conjecture \ref{pacc}. \ The first, in Section \ref{NUMERICS}, is that it is
supported numerically. \ The second, in Section \ref{DET}, is that the
analogous statement holds with the \textit{determinant} instead of the
permanent. \ The proof of this result makes essential use of geometric
properties of the determinant, which is why we do not know how to extend it to
the permanent. \ On the other hand, Godsil and Gutman \cite{godsilgutman}
observed that, for all matrices $X=\left(  x_{ij}\right)  $,%
\[
\operatorname*{Per}\left(  X\right)  =\operatorname*{E}\left[
\operatorname*{Det}\left(
\begin{array}
[c]{ccc}%
\pm\sqrt{x_{11}} & \cdots & \pm\sqrt{x_{1n}}\\
\vdots & \ddots & \vdots\\
\pm\sqrt{x_{n1}} & \cdots & \pm\sqrt{x_{nn}}%
\end{array}
\right)  ^{2}\right]  ,
\]
where the expectation is over all $2^{n^{2}}$\ ways of assigning $+$'s and
$-$'s\ to the entries. \ Because of this fact, together with our numerical
data, we suspect that the story for the permanent may be similar to that for
the determinant. \ The third piece of evidence is\ that a weaker form of
Conjecture \ref{pacc} holds: basically, $\left\vert \operatorname*{Per}\left(
X\right)  \right\vert $\ has at least a $\Omega\left(  1/n\right)
$\ probability of being $\Omega\left(  \sqrt{n!}\right)  $. \ We prove this by
calculating the fourth moment of $\operatorname*{Per}\left(  X\right)  $.
\ Unfortunately, extending the calculation to higher moments seems difficult.

Before going further, let us make some elementary remarks about the
distribution over $\operatorname*{Per}\left(  X\right)  $\ for $X\sim
\mathcal{G}^{n\times n}$. \ By symmetry, clearly $\operatorname*{E}\left[
\operatorname*{Per}\left(  X\right)  \right]  =0$. \ The second moment is also
easy to calculate:%
\begin{align*}
\operatorname*{E}\left[  \left\vert \operatorname*{Per}\left(  X\right)
\right\vert ^{2}\right]   &  =\operatorname*{E}\left[  \sum_{\sigma,\tau\in
S_{n}}\prod_{i=1}^{n}x_{i,\sigma\left(  i\right)  }\overline{x}_{i,\tau\left(
i\right)  }\right] \\
&  =\operatorname*{E}\left[  \sum_{\sigma\in S_{n}}\prod_{i=1}^{n}\left\vert
x_{i,\sigma\left(  i\right)  }\right\vert ^{2}\right] \\
&  =\sum_{\sigma\in S_{n}}\prod_{i=1}^{n}\operatorname*{E}\left[  \left\vert
x_{i,\sigma\left(  i\right)  }\right\vert ^{2}\right] \\
&  =n!.
\end{align*}
We will often find it convenient to work with the normalized random variable%
\[
P_{n}:=\frac{\left\vert \operatorname*{Per}\left(  X\right)  \right\vert ^{2}%
}{n!},
\]
so that $\operatorname*{E}\left[  P_{n}\right]  =1$.

\subsection{Numerical Data\label{NUMERICS}}

Figure \ref{pccfig}\ shows the numerically-computed probability density
function of $P_{n}$ when $n=6$. \ For comparison, we have also plotted the pdf
of $D_{n}:=\left\vert \operatorname*{Det}\left(  X\right)  \right\vert
^{2}/n!$.%
%TCIMACRO{\FRAME{ftbpFU}{4.0197in}{2.0167in}{0pt}{\Qcb{Probability density
%functions of the random variables $D_{n}=\left\vert \operatorname*{Det}\left(
%X\right)  \right\vert ^{2}/n!$\ and $P_{n}=\left\vert \operatorname*{Per}%
%\left(  X\right)  \right\vert ^{2}/n!$, where $X\sim\QTR{cal}{G}^{n\times n}$
%is a complex Gaussian random matrix, in the case $n=6$. \ Note that
%$\operatorname*{E}\left[  D_{n}\right]  =\operatorname*{E}\left[
%P_{n}\right]  =1$. \ As $n$ increases, the bends on the left become steeper.
%\ We do not know whether the pdfs diverge at the origin.}}{\Qlb{pccfig}%
%}{pdf2.eps}{\special{ language "Scientific Word";  type "GRAPHIC";
%width 4.0197in;  height 2.0167in;  depth 0pt;  original-width 8.1327in;
%original-height 10.6415in;  cropleft "0.0984";  croptop "0.7923";
%cropright "0.9134";  cropbottom "0.4816";
%filename 'pdf2.eps';file-properties "XNPEU";}} }%
%BeginExpansion
\begin{figure}[ptb]%
\centering
\includegraphics[
trim=0.800258in 5.124947in 0.704292in 2.210240in,
height=2.0167in,
width=4.0197in
]%
{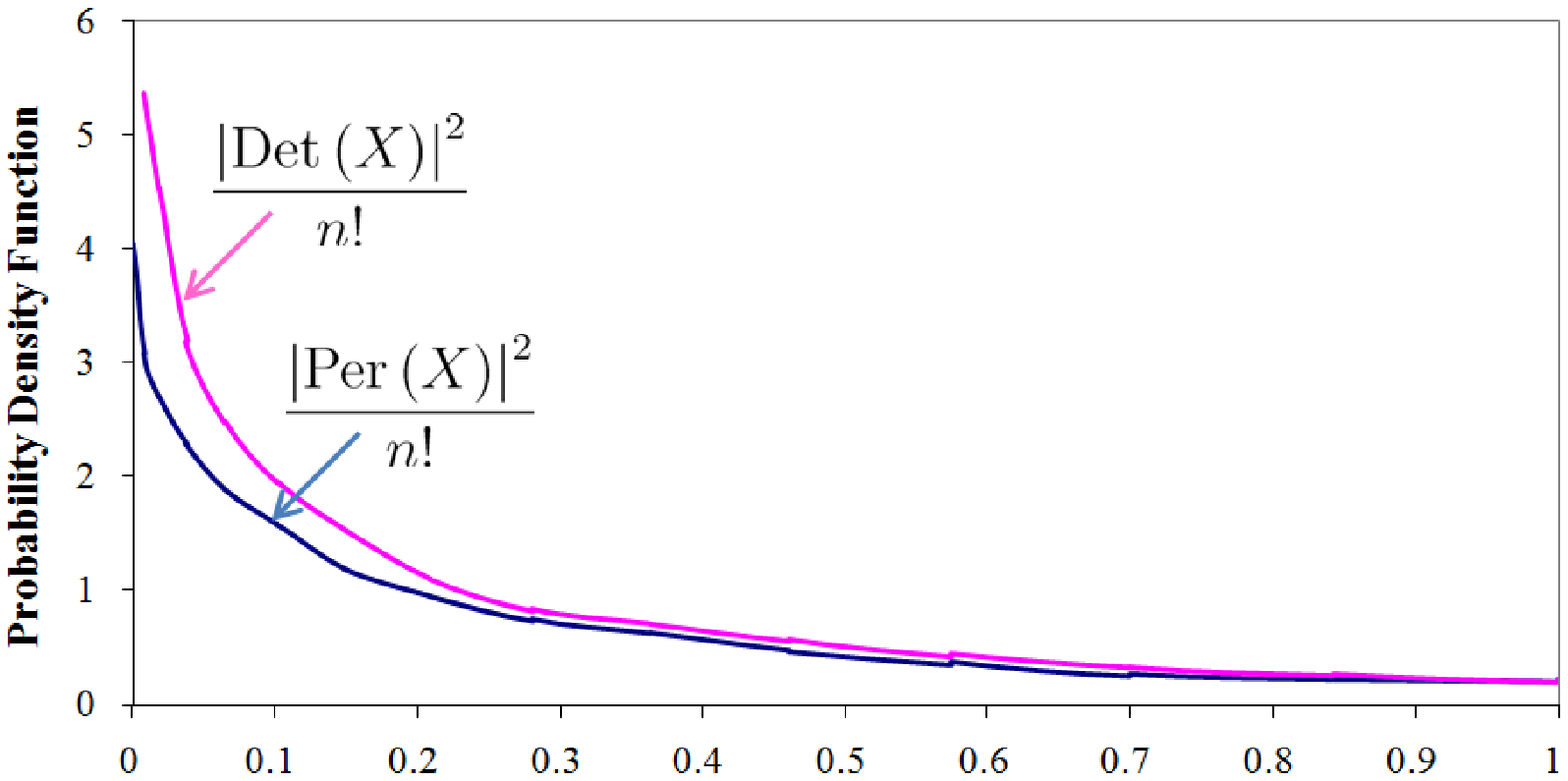}%
\caption{Probability density functions of the random variables $D_{n}%
=\left\vert \operatorname*{Det}\left(  X\right)  \right\vert ^{2}/n!$\ and
$P_{n}=\left\vert \operatorname*{Per}\left(  X\right)  \right\vert ^{2}/n!$,
where $X\sim\mathcal{G}^{n\times n}$ is a complex Gaussian random matrix, in
the case $n=6$. \ Note that $\operatorname*{E}\left[  D_{n}\right]
=\operatorname*{E}\left[  P_{n}\right]  =1$. \ As $n$ increases, the bends on
the left become steeper. \ We do not know whether the pdfs diverge at the
origin.}%
\label{pccfig}%
\end{figure}
%EndExpansion

The numerical evidence up to $n=10$ is strongly consistent with Conjecture
\ref{pacc}. \ Indeed, from the data it seems likely that for all $0\leq
\beta<2$, there exist constants $C,D$\ such that for all $n$\ and
$\varepsilon>0$,%
\[
\Pr_{X\sim\mathcal{G}^{n\times n}}\left[  \left\vert \operatorname*{Per}%
\left(  X\right)  \right\vert <\varepsilon\sqrt{n!}\right]  <Cn^{D}%
\varepsilon^{\beta},
\]
and perhaps the above even holds when $\beta=2$.

\subsection{The Analogue for Determinants\label{DET}}

We prove the following theorem, which at least settles Conjecture
\ref{pacc}\ with the determinant in place of the permanent:

\begin{theorem}
[Determinant Anti-Concentration Theorem]\label{detthm}For all $0\leq\beta<2$,
there exists a constant $C_{\beta}$\ such that for all $n$\ and $\varepsilon
>0$,%
\[
\Pr_{X\sim\mathcal{G}^{n\times n}}\left[  \left\vert \operatorname*{Det}%
\left(  X\right)  \right\vert <\varepsilon\sqrt{n!}\right]  <C_{\beta}%
n^{\beta\left(  \beta+2\right)  /8}\varepsilon^{\beta}.
\]

\end{theorem}

We leave as an open problem whether Theorem \ref{detthm}\ holds when $\beta=2$.

Compared to the permanent, a lot is known about the determinants of Gaussian
matrices. \ In particular, Girko \cite{girko} (see also Costello and Vu
\cite[Appendix A]{cv}) have shown that%
\[
\frac{\ln\left\vert \operatorname*{Det}\left(  X\right)  \right\vert -\ln
\sqrt{\left(  n-1\right)  !}}{\sqrt{\frac{\ln n}{2}}}%
\]
converges weakly to the normal distribution $\mathcal{N}\left(  0,1\right)
_{\mathbb{R}}$. \ Unfortunately, weak convergence is not enough to imply
Theorem \ref{detthm}, so we will have to do some more work. \ Indeed, we will
find that the probability density function of $\left\vert \operatorname*{Det}%
\left(  X\right)  \right\vert ^{2}$, in the critical regime where $\left\vert
\operatorname*{Det}\left(  X\right)  \right\vert ^{2}\approx0$, is
\textit{different} than one might guess from the above formula.

The key fact about $\operatorname*{Det}\left(  X\right)  $\ that we will use
is that we can compute its moments \textit{exactly}---even the fractional and
inverse moments. \ To do so, we use the following beautiful characterization,
which can be found (for example) in Costello and Vu \cite{cv}.

\begin{lemma}
[\cite{cv}]\label{chisquared}Let $X\sim\mathcal{G}^{n\times n}$\ be a complex
Gaussian random matrix. \ Then $\left\vert \operatorname*{Det}\left(
X\right)  \right\vert ^{2}$\ has the same distribution as%
\[
\prod_{i=1}^{n}\left(  \sum_{j=1}^{i}\left\vert \xi_{ij}\right\vert
^{2}\right)
\]
where the $\xi_{ij}$'s are independent $\mathcal{N}\left(  0,1\right)
_{\mathbb{C}}$\ Gaussians. \ (In other words, $\left\vert \operatorname*{Det}%
\left(  X\right)  \right\vert ^{2}$\ is distributed as $T_{1}\cdots T_{n}%
$,\ where each $T_{k}$\ is an independent $\chi^{2}$\ random variable with $k$
degrees of freedom.)
\end{lemma}

The proof of Lemma \ref{chisquared}\ (which we omit) uses the interpretation
of the determinant as the volume of a parallelepiped, together with the
spherical symmetry of the Gaussian distribution. \ 

As with the permanent, it will be convenient to work with the normalized
random variable%
\[
D_{n}:=\frac{\left\vert \operatorname*{Det}\left(  X\right)  \right\vert ^{2}%
}{n!},
\]
so that $\operatorname*{E}\left[  D_{n}\right]  =1$. \ Using Lemma
\ref{chisquared}, we now calculate the moments of $D_{n}$.

\begin{lemma}
\label{detmoments}For all real numbers $\alpha>-1$,%
\[
\operatorname*{E}\left[  D_{n}^{\alpha}\right]  =\frac{1}{\left(  n!\right)
^{\alpha}}\prod_{k=1}^{n}\frac{\Gamma\left(  k+\alpha\right)  }{\Gamma\left(
k\right)  }.
\]
(If $\alpha\leq-1$\ then $\operatorname*{E}\left[  D_{n}^{\alpha}\right]
=\infty$.)
\end{lemma}

\begin{proof}
By Lemma \ref{chisquared},%
\begin{align*}
\operatorname*{E}\left[  D_{n}^{\alpha}\right]   &  =\frac{1}{\left(
n!\right)  ^{\alpha}}\operatorname*{E}\left[  T_{1}^{\alpha}\cdots
T_{n}^{\alpha}\right] \\
&  =\frac{1}{\left(  n!\right)  ^{\alpha}}\prod_{k=1}^{n}\operatorname*{E}%
\left[  T_{k}^{\alpha}\right]  ,
\end{align*}
where each $T_{k}$\ is an independent $\chi^{2}$\ random variable with $k$
degrees of freedom. \ Now, $T_{k}$\ has probability density function%
\[
f\left(  x\right)  =\frac{e^{-x}x^{k-1}}{\Gamma\left(  k\right)  }%
\]
for $x\geq0$. \ So%
\begin{align*}
\operatorname*{E}\left[  T_{k}^{\alpha}\right]   &  =\frac{1}{\Gamma\left(
k\right)  }\int_{0}^{\infty}e^{-x}x^{k+\alpha-1}dx\\
&  =\frac{\Gamma\left(  k+\alpha\right)  }{\Gamma\left(  k\right)  }%
\end{align*}
as long as $k+\alpha>0$. \ (If $k+\alpha\leq0$, as can happen if $\alpha
\leq-1$,\ then the above integral diverges.)
\end{proof}

As a sample application of Lemma \ref{detmoments}, if $\alpha$\ is a positive
integer then we get%
\[
\operatorname*{E}\left[  D_{n}^{\alpha}\right]  =\prod_{i=1}^{\alpha-1}%
\binom{n+i}{i}=\Theta\left(  n^{\alpha\left(  \alpha-1\right)  /2}\right)  .
\]
For our application, though, we are interested in the dependence of
$\operatorname*{E}\left[  D_{n}^{\alpha}\right]  $ on $n$ when $\alpha$\ is
\textit{not} necessarily a positive integer. \ The next lemma shows that the
asymptotic behavior above generalizes to negative and fractional $\alpha$.

\begin{lemma}
\label{stirling}For all real numbers $\alpha>-1$, there exists a positive
constant $C_{\alpha}$\ such that%
\[
\lim_{n\rightarrow\infty}\frac{\operatorname*{E}\left[  D_{n}^{\alpha}\right]
}{n^{\alpha\left(  \alpha-1\right)  /2}}=C_{\alpha}.
\]

\end{lemma}

\begin{proof}
Let us write%
\[
\operatorname*{E}\left[  D_{n}^{\alpha}\right]  =\frac{\Gamma\left(
1+\alpha\right)  }{n^{\alpha}}\prod_{k=1}^{n-1}\frac{\Gamma\left(
k+\alpha+1\right)  }{k^{\alpha}\Gamma\left(  k+1\right)  }.
\]
Then by Stirling's approximation,%
\begin{align*}
\ln\prod_{k=1}^{n-1}\frac{\Gamma\left(  k+\alpha+1\right)  }{k^{\alpha}%
\Gamma\left(  k+1\right)  }  &  =\sum_{k=1}^{n-1}\left(  \ln\frac
{\Gamma\left(  k+\alpha+1\right)  }{\Gamma\left(  k+1\right)  }-\alpha\ln
k\right) \\
&  =H_{\alpha}+o\left(  1\right)  +\sum_{k=1}^{n-1}\left(  \ln\left(
\frac{\sqrt{2\pi\left(  k+\alpha\right)  }\left(  \frac{k+\alpha}{e}\right)
^{k+\alpha}}{\sqrt{2\pi k}\left(  \frac{k}{e}\right)  ^{k}}\right)  -\alpha\ln
k\right) \\
&  =H_{\alpha}+o\left(  1\right)  +\sum_{k=1}^{n-1}\left(  \left(
k+\alpha+\frac{1}{2}\right)  \ln\left(  \frac{k+\alpha}{k}\right)
-\alpha\right) \\
&  =H_{\alpha}+J_{\alpha}+o\left(  1\right)  +\sum_{k=1}^{n-1}\left(  \left(
k+\alpha+\frac{1}{2}\right)  \left(  \frac{\alpha}{k}-\frac{\alpha^{2}}%
{2k^{2}}\right)  -\alpha\right) \\
&  =H_{\alpha}+J_{\alpha}+o\left(  1\right)  +\sum_{k=1}^{n-1}\left(
\frac{\alpha\left(  \alpha+1\right)  }{2k}-\frac{\alpha^{2}\left(
2\alpha+1\right)  }{4k^{2}}\right) \\
&  =H_{\alpha}+J_{\alpha}+L_{\alpha}+o\left(  1\right)  +\frac{\alpha\left(
\alpha+1\right)  }{2}\ln n.
\end{align*}
In the above, $H_{\alpha}$, $J_{\alpha}$, and $L_{\alpha}$ are finite error
terms that depend only on $\alpha$ (and not $n$):%
\begin{align*}
H_{\alpha}  &  =\sum_{k=1}^{\infty}\ln\left(  \frac{\Gamma\left(
k+\alpha+1\right)  }{\Gamma\left(  k+1\right)  }\frac{\sqrt{k}\left(  \frac
{k}{e}\right)  ^{k}}{\sqrt{k+\alpha}\left(  \frac{k+\alpha}{e}\right)
^{k+\alpha}}\right)  ,\\
J_{\alpha}  &  =\sum_{k=1}^{\infty}\left(  k+\alpha+\frac{1}{2}\right)
\left(  \ln\left(  \frac{k+\alpha}{k}\right)  -\left(  \frac{\alpha}{k}%
-\frac{\alpha^{2}}{2k^{2}}\right)  \right)  ,\\
L_{\alpha}  &  =\frac{\alpha\left(  \alpha+1\right)  }{2}\left(
\lim_{n\rightarrow\infty}\sum_{k=1}^{\infty}\frac{1}{k}-\ln n\right)
-\sum_{k=1}^{\infty}\frac{\alpha^{2}\left(  2\alpha+1\right)  }{4k^{2}}\\
&  =\frac{\alpha\left(  \alpha+1\right)  \gamma}{2}-\frac{\alpha^{2}\left(
2\alpha+1\right)  \pi^{2}}{24k^{2}},
\end{align*}
where $\gamma\approx0.577$\ is the Euler-Mascheroni constant. \ The $o\left(
1\right)  $'s represent additional error terms that go to $0$
as\ $n\rightarrow\infty$. \ Hence%
\[
\prod_{k=1}^{n-1}\frac{\Gamma\left(  k+\alpha+1\right)  }{k^{\alpha}%
\Gamma\left(  k+1\right)  }=e^{H_{\alpha}+J_{\alpha}+L_{\alpha}+o\left(
1\right)  }n^{\alpha\left(  \alpha+1\right)  /2}%
\]
and%
\begin{align*}
\lim_{n\rightarrow\infty}\frac{\operatorname*{E}\left[  D_{n}^{\alpha}\right]
}{n^{\alpha\left(  \alpha-1\right)  /2}}  &  =\lim_{n\rightarrow\infty}\left(
\frac{1}{n^{\alpha\left(  \alpha-1\right)  /2}}\cdot\frac{\Gamma\left(
1+\alpha\right)  }{n^{\alpha}}e^{H_{\alpha}+J_{\alpha}+L_{\alpha}+o\left(
1\right)  }n^{\alpha\left(  \alpha+1\right)  /2}\right) \\
&  =\Gamma\left(  1+\alpha\right)  e^{H_{\alpha}+J_{\alpha}+L_{\alpha}},
\end{align*}
which is a positive constant $C_{\alpha}$\ depending on $\alpha$.
\end{proof}

We can now complete the proof of Theorem \ref{detthm}.

\begin{proof}
[Proof of Theorem \ref{detthm}]Let $\alpha:=-\beta/2$. \ Then by Markov's
inequality, for all $\varepsilon>0$\ we have%
\begin{align*}
\operatorname*{E}\left[  D_{n}^{\alpha}\right]   &  =\operatorname*{E}\left[
\left(  \frac{\sqrt{n!}}{\left\vert \operatorname*{Det}\left(  X\right)
\right\vert }\right)  ^{\beta}\right] \\
&  \geq\Pr_{X\sim\mathcal{G}^{n\times n}}\left[  \left\vert
\operatorname*{Det}\left(  X\right)  \right\vert <\varepsilon\sqrt{n!}\right]
\cdot\frac{1}{\varepsilon^{\beta}}.
\end{align*}
Hence%
\begin{align*}
\Pr_{X\sim\mathcal{G}^{n\times n}}\left[  \left\vert \operatorname*{Det}%
\left(  X\right)  \right\vert <\varepsilon\sqrt{n!}\right]   &  \leq
\operatorname*{E}\left[  D_{n}^{\alpha}\right]  \cdot\varepsilon^{\beta}\\
&  <C_{\alpha}n^{\alpha\left(  \alpha-1\right)  /2}\varepsilon^{\beta}\\
&  =C_{\beta}^{\prime}n^{\beta\left(  \beta+2\right)  /8}\varepsilon^{\beta}%
\end{align*}
for some positive constants $C_{\alpha},C_{\beta}^{\prime}$\ depending only on
$\alpha$\ and $\beta$\ respectively.
\end{proof}

\subsection{Weak Version of the PACC\label{WEAKPACC}}

We prove the following theorem about concentration of Gaussian permanents.

\begin{theorem}
[Weak Anti-Concentration of the Permanent]\label{weakpacc}For all $\alpha<1$,%
\[
\Pr_{X\sim\mathcal{G}^{n\times n}}\left[  \left\vert \operatorname*{Per}%
\left(  X\right)  \right\vert ^{2}\geq\alpha\cdot n!\right]  >\frac{\left(
1-\alpha\right)  ^{2}}{n+1}.
\]

\end{theorem}

While Theorem \ref{weakpacc}\ falls short of proving Conjecture \ref{pacc}, it
at least shows that $\left\vert \operatorname*{Per}\left(  X\right)
\right\vert $\ has a \textit{non-negligible} probability of being large enough
for our application when $X$ is a Gaussian random matrix. \ In other words, it
rules out the possibility that $\left\vert \operatorname*{Per}\left(
X\right)  \right\vert $ is almost always tiny compared to its expected value,
and that only for (say) a $1/\exp\left(  n\right)  $\ fraction of matrices
$X$\ does $\left\vert \operatorname*{Per}\left(  X\right)  \right\vert $
become enormous.

Recall that $P_{n}$\ denotes the random variable $\left\vert
\operatorname*{Per}\left(  X\right)  \right\vert ^{2}/n!$, and that
$\operatorname*{E}\left[  P_{n}\right]  =1$. \ Our proof of Theorem
\ref{weakpacc} will proceed by showing that $\operatorname*{E}\left[
P_{n}^{2}\right]  =n+1$. \ As we will see later, it is almost an
\textquotedblleft accident\textquotedblright\ that this is
true---$\operatorname*{E}\left[  P_{n}^{3}\right]  $, $\operatorname*{E}%
\left[  P_{n}^{4}\right]  $, and so on all grow exponentially with $n$---but
it is enough to imply Theorem \ref{weakpacc}.

To calculate $\operatorname*{E}\left[  P_{n}^{2}\right]  $, we first need a
proposition about the number of cycles in a random permutation, which can be
found in Lange \cite[p. 76]{lange}\ for example, though we prove it for
completeness. \ Given a permutation $\sigma\in S_{n}$, let
$\operatorname*{cyc}\left(  \sigma\right)  $ be the number of cycles in
$\sigma$.

\begin{proposition}
\label{permpower}For any constant $c\geq1$,%
\[
\operatorname*{E}_{\sigma\in S_{n}}\left[  c^{\operatorname*{cyc}\left(
\sigma\right)  }\right]  =\binom{n+c-1}{c-1}.
\]

\end{proposition}

\begin{proof}
Assume for simplicity that $c$ is a positive integer.\ \ Define a
$c$\textit{-colored permutation} (on $n$ elements) to be a permutation
$\sigma\in S_{n}$\ in which every cycle is colored one of $c$ possible colors.
\ Then clearly the number of $c$-colored permutations equals%
\[
f\left(  n\right)  :=\sum_{\sigma\in S_{n}}c^{\operatorname*{cyc}\left(
\sigma\right)  }.
\]
Now consider forming a $c$-colored permutation $\sigma$. \ There are $n$
possible choices for $\sigma\left(  1\right)  $. \ If $\sigma\left(  1\right)
=1$, then we have completed a cycle of length $1$, and there are $c$ possible
colors for that cycle. \ Therefore the number of $c$-colored permutations
$\sigma$\ such that $\sigma\left(  1\right)  =1$\ is $c\cdot f\left(
n-1\right)  $. \ On the other hand, if $\sigma\left(  1\right)  =b$ for some
$b\neq1$, then we can treat the pair $\left(  1,b\right)  $\ as though it were
a single element, with an incoming edge to $1$ and an outgoing edge from $b$.
\ Therefore the number of $c$-colored permutations $\sigma$\ such that
$\sigma\left(  1\right)  =b$\ is $f\left(  n-1\right)  $. \ Combining, we
obtain the recurrence relation%
\begin{align*}
f\left(  n\right)   &  =c\cdot f\left(  n-1\right)  +\left(  n-1\right)
f\left(  n-1\right) \\
&  =\left(  n+c-1\right)  f\left(  n-1\right)  .
\end{align*}
Together with the base case $f\left(  0\right)  =1$, this implies that%
\begin{align*}
f\left(  n\right)   &  =\left(  n+c-1\right)  \left(  n+c-2\right)
\cdot\cdots\cdot c\\
&  =\binom{n+c-1}{c-1}\cdot n!.
\end{align*}
Hence%
\[
\operatorname*{E}_{\sigma\in S_{n}}\left[  c^{\operatorname*{cyc}\left(
\sigma\right)  }\right]  =\frac{f\left(  n\right)  }{n!}=\binom{n+c-1}{c-1}.
\]
The above argument can be generalized to non-integer $c$ using standard tricks
(though we will not need that in the paper).
\end{proof}

We can now compute $\operatorname*{E}\left[  P_{n}^{2}\right]  $.

\begin{lemma}
\label{fourthmoment}$\operatorname*{E}\left[  P_{n}^{2}\right]  =n+1$.
\end{lemma}

\begin{proof}
We have%
\begin{align*}
\operatorname*{E}\left[  P_{n}^{2}\right]   &  =\frac{1}{\left(  n!\right)
^{2}}\operatorname*{E}_{X\sim\mathcal{G}^{n\times n}}\left[
\operatorname*{Per}\left(  X\right)  ^{2}\overline{\operatorname*{Per}\left(
X\right)  }^{2}\right] \\
&  =\frac{1}{\left(  n!\right)  ^{2}}\operatorname*{E}_{X\sim\mathcal{G}%
^{n\times n}}\left[  \sum_{\sigma,\tau,\alpha,\beta\in S_{n}}%
%TCIMACRO{\dprod \limits_{i=1}^{n}}%
%BeginExpansion
{\displaystyle\prod\limits_{i=1}^{n}}
%EndExpansion
x_{i,\sigma\left(  i\right)  }x_{i,\tau\left(  i\right)  }\overline
{x}_{i,\alpha\left(  i\right)  }\overline{x}_{i,\beta\left(  i\right)
}\right] \\
&  =\frac{1}{\left(  n!\right)  ^{2}}\sum_{\sigma,\tau,\alpha,\beta\in S_{n}%
}M\left(  \sigma,\tau,\alpha,\beta\right)
\end{align*}
where%
\begin{align*}
M\left(  \sigma,\tau,\alpha,\beta\right)  :=  &  \operatorname*{E}%
_{X\sim\mathcal{G}^{n\times n}}\left[
%TCIMACRO{\dprod \limits_{i=1}^{n}}%
%BeginExpansion
{\displaystyle\prod\limits_{i=1}^{n}}
%EndExpansion
x_{i,\sigma\left(  i\right)  }x_{i,\tau\left(  i\right)  }\overline
{x}_{i,\alpha\left(  i\right)  }\overline{x}_{i,\beta\left(  i\right)
}\right] \\
=  &
%TCIMACRO{\dprod \limits_{i=1}^{n}}%
%BeginExpansion
{\displaystyle\prod\limits_{i=1}^{n}}
%EndExpansion
\operatorname*{E}_{X\sim\mathcal{G}^{n\times n}}\left[  x_{i,\sigma\left(
i\right)  }x_{i,\tau\left(  i\right)  }\overline{x}_{i,\alpha\left(  i\right)
}\overline{x}_{i,\beta\left(  i\right)  }\right]  ,
\end{align*}
the last line following from the independence of the Gaussian variables
$x_{ij}$.

We now evaluate $M\left(  \sigma,\tau,\alpha,\beta\right)  $. \ Write
$\sigma\cup\tau=\alpha\cup\beta$ if%
\[
\left\{  \left(  1,\sigma\left(  1\right)  \right)  ,\left(  1,\tau\left(
1\right)  \right)  ,\ldots,\left(  n,\sigma\left(  n\right)  \right)  ,\left(
n,\tau\left(  n\right)  \right)  \right\}  =\left\{  \left(  1,\alpha\left(
1\right)  \right)  ,\left(  1,\beta\left(  1\right)  \right)  \ldots,\left(
n,\alpha\left(  n\right)  \right)  ,\left(  n,\beta\left(  n\right)  \right)
\right\}  .
\]
If $\sigma\cup\tau\neq\alpha\cup\beta$, then we claim that $M\left(
\sigma,\tau,\alpha,\beta\right)  =0$. \ This is because the Gaussian
distribution is uniform over phases---so if there exists an $x_{ij}$\ that is
not \textquotedblleft paired\textquotedblright\ with its complex conjugate
$\overline{x}_{ij}$\ (or vice versa), then the variations in that $x_{ij}%
$\ will cause the entire product to equal $0$. \ So suppose instead that
$\sigma\cup\tau=\alpha\cup\beta$. \ Then for each $i\in\left[  n\right]  $ in
the product, there are two cases. \ First, if $\sigma\left(  i\right)
\neq\tau\left(  i\right)  $, then%
\begin{align*}
\operatorname*{E}_{X\sim\mathcal{G}^{n\times n}}\left[  x_{i,\sigma\left(
i\right)  }x_{i,\tau\left(  i\right)  }\overline{x}_{i,\alpha\left(  i\right)
}\overline{x}_{i,\beta\left(  i\right)  }\right]   &  =\operatorname*{E}%
_{X\sim\mathcal{G}^{n\times n}}\left[  \left\vert x_{i,\sigma\left(  i\right)
}\right\vert ^{2}\left\vert x_{i,\tau\left(  i\right)  }\right\vert
^{2}\right] \\
&  =\operatorname*{E}_{X\sim\mathcal{G}^{n\times n}}\left[  \left\vert
x_{i,\sigma\left(  i\right)  }\right\vert ^{2}\right]  \operatorname*{E}%
_{X\sim\mathcal{G}^{n\times n}}\left[  \left\vert x_{i,\tau\left(  i\right)
}\right\vert ^{2}\right] \\
&  =1.
\end{align*}
Second, if $\sigma\left(  i\right)  =\tau\left(  i\right)  $, then%
\[
\operatorname*{E}_{X\sim\mathcal{G}^{n\times n}}\left[  x_{i,\sigma\left(
i\right)  }x_{i,\tau\left(  i\right)  }\overline{x}_{i,\alpha\left(  i\right)
}\overline{x}_{i,\beta\left(  i\right)  }\right]  =\operatorname*{E}%
_{X\sim\mathcal{G}^{n\times n}}\left[  \left\vert x_{i,\sigma\left(  i\right)
}\right\vert ^{4}\right]  =2.
\]
The result is that $M\left(  \sigma,\tau,\alpha,\beta\right)  =2^{K\left(
\sigma,\tau\right)  }$, where $K\left(  \sigma,\tau\right)  $\ is the number
of $i$'s such that $\sigma\left(  i\right)  =\tau\left(  i\right)  $.

Now let $N\left(  \sigma,\tau\right)  $\ be the number of pairs $\alpha
,\beta\in S_{n}$\ such that $\sigma\cup\tau=\alpha\cup\beta$. \ Then%
\begin{align*}
\operatorname*{E}\left[  P_{n}^{4}\right]   &  =\frac{1}{\left(  n!\right)
^{2}}\sum_{\sigma,\tau,\alpha,\beta\in S_{n}}M\left(  \sigma,\tau,\alpha
,\beta\right) \\
&  =\frac{1}{\left(  n!\right)  ^{2}}\sum_{\sigma,\tau\in S_{n}}2^{K\left(
\sigma,\tau\right)  }N\left(  \sigma,\tau\right) \\
&  =\operatorname*{E}_{\sigma,\tau\in S_{n}}\left[  2^{K\left(  \sigma
,\tau\right)  }N\left(  \sigma,\tau\right)  \right] \\
&  =\operatorname*{E}_{\sigma,\tau\in S_{n}}\left[  2^{K(\sigma^{-1}%
\sigma,\sigma^{-1}\tau)}N\left(  \sigma^{-1}\sigma,\sigma^{-1}\tau\right)
\right] \\
&  =\operatorname*{E}_{\xi\in S_{n}}\left[  2^{K\left(  e,\xi\right)
}N\left(  e,\xi\right)  \right]  ,
\end{align*}
where $e$\ denotes the identity permutation. \ Here the fourth line follows
from symmetry---specifically, from the easily-checked identities $K\left(
\sigma,\tau\right)  =K\left(  \alpha\sigma,\alpha\tau\right)  $\ and $N\left(
\sigma,\tau\right)  =N\left(  \alpha\sigma,\alpha\tau\right)  $.

We claim that the quantity $2^{K(e,\xi)}N\left(  e,\xi\right)  $\ has a simple
combinatorial interpretation as $2^{\operatorname*{cyc}\left(  \xi\right)  }$,
where $\operatorname*{cyc}\left(  \xi\right)  $\ is the number of cycles in
$\xi$. \ To see this, consider a bipartite multigraph $G$\ with $n$ vertices
on each side, and an edge from left-vertex \thinspace$i$\ to right-vertex $j$
if $i=j$\ or $\xi\left(  i\right)  =j$\ (or a double-edge from $i$\ to $j$\ if
$i=j$\ \textit{and} $\xi\left(  i\right)  =j$). \ Then since $e$\ and $\xi
$\ are both permutations, $G$ is a disjoint union of cycles. \ By definition,
$K\left(  e,\xi\right)  $\ equals the number of indices $i$ such that
$\xi\left(  i\right)  =i$---which is simply the number of double-edges in $G$,
or equivalently, the number of cycles in $\xi$ of length $1$. \ Also,
$N\left(  e,\xi\right)  $ equals the number of ways to partition the edges of
$G$ into two perfect matchings, corresponding to $\alpha$\ and $\beta
$\ respectively. \ In partitioning $G$, the only freedom we have is that each
cycle in $G$ of length at least $4$ can be decomposed in two inequvalent ways.
\ This implies that $N\left(  e,\xi\right)  =2^{L\left(  \xi\right)  }$, where
$L\left(  \xi\right)  $\ is the number of cycles in $\xi$ of length at least
$2$ (note that a cycle in $\xi$\ of length $k$ gives rise to a cycle in
$G$\ of length $2k$). \ Combining,%
\[
2^{K\left(  e,\xi\right)  }N\left(  e,\xi\right)  =2^{K\left(  e,\xi\right)
+L\left(  \xi\right)  }=2^{\operatorname*{cyc}\left(  \xi\right)  }.
\]
Hence%
\[
\operatorname*{E}\left[  P_{n}^{2}\right]  =\operatorname*{E}_{\xi\in S_{n}%
}\left[  2^{\operatorname*{cyc}\left(  \xi\right)  }\right]  =n+1
\]
by Proposition \ref{permpower}.
\end{proof}

Using Lemma \ref{fourthmoment}, we can now complete the proof of Theorem
\ref{weakpacc},\ that $\Pr\left[  P_{n}\geq\alpha\right]  >\frac{\left(
1-\alpha\right)  ^{2}}{n+1}$.

\begin{proof}
[Proof of Theorem \ref{weakpacc}]Let $F$ denote the event that $P_{n}%
\geq\alpha$, and let $\delta:=\Pr\left[  F\right]  $. \ Then%
\begin{align*}
1  &  =\operatorname*{E}\left[  P_{n}\right] \\
&  =\Pr\left[  F\right]  \operatorname*{E}\left[  P_{n}~|~F\right]
+\Pr\left[  \overline{F}\right]  \operatorname*{E}\left[  P_{n}~|~\overline
{F}\right] \\
&  <\delta\operatorname*{E}\left[  P_{n}~|~F\right]  +\alpha,
\end{align*}
so%
\[
\operatorname*{E}\left[  P_{n}~|~F\right]  >\frac{1-\alpha}{\delta}.
\]
By Cauchy-Schwarz, this implies%
\[
\operatorname*{E}\left[  P_{n}^{2}~|~F\right]  >\frac{\left(  1-\alpha\right)
^{2}}{\delta^{2}}%
\]
and hence%
\begin{align*}
\operatorname*{E}\left[  P_{n}^{2}\right]   &  =\Pr\left[  F\right]
\operatorname*{E}\left[  P_{n}^{2}~|~F\right]  +\Pr\left[  \overline
{F}\right]  \operatorname*{E}\left[  P_{n}^{2}~|~\overline{F}\right] \\
&  >\delta\cdot\frac{\left(  1-\alpha\right)  ^{2}}{\delta^{2}}+0\\
&  =\frac{\left(  1-\alpha\right)  ^{2}}{\delta}.
\end{align*}
Now, we know from Lemma \ref{fourthmoment}\ that $\operatorname*{E}\left[
P_{n}^{2}\right]  =n+1$. \ Rearranging, this means that%
\[
\delta>\frac{\left(  1-\alpha\right)  ^{2}}{n+1}%
\]
which is what we wanted to show.
\end{proof}

A natural approach to proving Conjecture \ref{pacc}\ would be to calculate the
\textit{higher} moments of $P_{n}$---$\operatorname*{E}\left[  P_{n}%
^{3}\right]  $, $\operatorname*{E}\left[  P_{n}^{4}\right]  $, and so on---by
generalizing Lemma \ref{fourthmoment}. \ In principle, these moments would
determine the probability density function of $P_{n}$\ completely.

When we do so, here is what we find. \ Given a bipartite $k$-regular
multigraph $G$ with $n$ vertices on each side, let $M\left(  G\right)  $\ be
the number of ways to decompose $G$ into an ordered list of $k$ disjoint
perfect matchings. \ Also, let $M_{k}$\ be the expectation of $M\left(
G\right)  $\ over a $k$-regular bipartite multigraph $G$\ chosen uniformly at
random. \ Then the proof of Lemma \ref{fourthmoment}\ extends to show the following:

\begin{theorem}
\label{gentm}$\operatorname*{E}\left[  P_{n}^{k}\right]  =M_{k}$ for all
positive integers $k$.
\end{theorem}

However, while $M_{1}=1$\ and $M_{2}=n+1$, it is also known that $M_{k}%
\sim\left(  k/e\right)  ^{n}$\ for all $k\geq3$: this follows from the
\textit{van der Waerden conjecture}, which was proved by Falikman
\cite{falikman} and Egorychev \cite{egorychev} in 1981. \ In other words, the
higher moments of $P_{n}$\ grow exponentially with $n$. \ Because of this, it
seems one would need to know the higher moments extremely precisely in order
to conclude anything about the quantities of interest, such as $\Pr\left[
P_{n}<\alpha\right]  $.

\section{The Hardness of Gaussian Permanents\label{HARDPER}}

In this section, we move on to discuss Conjecture \ref{pgc}, which says that
\textsc{GPE}$_{\times}$---the problem of multiplicatively estimating
$\operatorname*{Per}\left(  X\right)  $, where $X\sim\mathcal{G}^{n\times n}%
$\ is a Gaussian random matrix---is $\mathsf{\#P}$-hard.\ \ Proving Conjecture
\ref{pgc}\ is the central theoretical challenge that we leave.\footnote{Though
note that, for our \textsc{BosonSampling} hardness argument to work, all we
\textit{really} need is that estimating $\operatorname*{Per}\left(  X\right)
$ for Gaussian $X$ is not in the class $\mathsf{BPP}^{\mathsf{NP}}$, and one
could imagine giving evidence for this that fell short of $\mathsf{\#P}%
$-hardness.}

Intuitively, Conjecture \ref{pgc} implies that if $\mathsf{P}^{\mathsf{\#P}%
}\neq\mathsf{BPP}$, then no algorithm for \textsc{GPE}$_{\times}$\ can run in
time $\operatorname*{poly}\left(  n,1/\varepsilon,1/\delta\right)  $. \ Though
it will not be needed for this work, one could also consider a stronger
conjecture, which would say that if $\mathsf{P}^{\mathsf{\#P}}\neq
\mathsf{BPP}$, then no algorithm for \textsc{GPE}$_{\times}$\ can run in
time\ $n^{f\left(  \varepsilon,\delta\right)  }$ for any function $f$.

In contrast to the case of the Permanent Anti-Concentration Conjecture, the
question arises of why one should even expect Conjecture \ref{pgc} to be true.
\ Undoubtedly the main reason is that the analogous statement for
\textit{permanents over finite fields} is true: this is the \textit{random
self-reducibility of the permanent},\ first proved by Lipton \cite{lipton}.
\ Thus, we are \textquotedblleft merely\textquotedblright\ asking for the real
or complex analogue of something already known in the finite field case.

A second piece of evidence for Conjecture \ref{pgc} is that, if $X\sim
\mathcal{G}^{n\times n}$ is a Gaussian matrix, then all known approximation
algorithms fail to find any reasonable approximation to $\operatorname*{Per}%
\left(  X\right)  $. \ If $X$ were a \textit{nonnegative} matrix, then we
could use the celebrated approximation algorithm of Jerrum, Sinclair, and
Vigoda \cite{jsv}---but since $X$ has negative and complex entries, it is not
even clear how to estimate $\operatorname*{Per}\left(  X\right)  $\ in
$\mathsf{BPP}^{\mathsf{NP}}$, let alone in $\mathsf{BPP}$. \ Perhaps the most
relevant approximation algorithms are those of Gurvits \cite{gurvits:alg},
which we discuss in Appendix \ref{ALGS}. \ In particular, Theorem
\ref{gurvitsthm} will give a randomized algorithm due to Gurvits that
approximates $\operatorname*{Per}\left(  X\right)  $\ to within an additive
error $\pm\varepsilon\left\Vert X\right\Vert ^{n}$, in $O\left(
n^{2}/\varepsilon^{2}\right)  $\ time. \ For a Gaussian matrix $X\sim
\mathcal{G}^{n\times n}$, it is known that $\left\Vert X\right\Vert
\thickapprox2\sqrt{n}$ almost surely. \ So in $O\left(  n^{2}/\varepsilon
^{2}\right)  $\ time, we can approximate $\operatorname*{Per}\left(  X\right)
$ to within additive error $\pm\varepsilon\left(  2\sqrt{n}\right)  ^{n}$.
\ However, this is larger than what we need (namely $\pm\varepsilon\sqrt
{n!}/\operatorname*{poly}\left(  n\right)  $) by a $\thicksim\left(  2\sqrt
{e}\right)  ^{n}$ factor.

In the rest of this section, we discuss the prospects for proving Conjecture
\ref{pgc}. \ First, in Section \ref{EVPGC},\ we at least show that
\textit{exactly} computing $\operatorname*{Per}\left(  X\right)  $ for a
Gaussian random matrix $X\sim\mathcal{G}^{n\times n}$\ is $\mathsf{\#P}$-hard.
\ The proof is a simple extension of the classic result of Lipton
\cite{lipton},\ that the permanent over \textit{finite fields} is
\textquotedblleft random self-reducible\textquotedblright: that is, as hard to
compute on average as it is in the worst case. \ As in Lipton's proof, we use
the facts that (1) the permanent is a low-degree polynomial, and (2)
low-degree polynomials constitute excellent error-correcting codes. \ However,
in Section \ref{BARRIER}, we then explain why \textit{any extension of this
result to show average-case hardness of approximating }$\operatorname*{Per}%
\left(  X\right)  $\textit{\ will require a fundamentally new approach.} \ In
other words, the \textquotedblleft polynomial reconstruction
paradigm\textquotedblright\ cannot suffice, on its own, to prove Conjecture
\ref{pgc}.

\subsection{Evidence That \textsc{GPE}$_{\times}$\ Is $\mathsf{\#P}%
$-Hard\label{EVPGC}}

We already saw, in Theorem \ref{approxhard}, that approximating the
permanent\ (or even the magnitude of the permanent) of \textit{all} matrices
$X\in\mathbb{C}^{n\times n}$\ is a $\mathsf{\#P}$-hard\ problem. \ But what
about the \textquotedblleft opposite\textquotedblright\ problem:
\textit{exactly} computing the permanent of \textit{most} matrices
$X\sim\mathcal{G}^{n\times n}$? \ In this section, we will show that the
latter problem is $\mathsf{\#P}$-hard as well. \ This means that, if we want
to prove the Permanent-of-Gaussians Conjecture, then the difficulty really is
just to \textit{combine} approximation with an average-case assumption.

Our result will be an adaptation of a famous result on the
random-self-reducibility of the permanent over \textit{finite} fields:

\begin{theorem}
[Random-Self-Reducibility of the \textsc{Permanent} \cite{lipton}%
,\cite{glrsw},\cite{gemmellsudan},\cite{cps}]\label{liptonthm}For all
$\alpha\geq1/\operatorname*{poly}\left(  n\right)  $ and primes $p>\left(
3n/\alpha\right)  ^{2}$, the following problem is $\mathsf{\#P}$-hard: given a
uniform random matrix $M\in\mathbb{F}_{p}^{n\times n}$, output
$\operatorname*{Per}\left(  M\right)  $\ with probability at least $\alpha$
over $M$.
\end{theorem}

The proof of Theorem \ref{liptonthm}\ proceeds by reduction: suppose we had an
oracle $\mathcal{O}$\ such that%
\[
\Pr_{M\in\mathbb{F}_{p}^{n\times n}}\left[  \mathcal{O}\left(  M\right)
=\operatorname*{Per}\left(  M\right)  \right]  \geq\alpha.
\]
Using $\mathcal{O}$, we give a randomized algorithm that computes the
permanent of an \textit{arbitrary} matrix $X\in\mathbb{F}_{p}^{n\times n}$.
\ The latter is certainly a $\mathsf{\#P}$-hard problem, which implies that
computing $\operatorname*{Per}\left(  M\right)  $\ for even an $\alpha$
fraction of $M$'s must have been $\mathsf{\#P}$-hard\ as well.

There are actually \textit{four} variants of Theorem \ref{liptonthm},\ which
handle increasingly small values of $\alpha$. \ All four are based on the same
idea---namely, reconstructing a low-degree polynomial from noisy samples---but
as $\alpha$ gets smaller, one has to use more and more sophisticated
reconstruction methods. \ For convenience, we have summarized the variants in
the table below.%
\[%
\begin{tabular}
[c]{llll}%
\textbf{Success probability }$\alpha$ & \textbf{Reconstruction method} &
\textbf{Curve in }$\mathbb{F}^{n\times n}$ & \textbf{Reference}\\
$1-\frac{1}{3n}$ & Lagrange interpolation & Linear & Lipton \cite{lipton}\\
$\frac{3}{4}+\frac{1}{\operatorname*{poly}\left(  n\right)  }$ &
Berlekamp-Welch & Linear & Gemmell et al.\ \cite{glrsw}\\
$\frac{1}{2}+\frac{1}{\operatorname*{poly}\left(  n\right)  }$ &
Berlekamp-Welch & Polynomial & Gemmell-Sudan \cite{gemmellsudan}\\
$\frac{1}{\operatorname*{poly}\left(  n\right)  }$ & Sudan's list decoding
\cite{sudan} & Polynomial & Cai et al.\ \cite{cps}%
\end{tabular}
\ \ \ \
\]
In adapting Theorem \ref{liptonthm} to matrices over $\mathbb{C}$, we face a
choice of which variant to prove. \ For simplicity, we have chosen to prove
only the $\alpha=\frac{3}{4}+\frac{1}{\operatorname*{poly}\left(  n\right)  }%
$\ variant in this paper. \ However, we believe that it should be possible to
adapt the $\alpha=\frac{1}{2}+\frac{1}{\operatorname*{poly}\left(  n\right)
}$\ and $\alpha=\frac{1}{\operatorname*{poly}\left(  n\right)  }$\ variants to
the complex case as well; we leave this as a problem for future work.

Let us start by explaining how the reduction works in the finite field case,
when $\alpha=\frac{3}{4}+\delta$\ for some $\delta=\frac{1}%
{\operatorname*{poly}\left(  n\right)  }$. \ Assume we are given as input a
matrix $X\in\mathbb{F}_{p}^{n\times n}$, where $p\geq n/\delta$\ is a
prime.\ \ We are also given an oracle $\mathcal{O}$\ such that%
\[
\Pr_{M\in\mathbb{F}_{p}^{n\times n}}\left[  \mathcal{O}\left(  M\right)
=\operatorname*{Per}\left(  M\right)  \right]  \geq\frac{3}{4}+\delta.
\]
Then using $\mathcal{O}$, our goal is to compute $\operatorname*{Per}\left(
X\right)  $.

We do so using the following algorithm. \ First choose another matrix
$Y\in\mathbb{F}_{p}^{n\times n}$ uniformly at random. \ Then set%
\begin{align*}
X\left(  t\right)   &  :=X+tY,\\
q\left(  t\right)   &  :=\operatorname*{Per}\left(  X\left(  t\right)
\right)  .
\end{align*}
Notice that $q\left(  t\right)  $ is a univariate polynomial in $t$, of degree
at most $n$. \ Furthermore, $q\left(  0\right)  =\operatorname*{Per}\left(
X\left(  0\right)  \right)  =\operatorname*{Per}\left(  X\right)  $, whereas
for each $t\neq0$, the matrix $X\left(  t\right)  $\ is uniformly random. \ So
by assumption, for each $t\neq0$\ we have%
\[
\Pr\left[  \mathcal{O}\left(  X\left(  t\right)  \right)  =q\left(  t\right)
\right]  \geq\frac{3}{4}+\delta.
\]
Let $S$\ be the set of all nonzero $t$\ such that $\mathcal{O}\left(  X\left(
t\right)  \right)  =q\left(  t\right)  $. \ Then by Markov's inequality,%
\[
\Pr\left[  \left\vert S\right\vert \geq\left(  \frac{1}{2}+\delta\right)
\left(  p-1\right)  \right]  \geq1-\frac{\frac{1}{4}-\delta}{\frac{1}%
{2}-\delta}\geq\frac{1}{2}+\delta.
\]
So if we can just compute $\operatorname*{Per}\left(  X\right)  $\ in the case
where $\left\vert S\right\vert \geq\left(  1/2+\delta\right)  \left(
p-1\right)  $, then all we need to do is run our algorithm $O\left(
1/\delta^{2}\right)  $\ times (with different choices of the matrix $Y$), and
output the majority result.

So the problem reduces to the following: reconstruct a univariate polynomial
$q:\mathbb{F}_{p}\rightarrow\mathbb{F}_{p}$ of degree $n$, given
\textquotedblleft sample data\textquotedblright\ $\mathcal{O}\left(  X\left(
1\right)  \right)  ,\ldots,\mathcal{O}\left(  X\left(  p-1\right)  \right)  $
that satisfies $q\left(  t\right)  =\mathcal{O}\left(  X\left(  t\right)
\right)  $\ for at least a $\frac{1}{2}+\delta$\ fraction of $t$'s.
\ Fortunately, we can solve that problem efficiently using the well-known
\textit{Berlekamp-Welch algorithm}:

\begin{theorem}
[Berlekamp-Welch Algorithm]\label{bwalg}Let $q$\ be a univariate polynomial of
degree $d$, over any field $\mathbb{F}$. \ Suppose we are given $m$ pairs of
$\mathbb{F}$-elements $\left(  x_{1},y_{1}\right)  ,\ldots,\left(  x_{m}%
,y_{m}\right)  $\ (with the $x_{i}$'s\ all distinct), and are promised that
$y_{i}=q\left(  x_{i}\right)  $\ for more than $\frac{m+d}{2}$\ values of $i$.
\ Then there is a deterministic algorithm to reconstruct $q$, using
$\operatorname*{poly}\left(  n,d\right)  $\ field operations.
\end{theorem}

Theorem \ref{bwalg} applies to our scenario provided $p$ is large enough (say,
at least $n/\delta$). \ Once we have the polynomial $q$, we then simply
evaluate it at $0$\ to obtain $q\left(  0\right)  =\operatorname*{Per}\left(
X\right)  $.

The above argument shows that it is $\mathsf{\#P}$-hard to compute the
permanent of a \textquotedblleft random\textquotedblright\ matrix---but only
over a sufficiently-large \textit{finite} field $\mathbb{F}$, and with respect
to the uniform distribution over matrices. \ By contrast, what if $\mathbb{F}%
$\ is the field of complex numbers, and the distribution over matrices is the
Gaussian distribution, $\mathcal{G}^{n\times n}$?

In that case, one can check that the entire argument still goes through,
\textit{except} for the part where we asserted that the matrix $X\left(
t\right)  $\ was uniformly random.\ \ In the Gaussian case, it is easy enough
to arrange that $X\left(  t\right)  \sim\mathcal{G}^{n\times n}$\ for some
\textit{fixed} $t\neq0$,\ but we can no longer ensure that $X\left(  t\right)
\sim\mathcal{G}^{n\times n}$\ for \textit{all} $t\neq0$ simultaneously.
\ Indeed, $X\left(  t\right)  $\ becomes arbitrarily close to the input matrix
$X\left(  0\right)  =X$ as $t\rightarrow0$. \ Fortunately, we can deal with
that problem by means of Lemma \ref{autocor}, which implies that, if the
matrix $M\in\mathbb{C}^{n\times n}$\ is sampled from $\mathcal{G}^{n\times n}%
$\ and if $E$ is a small shift, then $M+E$\ is nearly indistinguishable from a
sample from $\mathcal{G}^{n\times n}$. \ Using Lemma \ref{autocor}, we now
adapt Theorem \ref{liptonthm} to the complex case.

\begin{theorem}
[Random Self-Reducibility of Gaussian Permanent]\label{exactpgc}For all
$\delta\geq1/\operatorname*{poly}\left(  n\right)  $, the following problem is
$\mathsf{\#P}$-hard. \ Given an $n\times n$\ matrix $M$\ drawn from
$\mathcal{G}^{n\times n}$, output $\operatorname*{Per}\left(  M\right)
$\ with probability at least $\frac{3}{4}+\delta$\ over $M$.
\end{theorem}

\begin{proof}
Let $X=\left(  x_{ij}\right)  \in\left\{  0,1\right\}  ^{n\times n}$ be an
arbitrary $0/1$ matrix. \ We will show how to compute $\operatorname*{Per}%
\left(  X\right)  $ in probabilistic polynomial time, given access to an
oracle $\mathcal{O}$\ such that%
\[
\Pr_{M\sim\mathcal{G}^{n\times n}}\left[  \mathcal{O}\left(  M\right)
=\operatorname*{Per}\left(  M\right)  \right]  \geq\frac{3}{4}+\delta.
\]
Clearly this suffices to prove the theorem.

The first step is to choose a matrix $Y\in\mathbb{C}^{n\times n}$ from the
Gaussian distribution $\mathcal{G}^{n\times n}$. \ Then define%
\[
X\left(  t\right)  :=\left(  1-t\right)  Y+tX,
\]
so that $X\left(  0\right)  =Y$ and $X\left(  1\right)  =X$. \ Next define
\[
q\left(  t\right)  :=\operatorname*{Per}\left(  X\left(  t\right)  \right)  ,
\]
so that $q\left(  t\right)  $ is a univariate polynomial in $t$ of degree at
most $n$, and $q\left(  1\right)  =\operatorname*{Per}\left(  X\left(
1\right)  \right)  =\operatorname*{Per}\left(  X\right)  $.

Now let $L:=\left\lceil n/\delta\right\rceil $\ and $\varepsilon:=\frac
{\delta}{\left(  4n^{2}+2n\right)  L}$. \ For each $\ell\in\left[  L\right]
$, call the oracle $\mathcal{O}$\ on input matrix $X\left(  \varepsilon
\ell\right)  $. \ Then, using the Berlekamp-Welch algorithm (Theorem
\ref{bwalg}), attempt to find a degree-$n$ polynomial $q^{\prime}%
:\mathbb{C}\rightarrow\mathbb{C}$ such that%
\[
q^{\prime}\left(  \varepsilon\ell\right)  =\mathcal{O}\left(  X\left(
\varepsilon\ell\right)  \right)
\]
for at least a $\frac{3}{4}+\delta$\ fraction of $\ell\in\left[  L\right]  $.
\ If no such $q^{\prime}$ is found, then fail; otherwise, output $q^{\prime
}\left(  1\right)  $\ as the\ guessed value of $\operatorname*{Per}\left(
X\right)  $.

We claim that the above algorithm succeeds (that is, outputs $q^{\prime
}\left(  1\right)  =\operatorname*{Per}\left(  X\right)  $) with probability
at least $\frac{1}{2}+\frac{\delta}{2}$ over $Y$. \ Provided that holds, it is
clear that the success probability can be boosted to (say) $2/3$, by simply
repeating the algorithm $O\left(  1/\delta^{2}\right)  $\ times with different
choices of $Y$\ and then outputting the majority result.

To prove the claim, note that for each $\ell\in\left[  L\right]  $, one can
think of the matrix $X\left(  \varepsilon\ell\right)  $\ as having been drawn
from the distribution%
\[
\mathcal{D}_{\ell}:=\prod_{i,j=1}^{n}\mathcal{N}\left(  \varepsilon\ell
a_{ij},\left(  1-\varepsilon\ell\right)  ^{2}\right)  _{\mathbb{C}}.
\]
Let%
\[
\mathcal{D}_{\ell}^{\prime}:=\prod_{i,j=1}^{n}\mathcal{N}\left(
\varepsilon\ell a_{ij},1\right)  _{\mathbb{C}}%
\]
Then by the triangle inequality together with Lemma \ref{autocor},%
\begin{align*}
\left\Vert \mathcal{D}_{\ell}-\mathcal{G}^{n\times n}\right\Vert  &
\leq\left\Vert \mathcal{D}_{\ell}-\mathcal{D}_{\ell}^{\prime}\right\Vert
+\left\Vert \mathcal{D}_{\ell}^{\prime}-\mathcal{G}^{n\times n}\right\Vert \\
&  \leq2n^{2}\varepsilon\ell+\sqrt{n^{2}\left(  \varepsilon\ell\right)  ^{2}%
}\\
&  \leq\left(  2n^{2}+n\right)  \varepsilon L\\
&  \leq\frac{\delta}{2}.
\end{align*}
Hence%
\begin{align*}
\Pr\left[  \mathcal{O}\left(  X\left(  \varepsilon\ell\right)  \right)
=q\left(  \varepsilon\ell\right)  \right]   &  \geq\frac{3}{4}+\delta
-\left\Vert \mathcal{D}_{\ell}-\mathcal{N}\left(  0,1\right)  _{\mathbb{C}%
}^{n\times n}\right\Vert \\
&  \geq\frac{3}{4}+\frac{\delta}{2}.
\end{align*}
Now let $S$\ be the set of all $\ell\in\left[  L\right]  $\ such that
$\mathcal{O}\left(  X\left(  \varepsilon\ell\right)  \right)  =q\left(
\varepsilon\ell\right)  $. \ Then by Markov's inequality,%
\[
\Pr\left[  \left\vert S\right\vert \geq\left(  \frac{1}{2}+\frac{\delta}%
{2}\right)  L\right]  \geq1-\frac{\frac{1}{4}-\frac{\delta}{2}}{\frac{1}%
{2}-\frac{\delta}{2}}\geq\frac{1}{2}+\frac{\delta}{2}.
\]
Furthermore, suppose $\left\vert S\right\vert \geq\left(  \frac{1}{2}%
+\frac{\delta}{2}\right)  L$. \ Then by Theorem \ref{bwalg}, the
Berlekamp-Welch algorithm will succeed; that is, its output polynomial
$q^{\prime}$\ will be equal to $q$. \ This proves the claim and hence the lemma.
\end{proof}

As mentioned before, we conjecture that it is possible to improve Theorem
\ref{exactpgc}, to show that it is $\mathsf{\#P}$-hard even to compute the
permanent of an $\alpha=\frac{1}{\operatorname*{poly}\left(  n\right)  }%
$\ fraction of matrices $X$\ drawn from the Gaussian distribution
$\mathcal{G}^{n\times n}$.

Let us mention two other interesting improvements that one can make to Theorem
\ref{exactpgc}. \ First, one can easily modify the proof to show that not just
$\operatorname*{Per}\left(  X\right)  $, but also $\left\vert
\operatorname*{Per}\left(  X\right)  \right\vert ^{2}$,\ is as hard to compute
for $X$ drawn from the Gaussian distribution $\mathcal{G}^{n\times n}$\ as it
is in the worst case. \ For this, one simply needs to observe that, just as
$\operatorname*{Per}\left(  X\right)  $\ is a degree-$n$ polynomial in the
entries of $X$, so $\left\vert \operatorname*{Per}\left(  X\right)
\right\vert ^{2}$\ is a degree-$2n$ polynomial in the entries of $X$ together
with their complex conjugates (or alternatively, in the real and imaginary
parts of the entries). \ The rest of the proof goes through as before. \ Since
$\left\vert \operatorname*{Per}\left(  X\right)  \right\vert ^{2}$\ is
$\mathsf{\#P}$-hard to compute in the worst case\ by Theorem \ref{approxhard},
it follows that $\left\vert \operatorname*{Per}\left(  X\right)  \right\vert
^{2}$\ is $\mathsf{\#P}$-hard to compute\ for $X$ drawn from the Gaussian
distribution as well.

Second, in the proof of Theorem \ref{exactpgc}, one can relax the requirement
that the oracle $\mathcal{O}$\ computes $\operatorname*{Per}\left(  X\right)
$\ \textit{exactly} with high probability over $X\sim\mathcal{G}^{n\times n}$,
and merely require that%
\[
\Pr_{X\sim\mathcal{G}^{n\times n}}\left[  \left\vert \mathcal{O}\left(
X\right)  -\operatorname*{Per}\left(  X\right)  \right\vert \leq2^{-q\left(
n\right)  }\right]  \geq\frac{3}{4}+\frac{1}{\operatorname*{poly}\left(
n\right)  },
\]
for some sufficiently large polynomial $q$. \ To do so, one can appeal to the
following lemma of Paturi.

\begin{lemma}
[Paturi \cite{paturi}; see also Buhrman et al.\ \cite{bcwz}]\label{paturilem}%
Let $p:\mathbb{R}\rightarrow\mathbb{R}$\ be a real polynomial of degree $d$,
and suppose $\left\vert p\left(  x\right)  \right\vert \leq\delta$\ for all
$\left\vert x\right\vert \leq\varepsilon$. \ Then $\left\vert p\left(
1\right)  \right\vert \leq\delta e^{2d\left(  1+1/\varepsilon\right)  }$.
\end{lemma}

From this perspective, the whole challenge in proving the
Permanent-of-Gaussians Conjecture is to replace the $2^{-q\left(  n\right)  }%
$\ approximation error with $1/q\left(  n\right)  $.

Combining, we obtain the following theorem, whose detailed proof we omit.

\begin{theorem}
\label{exactpgc2}There exists a polynomial $p$ for which the following problem
is $\mathsf{\#P}$-hard, for all $\delta\geq1/\operatorname*{poly}\left(
n\right)  $. \ Given an $n\times n$\ matrix $X$\ drawn from $\mathcal{G}%
^{n\times n}$, output a real number $y$ such that $\left\vert y-\left\vert
\operatorname*{Per}\left(  X\right)  \right\vert ^{2}\right\vert
\leq2^{-p\left(  n,1/\delta\right)  }$\ with probability at least $\frac{3}%
{4}+\delta$\ over $X$.
\end{theorem}

As a final observation, it is easy to find \textit{some} efficiently samplable
distribution $\mathcal{D}$\ over matrices $X\in\mathbb{C}^{n\times n}$, such
that estimating $\operatorname*{Per}\left(  X\right)  $\ or $\left\vert
\operatorname*{Per}\left(  X\right)  \right\vert ^{2}$ for most $X\sim
\mathcal{D}$\ is a $\mathsf{\#P}$-hard problem. \ To do so, simply start with
any problem that is known to be $\mathsf{\#P}$-hard\ on average: for example,
computing $\operatorname*{Per}\left(  M\right)  $\ for most matrices
$M\in\mathbb{F}_{p}^{n\times n}$\ over a \textit{finite} field $\mathbb{F}%
_{p}$. \ Next, use Theorem \ref{approxhard}\ to reduce the computation of
$\operatorname*{Per}\left(  M\right)  $\ (for a uniform random $M$) to the
estimation of $\left\vert \operatorname*{Per}\left(  X_{1}\right)  \right\vert
^{2},\ldots,\left\vert \operatorname*{Per}\left(  X_{m}\right)  \right\vert
^{2}$, for various matrices $X_{1},\ldots,X_{m}\in\mathbb{C}^{n\times n}$.
\ Finally, output a random $X_{i}$\ as one's sample from $\mathcal{D}$.
\ Clearly, if one could estimate $\left\vert \operatorname*{Per}\left(
X\right)  \right\vert ^{2}$\ for a $1-1/\operatorname*{poly}\left(  n\right)
$\ fraction of $X\sim\mathcal{D}$, one could also compute $\operatorname*{Per}%
\left(  M\right)  $\ for a $1-1/\operatorname*{poly}\left(  n\right)
$\ fraction of\ $M\in\mathbb{F}_{p}^{n\times n}$, and thereby solve a
$\mathsf{\#P}$-hard problem. \ Because of this, we see that the challenge is
\textquotedblleft merely\textquotedblright\ how to prove average-case
$\mathsf{\#P}$-hardness, in the specific case where the distribution
$\mathcal{D}$\ over matrices that interests us is the Gaussian distribution
$\mathcal{G}^{n\times n}$ (or more generally, some other \textquotedblleft
nice\textquotedblright\ or \textquotedblleft uniform-looking\textquotedblright\ distribution).

\subsection{The Barrier to Proving the PGC\label{BARRIER}}

In this section, we identify a significant barrier to proving Conjecture
\ref{pgc}, and explain why a new approach seems needed.

As Section \ref{EVPGC}\ discussed, all existing proofs of the
worst-case/average-case equivalence of the \textsc{Permanent}\ are based on
\textit{low-degree polynomial interpolation}. \ More concretely, given a
matrix $X\in\mathbb{F}^{n\times n}$ for which we want to compute
$\operatorname*{Per}\left(  X\right)  $, we first choose a random low-degree
curve $X\left(  t\right)  $\ through $\mathbb{F}^{n\times n}$\ satisfying
$X\left(  0\right)  =X$. \ We then choose nonzero points $t_{1},\ldots
,t_{m}\in\mathbb{R}$, and compute or approximate $\operatorname*{Per}\left(
X\left(  t_{i}\right)  \right)  $ for all $i\in\left[  m\right]  $, using the
assumption that the \textsc{Permanent}\ is easy on average. \ Finally, using
the fact that $q\left(  t\right)  :=\operatorname*{Per}\left(  X\left(
t\right)  \right)  $\ is a low-degree polynomial in $t$, we perform polynomial
interpolation on the noisy estimates%
\[
y_{1}\approx q\left(  t_{1}\right)  ,\ldots,y_{m}\approx q\left(
t_{m}\right)  ,
\]
in order to obtain an estimate of the worst-case permanent $q\left(  0\right)
=\operatorname*{Per}\left(  X\left(  0\right)  \right)  =\operatorname*{Per}%
\left(  X\right)  $.

The above approach is a very general one, with different instantiations
depending on the base field $\mathbb{F}$, the fraction of $X$'s for which we
can compute $\operatorname*{Per}\left(  X\right)  $,\ and so forth.
\ Nevertheless, we claim that, assuming the Permanent Anti-Concentration
Conjecture,\textit{ the usual polynomial interpolation approach cannot
possibly work to prove Conjecture \ref{pgc}.} \ Let us see why this is the case.

Let $X\in\mathbb{C}^{n\times n}$\ be a matrix where every entry has absolute
value at most $1$. \ Then certainly it is a $\mathsf{\#P}$-hard\ problem to
approximate $\operatorname*{Per}\left(  X\right)  $\ multiplicatively\ (as
shown by Theorem \ref{approxhard}, for example). \ Our goal is to reduce the
approximation of $\operatorname*{Per}\left(  X\right)  $ to the approximation
of $\operatorname*{Per}\left(  X_{1}\right)  ,\ldots,\operatorname*{Per}%
\left(  X_{m}\right)  $, for some matrices $X_{1},\ldots,X_{m}$\ that are
drawn from the Gaussian distribution $\mathcal{G}^{n\times n}$ or something
close to it.

Recall from Section \ref{DISTPER}\ that%
\[
\operatorname*{E}_{X\sim\mathcal{G}^{n\times n}}\left[  \left\vert
\operatorname*{Per}\left(  X\right)  \right\vert ^{2}\right]  =n!,
\]
which combined with Markov's inequality yields%
\begin{equation}
\Pr_{X\sim\mathcal{G}^{n\times n}}\left[  \left\vert \operatorname*{Per}%
\left(  X\right)  \right\vert >k\sqrt{n!}\right]  <\frac{1}{k^{2}}
\label{mineq}%
\end{equation}
for all $k>1$. \ But this already points to a problem: $\left\vert
\operatorname*{Per}\left(  X\right)  \right\vert $\textit{\ could, in general,
be larger than }$\left\vert \operatorname*{Per}\left(  X_{1}\right)
\right\vert ,\ldots,\left\vert \operatorname*{Per}\left(  X_{m}\right)
\right\vert $\textit{\ by an exponential factor}. \ Specifically, $\left\vert
\operatorname*{Per}\left(  X\right)  \right\vert $\ could be as large as
$n!$\ (for example, if $A$ is the all-$1$'s matrix). \ By contrast,
$\left\vert \operatorname*{Per}\left(  X_{1}\right)  \right\vert
,\ldots,\left\vert \operatorname*{Per}\left(  X_{m}\right)  \right\vert
$\ will typically be $O(\sqrt{n!})$ by equation (\ref{mineq}). \ And yet, from
constant-factor approximations to $\operatorname*{Per}\left(  X_{1}\right)
,\ldots,\operatorname*{Per}\left(  X_{m}\right)  $, we are supposed to recover
a constant-factor approximation to $\operatorname*{Per}\left(  X\right)  $,
\textit{even in the case that }$\left\vert \operatorname*{Per}\left(
X\right)  \right\vert $\textit{\ is much smaller than }$n!$\ (say, $\left\vert
\operatorname*{Per}\left(  X\right)  \right\vert \approx\sqrt{n!}$).

Why is this a problem? \ Because polynomial interpolation is linear with
respect to additive errors. \ And therefore, even modest errors in estimating
$\operatorname*{Per}\left(  X_{1}\right)  ,\ldots,\operatorname*{Per}\left(
X_{m}\right)  $\ could cause a large error in estimating $\operatorname*{Per}%
\left(  X\right)  $.

To see this concretely, let $X$\ be the $n\times n$\ all-$1$'s matrix, and
$X\left(  t\right)  $\ be a randomly-chosen curve through $\mathbb{C}^{n\times
n}$\ that satisfies $X\left(  0\right)  =X$. \ Also, let $t_{1},\ldots
,t_{m}\in\mathbb{R}$\ be nonzero points such that, as we vary $X$, each
$X\left(  t_{i}\right)  $\ is close to a Gaussian random matrix $X\sim
\mathcal{G}^{n\times n}$. \ (We need not assume that the $X\left(
t_{i}\right)  $'s are independent.) \ Finally, let $q_{0}\left(  t\right)
:=\operatorname*{Per}\left(  X\left(  t\right)  \right)  $. \ Then

\begin{enumerate}
\item[(i)] $\left\vert q_{0}\left(  t_{1}\right)  \right\vert ,\ldots
,\left\vert q_{0}\left(  t_{m}\right)  \right\vert $\ are each at most
$n^{O\left(  1\right)  }\sqrt{n!}$ with high probability over the choice of
$X$, but

\item[(ii)] $\left\vert q_{0}\left(  0\right)  \right\vert =\left\vert
\operatorname*{Per}\left(  X\left(  0\right)  \right)  \right\vert =\left\vert
\operatorname*{Per}\left(  X\right)  \right\vert =n!$.
\end{enumerate}

Here (i) holds by our assumption that each $X\left(  t_{i}\right)  $\ is close
to Gaussian, together with equation (\ref{mineq}).

All we need to retain from this is that a polynomial $q_{0}$\ with properties
(i) and (ii) \textit{exists}, within whatever class of polynomials is relevant
for our interpolation problem.

Now, suppose that instead of choosing $X$ to be the all-$1$'s matrix, we had
chosen an$\ X$ such that $\left\vert \operatorname*{Per}\left(  A\right)
\right\vert \leq\sqrt{n!}$. \ Then as before, we could choose a random curve
$X\left(  t\right)  $\ such that $X\left(  0\right)  =X$ and\ $X\left(
t_{1}\right)  ,\ldots,X\left(  t_{m}\right)  $\ are approximately Gaussian,
for some fixed interpolation points $t_{1},\ldots,t_{m}\in\mathbb{R}$. \ Then
letting $q\left(  t\right)  :=\operatorname*{Per}\left(  X\left(  t\right)
\right)  $, we would have

\begin{enumerate}
\item[(i)] $\left\vert q\left(  t_{1}\right)  \right\vert ,\ldots,\left\vert
q\left(  t_{m}\right)  \right\vert $ are each at least $\sqrt{n!}/n^{O\left(
1\right)  }$ with high probability over the choice of $X$, and

\item[(ii)] $\left\vert q\left(  0\right)  \right\vert =\left\vert
\operatorname*{Per}\left(  X\left(  0\right)  \right)  \right\vert =\left\vert
\operatorname*{Per}\left(  X\right)  \right\vert \leq\sqrt{n!}$.
\end{enumerate}

Here (i) holds by our assumption that each $X\left(  t_{i}\right)  $\ is close
to Gaussian, together with Conjecture \ref{pacc}\ (the Permanent
Anti-Concentration Conjecture).

Now define a new polynomial%
\[
\widetilde{q}\left(  t\right)  :=q\left(  t\right)  +\gamma q_{0}\left(
t\right)  ,
\]
where, say, $\left\vert \gamma\right\vert =2^{-n}$. \ Then for all
$i\in\left[  m\right]  $, the difference%
\[
\left\vert \widetilde{q}\left(  t_{i}\right)  -q\left(  t_{i}\right)
\right\vert =\left\vert \gamma q_{0}\left(  t_{i}\right)  \right\vert
\leq\frac{n^{O\left(  1\right)  }}{2^{n}}\sqrt{n!},
\]
is negligible compared to $\sqrt{n!}$. \ This means that \textit{it is
impossible to distinguish the two polynomials }$\widetilde{q}$\textit{\ and
}$q$\textit{, given their approximate values at the points }$t_{1}%
,\ldots,t_{m}$. \ And yet the two polynomials have completely different
behavior at the point $0$: by assumption $\left\vert q\left(  0\right)
\right\vert \leq\sqrt{n!}$, but%
\begin{align*}
\left\vert \widetilde{q}\left(  0\right)  \right\vert  &  \geq\left\vert
\gamma q_{0}\left(  0\right)  \right\vert -\left\vert q\left(  0\right)
\right\vert \\
&  \geq\frac{n!}{2^{n}}-\sqrt{n!}.
\end{align*}
We conclude that it is impossible, given only the approximate values of the
polynomial $q\left(  t\right)  :=\operatorname*{Per}\left(  X\left(  t\right)
\right)  $ at the points $t_{1},\ldots,t_{m}$, to deduce its approximate value
at $0$. \ And therefore, assuming the PACC, the usual polynomial interpolation
approach cannot suffice for proving Conjecture \ref{pgc}.

Nevertheless, we speculate that there \textit{is} a worst-case/average-case
reduction for approximating the permanents of Gaussian random matrices, and
that the barrier we have identified merely represents a limitation of current
techniques. \ So for example, perhaps one can do interpolation using a
\textit{restricted class} of low-degree polynomials, such as polynomials with
an upper bound on their coefficients. \ To evade the barrier, what seems to be
crucial is that the restricted class of polynomials one uses not be closed
under addition.

Of course, the above argument relied on the Permanent Anti-Concentration
Conjecture, so one conceivable way around the barrier would be if the PACC
were false. \ However, in that case, the results of Section \ref{GPE2GPE}
would fail: that is, we would not know how to use the hardness of
\textsc{GPE}$_{\times}$ to deduce the hardness of $\left\vert
\text{\textsc{GPE}}\right\vert _{\pm}^{2}$ that we need for our application.

\section{Open Problems\label{OPEN}}

The most exciting challenge we leave is to \textit{do} the experiments
discussed in Section \ref{EXPER}, whether in linear optics or in other
physical systems that contain excitations that behave as identical bosons.
\ If successful, such experiments have the potential to provide the strongest
evidence to date for violation of the Extended Church-Turing Thesis in nature.

We now list a few theoretical open problems.

\begin{enumerate}
\item[(1)] The most obvious problem is to prove Conjecture \ref{pgc} (the
Permanent-of-Gaussians Conjecture): that approximating the permanent of a
matrix of i.i.d.\ Gaussian entries is $\mathsf{\#P}$-hard. \ Failing that, can
we prove $\mathsf{\#P}$-hardness\ for \textit{any} problem with a similar
\textquotedblleft flavor\textquotedblright\ (roughly speaking, an average-case
approximate counting problem over $\mathbb{R}$\ or $\mathbb{C}$)? \ Can we at
least find evidence that such a problem is not in $\mathsf{BPP}^{\mathsf{NP}}$?

\item[(2)] Another obvious problem is to prove Conjecture \ref{pacc}\ (the
Permanent Anti-Concentration Conjecture), that $\left\vert \operatorname*{Per}%
\left(  X\right)  \right\vert $\ almost always exceeds $\sqrt{n!}%
/\operatorname*{poly}\left(  n\right)  $ for Gaussian random matrices
$X\thicksim\mathcal{N}\left(  0,1\right)  _{\mathbb{C}}^{n\times n}$.
\ Failing that, \textit{any} progress on understanding the distribution of
$\operatorname*{Per}\left(  X\right)  $\ for Gaussian $X$\ would be interesting.

\item[(3)] Can we reduce the number of modes needed for our linear-optics
experiment, perhaps from $O\left(  n^{2}\right)  $\ to $O\left(  n\right)  $?

\item[(4)] How does the noninteracting-boson model relate to other models of
computation that are believed to be intermediate between $\mathsf{BPP}$\ and
$\mathsf{BQP}$? \ To give one concrete question, can every boson computation
be simulated by a qubit-based quantum circuit of logarithmic depth?

\item[(5)] Using quantum fault-tolerance techniques, can one decrease the
effective error in our experiment to $1/\exp\left(  n\right)  $---thereby
obviating the need for the mathematical work we do in this paper to handle
$1/\operatorname*{poly}\left(  n\right)  $\ error in variation distance?
\ Note that, \textit{if} one had the resources for universal quantum
computation, then one could easily combine our experiment with standard
fault-tolerance schemes,\ which are known to push the effective error down to
$1/\exp\left(  n\right)  $\ using $\operatorname*{poly}\left(  n\right)
$\ computational overhead. \ So the interesting question is whether one can
make our experiment fault-tolerant using \textit{fewer} resources than are
needed for universal quantum computing---and in particular, whether one can do
so using linear optics alone.

\item[(6)] Can we give evidence against not merely an FPTAS (Fully Polynomial
Time Approximation Scheme) for the \textsc{BosonSampling} problem, but an
approximate sampling algorithm that works for some \textit{fixed}\ error
$\varepsilon>1/\operatorname*{poly}\left(  n\right)  $?

\item[(7)] For what other interesting quantum systems, besides linear optics,
do analogues of our hardness results hold? \ As mentioned in Section
\ref{RELATED}, the beautiful work of Bremner, Jozsa, and Shepherd \cite{bjs}
shows that exact simulation of \textquotedblleft commuting\ quantum
computations\textquotedblright\ in classical polynomial time would collapse
the polynomial hierarchy. \ What can we say about \textit{approximate}
classical simulation of their model?

\item[(8)] In this work, we showed that unlikely complexity consequences would
follow if classical computers could simulate quantum computers on all
\textit{sampling} or \textit{search} problems: that is, that $\mathsf{SampP}%
=\mathsf{SampBQP}$\ or $\mathsf{FBPP}=\mathsf{FBQP}$. \ An obvious question
that remains is, what about \textit{decision} problems? \ Can we derive some
unlikely collapse of classical complexity classes from the assumption that
$\mathsf{P}=\mathsf{BQP}$\ or $\mathsf{P{}romiseP}=\mathsf{P{}romiseBQP}$?

\item[(9)] Is there any plausible candidate for a \textit{decision} problem
that is efficiently solvable by a boson computer, but not by a classical computer?

\item[(10)] As discussed in Section \ref{EXPER}, it is not obvious how to
convince a skeptic that a quantum computer is really solving the
\textsc{BosonSampling}\ problem in a scalable way. \ This is because, unlike
with (say) \textsc{Factoring}, neither \textsc{BosonSampling}\ nor any related
problem seems to be in $\mathsf{NP}$. \ How much can we do to remedy this?
\ For example, can a prover with a \textsc{BosonSampling} oracle prove any
nontrivial statements to a $\mathsf{BPP}$\ verifier\ via an interactive protocol?

\item[(11)] Is there a polynomial-time classical algorithm to sample from a
probability distribution $\mathcal{D}^{\prime}$\ that \textit{cannot be
efficiently distinguished} from the distribution $\mathcal{D}$ sampled by a
boson computer?
\end{enumerate}

\section{Acknowledgments}

We thank Andy Drucker, Oded Goldreich, Aram Harrow, Matt Hastings, Greg
Kuperberg, Masoud Mohseni, Terry Rudolph, Barry Sanders, Madhu Sudan, Terry
Tao, Barbara Terhal, Lev Vaidman, Leslie Valiant, and Avi Wigderson for
helpful discussions. \ We especially thank Leonid Gurvits for explaining his
polynomial formalism\ and for allowing us to include several of his results in
Appendix \ref{ALGS}, and Mick Bremner and Richard Jozsa for discussions of
their work \cite{bjs}.

\bibliographystyle{plain}
\bibliography{thesis}

\begin{thebibliography}{10}

\bibitem{aar:bf}
S.~Aaronson.
\newblock Algorithms for {B}oolean function query properties.
\newblock {\em SIAM J. Comput.}, 32(5):1140--1157, 2003.

\bibitem{aar:pp}
S.~Aaronson.
\newblock Quantum computing, postselection, and probabilistic polynomial-time.
\newblock {\em Proc. Roy. Soc. London}, A461(2063):3473--3482, 2005.
\newblock quant-ph/0412187.

\bibitem{aar:ph}
S.~Aaronson.
\newblock {BQP} and the polynomial hierarchy.
\newblock In {\em Proc. ACM STOC}, 2010.
\newblock arXiv:0910.4698.

\bibitem{aar:samp}
S.~Aaronson.
\newblock The equivalence of sampling and searching.
\newblock arXiv:1009.5104, ECCC TR10-128, 2010.

\bibitem{ag}
S.~Aaronson and D.~Gottesman.
\newblock Improved simulation of stabilizer circuits.
\newblock {\em Phys. Rev. A}, 70(052328), 2004.
\newblock quant-ph/0406196.

\bibitem{al:fermi}
D.~S. Abrams and S.~Lloyd.
\newblock Simulation of many-body {F}ermi systems on a universal quantum
  computer.
\newblock {\em Phys. Rev. Lett.}, 79:2586--2589, 1997.
\newblock quant-ph/9703054.

\bibitem{ab}
D.~Aharonov and M.~Ben-Or.
\newblock Fault-tolerant quantum computation with constant error.
\newblock In {\em Proc. ACM STOC}, pages 176--188, 1997.
\newblock quant-ph/9906129.

\bibitem{bartlettsanders}
S.~D. Bartlett and B.~C. Sanders.
\newblock Requirement for quantum computation.
\newblock {\em Journal of Modern Optics}, 50:2331--2340, 2003.
\newblock quant-ph/0302125.

\bibitem{bv}
E.~Bernstein and U.~Vazirani.
\newblock Quantum complexity theory.
\newblock {\em SIAM J. Comput.}, 26(5):1411--1473, 1997.
\newblock First appeared in ACM STOC 1993.

\bibitem{bjs}
M.~Bremner, R.~Jozsa, and D.~Shepherd.
\newblock Classical simulation of commuting quantum computations implies
  collapse of the polynomial hierarchy.
\newblock {\em Proc. Roy. Soc. London}, 2010.
\newblock To appear. arXiv:1005.1407.

\bibitem{bcwz}
H.~Buhrman, R.~Cleve, R.~de Wolf, and Ch. Zalka.
\newblock Bounds for small-error and zero-error quantum algorithms.
\newblock In {\em Proc. IEEE FOCS}, pages 358--368, 1999.
\newblock cs.CC/9904019.

\bibitem{cps}
J.-Y. Cai, A.~Pavan, and D.~Sivakumar.
\newblock On the hardness of permanent.
\newblock In {\em Proc. Intl. Symp. on Theoretical Aspects of Computer Science
  (STACS)}, pages 90--99, 1999.

\bibitem{caianiello}
E.~R. Caianiello.
\newblock On quantum field theory, 1: explicit solution of {D}yson's equation
  in electrodynamics without use of {F}eynman graphs.
\newblock {\em Nuovo Cimento}, 10:1634--1652, 1953.

\bibitem{ceperley}
D.~M. Ceperley.
\newblock An overview of quantum {M}onte {C}arlo methods.
\newblock {\em Reviews in Mineralogy and Geochemistry}, 71(1):129--135, 2010.

\bibitem{cw}
R.~Cleve and J.~Watrous.
\newblock Fast parallel circuits for the quantum {F}ourier transform.
\newblock In {\em Proc. IEEE FOCS}, pages 526--536, 2000.
\newblock quant-ph/0006004.

\bibitem{cv}
K.~P. Costello and V.~H. Vu.
\newblock Concentration of random determinants and permanent estimators.
\newblock {\em SIAM J. Discrete Math}, 23(3).

\bibitem{dgp}
C.~Daskalakis, P.~W. Goldberg, and C.~H. Papadimitriou.
\newblock The complexity of computing a {N}ash equilibrium.
\newblock {\em Commun. ACM}, 52(2):89--97, 2009.
\newblock Earlier version in Proceedings of STOC'2006.

\bibitem{egorychev}
G.~P. Egorychev.
\newblock Proof of the van der {W}aerden conjecture for permanents.
\newblock {\em Sibirsk. Mat. Zh.}, 22(6):65--71, 1981.
\newblock English translation in Siberian Math. J. 22, pp. 854-859, 1981.

\bibitem{falikman}
D.~I. Falikman.
\newblock Proof of the van der {W}aerden conjecture regarding the permanent of
  a doubly stochastic matrix.
\newblock {\em Mat. Zametki}, 29:931--938, 1981.
\newblock English translation in Math. Notes 29, pp. 475-479, 1981.

\bibitem{feffumans}
B.~Fefferman and C.~Umans.
\newblock Pseudorandom generators and the {BQP} vs. {PH} problem.
\newblock http://www.cs.caltech.edu/\symbol{126}umans/papers/FU10.pdf, 2010.

\bibitem{fghp}
S.~Fenner, F.~Green, S.~Homer, and R.~Pruim.
\newblock Determining acceptance possibility for a quantum computation is hard
  for the polynomial hierarchy.
\newblock {\em Proc. Roy. Soc. London}, A455:3953--3966, 1999.
\newblock quant-ph/9812056.

\bibitem{feynman:qc}
R.~P. Feynman.
\newblock Simulating physics with computers.
\newblock {\em Int. J. Theoretical Physics}, 21(6-7):467--488, 1982.

\bibitem{glrsw}
P.~Gemmell, R.~Lipton, R.~Rubinfeld, M.~Sudan, and A.~Wigderson.
\newblock Self-testing/correcting for polynomials and for approximate
  functions.
\newblock In {\em Proc. ACM STOC}, pages 32--42, 1991.

\bibitem{gemmellsudan}
P.~Gemmell and M.~Sudan.
\newblock Highly resilient correctors for polynomials.
\newblock {\em Inform. Proc. Lett.}, 43:169--174, 1992.

\bibitem{girko}
V.~L. Girko.
\newblock A refinement of the {C}entral {L}imit {T}heorem for random
  determinants.
\newblock {\em Teor. Veroyatnost. i Primenen}, 42:63--73, 1997.
\newblock Translation in Theory Probab. Appl 42 (1998), 121-129.

\bibitem{godsilgutman}
C.~D. Godsil and I.~Gutman.
\newblock On the matching polynomial of a graph.
\newblock In {\em Algebraic Methods in Graph Theory I-II}, pages 67--83. North
  Holland, 1981.

\bibitem{gurvits:alg}
L.~Gurvits.
\newblock On the complexity of mixed discriminants and related problems.
\newblock In {\em Mathematical Foundations of Computer Science}, pages
  447--458, 2005.

\bibitem{hht}
Y.~Han, L.~Hemaspaandra, and T.~Thierauf.
\newblock Threshold computation and cryptographic security.
\newblock {\em SIAM J. Comput.}, 26(1):59--78, 1997.

\bibitem{hom}
C.~K. Hong, Z.~Y. Ou, and L.~Mandel.
\newblock Measurement of subpicosecond time intervals between two photons by
  interference.
\newblock {\em Phys. Rev. Lett.}, 59(18):2044--2046, 1987.

\bibitem{jsv}
M.~Jerrum, A.~Sinclair, and E.~Vigoda.
\newblock A polynomial-time approximation algorithm for the permanent of a
  matrix with non-negative entries.
\newblock {\em J. ACM}, 51(4):671--697, 2004.
\newblock Earlier version in STOC'2001.

\bibitem{jordan}
S.~P. Jordan.
\newblock Permutational quantum computing.
\newblock {\em Quantum Information and Computation}, 10(5/6):470--497, 2010.
\newblock arXiv:0906.2508.

\bibitem{khot}
S.~Khot.
\newblock On the {U}nique {G}ames {C}onjecture.
\newblock In {\em Proc. IEEE Conference on Computational Complexity}, pages
  99--121, 2010.

\bibitem{knill:matchgate}
E.~Knill.
\newblock Fermionic linear optics and matchgates.
\newblock quant-ph/0108033, 2001.

\bibitem{knilaflamme}
E.~Knill and R.~Laflamme.
\newblock Power of one bit of quantum information.
\newblock {\em Phys. Rev. Lett.}, 81(25):5672--5675, 1998.
\newblock quant-ph/9802037.

\bibitem{klm}
E.~Knill, R.~Laflamme, and G.~J. Milburn.
\newblock A scheme for efficient quantum computation with linear optics.
\newblock {\em Nature}, 409:46--52, 2001.
\newblock See also quant-ph/0006088.

\bibitem{klz}
E.~Knill, R.~Laflamme, and W.~Zurek.
\newblock Resilient quantum computation.
\newblock {\em Science}, 279:342--345, 1998.
\newblock quant-ph/9702058.

\bibitem{lange}
K.~Lange.
\newblock {\em Applied Probability}.
\newblock Springer, 2003.

\bibitem{limbeige}
Y.~L. Lim and A.~Beige.
\newblock Generalized {H}ong-{O}u-{M}andel experiments with bosons and
  fermions.
\newblock {\em New J. Phys.}, 7(155), 2005.
\newblock quant-ph/0505034.

\bibitem{lipton}
R.~J. Lipton.
\newblock New directions in testing.
\newblock In {\em Distributed Computing and Cryptography}, pages 191--202. AMS,
  1991.

\bibitem{lounis}
B.~Lounis and M.~Orrit.
\newblock Single-photon sources.
\newblock {\em Reports on Progress in Physics}, 68(5), 2005.

\bibitem{mastrodon}
C.~Mastrodonato and R.~Tumulka.
\newblock Elementary proof for asymptotics of large {H}aar-distributed unitary
  matrices.
\newblock {\em Letters in Mathematical Physics}, 82(1):51--59, 2007.
\newblock arXiv:0705.3146.

\bibitem{nc}
M.~Nielsen and I.~Chuang.
\newblock {\em Quantum Computation and Quantum Information}.
\newblock Cambridge University Press, 2000.

\bibitem{paturi}
R.~Paturi.
\newblock On the degree of polynomials that approximate symmetric {B}oolean
  functions.
\newblock In {\em Proc. ACM STOC}, pages 468--474, 1992.

\bibitem{petzreffy}
D.~Petz and J.~R\'{e}ffy.
\newblock On asymptotics of large {H}aar distributed unitary matrices.
\newblock {\em Periodica Mathematica Hungarica}, 49(1):103--117, 2004.
\newblock arXiv:math/0310338.

\bibitem{petzreffy2}
D.~Petz and J.~R\'{e}ffy.
\newblock Large deviation theorem for empirical eigenvalue distribution of
  truncated {H}aar unitary matrices.
\newblock {\em Prob. Theory and Related Fields}, 133(2):175--189, 2005.
\newblock arXiv:math/0409552.

\bibitem{rzbb}
M.~Reck, A.~Zeilinger, H.~J. Bernstein, and P.~Bertani.
\newblock Experimental realization of any discrete unitary operator.
\newblock {\em Phys. Rev. Lett.}, 73(1):58--61, 1994.

\bibitem{reffy}
J.~R\'{e}ffy.
\newblock {\em Asymptotics of random unitaries}.
\newblock PhD thesis, Budapest University of Technology and Economics, 2005.
\newblock http://www.math.bme.hu/\symbol{126}reffyj/disszer.pdf.

\bibitem{rudolph}
T.~Rudolph.
\newblock A simple encoding of a quantum circuit amplitude as a matrix
  permanent.
\newblock arXiv:0909.3005, 2009.

\bibitem{scheel}
S.~Scheel.
\newblock Permanents in linear optical networks.
\newblock quant-ph/0406127, 2004.

\bibitem{shepherd}
D.~Shepherd and M.~J. Bremner.
\newblock Temporally unstructured quantum computation.
\newblock {\em Proc. Roy. Soc. London}, A465(2105):1413--1439, 2009.
\newblock arXiv:0809.0847.

\bibitem{shi:gate}
Y.~Shi.
\newblock Both {T}offoli and controlled-{NOT} need little help to do universal
  quantum computation.
\newblock {\em Quantum Information and Computation}, 3(1):84--92, 2002.
\newblock quant-ph/0205115.

\bibitem{shor}
P.~W. Shor.
\newblock Polynomial-time algorithms for prime factorization and discrete
  logarithms on a quantum computer.
\newblock {\em SIAM J. Comput.}, 26(5):1484--1509, 1997.
\newblock Earlier version in IEEE FOCS 1994. quant-ph/9508027.

\bibitem{simon}
D.~Simon.
\newblock On the power of quantum computation.
\newblock In {\em Proc. IEEE FOCS}, pages 116--123, 1994.

\bibitem{hastad:book}
J.~H\aa stad.
\newblock {\em Computational Limitations for Small Depth Circuits}.
\newblock MIT Press, 1987.

\bibitem{stockmeyer}
L.~J. Stockmeyer.
\newblock The complexity of approximate counting.
\newblock In {\em Proc. ACM STOC}, pages 118--126, 1983.

\bibitem{sudan}
M.~Sudan.
\newblock Maximum likelihood decoding of {R}eed-{S}olomon codes.
\newblock In {\em Proc. IEEE FOCS}, pages 164--172, 1996.

\bibitem{taovu:perm}
T.~Tao and V.~Vu.
\newblock On the permanent of random {B}ernoulli matrices.
\newblock {\em Advances in Mathematics}, 220(3):657--669, 2009.
\newblock arXiv:0804.2362.

\bibitem{td:fermion}
B.~M. Terhal and D.~P. DiVincenzo.
\newblock Classical simulation of noninteracting-fermion quantum circuits.
\newblock {\em Phys. Rev. A}, 65(032325), 2002.
\newblock quant-ph/0108010.

\bibitem{td}
B.~M. Terhal and D.~P. DiVincenzo.
\newblock Adaptive quantum computation, constant-depth circuits and
  {A}rthur-{M}erlin games.
\newblock {\em Quantum Information and Computation}, 4(2):134--145, 2004.
\newblock quant-ph/0205133.

\bibitem{toda}
S.~Toda.
\newblock {PP} is as hard as the polynomial-time hierarchy.
\newblock {\em SIAM J. Comput.}, 20(5):865--877, 1991.

\bibitem{troyanskytishby}
L.~Troyansky and N.~Tishby.
\newblock Permanent uncertainty: On the quantum evaluation of the determinant
  and the permanent of a matrix.
\newblock In {\em Proceedings of PhysComp}, 1996.

\bibitem{valiant}
L.~G. Valiant.
\newblock The complexity of computing the permanent.
\newblock {\em Theoretical Comput. Sci.}, 8(2):189--201, 1979.

\bibitem{valiant:qc}
L.~G. Valiant.
\newblock Quantum circuits that can be simulated classically in polynomial
  time.
\newblock {\em SIAM J. Comput.}, 31(4):1229--1254, 2002.
\newblock Earlier version in STOC'2001.

\bibitem{vsbysc}
L.~Vandersypen, M.~Steffen, G.~Breyta, C.~S. Yannoni, M.~H. Sherwood, and I.~L.
  Chuang.
\newblock Experimental realization of {S}hor's quantum factoring algorithm
  using nuclear magnetic resonance.
\newblock {\em Nature}, 414:883--887, 2001.
\newblock quant-ph/0112176.

\end{thebibliography}

\section{Appendix: Positive Results for Simulation of Linear
Optics\label{ALGS}}

In this appendix, we present two results of Gurvits, both of which give
surprising classical polynomial-time algorithms for computing certain
properties of linear-optical networks. \ The first result, which appeared in
\cite{gurvits:alg},\ gives an efficient randomized algorithm to approximate
the permanent of a (sub)unitary matrix with $\pm1/\operatorname*{poly}\left(
n\right)  $ additive error, and as a consequence, to estimate final amplitudes
such as $\left\langle 1_{n}\right\vert \varphi\left(  U\right)  \left\vert
1_{n}\right\rangle =\operatorname*{Per}\left(  U_{n,n}\right)  $ with
$\pm1/\operatorname*{poly}\left(  n\right)  $ additive error, given any
linear-optical network $U$. \ This ability is of limited use in practice,
since $\left\langle 1_{n}\right\vert \varphi\left(  U\right)  \left\vert
1_{n}\right\rangle $\ will be exponentially small for most choices of $U$ (in
which case, $0$ is also a good additive estimate!). \ On the other hand, we
certainly do not know how to do anything similar for general, qubit-based
quantum circuits---indeed, if we could, then $\mathsf{BQP}$\ would equal
$\mathsf{BPP}$.

Gurvits's second result (unpublished) gives a way to compute the marginal
distribution over photon numbers for any $k$ modes, deterministically and in
$n^{O\left(  k\right)  }$\ time. \ Again, this is perfectly consistent with
our hardness conjectures, since if one wanted to \textit{sample} from the
distribution over photon numbers (or compute a final probability such as
$\left\vert \left\langle 1_{n}\right\vert \varphi\left(  U\right)  \left\vert
1_{n}\right\rangle \right\vert ^{2}$), one would need to take $k\geq n$.

To prove Gurvits's first result, our starting point will be the following
identity of Ryser, which is also used for computing the permanent of an
$n\times n$\ matrix\ in $O\left(  2^{n}n^{2}\right)  $\ time.

\begin{lemma}
[Ryser's Formula]\label{ryser}For all $V\in\mathbb{C}^{n\times n}$,%
\[
\operatorname*{Per}\left(  V\right)  =\operatorname*{E}_{x_{1},\ldots,x_{n}%
\in\left\{  -1,1\right\}  }\left[  x_{1}\cdots x_{n}%
%TCIMACRO{\dprod \limits_{i=1}^{n}}%
%BeginExpansion
{\displaystyle\prod\limits_{i=1}^{n}}
%EndExpansion
\left(  v_{i1}x_{1}+\cdots+v_{in}x_{n}\right)  \right]  .
\]

\end{lemma}

\begin{proof}
Let $p\left(  x_{1},\ldots,x_{n}\right)  $\ be the degree-$n$ polynomial that
corresponds to the product in the above expectation. \ Then the only monomial
of $p$ that can contribute to the expectation is $x_{1}\cdots x_{n}$, since
all the other monomials will be cancelled out by the multiplier of
$x_{1}\cdots x_{n}$\ (which is equally likely to be $1$\ or $-1$).
\ Furthermore, as in Lemma \ref{perlem1}, the coefficient of $x_{1}\cdots
x_{n}$\ is just%
\[
\sum_{\sigma\in S_{n}}\prod_{i=1}^{n}v_{i,\sigma\left(  i\right)
}=\operatorname*{Per}\left(  V\right)  .
\]
Therefore the expectation equals%
\[
\operatorname*{Per}\left(  V\right)  \operatorname*{E}_{x_{1},\ldots,x_{n}%
\in\left\{  -1,1\right\}  }\left[  x_{1}^{2}\cdots x_{n}^{2}\right]
=\operatorname*{Per}\left(  V\right)  .
\]
(Indeed, all we needed about the random variables $x_{1},\ldots,x_{n}$\ was
that they were independent and had mean $0$ and variance $1$.)
\end{proof}

Given $x=\left(  x_{1},\ldots,x_{n}\right)  \in\left\{  -1,1\right\}  ^{n}$,
let%
\[
\operatorname*{Rys}\nolimits_{x}\left(  V\right)  :=x_{1}\cdots x_{n}%
%TCIMACRO{\dprod \limits_{i=1}^{n}}%
%BeginExpansion
{\displaystyle\prod\limits_{i=1}^{n}}
%EndExpansion
\left(  v_{i1}x_{1}+\cdots+v_{in}x_{n}\right)  .
\]
Then Lemma \ref{ryser}\ says that $\operatorname*{Rys}\nolimits_{x}\left(
V\right)  $\ is an \textit{unbiased estimator} for the permanent, in the sense
that $\operatorname*{E}_{x}\left[  \operatorname*{Rys}\nolimits_{x}\left(
V\right)  \right]  =\operatorname*{Per}\left(  V\right)  $. \ Gurvits
\cite{gurvits:alg}\ observed the following key further fact about
$\operatorname*{Rys}\nolimits_{x}\left(  V\right)  $.

\begin{lemma}
\label{gurvitslem}$\left\vert \operatorname*{Rys}\nolimits_{x}\left(
V\right)  \right\vert \leq\left\Vert V\right\Vert ^{n}$ for all $x\in\left\{
-1,1\right\}  ^{n}$\ and all $V$.
\end{lemma}

\begin{proof}
Given a vector $x=\left(  x_{1},\ldots,x_{n}\right)  $ all of whose entries
are $1$\ or $-1$, let $y=Vx$, and let%
\[
y_{i}:=v_{i1}x_{1}+\cdots+v_{in}x_{n}%
\]
be the $i^{th}$\ component of $y$. \ Then $\left\Vert x\right\Vert =\sqrt{n}$,
so $\left\Vert y\right\Vert \leq\left\Vert V\right\Vert \left\Vert
x\right\Vert =\left\Vert V\right\Vert \sqrt{n}$. \ Hence%
\begin{align*}
\left\vert \operatorname*{Rys}\nolimits_{x}\left(  V\right)  \right\vert  &
=\left\vert x_{1}\cdots x_{n}y_{1}\cdots y_{n}\right\vert \\
&  =\left\vert y_{1}\cdots y_{n}\right\vert \\
&  \leq\left(  \frac{\left\vert y_{1}\right\vert +\cdots+\left\vert
y_{n}\right\vert }{n}\right)  ^{n}\\
&  \leq\left(  \frac{\left\Vert y\right\Vert }{\sqrt{n}}\right)  ^{n}\\
&  \leq\left\Vert V\right\Vert ^{n},
\end{align*}
where the third line follows from the arithmetic-geometric mean inequality,
and the fourth line follows from Cauchy-Schwarz.
\end{proof}

An immediate consequence of Lemma \ref{gurvitslem}\ is the following:

\begin{corollary}
\label{percor2}$\left\vert \operatorname*{Per}\left(  V\right)  \right\vert
\leq\left\Vert V\right\Vert ^{n}$ for all $V$.
\end{corollary}

Another consequence is a fast additive approximation algorithm for
$\operatorname*{Per}\left(  V\right)  $, which works whenever $\left\Vert
V\right\Vert $\ is small.

\begin{theorem}
[Gurvits's Permanent Approximation Algorithm \cite{gurvits:alg}]%
\label{gurvitsthm}There exists a randomized (classical) algorithm that takes a
matrix $V\in\mathbb{C}^{n\times n}$ as input, runs in $O\left(  n^{2}%
/\varepsilon^{2}\right)  $\ time, and with high probability, approximates
$\operatorname*{Per}\left(  V\right)  $\ to within an additive error
$\pm\varepsilon\left\Vert V\right\Vert ^{n}$.
\end{theorem}

\begin{proof}
By Lemma \ref{ryser},%
\[
\operatorname*{Per}\left(  V\right)  =\operatorname*{E}_{x\in\left\{
-1,1\right\}  ^{n}}\left[  \operatorname*{Rys}\nolimits_{x}\left(  V\right)
\right]  .
\]
Furthermore, we know from Lemma \ref{gurvitslem} that $\left\vert
\operatorname*{Rys}\nolimits_{x}\left(  V\right)  \right\vert \leq\left\Vert
V\right\Vert ^{n}$ for every $x$. \ So our approximation algorithm is simply
the following: for $T=O\left(  1/\varepsilon^{2}\right)  $, first choose
$T$\ vectors $x\left(  1\right)  ,\ldots,x\left(  T\right)  $\ uniformly at
random from $\left\{  -1,1\right\}  ^{n}$. \ Then output the empirical mean%
\[
\widetilde{p}:=\frac{1}{T}\sum_{t=1}^{T}\operatorname*{Rys}\nolimits_{x\left(
t\right)  }\left(  V\right)
\]
as our estimate of $\operatorname*{Per}\left(  V\right)  $. \ Since
$\operatorname*{Rys}\nolimits_{x}\left(  V\right)  $\ can be computed in
$O\left(  n^{2}\right)  $ time, this algorithm takes $O\left(  n^{2}%
/\varepsilon^{2}\right)  $ time. \ The failure probability,%
\[
\Pr_{x\left(  1\right)  ,\ldots,x\left(  T\right)  }\left[  \left\vert
\widetilde{p}-\operatorname*{Per}\left(  V\right)  \right\vert >\varepsilon
\left\Vert V\right\Vert ^{n}\right]  ,
\]
can be upper-bounded using a standard Chernoff bound.
\end{proof}

In particular, Theorem \ref{gurvitsthm} implies that, given an $n\times n$
unitary matrix $U$, one can approximate $\operatorname*{Per}\left(  U\right)
$\ to within an additive error $\pm\varepsilon$\ (with high probability) in
$\operatorname*{poly}\left(  n,1/\varepsilon\right)  $\ time.

We now sketch a proof of Gurvits's second result, giving an $n^{O\left(
k\right)  }$-time algorithm to compute the marginal distribution over any $k$
photon modes. \ We will assume the following lemma, whose proof will appear in
a forthcoming paper of Gurvits.

\begin{lemma}
[Gurvits]\label{lowranklem}Let $V\in\mathbb{C}^{n\times n}$\ be a matrix of
rank $k$. \ Then $\operatorname*{Per}\left(  V+I\right)  $\ can be computed
exactly in $n^{O\left(  k\right)  }$\ time.
\end{lemma}

We now show how to apply Lemma \ref{lowranklem} to the setting of linear optics.

\begin{theorem}
[Gurvits's $k$-Photon Marginal Algorithm]\label{marginalthm}There exists a
deterministic classical algorithm that, given a unitary matrix $U\in
\mathbb{C}^{m\times m}$, indices $i_{1},\ldots,i_{k}\in\left[  m\right]  $,
and occupation numbers $j_{1},\ldots,j_{k}\in\left\{  0,\ldots,n\right\}  $,
computes the joint probability%
\[
\Pr_{S=\left(  s_{1},\ldots,s_{m}\right)  \sim\mathcal{D}_{U}}\left[
s_{i_{1}}=j_{1}\wedge\cdots\wedge s_{i_{k}}=j_{k}\right]
\]
in $n^{O\left(  k\right)  }$ time.
\end{theorem}

\begin{proof}
By symmetry, we can assume without loss of generality that $\left(
i_{1},\ldots,i_{k}\right)  =\left(  1,\ldots,k\right)  $. \ Let $c=\left(
c_{1},\ldots,c_{k}\right)  $ be an arbitrary vector in $\mathbb{C}^{k}$.
\ Then the crucial claim is that we can compute the expectation%
\[
\operatorname*{E}_{S\sim\mathcal{D}_{U}}\left[  \left\vert c_{1}\right\vert
^{2s_{1}}\cdots\left\vert c_{k}\right\vert ^{2s_{k}}\right]  =\sum
_{s_{1},\ldots,s_{k}}\Pr\left[  s_{1},\ldots,s_{k}\right]  \left\vert
c_{1}\right\vert ^{2s_{1}}\cdots\left\vert c_{k}\right\vert ^{2s_{k}}%
\]
in $n^{O\left(  k\right)  }$\ time. \ Given this claim, the theorem follows
easily. \ We simply need to choose $\left(  n+1\right)  ^{k}$\ values for
$\left\vert c_{1}\right\vert ,\ldots,\left\vert c_{k}\right\vert $, compute
$\operatorname*{E}_{S\sim\mathcal{D}_{U}}\left[  \left\vert c_{1}\right\vert
^{2s_{1}}\cdots\left\vert c_{k}\right\vert ^{2s_{k}}\right]  $\ for each one,
and then solve the resulting system of $\left(  n+1\right)  ^{k}$\ independent
linear equations in $\left(  n+1\right)  ^{k}$\ unknowns to obtain the
probabilities $\Pr\left[  s_{1},\ldots,s_{k}\right]  $\ themselves.

We now prove the claim. \ Let $I_{c}:\mathbb{C}^{m}\rightarrow\mathbb{C}^{m}%
$\ be the diagonal linear transformation that maps the vector $\left(
x_{1},\ldots,x_{m}\right)  $ to $\left(  c_{1}x_{1},\ldots,c_{k}x_{k}%
,x_{k+1},\ldots,x_{m}\right)  $, and let $I_{\left\vert c\right\vert ^{2}%
}=I_{c}^{\dag}I_{c}$ be the linear transformation that maps $\left(
x_{1},\ldots,x_{m}\right)  $ to $\left(  \left\vert c_{1}\right\vert ^{2}%
x_{1},\ldots,\left\vert c_{k}\right\vert ^{2}x_{k},x_{k+1},\ldots
,x_{m}\right)  $. \ Also, let%
\[
U\left[  J_{m,n}\right]  \left(  x\right)  =\sum_{S\in\Phi_{m,n}}a_{S}x^{S}.
\]
Now define a polynomial $q$\ by%
\[
q\left(  x\right)  :=I_{c}U\left[  J_{m,n}\right]  \left(  x\right)  ,
\]
and note that%
\[
q\left(  x\right)  =\sum_{S\in\Phi_{m,n}}a_{S}x^{S}c_{1}^{s_{1}}\cdots
c_{k}^{s_{k}}.
\]
Hence%
\begin{align*}
\operatorname*{E}_{S=\left(  s_{1},\ldots,s_{m}\right)  \sim\mathcal{D}_{U}%
}\left[  \left\vert c_{1}\right\vert ^{2s_{1}}\cdots\left\vert c_{k}%
\right\vert ^{2s_{k}}\right]   &  =\sum_{S=\left(  s_{1},\ldots,s_{m}\right)
\in\Phi_{m,n}}\left(  \left\vert a_{S}\right\vert ^{2}s_{1}!\cdots
s_{m}!\right)  \left\vert c_{1}\right\vert ^{2s_{1}}\cdots\left\vert
c_{k}\right\vert ^{2s_{k}}\\
&  =\sum_{S=\left(  s_{1},\ldots,s_{m}\right)  \in\Phi_{m,n}}\left(
\overline{a}_{S}\overline{c}_{1}^{s_{1}}\cdots\overline{c}_{k}^{s_{k}}\right)
\left(  a_{S}c_{1}^{s_{1}}\cdots c_{k}^{s_{k}}\right)  s_{1}!\cdots s_{m}!\\
&  =\left\langle q,q\right\rangle .
\end{align*}
Now,%
\begin{align*}
\left\langle q,q\right\rangle  &  =\left\langle I_{c}U\left[  J_{m,n}\right]
,I_{c}U\left[  J_{m,n}\right]  \right\rangle \\
&  =\left\langle U\left[  J_{m,n}\right]  ,I_{\left\vert c\right\vert ^{2}%
}U\left[  J_{m,n}\right]  \right\rangle \\
&  =\left\langle J_{m,n},U^{\dag}I_{\left\vert c\right\vert ^{2}}U\left[
J_{m,n}\right]  \right\rangle \\
&  =\operatorname*{Per}\left(  \left(  U^{\dag}I_{\left\vert c\right\vert
^{2}}U\right)  _{n,n}\right)
\end{align*}
where the second and third lines follow from Theorem \ref{invarthm2}, and the
fourth line follows from Lemma \ref{perlem1}. \ Finally, let $\Lambda
:=I_{\left\vert c\right\vert ^{2}}-I$. \ Then $\Lambda$\ is a diagonal matrix
of rank at most $k$, and%
\begin{align*}
\left(  U^{\dag}I_{\left\vert c\right\vert ^{2}}U\right)  _{n,n}  &  =\left(
U^{\dag}\left(  \Lambda+I\right)  U\right)  _{n,n}\\
&  =\left(  U^{\dag}\Lambda U+I\right)  _{n,n}\\
&  =V+I,
\end{align*}
where $V:=\left(  U^{\dag}\Lambda U\right)  _{n,n}$ is an $n\times n$\ matrix
of rank at most $k$. \ So by Lemma \ref{lowranklem}, we can compute%
\[
\operatorname*{Per}\left(  V+I\right)  =\operatorname*{E}_{S=\left(
s_{1},\ldots,s_{m}\right)  \sim\mathcal{D}_{U}}\left[  \left\vert
c_{1}\right\vert ^{2s_{1}}\cdots\left\vert c_{k}\right\vert ^{2s_{k}}\right]
\]
in $n^{O\left(  k\right)  }$\ time. \ Furthermore, notice that we can compute
$V$ itself in $O\left(  n^{2}k\right)  =n^{O\left(  1\right)  }$\ time,
independent of $m$. \ Therefore the total time needed to compute the
expectation is $n^{O\left(  k\right)  +O\left(  1\right)  }=n^{O\left(
k\right)  }$. \ This proves the claim.
\end{proof}

\section{Appendix: The Bosonic Birthday Paradox\label{BIRTHDAY}}

By the \textit{birthday paradox}, we mean the statement that, if $n$ balls are
thrown uniformly and independently into $m$ bins, then with high probability
we will see a collision (i.e., two or more balls in the same bin) if
$m=O\left(  n^{2}\right)  $, but \textit{not otherwise}.

In this appendix, we prove the useful fact that the birthday paradox still
holds if the balls are identical bosons, and \textquotedblleft
throwing\textquotedblright\ the balls means applying a Haar-random unitary
matrix. \ More precisely, suppose there are $m$ modes, of which the first $n$
initially contain $n$ identical photons (with one photon in each mode) and the
remaining $m-n$\ are unoccupied. \ Suppose we mix the modes by applying an
$m\times m$\ unitary matrix\ $U$ chosen uniformly at random from the Haar
measure. \ Then if we measure the occupation number of each mode, we will
observe a collision (i.e., two or more photons in the same mode) with
probability bounded away from $0$ if $m=O\left(  n^{2}\right)  $\ but not otherwise.

It is well-known that identical bosons are \textquotedblleft
gregarious,\textquotedblright\ in the sense of being \textit{more} likely than
classical particles to occur in the same state. \ For example, if we throw two
balls uniformly and independently into two bins, then the probability of both
balls landing in the same bin is only $1/2$\ with classical balls, but
$2/3$\ if the balls are identical bosons.\footnote{This is in stark contrast
to the situation with identical fermions, no two of which ever occur in the
same state by the Pauli exclusion principle.} \ So the interesting part of the
bosonic birthday paradox is the \textquotedblleft converse
direction\textquotedblright: when $m\gg n^{2}$, the probability of two or more
bosons landing in the same mode is \textit{not} too large. \ In other words,
while bosons are \textquotedblleft somewhat\textquotedblright\ more gregarious
than classical particles, they are not \textit{so} gregarious as to require a
different asymptotic relation between $m$ and $n$.

The proof of our main result, Theorem \ref{mainresult}, implicitly used this
fact: we needed that when $m\gg n^{2}$, the basis states with two or more
photons in the same mode can safely be neglected. \ However, while in
principle one could extract a proof of the bosonic birthday paradox from the
proof of Theorem \ref{mainresult}, we thought it would be illuminating to
prove the bosonic birthday paradox directly.

The core of the proof is the following simple lemma about the transition
probabilities induced by unitary matrices.

\begin{lemma}
[Unitary Pigeonhole Principle]\label{goodbadlem}Partition a finite set
$\left[  M\right]  $\ into a \textquotedblleft good part\textquotedblright%
\ $G$\ and \textquotedblleft bad part\textquotedblright\ $B=\left[  M\right]
\setminus G$. \ Also, let $U=\left(  u_{xy}\right)  $ be any $M\times
M$\ unitary matrix. \ Suppose we choose an element $x\in G$\ uniformly at
random, apply $U$ to $\left\vert x\right\rangle $, then measure $U\left\vert
x\right\rangle $\ in the standard basis. \ Then letting $y$\ be the
measurement outcome, we have $\Pr\left[  y\in B\right]  \leq\left\vert
B\right\vert /\left\vert G\right\vert $.
\end{lemma}

\begin{proof}
Let\ $R$\ be an $M\times M$\ doubly-stochastic matrix whose $\left(
x,y\right)  $\ entry is $r_{xy}:=\left\vert u_{xy}\right\vert ^{2}$. \ Then
applying $U$ to a computational basis state $\left\vert x\right\rangle $\ and
measuring immediately afterward is the same as applying $R$; in particular,
$\Pr\left[  y\in B\right]  =r_{xy}$. \ Moreover,%
\begin{align*}
\sum_{x,y\in G}r_{xy}  &  =\sum_{x\in G,y\in\left[  M\right]  }r_{xy}%
+\sum_{x\in\left[  M\right]  ,y\in G}r_{xy}-\sum_{x,y\in\left[  M\right]
}r_{xy}+\sum_{x,y\in B}r_{xy}\\
&  =\left\vert G\right\vert +\left\vert G\right\vert -M+\sum_{x,y\in B}%
r_{xy}\\
&  \geq2\left\vert G\right\vert -M,
\end{align*}
where the first line follows from simple rearrangements and the second line
follows from the double-stochasticity of $R$. \ Hence%
\[
\Pr\left[  y\in G\right]  =\operatorname*{E}_{x\in G}\left[  \sum_{y\in
G}r_{xy}\right]  \geq\frac{2\left\vert G\right\vert -M}{\left\vert
G\right\vert }=1-\frac{\left\vert B\right\vert }{\left\vert G\right\vert },
\]
and%
\[
\Pr\left[  y\in B\right]  =1-\Pr\left[  y\in G\right]  \leq\frac{\left\vert
B\right\vert }{\left\vert G\right\vert }.
\]

\end{proof}

Lemma \ref{goodbadlem}\ has the following important corollary. \ Suppose we
draw the $M\times M$\ unitary matrix $U$\ from a probability distribution
$\mathcal{Z}$, where $\mathcal{Z}$\ is \textit{symmetric} with respect to some
transitive group of permutations on the good set $G$. \ Then $\Pr\left[  y\in
B\right]  $\ is clearly independent of the choice of initial state $x\in G$.
\ And therefore, in the statement of the lemma, we might as well \textit{fix}
$x\in G$ rather than choosing it randomly. \ The statement then becomes:

\begin{corollary}
\label{goodbadcor}Partition a finite set $\left[  M\right]  $\ into a
\textquotedblleft good part\textquotedblright\ $G$\ and \textquotedblleft bad
part\textquotedblright\ $B=\left[  M\right]  \setminus G$.\ \ Also, let
$\Gamma\leq S_{M}$\ be a permutation group that is transitive with respect to
$G$, and let $\mathcal{Z}$ be a probability distribution over $M\times
M$\ unitary matrices that is symmetric with respect to $\Gamma$.\ Fix an
element $x\in G$. \ Suppose we draw a unitary matrix $U$ from $\mathcal{Z}$,
apply $U$ to $\left\vert x\right\rangle $, and measure $U\left\vert
x\right\rangle $\ in the standard basis. \ Then the measurement outcome will
belong to $B$ with probability at most $\left\vert B\right\vert /\left\vert
G\right\vert $.
\end{corollary}

Given positive integers $m\geq n$, recall that $\Phi_{m,n}$ is the set of
lists of nonnegative integers $S=\left(  s_{1},\ldots,s_{m}\right)  $\ such
that $s_{1}+\cdots+s_{m}=n$. \ Also, recall from Theorem \ref{mainresult}%
\ that a basis state $S\in\Phi_{m,n}$\ is called \textit{collision-free} if
each $s_{i}$\ is either $0$ or $1$. \ Let $G_{m,n}$\ be the set of
collision-free $S$'s, and let $B_{m,n}=\Phi_{m,n}\setminus G_{m,n}$. \ Then we
have the following simple estimate.

\begin{proposition}
\label{goodbound}%
\[
\frac{\left\vert G_{m,n}\right\vert }{\left\vert \Phi_{m,n}\right\vert
}>1-\frac{n^{2}}{m}.
\]

\end{proposition}

\begin{proof}%
\begin{align*}
\frac{\left\vert G_{m,n}\right\vert }{\left\vert \Phi_{m,n}\right\vert }  &
=\frac{\binom{m}{n}}{\binom{m+n-1}{n}}\\
&  =\frac{m!\left(  m-1\right)  !}{\left(  m-n\right)  !\left(  m+n-1\right)
!}\\
&  =\left(  1-\frac{n-1}{m}\right)  \left(  1-\frac{n-1}{m+1}\right)
\cdot\cdots\cdot\left(  1-\frac{n-1}{m+n-1}\right) \\
&  >1-\frac{n^{2}}{m}.
\end{align*}

\end{proof}

Now let $U$ be an $m\times m$\ unitary matrix, and recall from Section
\ref{PHYSDEF} that $\varphi\left(  U\right)  $\ is the \textquotedblleft
lifting\textquotedblright\ of $U$ to the $n$-photon Hilbert space of dimension
$M=\binom{m+n-1}{n}$. \ Also, let $A=A\left(  U,n\right)  $\ be the $m\times
n$\ matrix corresponding to the first $n$ columns of $U$. \ Then recall that
$\mathcal{D}_{A}$\ is the probability distribution over $\Phi_{m,n}$ obtained
by drawing each basis state $S\in\Phi_{m,n}$ with probability equal to
$\left\vert \left\langle 1_{n}|\varphi\left(  U\right)  |S\right\rangle
\right\vert ^{2}$.

Using the previous results, we can upper-bound the probability that a
Haar-random unitary maps the basis state $\left\vert 1_{n}\right\rangle $\ to
a basis state containing two or more photons in the same mode.

\begin{theorem}
[Boson Birthday Bound]\label{birthdaythm}Recalling that $\mathcal{H}_{m,m}%
$\ is the Haar measure over $m\times m$\ unitary matrices,%
\[
\operatorname*{E}_{U\in\mathcal{H}_{m,m}}\left[  \Pr_{\mathcal{D}_{A\left(
U,n\right)  }}\left[  S\in B_{m,n}\right]  \right]  <\frac{2n^{2}}{m}.
\]

\end{theorem}

\begin{proof}
Given a permutation $\sigma\in S_{m}$ of single-photon states (or equivalently
of modes), let $\varphi\left(  \sigma\right)  $ be the permutation on the set
$\Phi_{m,n}$\ of $n$-photon states that is induced by $\sigma$, and let
$\Gamma:=\left\{  \varphi\left(  \sigma\right)  :\sigma\in S_{m}\right\}  $.
\ Then $\Gamma$\ is a subgroup of $S_{M}$\ of order $m!$ (where as before,
$M=\binom{m+n-1}{n}$). \ Furthermore, $\Gamma$\ is transitive with respect to
the set $G_{m,n}$, since we can map any collision-free basis state\ $S\in
G_{m,n}$\ to any other collision-free basis state $S^{\prime}\in G_{m,n}$\ via
a suitable permutation $\sigma\in S_{m}$\ of the underlying modes.

Now let $\mathcal{U}$\ be the probability distribution over $M\times
M$\ unitary matrices $V$ that is obtained by first drawing an $m\times
m$\ unitary matrix $U$\ from $\mathcal{H}_{m,m}$\ and then setting
$V:=\varphi\left(  U\right)  $. \ Then since $\mathcal{H}_{m,m}$\ is symmetric
with respect to permutations $\sigma\in S_{m}$, it follows that $\mathcal{U}%
$\ is symmetric with respect to permutations $\varphi\left(  \sigma\right)
\in S_{M}$.

We want to upper-bound $\operatorname*{E}_{U\in\mathcal{H}_{m,m}}\left[
\Pr_{\mathcal{D}_{A\left(  U,n\right)  }}\left[  S\in B_{m,n}\right]  \right]
$. \ This is simply the probability that, after choosing an $m\times m$
unitary $U$\ from $\mathcal{H}_{m,m}$, applying the $M\times M$\ unitary
$\varphi\left(  U\right)  $\ to the basis state $\left\vert 1_{n}\right\rangle
$, and then measuring in the Fock basis, we obtain an outcome in $G_{m,n}$.
\ So \
\[
\operatorname*{E}_{U\in\mathcal{H}_{m,m}}\left[  \Pr_{\mathcal{D}_{A\left(
U,n\right)  }}\left[  S\in B_{m,n}\right]  \right]  \leq\frac{\left\vert
B_{m,n}\right\vert }{\left\vert G_{m,n}\right\vert }<\frac{n^{2}/m}{1-n^{2}%
/m}.
\]
Here the first inequality follows from Corollary \ref{goodbadcor} together
with the fact that $1_{n}\in G_{m,n}$, while the second inequality follows
from\ Proposition \ref{goodbound}. \ Since the expectation is in any case at
most $1$, we therefore have an upper bound of%
\[
\min\left\{  \frac{n^{2}/m}{1-n^{2}/m},1\right\}  \leq\frac{2n^{2}}{m}.
\]

\end{proof}

\end{document}